\documentclass{article}
\pdfoutput=1 
\usepackage{fullpage}
\usepackage[utf8]{inputenc}
\usepackage[pdfencoding=auto,psdextra,hidelinks]{hyperref}
\usepackage{authblk}
\usepackage{amsmath}
\usepackage{amsthm}
\usepackage{amsfonts}
\usepackage{amssymb}
\usepackage[nopatch=eqnum]{microtype}
\usepackage{mathtools}
\usepackage[authoryear,square]{natbib}
\usepackage[svgnames,table,xcdraw]{xcolor}
\hypersetup{colorlinks={true},urlcolor={blue},linkcolor={DarkBlue},citecolor=[named]{DarkGreen}}
\usepackage{xspace}
\usepackage[capitalise,nameinlink,noabbrev]{cleveref}
\usepackage{crossreftools}
\usepackage{amsmath}
\usepackage{tikz}
\usepackage{subcaption}
\usepackage{doi}

\usepackage{algpseudocode}
\usepackage{algorithm}

\usepackage[sc]{mathpazo}

\newcommand{\TFNP}{\ensuremath{\mathrm{TFNP}}\xspace}
\newcommand{\FIXP}{\ensuremath{\mathrm{FIXP}}\xspace}

\newcommand{\poly}{\ensuremath{\mathrm{poly}}}
\newcommand{\size}{\ensuremath{\mathrm{size}}}

\DeclareMathOperator{\supp}{supp}
\DeclarePairedDelimiter\ceil{\lceil}{\rceil}
\DeclarePairedDelimiter\floor{\lfloor}{\rfloor}

\newcommand{\NOT}{\ensuremath{{\rm \sf NOT}}}
\newcommand{\AND}{\ensuremath{{\rm \sf AND}}}
\newcommand{\OR}{\ensuremath{{\rm \sf OR}}}

\newcommand{\allzeros}{\ensuremath{\mathbf{0}}}
\newcommand{\transpose}{\ensuremath{\mathsf{T}}}

\newcommand{\eps}{\ensuremath{\varepsilon}}
\newcommand{\NN}{\ensuremath{\mathbb{N}}}

\newcommand{\RR}{\ensuremath{\mathbb{R}}}
\newcommand{\RRnn}{\ensuremath{\mathbb{R}_{\geq0}}}

\newcommand{\zo}{\ensuremath{\{0,1\}}}
\DeclareMathOperator*{\argmax}{arg\,max}
\DeclareMathOperator*{\argmin}{arg\,min}

\DeclareMathOperator{\bitextract}{bitExt}
\DeclareMathOperator{\bitmult}{bitMult}
\DeclareMathOperator{\bitval}{bitVal}

\newcommand{\Heaviside}{\ensuremath{\operatorname{H}}}

\newcommand{\calP}{\ensuremath{\mathcal{P}}}
\newcommand{\calC}{\ensuremath{\mathcal{C}}}

\usepackage{enumitem}
\newlist{conditions}{enumerate}{2}
\setlist[conditions]{label={\arabic*.}}
\crefname{conditionsi}{Condition}{Conditions}

\theoremstyle{definition}
\newtheorem{definition}{Definition}[section]
\newtheorem{remark}{Remark}[section]
\newtheorem{example}{Example}[section]

\theoremstyle{plain}
\newtheorem{theorem}{Theorem}[section]
\newtheorem{proposition}{Proposition}[section]
\newtheorem{lemma}{Lemma}[section]
\newtheorem{corollary}{Corollary}[section]
\newtheorem{claim}{Claim}[section]

\newcommand{\fout}{{g^{\text{out}}_f}}
\newcommand{\fin}{{g^{\text{in}}_f}}
\newcommand{\linoptgate}{{linear-OPT-gate}\xspace}
\newcommand{\linoptgates}{{linear-OPT-gates}\xspace}
\newcommand{\circparam}{{gate input}\xspace}
\newcommand{\circparams}{{gate inputs}\xspace}
\newcommand{\Circparams}{{Gate Inputs}\xspace}
\newcommand{\fixpoptgate}{{OPT-gate for FIXP}\xspace}
\newcommand{\fixpoptgates}{{OPT-gates for FIXP}\xspace}
\newcommand{\linear}{PL\xspace}
\newcommand{\Linear}{PL\xspace}
\newcommand{\Sol}{\ensuremath{\operatorname{Sol}}}
\newcommand{\qgraph}{feasibility graph\xspace}
\newcommand{\pseudog}{\linear pseudo-circuit\xspace}
\newcommand{\pseudogs}{\linear pseudo-circuits\xspace}
\newcommand{\Pseudog}{\Linear pseudo-circuit\xspace}
\newcommand{\lbro}{PLBRO\xspace}

\newcommand{\seg}{tranche\xspace}
\newcommand{\segs}{tranches\xspace}

\newcommand{\us}{{\color{black!60!blue}\textbf{[Our Work]}}\xspace}

\title{PPAD-membership for Problems with Exact Rational Solutions: A General Approach via Convex Optimization}

\date{}

\author[1]{Aris Filos-Ratsikas}
\author[2]{Kristoffer Arnsfelt Hansen}
\author[2]{Kasper H{\o}gh}
\author[3]{Alexandros Hollender}
\affil[1]{University of Edinburgh, UK}
\affil[2]{Aarhus University, Denmark}
\affil[3]{University of Oxford, UK}

\begin{document}

\maketitle
\begin{abstract}
We introduce a general technique for proving membership of search problems with \emph{exact rational} solutions in PPAD, one of the most well-known classes containing total search problems with polynomial-time verifiable solutions. In particular, we construct a ``pseudogate'', coined the \emph{\linoptgate}, which can be used as a ``plug-and-play'' component in a piecewise-linear (\linear) arithmetic circuit, as an integral component of the ``Linear-FIXP'' equivalent definition of the class. The \linoptgate can solve several convex optimization programs, including quadratic programs, which often appear organically in the simplest existence proofs for these problems. This effectively transforms existence proofs to PPAD-membership proofs, and consequently establishes the existence of solutions described by rational numbers.

Using the \linoptgate, we are able to significantly simplify and generalize almost all known PPAD-membership proofs for finding exact solutions in the application domains of game theory, competitive markets, auto-bidding auctions, and fair division, as well as to obtain new PPAD-membership results for problems in these domains.
\end{abstract}

\section{Introduction}

Total search problems, i.e., search problems for which a solution is always guaranteed to exist, have been studied extensively over the better part of the last century, in the intersection of mathematics, economics and computer science. Famous examples of such problems are finding Nash equilibria in games \citep{Nash50}, competitive equilibria in markets \citep{arrow1954existence} and envy-free divisions of resources \citep{AMM:Stromquist1980}. While the classic works in mathematics and economics have been primarily concerned with establishing the existence as well as desirable properties of these solutions, the literature of computer science over the past 35 years has been instrumental in formulating and answering questions about the computational complexity of finding them.

More precisely, \citet{megiddo1991total} defined the class TFNP to include all total search problems for which a solution is verifiable in polynomial time. To capture the computational complexity of many problems including the aforementioned ones, several subclasses of TFNP were subsequently defined. Among those, one that has been extremely successful in this regard is the class PPAD of \citet{JCSS:Papadimitriou1994}, which was proven to characterize the complexity of computing Nash equilibria in games \citep{SICOMP:DaskalakisGP2009,chen2009settling}, as well as competitive equilibria for several types of markets \citep{ChenPY17-non-monotone-markets}, among many others. 

In reality, when making statements like the above, i.e., general statements of the form, ``finding a Nash equilibrium is in PPAD'', or similarly for a solution to some other total search problem, it is most often meant that what lies in the class is the problem of finding \emph{approximate} solutions. For strategic games for example, that would mean strategy profiles which are \emph{almost} Nash equilibria, up to some additive parameter $\varepsilon$. This is actually quite often necessary, as it has been shown that for many of these problems, there are cases where all of their solutions can only be described by irrational numbers, and hence we can not hope to compute them exactly on a computer.  

Still, there is a large number of important variants of these domains for which \emph{exact rational} solutions exist. For example, several strategic games always have equilibria in rational numbers, and so do certain markets for their competitive equilibria. There are also examples from fair division where rational partitions of the resources can be achieved. In all of those cases, PPAD-membership results for their approximate versions are unsatisfactory; we would like to place the \emph{exact} problems in PPAD instead.

Indeed, coming up with proofs of existence that also guarantee rationality of solutions has been a topic of interest in the area since the very early days, way before the introduction of the relevant computational complexity classes, e.g., see \citep{eaves1976finite,lemke1964equilibrium,lemke1965bimatrix,howson1972equilibria}. Driven by those classic results, a significant literature in computer science has attempted, and quite often has succeeded in placing the corresponding computational problems in PPAD, for several of the application domains mentioned above, including games \citep{EC:Sorensen12,hansen2018computational,SICOMP:KintaliPRST13,klimm2020complexity,meunier2013lemke}, markets \citep{vazirani2011market,SODA:GargV14,garg2017market,garg2018substitution}, as well as the more recent domain of auto-bidding auctions \citep{chen2021complexity}.

While these PPAD-membership proofs typically do follow one of a few common approaches, in essence they are rather domain-specific and require reconstructing a set of arguments again for each application at hand (see \Cref{sec:Other-Approaches} below for a detailed discussion). Instead, we would like to have \emph{one general technique for proving PPAD-membership of problems with exact solutions}, and ideally one that arises ``organically'' as the computational equivalent of the standard proofs of existence. To do this, a very promising avenue seems to be via a characterization of PPAD, coined \emph{Linear-FIXP}, due to \citet{etessami2010complexity}, which defines the class in terms of fixed points of problems represented by piecewise-linear arithmetic circuits. This is because a standard existence proof, e.g., via the Kakutani fixed point theorem \citep{kakutani1941generalization} or via Brouwer's fixed point theorem \citep{MA:Brouwer1911}, often obtains the solution as a fixed point of a set of local optimization problems, in which each agent or player is independently maximizing a piecewise utility/payoff function. If we could ``embed'' these optimization problems into a piecewise-linear circuit, that would essentially translate the existence proof into a PPAD-membership proof. This is crisply captured in the following quote from \citet{vazirani2011market}, in the context of proving PPAD-membership for competitive equilibria in certain markets: 

\begin{quote}
    \emph{``There are very few ways for showing membership in PPAD. A promising approach
for our case is to use the characterization of PPAD of Etessami and Yannakakis [2010]
as the class of exact fixed-point computation problems for piecewise-linear, polynomial
time computable Brouwer functions. $[\ldots]$ Unfortunately, we do not see how to do this $[\ldots]$ it is not clear how to transfer the piecewise-linearity of
the utility functions to the Brouwer function.'' \citep{vazirani2011market}.}
\end{quote}

\noindent Recently, \citet{SICOMP:Filos-RatsikasH2023} in fact developed a general technique along those lines: they designed an \emph{optimization gate}, which can be used as part of a circuit to substitute the aforementioned optimization problems and obtain membership results. Crucially however, their membership results are \emph{not} for the class PPAD, but rather for the class FIXP \citep{etessami2010complexity}, a superclass of Linear-FIXP in which the main computational device is a (general) arithmetic circuit, not a piecewise-linear one. These circuits are particularly powerful and can capture solutions with irrationalities. Using their ``\fixpoptgate'', \citet{SICOMP:Filos-RatsikasH2023} showed the FIXP-membership of several very general problems related to strategic games, markets and fair division. 

While FIXP is certainly a natural class, it has not enjoyed the same success as PPAD, even in the context of classifying problems with exact solutions. Besides, in the standard (Turing) model of computation, a FIXP-membership result can be interpreted as finding a point that is \emph{close} to a solution (e.g., in the max norm). This is often a stronger guarantee than an approximate solution as described earlier, but it it still very much only an approximation. Again, this is unsatisfactory for those problems with exact rational solutions that should be in PPAD.

Could we hope to use \citeauthor{SICOMP:Filos-RatsikasH2023}'s optimization gate to obtain PPAD-membership? This is actually practically impossible, for reasons which are deeply rooted in the definitions of the classes; we highlight those in \Cref{sec:linopt-vs-opt} below. In short, the power of general arithmetic circuits over piecewise-linear ones lies in their capability to multiply and divide input variables, and this is of vital importance in the design of the \fixpoptgate in \citep{SICOMP:Filos-RatsikasH2023}. What we really need is \emph{a new gate}, one which avoids such multiplications/divisions and hence can be used in a piecewise-linear arithmetic circuit. Designing such a gate poses significant technical challenges, which we highlight in \Cref{sec:linopt-vs-opt} and present in more detail in \Cref{sec:lin-opt-gate}. Additionally, clearly, the gate cannot capture the generality of applications that the \fixpoptgate does, as, as we said earlier, problems with irrational solutions cannot be in PPAD. It should however be general enough to capture any problem for which exact rational solutions are possible.  

This is the main technical contribution of our paper. We introduce the \linoptgate,\footnote{The term ``linear'' here refers to piecewise-linear functions, as in the class Linear-FIXP.} which can be used as a general purpose tool for proving PPAD-membership of problems with exact rational solutions. We demonstrate its strength and generality on a host of different applications in game theory, markets, auctions and fair division. Via its use, we are able to \emph{significantly simplify} or \emph{generalize} virtually all of the PPAD-membership proofs for problems with exact solutions in the literature, as well as to prove \emph{new} membership results for problems for which PPAD-membership was not known; we offer more details in the following subsection.

\subsection{A Powerful Tool for PPAD-membership: The \linoptgate}

We introduce the \linoptgate for proving membership of problems in PPAD. The \linoptgate can be used as a ``plug-and-play'' component in a \linear arithmetic circuit, i.e., similarly to any of the other gates $\{+,-,\max,\min, \times \zeta\}$ of the circuit (see \Cref{def:linear-arithmetic circuit}). The gate is guaranteed to work correctly at a fixed point of the function that the circuit encodes, which, for the purposes of proving PPAD-membership of a problem, is equivalent to a standard gate. 

The \linoptgate allows us to compute solutions to optimization programs of a certain form, like those shown in the left-hand side of \Cref{fig:linoptgate-intro}. In particular, these are optimization programs with a non-empty and bounded feasible domain given by a set of linear inequalities, and the subgradient of the convex objective function (in the variables $x$) is given by a \linear (piecewise-linear) arithmetic circuit. In particular, the \linoptgate can compute the solution to any linear program, but also to more general convex programs, e.g., those with quadratic objective functions. The inherent strength of the technique lies in the fact that these types of optimization programs arise naturally in several of the applications in game theory, competitive markets and fair division. Now, for the purpose of showing membership in PPAD, they may effectively be substituted by \linoptgates.

From the ability of the \linoptgate to solve optimization programs of the form $\mathcal{C}$ of \Cref{fig:linoptgate-intro}, we can also derive \emph{feasibility programs} with \emph{conditional constraints}, like the program $\mathcal{Q}$ on the right-hand side of \Cref{fig:linoptgate-intro}. These feasibility programs also often arise naturally in the context of existence proofs, and can be also thought of as being solved in a black-box manner by a gate, which is constructed using the \linoptgate. \medskip 

\noindent Our \linoptgate has a wealth of applications, which we discuss below.

\begin{figure}
\centering 
\fbox{
\centering
    \begin{minipage}{.45\textwidth}
\begin{center}\underline{Optimization Program $\mathcal{C}$}\end{center}
\begin{equation*}
\begin{split}
\min \quad &f(x;c) \\
\text{ s.t.} \quad & Ax \leq b\\
& x \in [-R,R]^n
\end{split}
\end{equation*}
    \end{minipage}%
    \hfill\vline\hfill
\begin{minipage}{.45\textwidth}
\vspace{-0.5cm}
\begin{center}\underline{Feasibility Program $\mathcal{Q}$}\end{center}
\begin{equation*}
\begin{split}
h_i(y) > 0 \implies a_i^\transpose x \leq b_i\\
x \in [-R,R]^n
\end{split}
\end{equation*}
\end{minipage}
}
\caption{The optimization programs and feasibility programs that can be solved by the \linoptgate.}
\label{fig:linoptgate-intro}
\end{figure}

\subsection{Applications of the \linoptgate}

We apply our \linoptgate to a plethora of different domains, and obtain PPAD-membership for finding solutions in several strategic games, competitive markets, auto-bidding auctions, as well as problems in fair division. We detail those applications in the corresponding sections below. Our results achieve the following three desired objectives simultaneously:
\begin{itemize}
    \item[-] Proofs of existence of solutions.
    \item[-] Proofs of rationality of solutions.
    \item[-] PPAD-membership of the corresponding problems.
\end{itemize}

\noindent For some of these domains, PPAD-membership results for the corresponding problems were known; still, the proofs to establish those were often rather involved. With the employment of our \linoptgate, they become \emph{conceptually and technically significantly simpler}. In essence, the \linoptgate allows us to turn a simple existence proof into a PPAD-membership result.
For some of our applications such simple existence proofs already existed, and are transformed to PPAD-membership proofs via the \linoptgate. For others, developing these simpler existence proofs is also part of our contribution; we provide more details in the sections below. The \linoptgate also allows us to \emph{straightforwardly} obtain generalizations of some of the known PPAD-membership results, to cases beyond what was known in the literature. Finally, we also obtain the PPAD-membership of some problems whose complexity had not been studied in the literature before.  

We summarize our results in \Cref{tab:results}, where we indicate which results were known in the literature before, which are generalizations, and which concern problems for which we did not know any results about their computational complexity.\medskip  

\noindent Before we proceed with the applications, we present the main techniques that have been used previously for proving PPAD-membership results, and highlight the main technical challenges of using those techniques as opposed to the ``plug-and-play'' nature of our \linoptgate.

\subsubsection{Main Previous Approaches}\label{sec:Other-Approaches}

\paragraph{Linear Complementarity Programs and Lemke's Algorithm.} The first main approach for establishing rationality of solutions and PPAD-membership is that of \emph{linear complementarity programs (LCPs)} \citep{cottle1968complementarity,cottle2009linear}. Given an $n \times n$ matrix $\mathbf{M}$ and a vector $\mathbf{q}$, an LCP seeks to find two vectors $\mathbf{y}$ and $\mathbf{v}$ satisfying:
\[
\mathbf{M}\cdot \mathbf{y} + \mathbf{v} = \mathbf{q}, \ \ \mathbf{y} \geq 0, \ \ \mathbf{v} \geq 0, \ \ \text{ and } \ \ \mathbf{y}^\transpose \cdot \mathbf{v} = 0
\]
The term ``complementarity'' stems from the fact that in a solution, we may have either $\mathbf{y}_i > 0$ or $\mathbf{v}_i > 0$, but not both. \citet{lemke1965bimatrix} designed an algorithm (based on the previously designed Lemke-Howson algorithm \citep{lemke1964equilibrium}) to solve LCPs via a series of \emph{complementary pivoting} steps, i.e., steps in which when a variable enters the basis, a complementary variable exits. Interestingly, the algorithm was designed in the context of computing Nash equilibria in bimatrix games, long before the associated computational complexity classes were defined. LCP-based formulations of equilibria and other fixed point problems have in fact been a subject of study in classic works (e.g, see \cite{eaves1976finite,howson1972equilibria}) as a means to obtain existence proofs that guarantee the rationality of solutions. PPAD membership can be obtained by pairing the algorithm with an appropriate local orientation of its complementarity paths \citep{todd1976orientation}. 

Quite importantly, \citeauthor{lemke1965bimatrix}'s algorithm terminates with either finding a solution to the LCP, or without finding a solution, in what is referred to as a \emph{secondary ray}. This feature of the algorithm is well-documented (e.g., see \citep{savani2006finding} for an excellent exposition) and is known as \emph{ray termination}. In terms of proving PPAD-membership, it seems almost inevitable that every PPAD-membership proof that uses this approach has to argue against ray termination. As \citet{SODA:GargV14} pointedly remark, in the context of a succession of papers on equilibrium computation in competitive markets:
\begin{quote}
\emph{``In the progression of these three works, the LCPs have become more involved and proving the lack of secondary rays has become increasingly harder.'' \citep{SODA:GargV14}}.
\end{quote}
This is not particular to markets either. For example, in \citeauthor{hansen2018computational}'s [\citeyear{hansen2018computational}] generalization of the results of \citet{EC:Sorensen12} from bimatrix to polymatrix games, those concerning $\varepsilon$-proper equilibria, a new LCP needs to be devised, together with a new argument against ray termination. Additionally, there are often significant challenges in even appropriately formulating the problems in question as LCPs. In some cases, the naive formulations may lead to inefficient representations, e.g., see \citep{EC:Sorensen12}. In other cases, all known formulations lead to \emph{nonstandard} LCPs, which cannot be handled by the ``vanilla'' version of Lemke's algorithm, and require variants of the algorithm to be devised, e.g., see \citep{garg2018substitution,meunier2013lemke}. Finally in some cases, it is not known if the derived LCPs can be solved via any variant of Lemke's algorithm, thus leading to the development of entirely new pivoting algorithms \citep{klimm2020complexity}. These characteristics of the LCP approach make it somewhat insufficient as a general purpose PPAD-membership technique. 

One advantage of LCP-based approaches is that they have been shown to perform well in practice, e.g., see \citep{garg2018substitution} and references therein. However, for the purpose of proving PPAD-membership, we do not see any general advantage of the LCP method over our \linoptgate. 

\paragraph{Approximation and Rounding.} The second general technique that has been used in several applications to prove the PPAD-membership of exact solutions is that of \emph{approximation and rounding}. This generally consists of the two following main steps:
\begin{itemize}
    \item[-] consider an \emph{approximation} or a \emph{relaxation} of the solution (e.g., $\varepsilon$-approximate equilibria) and prove that the approximate version is in PPAD, and
    \item[-] devise a rounding procedure to transform approximate solutions to exact solutions, while maintaining membership in the class. 
\end{itemize}
This rather indirect approach certainly suffers in terms of elegance. More importantly however, it is very much domain-specific. First, showing the PPAD-membership for the approximate version typically still requires a non-trivial proof, often even a rather involved one, e.g., via some reduction to one of the well-known problems in PPAD, like \textsc{End-of-Line} (see \Cref{def:end-of-line}) or the computational version of Sperner's lemma \citep{sperner1928neuer}. Also, the rounding procedure itself may be rather complicated, and of an ad hoc nature. For certain applications, there is a general linear programming-based technique developed by \citet{etessami2010complexity} to transform $\varepsilon$-approximate solutions to exact ones, for sufficiently small values of $\varepsilon$. Still, this does not apply to all problems, and it may need to be used in conjunction with other tailor-made rounding steps, e.g., see \citep{chen2021complexity,vazirani2011market}. 

\paragraph{The \linoptgate as a ``plug-and-play'' component.} As we will explain in the following, and as it will be evident via inspection of our proofs throughout the paper, the \linoptgate allows us to develop proofs which are very simple and streamlined, essentially mimicking the easiest proofs of existence. Clearly, most of the technical complications are ``hidden'' in the ``inner workings'' of the \linoptgate. This is the advantage of having a ``plug-and-play'' component readily available for the proofs: one does not need to even be concerned about how the \linoptgate works, but only to understand what kind of optimization programs it can solve. We consider this to be a significant advantage over the two aforementioned techniques, which require to devise application-specific arguments (be it arguments about ray termination or appropriate approximation and rounding). These arguments may be of a standard general nature, but they have to be devised anew for each application, as evidenced by all the different PPAD-membership results that employ these techniques. 

\renewcommand{\arraystretch}{1.2}

\begin{table}[]
\centering
\begin{tabular}{|ll|}
\hline
\rowcolor[HTML]{ECF4FF}
\multicolumn{2}{|c|}{Applications to Game Theory} \\[2ex] \hline
Games with \Linear Best Response Oracles (\lbro)    & \us                                                                            \\ \hline
PL Concave Games                      & \us                                                                            \\ \hline
Bilinear Games                                           & \citep{koller1996efficient}, \emph{implicitly}                                                                            \\ \hline
General Threshold Games                                 & \us                                                                            \\ \hline
Bimatrix Games                                           & \begin{tabular}[c]{@{}l@{}}\citep{JCSS:Papadimitriou1994}\\  \citep{cottle1968complementarity}, \emph{implicitly}  \end{tabular} \\
Polymatrix Games                                         &    \citep{howson1972equilibria}, \emph{implicitly}        
                                                       \\ 
\Linear Succinct Games                                    & \us                                                                            \\ \hline

\rowcolor[HTML]{FFECD5} 
\multicolumn{2}{|l|}{Multi-class Congestion games with piecewise-linear latency functions}                                                                             \\ \hline
Non-atomic Network Congestion Games           & \begin{tabular}[c]{@{}l@{}}\us\\ linear latencies {\citep{meunier2013lemke}}\end{tabular}    \\ \hline
Atomic Splittable  Network Congestion Games   & \begin{tabular}[c]{@{}l@{}}\us\\ linear latencies {\citep{klimm2020complexity}}\end{tabular} \\ \hline
Congestion Games with Malicious Players                  & \us                                                                             \\ \hline
\rowcolor[HTML]{FFECD5} 
\multicolumn{2}{|l|}{\cellcolor[HTML]{FFECD5}Other equilibrium notions}                                                                  \\ \hline
$\varepsilon$-proper Equilibria in Bimatrix Games        & \citep{EC:Sorensen12}                                                                      \\
$\varepsilon$-proper Equilibria in Polymatrix Games                  & \citep{hansen2018computational}                                                               \\
$\varepsilon$-proper Equilibria in \Linear Succinct Games             & \us                                                                             \\ \hline
Personalized Equilibria                                  & \citep{SICOMP:KintaliPRST13}                                                                \\ \hline
\rowcolor[HTML]{ECF4FF}
\multicolumn{2}{|c|}{Applications to Competitive Markets} \\[2ex] \hline

Exchange Markets with Linear Utilities        & \citep{eaves1976finite}, \emph{implicitly}                                                                   \\   
Arrow-Debreu Markets with SPLC Utilities             & \citep{garg2015complementary}                                                               \\
Arrow-Debreu Markets with SPLC Utilities/Productions            & \begin{tabular}[c]{@{}l@{}} \citep{vazirani2011market} \\  {\citep{SODA:GargV14}}\end{tabular}    \\
\begin{tabular}[c]{@{}l@{}} Arrow Debreu Markets with Leontief-free \\    Utilities/Productions \citep{garg2018substitution} \end{tabular}            &   \begin{tabular}[c]{@{}l@{}} Arrow-Debreu Markets with \emph{Succinct} SPLC \\Utilities/SPLC Productions     \us  \end{tabular}                                                                \\ \hline

\rowcolor[HTML]{ECF4FF}
\multicolumn{2}{|c|}{Applications to Auto-Bidding Auctions} \\[2ex] \hline
Pacing Equilibria in Second-Price Auctions with Budgets        & \citep{chen2021complexity}  \\ \hline

\rowcolor[HTML]{ECF4FF}
\multicolumn{2}{|c|}{Applications to Fair Division} \\[2ex] \hline
Envy-free Cake Cutting        & \citep{goldberg2020contiguous}, \emph{implicitly}  \\ \hline
Rental Harmony        & \us  \\ \hline

\end{tabular}
\caption{A summary of our PPAD-membership results - for other complementary results please see the respective sections/paragraphs in the introduction. Classes of domains that are within the same frame in the table (i.e., not separated by borders) are of increasing generality from top to bottom. Domains that appear in the same row of a frame are incomparable in terms of their generality. For the applications to game theory, all of the domains are special cases of \linear concave games which in turn are a special case of \lbro games. For those applications, the PPAD-membership extends to \emph{generalized equilibria}. For all of the results in the table, regardless of whether we obtain entirely new results, generalizations, or simply results which were known in the literature, we obtain \emph{significant simplifications} in the proofs.}
\label{tab:results}
\end{table}

\subsubsection{Implicit Functions and Correspondences} As a final remark before we present our applications, we point out that, via machinery that we develop in \Cref{sec:implicit}, our \linoptgate can be used to show the PPAD-membership of problems for which the inputs (e.g., utilities or latency functions) are given \emph{implicitly} in the input. In particular, we show how we can construct \linear arithmetic circuits computing these functions, when those are inputted \emph{succinctly} via Boolean circuits. In terms of the applications, this allows us to effectively consider functions of exponential size (in the size of the circuits), e.g., piecewise-linear utility functions with exponentially-many pieces. We provide details on how this capability of the \linoptgate is used in each application in the corresponding sections below. We present applications for which the aforementioned techniques of \Cref{sec:Other-Approaches} are \emph{inherently insufficient} for obtaining PPAD-membership results for those implicit functions, when these results are in fact enabled by the use of the \linoptgate.  

\subsubsection{PPAD-membership for Strategic Games}

We start our discussion from the applications of the \linoptgate to the problem of computing (exact) equilibria in strategic games. To provide some initial intuition, before the technical sections of the paper, we provide an informal example of the use of the \linoptgate to compute mixed Nash equilibria in bimatrix games; this is exposed in more detail in \Cref{sec:bimatrix-games}. 

\subsubsection*{An Example: Bimatrix Games} A bimatrix game is a game played between two players, in which the payoffs are given by two matrices $\mathbf{A}_1$ and $\mathbf{A}_2$, one for each player, denoting the payoff of the players when they each choose certain actions. Each player chooses a \emph{mixed strategy}, i.e., a probability distribution over actions in the game, aiming to maximize their expected payoff, against the choice of the opponent. A mixed Nash equilibrium is a pair of mixed strategies for which every player is \emph{best responding}, i.e., she is maximizing her payoff, given the strategy of the other player. The existence of mixed Nash equilibria for bimatrix games follows from \citeauthor{Nash50}'s general existence theorem \citep{Nash50}. The proof of the theorem that employs the Kakutani fixed point theorem \citep{kakutani1941generalization} constructs a fixed point of a function $F$ from the domain of mixed strategies to itself, for which each coordinate $F_i$ is a best response for player $i$ in the game. These best responses can be captured by optimization programs of the form $\mathcal{C}$ in \Cref{fig:linoptgate-intro} and in particular for the case of bimatrix games, these are linear programs in which the subgradients of the objective functions are linear functions. The existence proof then immediately yields a PPAD-membership proof if one substitutes those programs with \linoptgates that compute them. \medskip

\noindent We remark that for bimatrix games, the original PPAD-membership proof of \citet{JCSS:Papadimitriou1994} adopts the ``LCP approach'' that we mentioned earlier, i.e., it appeals to an alternative proof of Nash equilibrium existence due to \citet{cottle1968complementarity} (see also \citep{lemke1964equilibrium}) that formulates the problem as an LCP. This is a good example of what we mentioned earlier; the \linoptgate allows us to organically retrieve PPAD-membership from the standard, textbook existence proof of \citet{Nash50}. 

\subsubsection*{Best Response Oracles, \Linear Concave Games and Generalized Equilibria}

\paragraph{\Linear Best Response Oracles.} The approach that we highlighted above is not restricted to bimatrix games, but it actually captures a large class of strategic games. In \Cref{sec:LBRO-games} we provide a technical definition for a very general class of games, in which the best response of each agent is given by an oracle that can be computed by a \linear arithmetic circuit. We refer to these games as \emph{games with \linear best response oracles (\lbro games)}. An equilibrium of any \lbro game can straightforwardly be formulated as a fixed point of a function like the function $F$ above, where each coordinate $F_i$ computes the best response of player $i$ via the oracle. By using \linoptgates as oracles, we immediately obtain PPAD-membership results for a wealth of different games.

\paragraph{\Linear Concave Games.} The class of \emph{concave games} is a very large class of games, studied notably by \citet{rosen1965concave} and \citet{debreu1952social}. These are games with continuous strategy spaces, for which the existence of an equilibrium is guaranteed under certain continuity and concavity assumptions on the utility functions. This was proven by \citet{rosen1965concave} but also earlier independently by \citet{debreu1952social}, \cite{fan1952fixed}, and \citet{glicksberg1952further}, and for that reason the existence result is often referred to as the \citeauthor{debreu1952social}-\citeauthor{fan1952fixed}-\citeauthor{glicksberg1952further} theorem for continuous games. 

In \Cref{sec:generalized-concave-games} we prove that as long as the supergradient of the (concave) utility function can be computed by a \linear arithmetic circuit, concave games are \lbro games, and hence finding an equilibrium is in PPAD. We refer to those games as \linear concave games, and emphasize again that the utility function does not have to be piecewise linear, but only its (super)gradient; in particular, it could for example be a quadratic function. Bimatrix games are \linear concave games, and so are \emph{polymatrix games} \citep{janovskaja1968equilibrium, howson1972equilibria}, \emph{bilinear games} \citep{garg2011bilinear}, as well as generalizations of \emph{(digraph) threshold games} \citep{papadimitriou2021public}, and thus we obtain membership of finding equilibria in all of these games in PPAD. The latter two games have continuous strategy spaces, and thus the equilibria that we compute are pure, whereas for polymatrix games (and as a result, for bimatrix games) we compute equilibria in mixed strategies.

\paragraph{\Linear succinct games.} In fact, we define a large class of games, which generalize polymatrix games, one which we coin \emph{\linear succinct games}. In these games, the expected utility of a player, given a pure strategy $j$ and a mixed strategy $\mathbf{x}_{-i}$ of the other players, can be computed by a \linear arithmetic circuit. These are \linear concave games, and the PPAD-membership of finding their mixed Nash equilibria is a corollary of the results mentioned above. 

 We draw parallels between \linear succinct games and those defined in \citet{daskalakis2006game} and \citet{papadimitriou2008computing}. Those works define classes of succinct games for which there is an oracle for computing the expected utility of the player. In \citep{papadimitriou2008computing}, this oracle is referred to as the \emph{polynomial expectation property} and is used to show that correlated equilibria \citep{aumann1974subjectivity} of games with this property can be computed in polynomial time. In \citep{daskalakis2006game}, it is shown that if the oracle is given by a \emph{bounded division free straight-line program of polynomial length}, then these games are in PPAD. Crucially, this latter result concerns \emph{approximate equilibria}. One could view our result as a complement to those two results, one which concerns \emph{exact} equilibria in \emph{rational} numbers. \medskip

\noindent Our PPAD-membership result for \linear concave games captures the limits of the class of concave games for which rational equilibria exist, and thus membership in PPAD is possible. The only other known complexity results for the general class of concave games are a FIXP-completeness result due to \citet{SICOMP:Filos-RatsikasH2023}, and a very recent PPAD-membership result for \emph{approximate} equilibria due to \cite{PapadimitriouVZ23-kakutani}.  

\paragraph{Generalized Equilibrium.} \citet{debreu1952social} did not only consider concave games, but in fact a more general equilibrium notion, one in which the strategy space of each player is dependent on the set of strategies chosen by the other players. This was coined a ``social equilibrium'' by \citet{debreu1952social} (see also \citet{dasgupta2015debreu}) but over the years has been better known by the term \emph{generalized equilibrium}. For our purposes, the dependence on other strategies can be embedded in the constraints of the optimization programs that we use as oracles in \lbro games, in a way that can be handled by the \linoptgate. As a corollary, we obtain all of the aforementioned PPAD-membership results for generalized equilibria (rather than standard equilibria) as well, see \Cref{sec:generalized-concave-games}. To the best of our knowledge, these are the first PPAD-membership results for generalized equilibria in the literature. 

\subsubsection*{Personalized Equilibria} The notion of \emph{personalized equilibrium} was introduced by \citet{SICOMP:KintaliPRST13} in the context of games played on hypergraphs, with an equivalent strategic form. Intuitively speaking, these equilibria allow players to ``match'' their strategies with those of their opponents, without obeying a product distribution. \citet{SICOMP:KintaliPRST13} showed the PPAD-membership (and as a result, rationality of equilibria) of personalized equilibria via the ``relaxation and rounding approach'' (see \Cref{sec:Other-Approaches}). In particular, they first define an approximate version of the problem (the $\varepsilon$-personalized equilibrium), and reduce that problem to \textsc{End-Of-Line} (see \Cref{def:end-of-line} in \Cref{sec:preliminaries}), via a relatively involved construction. To obtain PPAD-membership for the exact problem (i.e., when $\varepsilon=0$) \citet{SICOMP:KintaliPRST13} construct an elaborate argument that appeals to linear programming compactness, by first showing that for sufficiently small $\varepsilon$, $\varepsilon$-personalized equilibria ``almost satisfy'' the constraints of the linear programs, and then carefully rounding the solution to obtain an exact equilibrium. 

The use of the \linoptgate allows us to obtain the PPAD-membership of the problem via an extremely simple argument. Essentially, each player computes their best response via a linear program which is computed by the \linoptgate, which reduces the problem to finding an equilibrium of an \lbro game, see \Cref{sec:personalized}.

\subsubsection*{$\varepsilon$-proper Equilibria} We also consider an alternative equilibrium notion, that of \emph{$\varepsilon$-proper equilibria}. This notion was introduced by \citet{IJGT:Myerson78} to refine the notion of $\varepsilon$-perfect equilibrium of \citet{selten1975reexamination}, and captures situations in which the players can make small mistakes (``trembles'') in the choice of their mixed strategies. The PPAD-membership of computing $\varepsilon$-proper equilibria was known for bimatrix games due to \citet{EC:Sorensen12} and for polymatrix games due to \citet{hansen2018computational}. Both of these works adopt the LCP approach, which means that they need to go through the hassle of establishing the properties of Lemke's algorithm, as discussed in \Cref{sec:Other-Approaches} above. Additionally, formulating the problem as an LCP in this case is far from trivial, and requires an extended formulation of the generalized permutahedron due to \citet{goemans2015smallest}, to make sure that the LCP has polynomially-many constraints. 

The use of our \linoptgate distinctly avoids all this labor. We formulate the problem of computing a best response for each player (where the best response is defined with respect to the $\varepsilon$-proper equilibrium notion) as a feasibility program of the form $\mathcal{Q}$ in \Cref{fig:linoptgate-intro}, which can be solved by the \linoptgate. This essentially renders the game a \lbro game, and the PPAD-membership follows simply as a corollary of our main theorem for \lbro games.

\subsubsection*{Network Congestion Games}

Our last application in the area of game theory is to multi-class congestion games. In particular, we will consider two models, \emph{non-atomic congestion games} and \emph{atomic splittable congestion games}. In the former case, there is a continuum of players who collectively form a class controlling a certain load allocation to different resources. In the latter case, each class is represented by a single (atomic) player, who controls the load and distributes it to the resources. For both of those settings, we will also consider the subclass of \emph{network congestion games}, where the strategies can be represented more succinctly using flows over a directed network.

The existence of equilibria in those games was established in classic works, e.g., see \citep{schmeidler1973equilibrium} or \citep{milchtaich2000generic}, originally via the employment of the \citeauthor{debreu1952social}-\citeauthor{fan1952fixed}-\citeauthor{glicksberg1952further} theorem [\citeyear{debreu1952social}] for continuous games, assuming that the latencies on the resources are concave functions. Relevant to us are the works on their computational complexity, namely \citep{meunier2013lemke} (for non-atomic network congestion games) and \citep{klimm2020complexity} (for atomic splittable network congestion games). Both papers showed the PPAD-membership of finding pure equilibria in their respective settings, when the latency functions are \emph{linear}. We remark that these games are different from atomic (non-splittable) congestion games, for which finding pure Nash equilibria is known to be in the class PLS defined by \citet{johnson1988easy}.

\citeauthor{meunier2013lemke} obtain their PPAD-membership result via the ``LCP approach'' mentioned in \Cref{sec:Other-Approaches}. Interestingly, their LCP formulation turns out to not be amenable to the use of Lemke's algorithm, so they have to devise a ``Lemke-like'' complementary pivoting algorithm, tailored to their problem. As in the case of Lemke's algorithm, they argue explicitly against ray termination. \citeauthor{klimm2020complexity} note that in their case, the problem of finding an equilibrium can be formulated as an LCP, however, it is not known or clear whether this LCP can be solved using any known variant of Lemke's algorithm. For that, they devise a rather involved proof, based on a new homotopy method, essentially a new pivoting algorithm. Their algorithm solves the problem of finding a Nash equilibrium as a system of linear equations involving notions such as \emph{excess flows}, \emph{vertex potentials} and \emph{block Laplacians}. At a very high level, the authors use the excess and potentials to define an undirected version of the \textsc{End-of-Line} graph (see \Cref{def:end-of-line} in \Cref{sec:preliminaries}), and the determinant of the block Laplacians to define a unique orientiation of the edges, effectively reducing the problem to \textsc{End-of-Line}. 

The \linoptgate allows us to avoid all of the technical complications of the proofs of \citet{meunier2013lemke} and \cite{klimm2020complexity} (which are rather involved, especially the latter), and essentially obtain the PPAD-membership for both of these problems as simple corollaries of our main results for \lbro games or concave games. In fact, we obtain \emph{generalizations} of those PPAD-membership results to games with more general latency functions, notably piecewise-linear latency functions (implicitly or explicitly represented). In exactly the same fashion, we can use the \linoptgate to obtain the PPAD-membership of \emph{congestion games with malicious players}, a setting studied by \citet{babaioff2009congestion}, for which computational complexity results had not been previously proven. \medskip  

\noindent All of our results on congestion games are presented in \Cref{sec:congestion-games}.

\subsubsection{PPAD-membership for Competitive Markets}

We now move on to the application of our technique to the domain of competitive markets. The standard market model in the literature is that of the Arrow-Debreu market \citep{arrow1954existence}, where a set of consumers compete for goods endowed by them and other consumers and goods produced by a set of firms. A \emph{competitive equilibrium} of the market is a set of allocations of goods to the consumers, a set of production quantities and a set of prices, such that at those prices, (a) all consumers maximize their individual utilities, (b) all firms produce optimal amounts, and (c) the market clears, i.e., supply equals demand. The existence of an equilibrium for the general market model was established by \citet{arrow1954existence} via the employment of \citeauthor{debreu1952social}'s social equilibrium theorem [\citeyear{debreu1952social}], under some standard assumptions on the utilities of the consumers and the production sets of the firms. 

\paragraph{Previous results and proofs.} It has been well-known since the early works in the area \citep{eaves1976finite} that in general Arrow-Debreu markets, competitive equilibria may be irrational. A significant literature, starting with the work of \citet{eaves1976finite} aimed at identifying special cases of the Arrow-Debreu market for which \emph{exact rational} solutions are always possible. When computer science took over in this quest, the related question of establishing the PPAD-membership of finding those exact solutions was also brought forward. Most of the PPAD-membership proofs that were developed through the years followed the ``LCP approach'', see \Cref{sec:Other-Approaches}. We present them here in succession:
\begin{itemize}[leftmargin=*]
    \item[-] \citet{eaves1976finite} considered the simplest case of exchange markets (no production) with linear utilities for the consumers and devised an LCP that can be solved by Lemke's algorithm. To establish the latter fact, he argued against ray termination, a characterstic of this approach that we emphasized in \Cref{sec:Other-Approaches}. A PPAD-membership proof is implicit in his result.\footnote{Note that for exchange markets with linear utilities and no production the problem is in fact known to be polynomial-time solvable \citep{Jain2007polynomial}.}
    \item[-] \citet{garg2015complementary} considered exchange markets with \emph{separable piecewise-linear concave (SPLC)} utilities, a generalization of linear utilities in which every agent has a piecewise linear concave utility for the amount of a good $j$ that she receives, and her total utility for her bundle is additive over goods. The authors proved the PPAD-membership of finding competitive equilibria in those markets via devising an LCP that was ``quite complex'' \citep{garg2015complementary}, and naturally had to argue against ray termination, to establish that Lemke's algorithm will terminate on this LCP with a valid solution.
    \item[-] \citet{SODA:GargV14} considered Arrow-Debreu markets with SPLC utilities as well as SPLC production functions. This is in fact the work from which the quote of \Cref{sec:Other-Approaches} is taken. The quote highlights the increasing challenge of developing these LCPs and establishing their successful termination. Indeed, for this LCP, \citet{SODA:GargV14} devise a set of linear programs, and then use the complementary slackness and their feasibility conditions to develop the LCP needed for production. The non-homogeneity of the resulting LCP for the equilibrium problem is dealt with in a manner which is different from previous works \citep{eaves1976finite,garg2015complementary} and, naturally, since the developed LCP is different, \citeauthor{SODA:GargV14} again need to argue against ray termination. 
    \item[-] The most general class of utility/production functions for which a PPAD-membership of exact competitive equilibria was proven is that of \emph{Leontief-free} functions \citep{garg2018substitution}, which generalize SPLC functions. For this, the authors devise yet another LCP formulation, which turns out to be even more complex than those of previous works.  This is because it has to differentiate between ``normal'' and ``abnormal'' variables, the latter preventing the employment of Lemke's algorithm. To circumvent this, they exploit some additional structure of their \emph{nonstandard} LCP, and then they also \emph{modify} Lemke's algorithm, to account for the possibility of abnormal variables becoming zero. Finally, as they devise a new LCP, they also have to argue once again against ray termination. 
\end{itemize}
\noindent Besides those works, the first work in computer science to prove PPAD-membership for markets with SPLC utilities/productions was \citep{vazirani2011market}. The approach in that paper is not the ``LCP approach'' but the ``approximation and rounding approach'' (again, see \Cref{sec:Other-Approaches}). An issue with this method is that very small changes in the prices may result in drastic changes in the optimal bundles of the consumers, which makes the proof quite challenging. To deal with this, \citet{vazirani2011market} devise a set of technical lemmas that allow them to ``force'' certain allocations over others.

\paragraph{Our results.} Our results in this section are twofold. 
\begin{itemize}[leftmargin=*]
\item[-] \textbf{Simplified proofs.} First, we employ the \linoptgate to recover all of the aforementioned PPAD-membership results via proofs which are conceptually and technically quite simpler. In particular, we formulate the optimal consumption and the optimal production as linear programs similar to program $\mathcal{C}$ of \Cref{fig:linoptgate-intro}, which can be effectively substituted by \linoptgates in a \linear arithmetic circuit. We also apply a standard variable change which was first used by \citet{eaves1976finite}, and which we refer to as \emph{Gale's substitution}, see \Cref{rem:gale}. For the prices, we develop a feasibility program, similar to program $\mathcal{Q}$ of \Cref{fig:linoptgate-intro}. In a fixed point of the circuit, the optimality of consumption and production follows almost immediately by design. The main technical challenge of the proofs lies in arguing the market clearing of the outputted prices, which however still requires a relatively short proof. 

To introduce the reader gently to our proof technique, we first apply it to the simple setting of exchange markets with linear utilities in \Cref{sec:exchange-markets-linear}, then to the setting of Arrow-Debreu markets with linear utilities and productions in \Cref{sec:prod-markets-linear}, and finally to the general case of Arrow-Debreu markets with Leontief-free utilities and productions in \Cref{sec:general-markets}. 

\item[-] \textbf{PPAD-membership for \emph{Succinct} SPLC (SSPLC) utilities.} In \Cref{sec:SSPLC-markets} we introduce a new class of utility functions, which we coin \emph{succinct separable piecewise-linear (SSPLC) utilities}. These are SPLC utilities in which the different segments of the utility function need not be given explicitly in the input (as in the case of (explicit) SPLC utilities), but can be accessed implicitly  via a boolean circuit. Effectively, this allows us to \emph{succinctly} represent SPLC functions with exponentially many pieces, where the input size is the size of the given circuits. We remark that the ``LCP-approach'' developed in the aforementioned papers is inherently limited in providing PPAD-membership results for this class. Indeed, one could formulate the problem as a large LCP in exponentially-many variables, and that would establish the existence of rational solutions. However, this formulation would no longer be a polynomial time reduction (since now we do not have explicit input parameters $u_{jk}^i$ for the utility of each piece) and hence it would not imply the PPAD-membership of the problem. In contrast, using our machinery from \Cref{sec:implicit} we can make sure that our \linoptgate can be used to obtain PPAD-membership for markets with SSPLC utilities as well. In our result we also add (explicit) SPLC production, which our technique clearly can handle. We provide a discussion on the challenges of extending our results to also capture SSPLC production functions at the end of \Cref{sec:SSPLC-markets}. Note that the SSPLC functions and the Leontief-free functions are of incomparable generality (and hence they appear on the same line of \Cref{tab:results}). Whether we can prove PPAD-membership for a class of ``succinct Leontief-free functions'', which would generalize both settings, is an interesting technical question.
\end{itemize}

\subsubsection{PPAD-membership for Auto-bidding Auctions}

Our next application is on the domain of auto-bidding auctions, which has received a lot of attention recently, due to its applicability in real-world scenarios \citep{balseiro2021budget, balseiro2021landscape,balseiro2021robust,balseiro2019learning,conitzer2022multiplicative,conitzer2022pacing,li2022auto,chen2021complexity,borgs2007dynamics}. In particular, in \Cref{sec:pacing} we consider the settings studied by \citet{conitzer2022multiplicative,conitzer2022pacing}, \citet{chen2009settling} and \citet{li2022auto}, in which buyers participate in several parallel single-item auctions, via scaling their valuations by a chosen parameter $\alpha$, called the \emph{pacing multiplier}. The buyers do that while facing constraints on their feasible expenditure, typically provided by budgets or return-on-investment (ROI) thresholds. The objective is to find a \emph{pacing equilibrium}, i.e., pacing multipliers and allocations for the buyers that are consistent with the format of the auction run (e.g, first-price or second-price) and satisfy the expenditure constraints of all the buyers simultaneously. Pacing equilibria have a similar flavor to the competitive equilibria discussed earlier, but are sufficiently different, and thus require separate handling.

\paragraph{Our proof vs the previous approach.} We prove that computing pacing equilibria in parallel second-price auctions with budgets is in PPAD. The problem was already known to be in PPAD (in fact, PPAD-complete) by the recent results of \citet{chen2021complexity}, building on the original existence result of \citet{conitzer2022multiplicative}. \citeauthor{chen2021complexity}'s proof rather heavily applies the ``approximation and rounding'' paradigm highlighted in \Cref{sec:Other-Approaches}. In particular, \citeauthor{chen2021complexity} define a $(\delta,\gamma)$-approximate variant of the pacing equilibrium, where $\delta, \gamma >0$ are two approximation parameters. Intuitively, this equilibrium corresponds to an ``almost equilibrium'' (i.e., the expenditure constraints are ``almost'' satisfied) of an ``almost second-price auction'' (i.e., an auction in which the set of winners is those with ``almost'' the highest bid). The authors prove that finding these approximate equilibria is in PPAD, via a reduction to a computational version of Sperner's lemma \citep{sperner1928neuer}, and then devise an intrictate rounding procedure to convert $(\delta,\gamma)$-equilibria into $\gamma$-equilibria. The final step in their proof applies the aforementioned technique of \citet{etessami2010complexity} (see \Cref{sec:Other-Approaches}) to further round these equilibria to pacing equilibria (i.e., where $\gamma = 0$).

Our proof employs the \linoptgate and is conceptually and technically much simpler, without needing to use approximations. Instead, we again apply the standard variable change in Gale's substitution (see \Cref{rem:gale}) which we also used for the case of competitive markets, to work with the expenditures rather than the allocations directly. From there, we can formulate the task of finding the optimal expenditures as a set of linear programs (one for each buyer), and the pacing multipliers will be obtained as a fixed point solution of a single simple equation. These linear programs can be solved by \linoptgates which essentially establishes the PPAD-membership of the problem. The proof is detailed in \Cref{sec:pacing-sp-budgets}.

\paragraph{ROI-constrained buyers.} We observe that the existence proof underlying our PPAD-membership proof in this section can in fact almost straightforwardly be modified to yield the existence of pacing equilibria for a different setting in auto-bidding auctions, that of second-price auctions with average return-on-investment (ROI) constraints, studied by \citet{li2022auto}. \citeauthor{li2022auto} established the existence of pacing equilibria via a rather indirect proof, which first reduces the problem to a somewhat convoluted concave game and applies the \citeauthor{debreu1952social}-\citeauthor{fan1952fixed}-\citeauthor{glicksberg1952further} theorem \citep{debreu1952social} to obtain Nash equilibrium existence, and then recovers a pacing equilibrium as a limit point of such a Nash equilibrium. This proof in fact closely follows the original proof of \citet{conitzer2022multiplicative} for the budgeted setting, and clearly does not have any implications on the computational complexity of the problem. 

Our proof, besides its advantages in terms of simplicity, also allows us for the first time to obtain computational membership results for pacing equilibria in the ROI-constrained buyer case. It turns out that for this setting, all pacing equilibria may be irrational (see \Cref{ex:irrational-roi} in \Cref{sec:RPE-irrational}), and hence membership in PPAD is not possible. Instead, we employ the OPT-gate for FIXP developed by \citet{SICOMP:Filos-RatsikasH2023} to easily transform our existence proof into a FIXP-membership proof.

\subsubsection{PPAD-membership for Fair Division}

The last applications of our \linoptgate are related to the task of fairly partitioning a resource among a set of agents with different preferences over its parts. In particular, we show the PPAD-membership of computing \emph{exact} envy-free solutions in two fundamental problems, namely \emph{envy-free cake cutting} \citep{gamow1958puzzle} and \emph{rental harmony} \citep{AMM:Su1999}, when the preferences of the agents ensure the existence of rational partitions.

\paragraph{Envy-free cake cutting.} The envy-free division of a continuous resource (metaphorically, a ``cake'') is one of the most fundamental and well-studied mathematical problems of the last century. The origins of the theory of the problem can be traced back to the pioneering work of \cite{steinhaus1949division}, with different variants being studied over the years in a large body of literature in mathematics, economics, and computer science; see \citep{brams1996fair,robertson1998cake,procaccia2013cake} for some excellent textbooks on the topic. The existence of an envy-free division was established in \citeyear{AMM:Stromquist1980} independently by \citet{AMM:Stromquist1980}, by \citet{woodall1980dividing}, and by Simmons (credited in \citep{AMM:Su1999}), even when the division is required to be \emph{contiguous}, i.e., when each agent receives a single, connected piece of the resource. These proofs proceed by first establishing the existence of divisions that are \emph{approximately} envy-free and then obtaining exact solutions as limit points of these approximations.

It is known that in general, envy-free divisions might be irrational (e.g., see \citep{bei2012optimal}, or \Cref{ex:cake-irrational} for a simpler example), and hence the problem of computing them cannot be in PPAD. \citet{SICOMP:Filos-RatsikasH2023} showed that envy-free cake cutting is in the class FIXP, which, recall, is appropriate for capturing the complexity of such problems.  Still, there are interesting cases for which rational divisions always exist. This is the case for example when the agents' preferences are captured by piecewise constant density functions \citep{goldberg2020contiguous}, a class of functions which is general enough to capture many problems of interest. A FIXP-membership result for these variants is unsatisfactory, and we would like to obtain a PPAD-membership result instead.

Without the convenience of using our \linoptgate, one can establish such a membership result via the ``approximation and rounding'' technique, see \Cref{sec:Other-Approaches}. \citet{OR:DengQS2012} showed that \emph{approximately} envy-free cake cutting is in PPAD, by transforming Simmons' proof into a computational reduction. \citet{goldberg2020contiguous} showed how to ``round'' the approximate solution to obtain an exact envy-free division for preferences captured by piecewise-constant densities, as long as $\varepsilon$ is sufficiently small. 

Luckily, the \linoptgate allows us to avoid having to do that, and instead directly obtain a PPAD-membership result without any need for approximations. In particular, we revisit the FIXP-membership proof of \citet{SICOMP:Filos-RatsikasH2023}; similarly to our approach in this paper, they essentially first construct an existence proof for the problem, one which involves a pair of optimization programs, and then substitute those programs with their \fixpoptgates. One might wonder if, by simply following the steps of the proof and substituting those programs with \linoptgates instead, we can recover the PPAD-membership of the problem, for those classes of preferences for which it is possible. This is almost true, apart from the fact that there is a step in their proof that cannot be done in a \linear arithmetic circuit. 

Still, we manage to substitute that part by a third optimization program, which is in fact a rather simple linear program, and can effectively be substituted by a \linoptgate. This allows us to obtain the PPAD-membership of the problem for the general class of valuation functions (i.e., functions expressing the preferences via numerical values) that can be computed by a \linear arithmetic circuit, see \Cref{sec:ef-cake-cutting}, capturing the aforementioned case of valuations with piecewise-constant densities. 

\paragraph{Rental harmony.} The rental harmony problem, notably studied by \citet{AMM:Su1999}, is concerned with the partition of rent among a set of tenants which have different preferences over combinations of rooms and rent partitions. In the generality studied by \citet{AMM:Su1999}, this problem is in fact equivalent to that of finding an envy-free division of a \emph{chore} among a set of agents. \citeauthor{AMM:Su1999}'s existence proof is inspired by Simmons' proof for envy-free cake cutting, but employs a ``dual Sperner labelling'' \citep{sperner1928neuer}. Similarly to the proofs for cake-cutting, the proof also appeals to limits of approximate solutions. In contrast to cake-cutting however, computational complexity results about this general version of the problem were not known, not even for approximate partitions. 

In \Cref{sec:rental-harmony}, we prove that the problem of finding a solution to rental harmony is in PPAD, as long as the valuations of the tenants for the rent partition are given by \linear arithmetic circuits. Interestingly, this is established via very much the same approach as the proof for envy-free cake cutting, thus providing for the first time a unified proof of existence for those two problems. If one goes beyond the aforementioned valuation functions, all rental harmony solutions may be irrational, as we show in \Cref{ex:irrational-rental-harmony}. For those cases, we explain how the existence proof can be coupled with the \fixpoptgate of \citet{SICOMP:Filos-RatsikasH2023} to establish the FIXP-membership of the problem. 

\paragraph{Computing Envy-free and Pareto-optimal allocations.} We remark that very recently \citet{caragiannis2023complexity} used our \linoptgates to establish that computing probabilistic envy-free and Pareto-optimal allocations of multiple divisible goods is in PPAD.

\subsection{The \linoptgate vs the \fixpoptgate}\label{sec:linopt-vs-opt}

As we mentioned in the introduction, \citet{SICOMP:Filos-RatsikasH2023} were the first to develop an OPT-gate for the computational class FIXP \citep{etessami2010complexity}. FIXP is the class that captures the complexity of computing a fixed point of an arithmetic circuit, i.e., a circuit over the basis $\{+,-,\max,\min,\div,*\}$ with rational constants, see \Cref{def:arithmetic circuit}. FIXP is a larger class than Linear-FIXP, due to the fact that we can multiply and divide inside the circuit. 

The tools that our \linoptgate provides are conceptually very similar to those of the \fixpoptgate of \citet{SICOMP:Filos-RatsikasH2023}, in that they can substitute convex optimization programs within existence proofs, when constructing a circuit whose fixed points are the solutions that we are looking for. However, the design of the gate itself is much more challenging.

The reason for this is the absence of the general multiplication gate $*$. While we can multiply any two circuit variables in a general arithmetic circuit, we can only multiply variables by constants in a \linear arithmetic circuit. The construction of the \fixpoptgate by \citet{SICOMP:Filos-RatsikasH2023} makes extensive usage of the multiplication gate $*$ and can thus not directly be used for creating the \linoptgate. In our case, the constraint matrix $A$ is fixed (i.e., not an input to the \linoptgate) and this does help to eliminate some of the general multiplication gates, but not all of them. At a high level, the construction of \citet{SICOMP:Filos-RatsikasH2023} ensures that the output $x$ of the gate satisfies
$$\mu_0 \cdot \partial f(x) + A^\transpose \mu = 0$$
where $\mu$ satisfies some standard KKT conditions.
If $x$ is feasible and if $\mu_0 > 0$, then it follows that $x$ is an optimal solution by standard arguments (using the convexity of $f$). The term $\mu_0$ is carefully constructed as a function of $\mu$ and $x$ in order to ensure that $x$ must be feasible and that $\mu_0 > 0$ when $x$ is feasible. However, since both $\mu_0$ and $\partial f(x)$ depend on $x$, in our case we cannot construct the term $\mu_0 \cdot \partial f(x)$, because that would entail multiplying two variables in the circuit. As a result, our construction instead ensures that the output $x$ of the gate satisfies
$$\eps \cdot \partial f(x) + A^\transpose \mu = 0$$
where $\mu$ again satisfies some standard KKT conditions, and where $\eps > 0$ is some sufficiently small constant that is picked when constructing the gate. By standard arguments it still holds that if $x$ is feasible, then it is an optimal solution. The challenge however is to ensure that $x$ is indeed feasible. While the argument is relatively straightforward in the work of \citet{SICOMP:Filos-RatsikasH2023}, because $\mu_0$ can depend on $x$, here $\mu_0$ has been replaced by a constant $\eps$. Our main technical contribution in the construction of the \linoptgate is to show that there exists a sufficiently small $\eps > 0$ that forces $x$ to be feasible, and that such $\eps$ can be constructed efficiently given the parameters of the gate (but, importantly, not its inputs!). As a bonus, our modified construction and analysis allows us to obtain a \linoptgate that does not require any constraint qualification, whereas the construction of \citet{SICOMP:Filos-RatsikasH2023} required an explicit Slater condition (which of course, as they show, is necessary in the case where the matrix $A$ is not fixed).\medskip

\noindent From the standpoint of applications, the \linoptgate can be used in almost the same direct manner as the \fixpoptgate of \citet{SICOMP:Filos-RatsikasH2023}. In some cases, precisely because we cannot multiply within a \linear arithmetic circuit, we may have to apply some standard variable changes, to ``linearize'' certain constraints. Still, the \linoptgate can effectively substitute appropriate optimization programs in the same way that the \fixpoptgate can. 
In a nutshell, one can view the \linoptgate as a more powerful tool for those applications for which rational exact solutions exist.

\subsection{Organization of the Paper}

In \Cref{sec:preliminaries} we provide the main definitions and terminology needed for our paper. In \Cref{sec:lin-opt-gate}, we detail the construction of the \linoptgate, and prove its correctness. In the same section (\Cref{sec:implicit}), we also develop the necessary machinery to show how the \linoptgate can be used to obtain PPAD-membership of problems where certain functions are given implicitly in the input to the problem. In \Cref{sec:games1}, we provide the first applications of the \linoptgate to several important classes of games and to different equilibrium notions, besides Nash equilibria. In \Cref{sec:congestion-games} we explain how to apply the machinery that we develop in \Cref{sec:games1} in order to obtain PPAD-membership results for equilibrium computation in nonatomic and atomic splittable congestion games. In \Cref{sec:markets}, we present the applications of our gate to finding competitive equilibria in Arrow-Debreu markets with different utility and production functions. In \Cref{sec:pacing} we demonstrate the applicability of the \linoptgate to obtain membership results for the auto-bidding auctions with pacing strategies. In \Cref{sec:fair-division} we obtain membership results for the two fundamental fair division problems of envy-free cake cutting and rental harmony. We offer some discussion and some directions for future work in \Cref{sec:conclusion}. 

We would like to emphasize that while our paper is very long, this is almost exclusively due to the fact that it covers so many applications, rather than due to the proofs that we develop for those applications, which in reality range from being very short to relatively short. For each of all of the domains that we consider, (a) we provide the appropriate definition and place the setting in context within the rest of the paper, (b) we discuss the related work and possibly the previous PPAD-membership results (if any), (c) we provide detailed comparisons with those previous proofs to demonstrate the effectiveness of our \linoptgate as a general-purpose proof technique, and finally (d) we develop the proofs themselves. In some cases in fact, we first apply the technique to simpler settings for a gentle introduction, and then move on to study those settings in their full generality. We believe that all of our application sections are largely self-contained, and can be read almost in isolation, even after only reading the introduction of the paper, and by referring only to certain clearly referenced parts in other sections.

\section{Preliminaries} \label{sec:preliminaries}
In this section, we introduce the computational class PPAD, as well as the main machinery that will be used throughout the paper. The details for the specific settings that we will consider in our applications will be defined in the corresponding sections. We start with the definitions of the relevant computational complexity classes. 

\subsection{The class PPAD} \label{sec:ppad}

All of the problems that we will consider in this paper will be \emph{total search problems}. A total search problem is one in which a solution is always guaranteed to exist. For example, finding a Nash equilibrium in a game is a total search problem, by \citeauthor{Nash50}'s theorem \citep{Nash50}. Similarly, competitive equilibria in markets always exist (e.g., see \citep{arrow1954existence}).  The class TFNP \citep{megiddo1991total} contains all total search problems \emph{in NP}, i.e., those for which a candidate solution can be verified in polynomial time. For example, verifying whether a given set of strategies is a Nash equilibrium in a bimatrix game (see \Cref{def:bimatrix-game}) can be done in polynomial time, and hence the problem of finding Nash equilibria in bimatrix games is in TFNP. For a formal definition of the class TFNP, we refer the reader to \citep{JCSS:Papadimitriou1994}.\medskip

\noindent The class PPAD, introduced by \cite{JCSS:Papadimitriou1994}, is defined with respect to its canonical problem, called \textsc{End-of-Line}, see \Cref{def:end-of-line} below. PPAD is the class of all problems in TFNP that are polynomial-time reducible to \textsc{End-of-Line}.

\begin{definition}[\textsc{End-of-Line}]\label{def:end-of-line}
The \textsc{End-of-Line} problem is defined as: given Boolean
circuits $P$ and  $S$ with $n$ input bits and $n$ output bits such that $P(0)=0 \neq S(0)$, find $x$
such that $P(S(x)) \neq x$ or $S(P(x)) \neq x \neq 0$.
\end{definition}

\noindent Intuitively, PPAD captures the following problem. We are given a directed graph in which every node has indegree and outdegree at most $1$ and a source of this graph, and we are asked to find another source or a sink. Such a node exists by the parity argument on the degrees of the nodes, which is the underlying principle of the class PPAD. Importantly, we are not given this graph explicitly in the input (otherwise the problem would be trivially in P), but we can access the predecessor and the successor of a given node via Boolean circuits; these are the circuits $P$ and $S$ of \Cref{def:end-of-line} above. We will be using an alternative definition of the class, via \linear arithmetic circuits, which we will define in \Cref{sec:linear-FIXP} next.

\subsection{The classes FIXP and Linear-FIXP} \label{sec:linear-FIXP}

We start by defining arithmetic circuits and \linear arithmetic circuits.\footnote{Sometimes in the literature, these are also referred to as \emph{algebraic circuits}, e.g., see \citep{SICOMP:Filos-RatsikasH2023}. We use the term ``\linear arithmetic circuit'' for circuits over the basis $\{+,-,\max,\min, \times \zeta\}$ with rational constants. Some recent works call them ``linear'' arithmetic circuits instead~\citep{deligkas2021BU,FearnleyGHS22-gradient}.}

\begin{definition}[Arithmetic Circuit] \label{def:arithmetic circuit}
An \emph{arithmetic circuit} is a circuit using gates in $\{+,-,*,\div,\max,\min\}$ as well as rational constants.
\end{definition}

\noindent A \linear arithmetic circuit is simply an arithmetic circuit were multiplication and division are not allowed.

\begin{definition}[\Linear Arithmetic Circuit] \label{def:linear-arithmetic circuit}
A {\linear arithmetic circuit} is a circuit using gates in $\{+, -, \max, \allowbreak \min, \allowbreak \times \zeta\}$ as well as rational constants, where $\times \zeta$ denotes multiplication by a rational constant.
\end{definition}

\noindent We will use \linear arithmetic circuits to provide an alternative definition of the class PPAD. First, we state and prove the next simple lemma, which will be useful in \Cref{sec:lin-opt-gate}. For a rational number $a$, we let $\size(a)$ denote the number of bits needed to describe $a$ in the standard representation, where $a$ is written as an irreducible fraction, and the numerator and denominator are written in binary. We let $\size(f)$ denote the number of bits needed to describe a \linear arithmetic circuit $f$ (in particular, this includes the length of the description of any constants used in $f$).

\begin{lemma}\label{lem:linear-circuit-bound}
	For any \linear arithmetic circuit $f: \mathbb{R}^n \to \mathbb{R}^m$ and any rational $B \geq 0$ it holds that
	$$\max_{x \in [-B,B]^n} \|f(x)\|_\infty \leq 2^{\poly(\size(B),\size(f))}.$$
\end{lemma}

\begin{proof}
	Since a \linear arithmetic circuit can be evaluated efficiently (see, e.g., \citep[Lemma~3.3]{FearnleyGHS22-gradient}), we have $\|f(0)\|_\infty \leq 2^{\poly(\size(f))}$. Additionally, $f$ is $L$-Lipschitz-continuous over $\mathbb{R}^n$ with Lipschitz constant $L = 2^{\poly(\size(f))}$, see, e.g., \citep[Lemma~A.1]{FearnleyGHS22-gradient}. As a result, for any $x \in [-B,B]^n$
	\[\|f(x)\|_\infty \leq \|f(0)\|_\infty + L \|x\|_\infty \leq 2^{\poly(\size(B),\size(f))}. \qedhere\]
\end{proof}

\noindent We now move on to the definition of the related computational classes, in the context of arithmetic circuits. We mentioned the class FIXP in the introduction; we proceed to formally define it below. The following definitions follow those of \citep{SICOMP:Filos-RatsikasH2023}. \medskip 

\noindent A search problem $\Pi$ with real-valued search space is defined by associating to
any input instance $I$ (encoded as a string over a finite alphabet $\Sigma$)
a search space $D_I\subseteq\RR^{d_I}$ and a set of solutions $\Sol(I)$. 
We assume there is a polynomial time algorithm that given $I$ computes a description of $D_I$.\medskip

\noindent Next, we define \emph{basic \FIXP problems} and 
\emph{basic Linear-\FIXP problems.}

\begin{definition}[Basic (linear)-\FIXP problem]\label{def:basic-fixp}
A search problem $\Pi$ is a basic (Linear)-$\FIXP$ problem
if every instance $I$ describes a nonempty compact convex domain
$D_I$ described by a set of linear inequalities with rational coefficients and a continuous map $F_I\colon D_I\rightarrow D_I$ given
by a (\linear) arithmetic circuit $C_I$, and the solution set is $\Sol(I) = \{x\in D_I\mid
F_I(x)=x\}$.  We assume that $C_I$ well-defined, i.e., it does not divide by zero and that it indeed represents a function $F_I$ with $F_I(D_I) \subseteq D_I$.\footnote{Given an arithmetic circuit, it is not clear how to check whether it does satisfy these properties, so we assume that it does, i.e., we consider promise problems, see also \citep{SICOMP:Filos-RatsikasH2023}. For the case of basic Linear-FIXP problems, the first condition is trivially satisfied (since division is not allowed), but for the second condition we still require a promise. Note that this means that the problem is formally not a TFNP problem, but instead a promise-TFNP problem. This is however not an issue for proving PPAD-membership, since the problem ultimately reduces to the TFNP problem End-of-Line \citep{etessami2010complexity}.}
\end{definition} 

\noindent Next, we define reductions between search problems. Let $\Pi$ and $\Gamma$ be search problems with real-valued search space. A $\emph{many-one reduction}$ 
from $\Pi$ to $\Gamma$ is a pair of maps $(f,g)$. The instance mapping $f$ maps
instances $I$ of $\Pi$ to instances $f(I)$ of $\Gamma$, and for any solution 
$y\in\Sol(f(I))$ the solution mapping $g$ maps the pair $(I,y)$ to a solution $g(I,y)\in\Sol(I)$
of $\Pi$. In order to avoid meaningless reductions, it is required that $\Sol(f(I))\neq\emptyset$
if $\Sol(I)\neq\emptyset$. 
We require that the instance mapping $f$ is computable in polynomial time.
\citet{etessami2010complexity} defined the notion of \emph{SL}-reductions where
the solution mapping $g$ is \emph{separable linear}. This means there 
exists a map $\pi\colon\{1,\dots, d_I\}\rightarrow \{1,\dots, d_{f(I)}\}$
and rational constants $a_i,b_i$, $i=1,\dots, d_I,$ such that for $y\in\Sol(f(I))$
one has that $x=g(I,y)$ is given by $x_i = a_iy_{\pi(i)}+b_i$ for all~$i$. The map $\pi$ and the constants $a_i,b_i$ should be computable from $I$ in polynomial time.\medskip

\noindent We now define the class FIXP.

\begin{definition}[\FIXP]\label{def:FIXP}
The class $\FIXP$ consists of all search problems with real-valued search space that SL-reduce
to a basic $\FIXP$ problem for which the domain $D_I$ is a convex polytope
described by a set of linear inequalities with rational coefficients and the function
$F_I$ is defined by an arithmetic circuit $C_I$. 
\end{definition}

\noindent The class Linear-FIXP is the ``piecewise-linear fragment of FIXP'' \citep{etessami2010complexity}, defined below.

\begin{definition}[Linear-\FIXP]\label{def:linear-FIXP}
The class Linear-\FIXP consists of all search problems with real-valued search space that SL-reduce
to a basic Linear-\FIXP problem for which the domain $D_I$ is a convex polytope
described by a set of linear inequalities with rational coefficients and the function
$F_I$ is defined by a \linear arithmetic circuit $C_I$. 
\end{definition}

\noindent Above we have formally defined the class Linear-\FIXP as a class of search problems with real-valued search space forming a subclass of $\FIXP$. However, Linear-\FIXP also naturally defines a subclass of \TFNP since for any instance $I$ of a basic Linear-\FIXP problem $\Pi$, $\Sol(I)$ always contains rational-valued solutions of polynomial bit-length. Following the convention in the literature, we will denote the class of search problems in \TFNP reducible to the exact fixed-point computation defined by $\Pi$ by Linear-\FIXP as well. \medskip

\noindent With this convention, \citet{etessami2010complexity} showed the following equivalence result:

\begin{theorem}[\citep{etessami2010complexity}]\label{thm:PPAD=linear-FIXP}
\emph{Linear-}\FIXP = \emph{PPAD}.
\end{theorem}

\noindent \Cref{thm:PPAD=linear-FIXP} provides an alternative definition of PPAD which we will be using throughout the paper. Roughly speaking, to show that a problem is in PPAD, it suffices to show that it can be reduced to computing a fixed point of a function encoded by a \linear arithmetic circuit. 

\begin{remark}[Linear vs \Linear]\label{rem:linear-vs-linear}
The term ``linear'' in ``linear-FIXP'' might be a bit misleading, as it may suggest that it only refers to linear functions. From \Cref{def:linear-arithmetic circuit}, it should be obvious that they capture piecewise-linear functions instead, and are hence more general. We believe that the term ``linear'' was used in the related literature rather than ``piecewise-linear'' for succinctness and brevity. In this paper, we will call these arithmetic circuits \linear arithmetic circuits (for ``piecewise-linear'') but still refer to the piecewise-linear fragment of \FIXP as Linear-\FIXP for consistency with the literature. We also call our new gate ``\linoptgate'' to remain consistent with Linear-FIXP.
\end{remark}

\noindent We conclude the section with a very useful definition, that of \emph{\pseudogs}.

\begin{definition}[\Pseudog]\label{def:pseudo-circuit}
A \pseudog with $n$ inputs and $m$ outputs is a \linear arithmetic circuit $F: \mathbb{R}^n \times [0,1]^\ell \to \mathbb{R}^m \times [0,1]^\ell$. The output of the \pseudog on input $\mathbf{x}$ is any $\mathbf{y}$ that satisfies $F(\mathbf{x},\mathbf{z}) = (\mathbf{y},\mathbf{z})$ for some  $\mathbf{z} \in [0,1]^\ell$. Note that a \pseudog can have multiple possible outputs.
\end{definition}

\noindent Intuitively, \pseudogs are \linear arithmetic circuits which are only required to work correctly at a fixed point of the encoded function (or, to be more precise, when its ``auxiliary'' variables $\mathbf{z}$ satisfy a fixed point condition). In particular, \linear arithmetic circuits are \pseudogs, and in fact, for the purpose of proving membership in PPAD, those two are equivalent. The \linoptgates that we will define in the next section are in fact \pseudogs that are used as \emph{primitives} or \emph{subroutines} in larger \pseudogs. The following is an example of a simple \pseudog, that was already used as an important primitive by \citet{SICOMP:Filos-RatsikasH2023}.

\begin{example}[\Pseudog computing the Heaviside function]\label{ex:Heaviside}
The Heaviside function is the following correspondence
\begin{equation*}
\Heaviside(x) = \begin{cases}
1 & \text{ if } x>0 \\
[0,1] & \text{ if } x=0 \\
0 & \text{ if } x<0 
\end{cases} \enspace .
\end{equation*}
We can construct a \pseudog $F: \mathbb{R} \times [0,1] \to \mathbb{R} \times [0,1]$ computing $\Heaviside$ by letting
$$F(x,z) := (z, \min\{1, \max\{0,z + x\}\}).$$
It is easy to check that $F$ indeed computes $\Heaviside$, i.e., $F(x,z) = (y,z) \implies y \in \Heaviside(x)$.
\end{example}

Computing a fixed point of a \pseudog corresponds to computing a fixed point of the \linear arithmetic circuit representing it. Thus, a \pseudog is guaranteed to have at least one rational fixed point, and the problem of computing such a fixed point lies in PPAD.

\section{A Powerful Tool for PPAD-membership: The \linoptgate}\label{sec:lin-opt-gate}
In this section, we develop our main tool, the \linoptgate. \medskip

\begin{definition}[\linoptgate]\label{def:linoptgate}
The \linoptgate is a gate which:
\begin{itemize}
\item[-] is parameterized\footnote{The parameters of a gate determine its behavior and must be provided every time a gate of this type is used in a circuit (and thus also count towards the representation size of the circuit). For example, whenever we use a ``multiplication by a constant'' gate $\times \zeta$, we have to specify the constant parameter $\zeta$ of the gate. The same also applies to the \linoptgate, except that it has (many) more parameters.} by $n, m, k \in \mathbb{N}$, a rational matrix $A \in \mathbb{R}^{m \times n}$, and a \linear arithmetic circuit $G_{\partial f} : \mathbb{R}^n \times \mathbb{R}^k \times [0,1]^\ell \to \mathbb{R}^n \times [0,1]^\ell$.
\item[-] takes as input $b \in \mathbb{R}^m$, $c \in \mathbb{R}^k$, and $R \in \mathbb{R}$.
\end{itemize}
The \linoptgate outputs an optimal solution of the following optimization problem (over variables $x \in \mathbb{R}^n$):

\begin{center}\underline{Optimization Program $\mathcal{C}$}\end{center}
\begin{equation}\label{eq:OPT-gate-general}
\begin{split}
\min \quad &f(x;c) \\
\text{ s.t.} \quad & Ax \leq b\\
& x \in [-R,R]^n
\end{split}
\end{equation}
whenever the two following conditions hold for the given inputs $b \in \mathbb{R}^m$, $c \in \mathbb{R}^k$, and $R \in \mathbb{R}$:
\begin{enumerate}
\item \label{enum:optgate-1}The feasible domain $\{x \in [-R,R]^n : Ax \leq b\}$ is not empty.
\item \label{enum:optgate-2}The map $x \mapsto f(x;c)$ is a convex function on the feasible domain and its subgradient is given by the \pseudog $G_{\partial f}$.
\end{enumerate}
If the two conditions are not satisfied then the \linoptgate can have arbitrary output. If the conditions are satisfied and $\mathcal{C}$ has multiple optimal solutions, then any such optimal solution is an acceptable output for the gate.
\end{definition}

The following theorem is our main result. It shows that \linoptgates can be simulated by standard gates by adding some auxiliary inputs and outputs to the circuit.

\begin{theorem}\label{thm:linoptgate}
When constructing a \linear arithmetic circuit for the purpose of proving membership in PPAD, we can also use \linoptgates, in addition to the standard gates. More formally, given a \linear arithmetic circuit\footnote{Here $D$ is a nonempty compact convex domain represented by linear inequalities.} $F: D \to D$ that uses \linoptgates, we can construct in polynomial time a \linear arithmetic circuit $G: D \times [0,1]^t \to D \times [0,1]^t$ that does not use \linoptgates, but which is equivalent to $F$, in the following sense:
$$G(y,\alpha) = (y,\alpha) \implies F(x) = x$$
for all $x \in D$ and $\alpha \in [0,1]^t$.
\end{theorem}

\begin{remark}\label{rem:linoptgate}
Note that $G$ must have a fixed point by Brouwer's fixed point theorem and the problem of computing such a fixed point lies in PPAD by \cref{thm:PPAD=linear-FIXP}. By \cref{thm:linoptgate} it follows that $F$ must also have a fixed point and computing one also lies in PPAD. The expression ``$F(x) = x$'' in the theorem should be understood as: ``On input $x$, there exists a valid assignment to all the gates of $F$ such that the output of $F$ is $x$''. The value assigned to a standard gate is fully specified by its inputs. However, for a \linoptgate there might be multiple acceptable assignments (e.g., if it solves an LP that has multiple solutions), and the expression ``$F(x) = x$'' states that there exists at least one acceptable assignment such that the output of $F$ is $x$. It does not say that any acceptable assignment will work, but merely that there exists one that does.
\end{remark}

The proof of \cref{thm:linoptgate} can be found in \cref{sec:linoptgate-proof}. We continue with some additional remarks about the \linoptgate.

\begin{remark}
Note that the objective function $f$ does not need to be computable by a \linear arithmetic circuit. We only require that its subgradient (on the feasible domain) is given as a \pseudog. In particular, $f$ can be a (convex) quadratic polynomial.
\end{remark}

\noindent In most of our applications, it will suffice for optimization program $\mathcal{C}$ to be a linear program, i.e., for $f$ to be linear function. We will use $\mathcal{P}$ to refer to the general form of this linear program, and we will reference that in our applications.
\begin{center}\underline{Linear Program $\mathcal{P}$}\end{center}
\begin{equation}\label{eq:OPT-gate-linear}
\begin{split}
\min \quad &c^\transpose x\\
\text{ s.t.} \quad & Ax \leq b\\
& x \in [-R,R]^n
\end{split}
\end{equation}
Note that here $f(x;c) = c^\transpose x$, and its subgradient is simply computed by the \pseudog $G_{\partial f} : \mathbb{R}^n \times \mathbb{R}^k \to \mathbb{R}^n$, $(x;c) \mapsto c$. The \pseudog $G_{\partial f}$ is in fact just a normal \linear arithmetic circuit here, i.e., $\ell = 0$, since no auxiliary fixed point variables are needed to compute the subgradient.

\begin{remark}\label{rem:subgradient-supergradient}
The \linoptgate can of course also solve maximization problems $\max f(x)$ where the objective function $f$ is concave, since this is equivalent to the problem $\min -f(x)$. In that case we have to provide a \pseudog computing the supergradient of $f$, or equivalently the subgradient of $-f$.
\end{remark}

\begin{remark}
Note that we require that the constraint matrix $A$ be fixed, whereas the right-hand side of the constraints $b$ can be given as an input to the gate. This is in fact necessary, as the following example shows. If the \linoptgate could solve the LP
\begin{equation*}
\begin{split}
\min \quad &x\\
\text{ s.t.} \quad & a \cdot x \geq 1\\
& x \in [-2,2]
\end{split}
\end{equation*}
where $a$ is not fixed, then we would obtain a \pseudog computing $1/a$ for $a \in [1,2]$. But then, we would be able to construct a \pseudog $F: [1,2] \to [1,2]$ computing $y \mapsto \min\{2, \max\{1, 2/y\}\}$. The only fixed point of $F$ is at $y = \sqrt{2}$, which is a contradiction, since \pseudogs always have at least one rational fixed point.
\end{remark}

\subsection{Feasibility Program with Conditional Constraints}

Using the \linoptgate, we can also solve feasibility programs, which will be very useful throughout our applications. In particular, when constructing a \linear arithmetic circuit \emph{for the purpose of proving membership in PPAD}, we can assume without loss of generality that we have access to an additional gate solving feasibility programs with conditional constraints, which:
\begin{itemize}
	\item[-] is parameterized by $n, m, k \in \mathbb{N}$, a rational matrix $A \in \mathbb{R}^{m \times n}$, and \linear arithmetic circuits $h_i : \mathbb{R}^k \to \mathbb{R}$ for $i = 1, \dots, m$.
	\item[-] takes as input $b \in \mathbb{R}^m$, $y \in \mathbb{R}^k$, and $R \in \mathbb{R}$.
\end{itemize}
The gate outputs a feasible solution of the following feasibility problem (over variables $x \in \mathbb{R}^n$):
\begin{center}\underline{Feasibility Program $\mathcal{Q}$}\end{center}
\begin{equation}\label{eq:feasibility-general}
\begin{split}
h_i(y) > 0 \implies a_i^\transpose x \leq b_i\\
x \in [-R,R]^n
\end{split}
\end{equation}
whenever it is feasible. Note that we may add \emph{unconditional} constraints to feasibility program $\mathcal{Q}$ above by simply setting $h_i(y) = 1$ in the conditional constraints above.

\paragraph{\bf Construction.}
We can solve the feasibility problem $\mathcal{Q}$ in \eqref{eq:feasibility-general} by solving the following optimization problem (over variables $x \in \mathbb{R}^n$)
\begin{equation}\label{eq:feasibility-as-OPT}
\begin{split}
\min \quad &\sum_{i=1}^m \max\{0, h_i(y)\} \cdot \max\{0, a_i^\transpose x-b_i\} \\
\text{ s.t.} \quad & x \in [-R,R]^n
\end{split}
\end{equation}
Note that if $x$ is an optimal solution for \eqref{eq:feasibility-as-OPT} with objective function value $0$, then $x$ is feasible for $\mathcal{Q}$ in \eqref{eq:feasibility-general}. Furthermore, if $\mathcal{Q}$ is feasible, then the optimal value of \eqref{eq:feasibility-as-OPT} is $0$. Thus, if $\mathcal{Q}$ is feasible, then any optimal solution to \eqref{eq:feasibility-as-OPT} will also be a feasible solution to $\mathcal{Q}$ in \eqref{eq:feasibility-general}.

As a result, it suffices to show that we can use the \linoptgate to solve \eqref{eq:feasibility-as-OPT}. Clearly, the feasible domain of \eqref{eq:feasibility-as-OPT} is nonempty. Thus, it remains to show that we can construct a \pseudog computing the subgradient of $x \mapsto f(x;y,b)$, where
$$f(x;y,b) = \sum_{i=1}^m \max\{0, h_i(y)\} \cdot \max\{0, a_i^\transpose x-b_i\}.$$
Note that this function is indeed convex in $x$.

The subgradient of $x \mapsto \max\{0, a_i^\transpose x-b_i\}$ can be expressed as $\Heaviside(a_i^\transpose x-b_i) \cdot a_i$, where we recall that $\Heaviside$ is the Heaviside function defined as 

\begin{equation*}
\Heaviside(z) = \begin{cases}
1 & \text{ if } z>0 \\
[0,1] & \text{ if } z=0 \\
0 & \text{ if } z<0 
\end{cases} \enspace .
\end{equation*}

\noindent Thus, the subgradient of $x \mapsto f(x;y,b)$ can be written as $\sum_{i=1}^m \max\{0, h_i(y)\} \cdot \Heaviside(a_i^\transpose x-b_i) \cdot a_i$. For each $i \in [m]$, we can compute the term $\max\{0, h_i(y)\} \cdot \Heaviside(a_i^\transpose x-b_i)$ by using \cref{lem:mult-by-Heaviside} below. Since the vectors $a_i$ are fixed, we can then compute the product $\max\{0, h_i(y)\} \cdot \Heaviside(a_i^\transpose x-b_i) \cdot a_i$. Doing this for every $i \in [m]$ and then summing up yields an element in the subgradient of $x \mapsto f(x;y,b)$. Thus, we have successfully constructed a \pseudog computing this subgradient. It remains to prove the following lemma we used, which will also be useful later.

\begin{lemma}\label{lem:mult-by-Heaviside}
For the purpose of proving PPAD-membership, we can construct a \pseudog computing $(x,y) \mapsto H(x) \cdot y$.
\end{lemma}

\begin{proof}
Note that $H(x) \cdot y$ can be obtained by computing $H(x) \cdot \max\{0,y\} - H(x) \cdot \max\{0,-y\}$. Thus, it suffices to prove that we can compute $H(x) \cdot \max\{0,y\}$. This can indeed be achieved by using the \linoptgate to solve the following LP (in variable $v \in \mathbb{R}$):
\begin{equation*}
\begin{split}
\max \quad &v \cdot x \\
\text{ s.t.} \quad & 0 \leq v \leq \max\{0,y\} \\
\end{split}
\end{equation*}
Note that the feasible domain is nonempty, and the gradient of the objective function is $x$, which can trivially be computed by a \linear arithmetic circuit. It is straightforward to verify that any optimal solution $v$ satisfies $v \in H(x) \cdot \max\{0,y\}$, as desired.
\end{proof}

\subsection{Using the \linoptgate in applications}\label{sec:linopt-applications}

In our applications in subsequent sections we will be constructing \linear arithmetic circuits containing several \linoptgates, corresponding to multiple optimization programs like the program $\mathcal{C}$ above, as well as feasibility programs like the program $\mathcal{Q}$. It will be helpful to be able to reference the inputs to those \linoptgates as opposed to the variables of the corresponding programs, particularly because variables for one program would be inputs to the \linoptgate corresponding to another program and vice versa. We will use the term \emph{\circparams} to refer to those inputs.

\begin{definition}[\Circparams]\label{def:circuit-parameters}
Consider an optimization program in the form of $\mathcal{C}$ or a feasibility program in the form of $\mathcal{Q}$ and let $C$ be its corresponding \linoptgate. We will refer to the inputs of $C$ as \emph{\circparams} of the program $\mathcal{C}$ or $\mathcal{Q}$.
\end{definition}

\noindent Using this terminology, we can argue that a specific program can be solved by the \linoptgate as follows.\medskip

\noindent \textbf{For optimization programs of the form $\mathcal{C}$} we need to argue
    \begin{itemize}
        \item[-] Conditions~\ref{enum:optgate-1} and \ref{enum:optgate-2} in the definition of the \linoptgate for $\mathcal{C}$ above, namely that the domain is non-empty and that the subgradient of the convex objective function is given by a \pseudog $G_{\partial f}$,
        \item[-] that \circparams appear only on the right-hand side of the constraints, but not on the left-hand side.
    \end{itemize}

\noindent \textbf{For feasibility programs of the form $\mathcal{Q}$} we need to argue that
    \begin{itemize}
        \item[-] the feasibility program is \emph{solvable} (i.e., feasible),
        \item[-] the \circparams appear only on the right-hand side of the constraints $a_i^\transpose x \leq b_i$,
        \item[-] only \circparams appear on the left-hand side of the conditional constraints, i.e., in the function $h_i(y) >0$.
    \end{itemize}
The conditions above are obviously equivalent to the optimization and feasibility programs having the form of \eqref{eq:OPT-gate-general} and \eqref{eq:feasibility-general}, since the \circparams are the inputs to the \linoptgate. Here the conditions are simply ``spelled-out'', because it is easier to refer to them in subsequent sections. What is not obvious is how one may argue about the solvability of a feasibility program $\mathcal{Q}$. 

\paragraph{Solvability of $\mathcal{Q}$.} The feasibility programs that will appear in most of our applications will have the following same general form; it will be easy to argue the solvability of those that do not.

\begin{center}\underline{Feasibility Program $\mathcal{Q}_\text{app}$}\end{center}
\begin{align*}
\begin{split}
&h_k(y) - h_{k'}(y) > 0 \implies w_k \leq \rho \cdot w_{k'} \  \text{ for all } k,k' \in [m]\\
&\sum_{j=1}^m w_j = 1,  \ \ \ 
w_i \geq \frac{\rho^m}{m}, \  \text{ for all } i \in [m]
\end{split}
\end{align*}
\noindent for some $0 <  \rho \leq 1$, where $w \in \mathbb{R}^m$ are the variables. For this type of feasibility program $\mathcal{Q}_\text{app}$, we can define the notion of a \emph{\qgraph}.

\begin{definition}[The \qgraph $G_\mathcal{Q}$]\label{def:feasibility-graph-q}
Consider a feasibility program of the form $\mathcal{Q}_\text{app}$. Let $G_\mathcal{Q_\text{app}}$ be the graph that has nodes for each $k \in [m]$, and a directed edge $(k,k')$ if and only if $h_k(y) - h_{k'}(y) > 0$.  We will refer to $G_\mathcal{Q_\text{app}}$ as the \emph{\qgraph} of $\mathcal{Q}_\text{app}$.
\end{definition}

\noindent The following lemma provides a general condition for solvability of $\mathcal{Q}_\text{app}$.

\begin{lemma}[Solvability of $\mathcal{Q}_\text{app}$]\label{lem:feasibility-of-q-graph}
A feasibility program of the form $\mathcal{Q}_\text{app}$ is solvable as long as its \qgraph $G_{\mathcal{Q}_{\text{app}}}$ is acyclic. 
\end{lemma}

\begin{proof}
Assume that $G_{\mathcal{Q}_\text{app}}$ is acyclic, and let $d_k$ be the length of the \emph{longest} path from node $k$ to a sink node in $\mathcal{Q}_\text{app}$. Let 
\begin{equation*}
w_k = \frac{\rho^{d_k}}{\sum_{j=1}^m \rho^{d_j}}, \ \text{ for all } k \in [m].
\end{equation*}
We will argue that these values of $w_k$, for $k \in [m]$, satisfy the constraints of the feasibility program $\mathcal{Q}_\text{app}$. Obviously, $\sum_{k=1}^m w_k =1$ by definition. Since the graph has $m$ nodes, it holds that $d_k \leq m$. Since $\rho \leq 1$, this implies that $\rho^{d_k} \geq \rho^m$ and that $\sum_{j=1}^n \rho^{d_j} \leq m$. Therefore, we obtain that $w_k \geq \frac{\rho^m}{m}$. It remains to show these values of $w_k$ satisfy the conditional constraints. Indeed, consider an edge $(k,k')$ in the \qgraph $\mathcal{Q}_\text{app}$, which, recall, corresponds to a constraint where $h_k(y) - h_{k'}(y) > 0$. Since $(k,k')$ is an edge in $\mathcal{Q}_\text{app}$, we have that $d_k \geq d_{k'}+1$, as there is a path from $k$ to a sink of $\mathcal{Q}_\text{app}$ that starts with the edge $(k,k')$. This implies that
\begin{equation*}
w_k = \frac{\rho^{d_k}}{\sum_{j=1}^m \rho^{d_j}} \leq \frac{\rho^{d_k'+1}}{\sum_{j=1}^m \rho^{d_j}} = \frac{\rho \cdot \rho^{d_k'}}{\sum_{j=1}^m \rho^{d_j}} = \rho \cdot w_{k'},
\end{equation*}
and hence the corresponding conditional constraint is satisfied.
\end{proof}

\noindent Thus, in our applications in which the feasibility programs that we construct are of the form $\mathcal{Q}_\text{app}$ above, it suffices to show that their corresponding \qgraph $\mathcal{Q}_\text{app}$ is acyclic, in order to establish their solvability by \Cref{lem:feasibility-of-q-graph}.

\subsection{Construction and proof for the \linoptgate}\label{sec:linoptgate-proof}

In this section, we prove our main result stated earlier, \cref{thm:linoptgate}. The theorem follows from the following proposition, which proves that a single \linoptgate can be simulated by a standard \linear arithmetic circuit.

\begin{proposition}\label{prop:lin-opt-gate-formal}
Given $n, m, k \in \mathbb{N}$, a rational matrix $A \in \mathbb{R}^{m \times n}$, rational bounds $R > 0$ and $C > 0$, as well as a \linear arithmetic circuit $G_{\partial f} : \mathbb{R}^n \times \mathbb{R}^k \times [0,1]^\ell \to \mathbb{R}^n \times [0,1]^\ell$, we can construct a \linear arithmetic circuit $F: \mathbb{R}^m \times \mathbb{R}^k \times [0,1]^t \to [-R,R]^n \times [0,1]^t$ in time
$$\poly(n, m, k, \size(A), \size(R), \size(C), \size(G_{\partial f}))$$
which satisfies
$$F(b,c,\alpha) = (x,\alpha) \implies x \text{ is an optimal solution to optimization problem $\mathcal{C}$ in \eqref{eq:OPT-gate-general} at } (b,c)$$
whenever $b$ and $c$ satisfy the following three conditions:
\begin{enumerate}
	\item The feasible domain $\{x \in [-R,R]^n : Ax \leq b\}$ is not empty.
	\item The map $x \mapsto f(x;c)$ is a convex function on the feasible domain and its subgradient is given by the \pseudog $G_{\partial f}$.
	\item $\|c\|_\infty \leq C$.
\end{enumerate}
\end{proposition}

\cref{thm:linoptgate} follows from \cref{prop:lin-opt-gate-formal} by simply repeatedly replacing every \linoptgate by its corresponding standard \linear arithmetic circuit with auxiliary inputs and outputs, until no \linoptgates are left in the circuit.

The attentive reader might have noticed that the statement of \cref{prop:lin-opt-gate-formal} does not completely correspond to what was claimed in the definition of the \linoptgate (\cref{def:linoptgate}). Namely, the proposition assumes that we are given $R$ as a parameter when we construct the gadget, whereas the original definition allowed $R$ to be an input to the \linoptgate. Furthermore, the proposition also asks for an upper bound $C$ on the length of $c$ to be given as a parameter when constructing the \linoptgate, whereas no such $C$ was mentioned earlier. This is without loss of generality, as we argue next.

Instead of assuming that $R$ is a fixed parameter of the \linoptgate, we can assume that $R$ is an input, but that we are also given an upper bound $R'$ on $R$. This can easily be achieved by applying \cref{prop:lin-opt-gate-formal} using $R'$ as the value of the parameter $R$, and explicitly adding constraints $x_i \leq R$, $-x_i \leq R$ (where $R$ can now indeed be given as an input, since it only appears on the right hand side of constraints). Note that as long as we indeed have $R \leq R'$, the new optimization problem is equivalent to the previous one, and the \linoptgate will correctly solve it.

We still have to provide upper bounds $R'$ and $C$ when constructing the \linoptgate. The crucial observation here is that the statement of \cref{thm:linoptgate} explicitly mentions that \linoptgates can only be used \emph{for the purpose of proving membership in PPAD}. More formally, this should be interpreted as saying: as long as we are ultimately only using the \linoptgates inside a \linear arithmetic circuit with bounded domain (which is the case in \cref{thm:linoptgate}), we can assume that we do not need to explicitly provide upper bounds $R'$ and $C$. The reason for this stems from \cref{lem:linear-circuit-bound}, which states that we can bound the magnitude of any value inside a \linear arithmetic circuit with bounded domain, in terms of the size of the description of the circuit, and the bound on the domain. Thus, whenever we use a \linoptgate in a \linear arithmetic circuit with bounded domain, since the inputs $R$ and $c$ are computed by a \linear arithmetic circuit with bounded domain, we can compute corresponding upper bounds $R'$ and $C$ for the actual construction of the \linoptgate gadget according to \cref{prop:lin-opt-gate-formal}. An important but subtle point is that the \linoptgate is always guaranteed to output an element in $[-R',R']^n$. As a result, even in a \linear arithmetic circuit that uses multiple \linoptgates, we can efficiently compute upper bounds on the magnitudes of numbers given as input to the \linoptgate, even before replacing the \linoptgates by the actual gadgets implementing them (namely the construction of \cref{prop:lin-opt-gate-formal}).

\subsubsection{Construction of the \linear arithmetic circuit \texorpdfstring{$F$}{F}}

We begin with the description of how the circuit $F$ is constructed. For this we use a sufficiently small value $\eps > 0$, which we construct in the next section. We let $t := n + \ell + m$ and write $[0,1]^t = [0,1]^n \times [0,1]^{\ell + m}$.\medskip

\noindent On input $(b,c,\alpha, \beta) \in \mathbb{R}^m \times \mathbb{R}^k \times [0,1]^n \times [0,1]^{\ell + m}$, the circuit $F$ proceeds as follows:
\begin{enumerate}
\item Compute $x := 2R\alpha-R$. In other words, we scale $\alpha \in [0,1]^n$ into a point $x \in [-R,R]^n$.
\item Compute $(v,\overline{\beta}_1,\dots,\overline{\beta}_\ell) := G_{\partial f}(x,c,\beta_1,\dots,\beta_\ell)$.
\item Compute $\mu_i \in \Heaviside(a_i^\transpose  x - b_i)$ for $i = 1, \dots, m$, using auxiliary variables $\beta_{\ell+1}, \dots, \beta_{\ell+m}$ and corresponding outputs $\overline{\beta}_{\ell+1}, \dots, \overline{\beta}_{\ell+m}$. More formally, compute $\overline{\beta}_{\ell+i} := \min\{1,\max\{0,\beta_{\ell+i} + a_i^\transpose  x - b_i\}\}$ and $\mu_i := \beta_{\ell+i}$ for $i = 1, \dots, m$.
\item Compute
$$\overline{x} := \Pi_R \left( x - \eps \cdot v - A^\transpose \mu  \right)$$
where $\Pi_R$ denotes projection to $[-R,R]^n$.
\item Compute $\overline{\alpha} := (\overline{x} + R)/2R$ (i.e., scale $\overline{x} \in [-R,R]^n$ into a point $\overline{\alpha} \in [0,1]^n$).
\item Output $(\overline{x},\overline{\alpha},\overline{\beta})$.
\end{enumerate}
Here $a_i \in \mathbb{R}^n$ denotes the $i$th row of matrix $A$.\medskip 

\noindent Note that we can construct a \linear arithmetic circuit computing $F$ in time 
\[\poly(n, m, k, \size(A), \size(R), \size(C), \size(G_{\partial f})),\] assuming that $\eps$ can computed in polynomial time in these quantities (which we argue in the section). In particular, the projection $\Pi_R(y)$ of some vector $y$ can be obtained by computing $\min\{R,\max\{-R,y_i\}\}$ for each coordinate of $y$.  Furthermore, note that the first output $\overline{x}$ of $F$ always satisfies $\overline{x} \in [-R,R]^n$, even when the fixed point constraints $\alpha = \overline{\alpha}$ and $\beta = \overline{\beta}$ are not satisfied.

\subsubsection{Construction of \eps}

We now describe how $\eps$ is constructed. Since $G_{\partial f}$ is a \linear arithmetic circuit, by \Cref{lem:linear-circuit-bound} we can compute in time $\poly(\size(R), \size(C), \size(G_{\partial f}))$ a rational $K > 0$ such that
\begin{equation}\label{eq:def-of-K}
\max_{(x,c,\beta) \in [-R,R]^n \times [-C,C]^n \times [0,1]^\ell} \|G_{\partial f}(x,c,\beta)\|_2 \leq K.
\end{equation}
For the construction of $\eps$ we will also require the following notation. We define $\widetilde{A} \in \mathbb{R}^{(m+2n) \times n}$ and $\widetilde{b} \in \mathbb{R}^{m+2n}$ such that the system ``$\widetilde{A}x \leq \widetilde{b}$'' corresponds to the system ``$Ax \leq b$'' with the additional constraints ``$x_i \leq R$'' and ``$-x_i \leq R$'' for $i = 1, \dots, n$. In particular, the first $m$ rows of $\widetilde{A}$ correspond to $A$ and the first $m$ entries in $\widetilde{b}$ correspond to $b$. We use $\widetilde{a}_i$ to denote the $i$th row of $\widetilde{A}$.\medskip

\noindent We set $\eps := \gamma^*/K$, where $\gamma^* > 0$ is computed in time $\poly(\size(A))$ using the following lemma. 

\begin{lemma}\label{lem:choice-of-gamma}
Given $A \in \mathbb{R}^{m \times n}$, we can construct in polynomial time a sufficiently small number $\gamma^* > 0$ such that for any nonempty $I \subseteq [m+2n]$ and for every partition of $I$ into $I_0$ and $I_1$, such that $I_1 \neq \emptyset$, if the following optimization problem (in variables $u \in \mathbb{R}^n$ and $\lambda \in \mathbb{R}^{|I|}$)
\begin{equation}\label{eq:OPT-for-gamma}
\begin{alignedat}{2}
\min \quad &\sum_{i \in I_1} \widetilde{a}_i^\transpose u && \\
\text{ s.t.} \quad & \widetilde{a}_i^\transpose u = 0 &&\forall i \in I_0\\
& \widetilde{a}_i^\transpose u \geq 0 &&\forall i \in I_1\\
& \|u\|_2 = 1 && \\
& u = \sum_{i \in I} \lambda_i \widetilde{a}_i && \\
& \lambda_i \geq 0 &&\forall i \in I
\end{alignedat}
\end{equation}
is feasible, then its optimal value $\gamma$ satisfies $\gamma > \gamma^*$.
\end{lemma}

\begin{proof}
Let $I$, $I_0$, and $I_1$ be as specified in the statement of the lemma, and such that optimization problem \eqref{eq:OPT-for-gamma} is feasible. Letting $\gamma$ denote the optimal value of \eqref{eq:OPT-for-gamma}, note that $\gamma \geq \gamma'/\sqrt{n}$, where $\gamma'$ is the optimal value of the same optimization problem, except that we replace the constraint ``$\|u\|_2 = 1$'' by ``$\|u\|_\infty = 1$'', namely:
\begin{equation}\label{eq:OPT-for-gamma-infty}
\begin{alignedat}{2}
\min \quad &\sum_{i \in I_1} \widetilde{a}_i^\transpose u && \\
\text{ s.t.} \quad & \widetilde{a}_i^\transpose u = 0 &&\forall i \in I_0\\
& \widetilde{a}_i^\transpose u \geq 0 &&\forall i \in I_1\\
& \|u\|_\infty = 1 && \\
& u = \sum_{i \in I} \lambda_i \widetilde{a}_i && \\
& \lambda_i \geq 0 &&\forall i \in I
\end{alignedat}
\end{equation}
This is due to the fact that all other constraints are invariant to scaling of $(u,\lambda)$. In particular, note that \eqref{eq:OPT-for-gamma-infty} is also feasible.

The optimal value $\gamma'$ of \eqref{eq:OPT-for-gamma-infty} satisfies $\gamma' = \min_{k \in I, s \in \{+1,-1\}} \gamma'_{k,s}$, where $\gamma'_{k,s}$ is the optimal value of the following LP
\begin{equation*}
\begin{alignedat}{2}
\min \quad &\sum_{i \in I_1} \widetilde{a}_i^\transpose u && \\
\text{ s.t.} \quad & \widetilde{a}_i^\transpose u = 0 &&\forall i \in I_0\\
& \widetilde{a}_i^\transpose u \geq 0 &&\forall i \in I_1\\
& -1 \leq u_j \leq 1 \quad &&\forall j \in [n] \\
& u_k = s && \\
& u = \sum_{i \in I} \lambda_i \widetilde{a}_i && \\
& \lambda_i \geq 0 &&\forall i \in I
\end{alignedat}
\end{equation*}
if it is feasible, and $\gamma'_{k,s} := + \infty$ otherwise.

Next, we show that $\gamma'_{k,s} > 0$. This clearly holds if the LP is infeasible. Assume towards a contradiction that the LP is feasible, but $\gamma'_{k,s} \leq 0$. Then, $\gamma'_{k,s} = 0$ and it is achieved by a feasible $u$ with $\widetilde{a}_i^\transpose u = 0$ for all $i \in I$. Since $u$ is feasible, there exists $\lambda \geq 0$ with $u = \sum_{i \in I} \lambda_i \widetilde{a}_i$. But then $\|u\|_2^2 = u^\transpose u = \sum_{i \in I} \lambda_i \widetilde{a}_i^\transpose u = 0$, which contradicts $|u_k| = 1$.

Finally, note that the bit-complexity of the optimal value of any LP of this form is polynomially bounded by the bit-representation of matrix $A$. Indeed, the number of bits needed to write down such an LP for any choice of $k \in I, s \in \{+1,-1\}$, and for any $I$, $I_0$, and $I_1$ is bounded by some polynomial quantity in $\size(A)$. In particular, given $A$, we can compute in polynomial time a rational value $\gamma^* > 0$ such that $\gamma^* < \gamma'_{k,s}/\sqrt{n}$ for all $k \in I, s \in \{+1,-1\}$, and for all possible $I$, $I_0$, and $I_1$. By the arguments above, it then follows that $\gamma^* < \gamma$.
\end{proof}

\subsubsection{Analysis: Fixed point constraints}

In this section, we state and prove some simple properties that follow from the fixed point constraints $(\alpha, \beta) = (\overline{\alpha}, \overline{\beta})$. Recall that the system ``$\widetilde{A}x \leq \widetilde{b}$'' corresponds to the system ``$Ax \leq b$'' with the additional constraints ``$x_i \leq R$'' and ``$-x_i \leq R$'' for $i = 1, \dots, n$. In particular, the first $m$ rows of $\widetilde{A}$ correspond to $A$ and the first $m$ entries in $\widetilde{b}$ correspond to $b$.\medskip

\noindent The circuit $F$ has been constructed in order to ensure that the following properties hold.

\begin{claim}\label{clm:OPT-gate-FP-constraints}
If $(\alpha, \beta) = (\overline{\alpha}, \overline{\beta})$, then:
\begin{enumerate}
\item $\mu_i \in \Heaviside(a_i^\transpose x - b_i)$ for $i = 1, \dots, m$.
\item $\|v\|_2 \leq K$ and, if $Ax \leq b$, then $v \in \partial f(x; c)$.
\item $x = \overline{x}$ and
\begin{equation}\label{eq:KKT-in-claim}
\eps \cdot v + \widetilde{A}^\transpose \widetilde{\mu} = 0
\end{equation}
where $\widetilde{\mu} \in \mathbb{R}^{m+2n}$ satisfies $\widetilde{\mu} \geq 0$, as well as
$$\widetilde{\mu}_i > 0 \implies \widetilde{a}_i^\transpose x \geq \widetilde{b}_i$$
for all $i = 1, \dots, m+2n$. Furthermore, for $i \in [m]$ we also have $\widetilde{\mu}_i = \mu_i$.
\end{enumerate}
\end{claim}

\begin{proof}
The first statement follows from the construction of the Heaviside \pseudog. The fact that $\|v\|_2 \leq K$ follows from the definition of $K$ (see~\eqref{eq:def-of-K}) and the assumption that $\|c\|_\infty \leq C$. The fact that $v \in \partial f(x; c)$ when $Ax \leq b$ follows from the assumption that $G_{\partial f}$ is a \pseudog computing $\partial f(x; c)$, whenever $x$ lies in feasible domain.

Since $\alpha = \overline{\alpha}$, it follows that $x = \overline{x}$, i.e., $x = \Pi_R(x - \eps \cdot v - A^\transpose \mu)$. This implies that
$$x - \eps \cdot v - A^\transpose \mu - x = I_n \lambda^+ - I_n \lambda^-$$
where $\lambda^+, \lambda^- \in \mathbb{R}^n$ are nonnegative, and additionally satisfy $\lambda^+_j > 0 \implies x_j = R$, and $\lambda^-_j > 0 \implies x_j = -R$. Here $I_n \in \mathbb{R}^{n \times n}$ denotes the identity matrix. Finally, noting that
$$\widetilde{A} = \begin{bmatrix}
A\\
I_n\\
-I_n
\end{bmatrix}$$
we can rewrite the equation as
$$\eps \cdot v + \widetilde{A}^\transpose \widetilde{\mu} = 0$$
where we let
$$\widetilde{\mu} = \begin{bmatrix}
\mu\\
\lambda^+\\
\lambda^-
\end{bmatrix}$$
Note that we indeed have $\widetilde{\mu} \geq 0$ and $\widetilde{\mu}_i > 0 \implies \widetilde{a}_i^\transpose x \geq \widetilde{b}_i$. In particular, for $i \in [m]$ this follows from $\widetilde{\mu}_i = \mu_i \in \Heaviside(a_i^\transpose x - b_i)$.
\end{proof}

\noindent In the remainder of the proof, we show that $x$ must necessarily be an optimal solution to the optimization problem. We first show that if $x$ is feasible, it is necessarily optimal. We complete the proof by proving that $x$ must necessarily be feasible.

\subsubsection{Analysis: Feasibility implies optimality}

Consider the case where $x$ is feasible, i.e., $Ax \leq b$. We will show that $x$ must be an optimal solution. Since $x \in [-R,R]^n$ and $Ax \leq b$, it follows that $\widetilde{A}x \leq \widetilde{b}$. Thus, the third statement in \Cref{clm:OPT-gate-FP-constraints} yields that
\begin{equation}\label{eq:KKT}
\eps \cdot v + \widetilde{A}^\transpose \widetilde{\mu} = 0
\end{equation}
where $\widetilde\mu \geq 0$ and where
$$\widetilde{\mu}_i > 0 \implies \widetilde{a}_i^\transpose x \geq \widetilde{b}_i \implies \widetilde{a}_i^\transpose x = \widetilde{b}_i.$$
In other words, $\widetilde{\mu}_i$ can only be strictly positive if the $i$th constraint is tight. Furthermore, since $x$ is feasible, the second statement in \Cref{clm:OPT-gate-FP-constraints} yields that $v \in \partial f(x;c)$. As a result, given that $\eps > 0$, the equality \eqref{eq:KKT} can be interpreted as saying that the Karush-Kuhn-Tucker (KKT) conditions hold at point $x$ for the following constrained minimization problem
\begin{equation*}
\begin{split}
\min \quad &f(x;c)\\
\text{ s.t.} \quad & \widetilde{A}x \leq \widetilde{b}\\
\end{split}
\end{equation*}
which is the same as our optimization problem $\mathcal{C}$ in \eqref{eq:OPT-gate-general}. Since, by assumption, $f$ is convex on the feasible domain, the KKT conditions are also sufficient for optimality, and thus $x$ is an optimal solution.

Formally, consider any feasible point $z$. We will show that $f(z) \geq f(x)$. Taking the inner product of \eqref{eq:KKT} with $z-x$ yields
$$\eps \cdot v^\transpose (z-x) = (\widetilde{A}^\transpose \widetilde{\mu})^\transpose (x-z) = \widetilde{\mu}^\transpose \widetilde{A}(x-z) \geq 0$$
since $\widetilde{\mu} \geq 0$ and $\widetilde{\mu}_i > 0 \implies \widetilde{a}_i^\transpose x = \widetilde{b}_i \geq \widetilde{a}_i^\transpose z$. Given that $\eps > 0$, it follows that $v^\transpose (z-x) \geq 0$. Finally, using the fact that $v \in \partial f(x;c)$, i.e., $v$ is a subgradient of $f(\cdot;c)$ at $x$, together with the definition of subgradients, we obtain
$$f(z) \geq f(x) + v^\transpose (z-x) \geq f(x)$$
i.e., $x$ is a global minimum of $f$ on the feasible domain.

\subsubsection{Analysis: Feasibility}

We show that $x$ must necessarily be feasible, i.e., $Ax \leq b$. Assume towards a contradiction that $x$ is not feasible. 

Let $J := \{i \in [m+2n]: \widetilde{a}_i^\transpose x \geq \widetilde{b}_i\}$ and note that $J$ is nonempty, since $x$ is infeasible. By \Cref{clm:OPT-gate-FP-constraints} we have $\widetilde{\mu}_i > 0 \implies \widetilde{a}_i^\transpose x \geq \widetilde{b}_i \implies i \in J$, and thus $\widetilde{\mu}_i = 0$ for all $i \in [m+2n] \setminus J$. In particular, $\widetilde{A}^\transpose \widetilde{\mu} = \widetilde{A}_J^\transpose \widetilde{\mu}_J$, where $\widetilde{A}_J$ denotes the restriction of $\widetilde{A}$ to the subset of rows $J$, and similarly for $\widetilde{\mu}_J$. As a result, we can rewrite \eqref{eq:KKT-in-claim} from \Cref{clm:OPT-gate-FP-constraints} as
\begin{equation}\label{eq:OPT-gate-analysis-feas}
\eps \cdot v + \widetilde{A}_J^\transpose \widetilde{\mu}_J = 0.
\end{equation}

Let $z^*$ be the projection of $x$ onto the convex set $D_J := \{y \in \mathbb{R}^n: \widetilde{A}_J y \leq \widetilde{b}_J\}$, which is nonempty, since the feasible region $D := \{y \in \mathbb{R}^n: \widetilde{A} y \leq \widetilde{b}\}$ is nonempty by assumption. Note that $z^* \neq x$, because if $x \in D_J$, then we would also have $x \in D$ by the definition of $J$.

Let $I := \{i \in J: \widetilde{a}_i^\transpose z^* = \widetilde{b}_i\}$, i.e., the set of all constraints in $J$ that are tight at $z^*$. Note that $z^*$, which is the projection of $x$ onto $D_J$, is also the projection of $x$ onto $D_I := \{y \in \mathbb{R}^n: \widetilde{A}_I y \leq \widetilde{b}_I\}$. Indeed, assume that $z^*$ is not the projection of $x$ onto $D_I$. Then there exists $y \in D_I$ with $(x - z^*)^\transpose (y - z^*) > 0$. But since all constraints in $J \setminus I$ are strictly satisfied at $z^*$, there exists a point $y^*$ on the segment $[z^*,y]$ that satisfies $y^* \in D_J$ and $(x - z^*)^\transpose (y^* - z^*) > 0$. This is a contradiction to the fact that $z^*$ is the projection of $x$ onto $D_J$. Thus, $z^*$ is indeed the projection of $x$ onto $D_I$.

Next, we partition $I$ into two sets: $I_0 := \{i \in I: \widetilde{a}_i^\transpose x = \widetilde{b}_i\}$ and $I_1 := \{i \in I: \widetilde{a}_i^\transpose x > \widetilde{b}_i\}$. Note that by construction $I_0 \cap I_1 = \emptyset$ and $I = I_0 \cup I_1$, since $I \subseteq J$. Furthermore, $I_1 \subseteq [m]$, because we have $x \in [-R,R]^n$, which means that $\widetilde{a}_i^\transpose x \leq \widetilde{b}_i$ for all $i \in [m+2n] \setminus [m]$. Finally, $I_1 \neq \emptyset$. Indeed, if $I_1 = \emptyset$, then $x \in D_I$, which contradicts $z^* \neq x$.

Taking the inner product of \eqref{eq:OPT-gate-analysis-feas} with $z^*-x$ yields
\begin{equation*}
\eps \cdot v^\transpose (z^*-x) = (\widetilde{A}_J^\transpose \widetilde{\mu}_J)^\transpose (x - z^*) = \widetilde{\mu}_J^\transpose \widetilde{A}_J(x - z^*) \geq \widetilde{\mu}_{I_1}^\transpose \widetilde{A}_{I_1}(x - z^*) = \mu_{I_1}^\transpose A_{I_1}(x - z^*) = \sum_{i \in I_1} a_i^\transpose (x - z^*)
\end{equation*}
where we used $\mu_i \geq 0$ and $\widetilde{a}_i^\transpose x \geq \widetilde{b}_i \geq \widetilde{a}_i^\transpose z^*$ for all $i \in J$, $\widetilde{\mu}_{I_1} = \mu_{I_1}$, $\widetilde{A}_{I_1} = A_{I_1}$, and $\mu_i = 1$ for all $i \in I_1 \subseteq [m]$, since $\mu_i \in \Heaviside(a_i^\transpose x - b_i)$. As a result
\begin{equation}\label{eq:OPT-gate-contradiction-key-lemma}
\sum_{i \in I} a_i^\transpose (x - z^*) \leq \eps \cdot v^\transpose (z^*-x) \leq \eps \|v\|_2 \|x - z^*\|_2 \leq \eps K \|x - z^*\|_2 \leq \gamma^* \|x - z^*\|_2
\end{equation}
where we used $\|v\|_2 \leq K$ (\Cref{clm:OPT-gate-FP-constraints}) and $\eps = \gamma^*/K$.

We now show that this is a contradiction to our choice of $\gamma^*$. Let $u := (x - z^*)/\|x-z^*\|_2$ and note that $u$ is well defined, since $z^* \neq x$.  Assume, for now, that $u$ is feasible for optimization problem \eqref{eq:OPT-for-gamma} from \Cref{lem:choice-of-gamma}, which we repeat here for convenience:
\begin{equation*}
\begin{alignedat}{2}
\min \quad &\sum_{i \in I_1} \widetilde{a}_i^\transpose u && \\
\text{ s.t.} \quad & \widetilde{a}_i^\transpose u = 0 &&\forall i \in I_0\\
& \widetilde{a}_i^\transpose u \geq 0 &&\forall i \in I_1\\
& \|u\|_2 = 1 && \\
& u = \sum_{i \in I} \lambda_i \widetilde{a}_i && \\
& \lambda_i \geq 0 &&\forall i \in I
\end{alignedat}
\end{equation*}
In particular, the optimization problem is feasible, and thus by \Cref{lem:choice-of-gamma} its optimal value is (strictly) lower bounded by $\gamma^*$. As a result, $\sum_{i \in I_1} \widetilde{a}_i^\transpose (x-z^*)/\|x-z^*\|_2 > \gamma^*$, which yields $\sum_{i \in I_1} \widetilde{a}_i^\transpose (x-z^*) > \gamma^* \|x-z^*\|_2$. But this is a contradiction to \eqref{eq:OPT-gate-contradiction-key-lemma}.

It remains to show that $u = (x - z^*)/\|x-z^*\|_2$ is indeed feasible for optimization problem \eqref{eq:OPT-for-gamma}. The first two constraints are satisfied, because $\widetilde{a}_i^\transpose x = \widetilde{b}_i = \widetilde{a}_i^\transpose z^*$ for all $i \in I_0$, and $\widetilde{a}_i^\transpose x > \widetilde{b}_i = \widetilde{a}_i^\transpose z^*$ for all $i \in I_1$. Clearly, $\|u\|_2 = 1$. Thus, it remains to prove that there exists $\lambda \geq 0$ such that $x-z^* = \widetilde{A}_I^\transpose \lambda$. Since $z^*$ is the projection of $x$ onto $D_I = \{y \in \mathbb{R}^n: \widetilde{A}_I y \leq \widetilde{b}_I\}$, it follows that $(x-z^*)^\transpose (y-z^*) \leq 0$ for all $y \in D_I$. Given that $\widetilde{A}_I z^* = \widetilde{b}_I$, we thus obtain that $(x-z^*)^\transpose y \leq 0$ for all $y$ with $\widetilde{A}_I y \leq 0$. However, by Farkas' lemma \citep{Boyd2004convex}, we know that exactly one of the following two statements holds:
\begin{enumerate}
	\item $\exists y \in \mathbb{R}^n: \quad \widetilde{A}_I y \leq 0$ and $(x-z^*)^\transpose y > 0$
	\item $\exists \lambda \in \mathbb{R}^{|I|}: \quad \widetilde{A}_I^\transpose \lambda= x-z^*$ and $\lambda \geq 0$
\end{enumerate}
Since we have shown that the first statement does not hold, we deduce that there exists $\lambda \geq 0$ such that $\widetilde{A}_I^\transpose \lambda = x-z^*$, as desired. As a result, $u = (x - z^*)/\|x-z^*\|_2$ is indeed feasible for optimization problem \eqref{eq:OPT-for-gamma}, and the proof is complete.

\subsection{Implicit Functions and Correspondences}\label{sec:implicit}

In this section we will construct \linear arithmetic circuits computing univariate piecewise-linear functions and piecewise-constant correspondences, based on a succinct representation given by a Boolean circuit, built from Boolean \AND, \OR\ and \NOT\ gates. We remark that conceptually, the ideas around the use of the triangle-wave function (see \Cref{def:36-H}) and the bit multiplication (see \Cref{def:bit-mult} and the subsequent discussion) are attributed to \citet{FearnleyGHS23-quadratic}, who developed them, albeit for a different setting.

By interpreting non-negative integers by their binary representation, a Boolean function may be viewed as computing an integer-valued function. For $b=(b_{n-1},\dots,b_0) \in \{0,1\}^n$, we denote by $\bitval_n(b)$ the number $\bitval_n(b)=(b_{n-1},\dots,b_0)_2=\sum_{i=0}^{n-1}b_i2^i$ encoded by $b$ in binary.
\begin{definition}
Let $C$ be a Boolean circuit with $n$ inputs and $n$ outputs, thereby computing a function $C \colon \zo^n \to \zo^n$, and let $1 \leq N < 2^n$ be an integer. We index the inputs and outputs of $C$ with $\{0,\dots,n-1\}$ and we consider $b \in \zo^n$ to represent the integer $\bitval_n(b)$.
The integer-valued function represented by $(C,N)$ is the function $f \colon \{0,1,\dots,N\} \to \{0,1,\dots,N\}$ given by 
 \[
f(\bitval_n(b)) = \min\{N, \bitval_n(C(b))\} = \min\{N, (C(b)_{n-1},\dots,C(x)_0)_2\} \enspace ,
 \]
 for $b \in \zo^n$ such that $\bitval_n(b)\leq N$.
\end{definition}
A natural piecewise-linear function is then given by linear interpolation between function values.
\begin{definition}\label{def:implicit-piecewise-linear}
Let $C$ be a Boolean circuit with $n$ inputs and $n$ outputs, and let $1 \leq N < 2^n$ be an integer. The piecewise-linear function represented by $(C,N)$ is the function $g \colon \RR \to \RR$ given by
\[
g(x) = 
\begin{cases}
f(0) & \text{for } x < 0\\
(1- (x-\floor{x}) ) f(\floor{x}) + ( x-\floor{x} ) f(\floor{x}+1) & \text{for } 0 \leq x < N\\
f(N) & \text{for } N \leq x
\end{cases} \enspace ,
\]
where $f$ is the integer-valued function represented by $(C,N)$.
\end{definition}
In a similar way, a natural piecewise-constant correspondence is given by extending function values to the right, until the next function value is defined.
\begin{definition}\label{def:implicit-piecewise-constant}
  Let $C$ be a Boolean circuit with $n$ inputs and $n$ outputs, and let $1 \leq N < 2^n$ be an integer. The piecewise-constant correspondence represented by $(C,N)$ is the correspondence $g \colon \RR \rightrightarrows \RR$ given by
\[
g(x) = 
\begin{cases}
f(0) & \text{for } x \leq 0\\
[\min(f(x-1),f(x)),\max(f(x-1),f(x))] & \text{for } x \in \{1,2,\dots,N\}\\
f(\floor{x}) & \text{for } x \in [0,N] \setminus \{0,1,\dots,N\}\\
f(N) & \text{for } N < x
\end{cases} \enspace ,
\]
where $f$ is the integer-valued function represented by $(C,N)$.
\end{definition}
An example illustrating \cref{def:implicit-piecewise-linear} and \cref{def:implicit-piecewise-constant} is given in \cref{fig:implicit-examples}. 
\begin{figure}[h]
\begin{subfigure}[b]{0.45\textwidth}
\centering
\begin{tikzpicture}[scale=0.5]
  \draw[->] (-1.5,0) -- (10.5,0) node[right] {};
  \draw[->] (0,0) -- (0,4.5) node[above] {};
  \foreach \x in {-1,0,...,10}
    \draw (\x,1pt) -- (\x,-3pt)
    node[anchor=north] {\x};
  \foreach \y in {1,...,4}
    \draw (1pt,\y) -- (-3pt,\y) 
    node[anchor=east] {\y}; 
  \def\fval{1,3,2,0,1,3,1,2,3,2}
  \foreach \y [count=\x from 0] in \fval{
    \fill (\x,\y) circle (4pt);
  }
  \draw[blue, line width=1pt] (-1.5,1) \foreach \y [count=\x from 0] in \fval {
-- (\x,\y)
} -- (10.5,2);
\end{tikzpicture}
\end{subfigure}
\begin{subfigure}[b]{.45\textwidth}
\centering
\begin{tikzpicture}[scale=0.5]
  \draw[->] (-1.5,0) -- (10.5,0) node[right] {};
  \draw[->] (0,0) -- (0,4.5) node[above] {};  
  \foreach \x in {-1,0,...,10}
    \draw (\x,1pt) -- (\x,-3pt)
    node[anchor=north] {\x};
  \foreach \y in {1,...,4}
    \draw (1pt,\y) -- (-3pt,\y) 
    node[anchor=east] {\y}; 
  \def\fval{1,3,2,0,1,3,1,2,3,2}
  \foreach \y [count=\x from 0] in \fval{
    \fill (\x,\y) circle (4pt);
  }
  \draw[blue, line width=1pt] (-1.5,1) \foreach \y [count=\x from 0] in \fval {
-- (\x,\y) -- (\x+1,\y)
} -- (10.5,2);
\end{tikzpicture}
  \end{subfigure}
\caption{The piecewise-linear function (left) and piecewise-constant correspondence (right) defined implicitly by a Boolean circuit $C$ and $N=9$ computing an integer valued function $f$ as shown with dots.}
\label{fig:implicit-examples}
\end{figure}
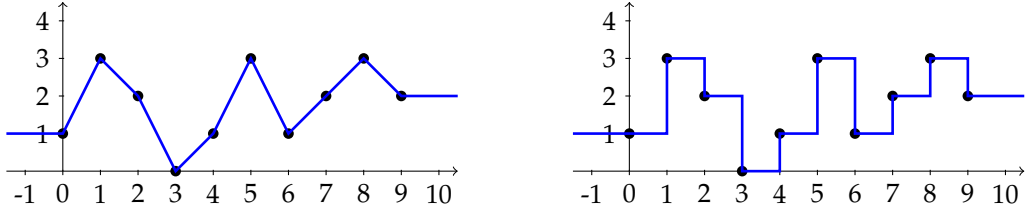

A function that will be used in both constructions is the familiar square-wave function, restricted to a bounded number of periods, as illustrated in \cref{fig:square-wave}.
\begin{definition}\label{def:36-H}
  For a non-negative integer $n$, the correspondence
  $S_n \colon \RR^2 \rightrightarrows [0,1]$, consisting of $2^n+1/2$
  periods of length~$p$, in the interval $[0,(2^n+1/2)p]$, is defined as $\Heaviside(T_n(x,p))$, where $T_n$ is the triangle-wave function defined
  inductively as follows.  For $n=0$, we define
  \[
    T_0(x,p) = \max(\min(x,p/2-x),\min(x-p,3p/2-x)) \enspace ,
  \]
  and for $n>0$ we define $T_n(x,p)=T_{n-1}(\min(x,(2^n+1/2)p-x),p)$.
\end{definition}

\begin{figure}[h]
 \begin{subfigure}[b]{.45\textwidth}
  \centering
\begin{tikzpicture}[scale=1]
  \draw[->] (-0.5,0) -- (3.5,0) node[right] {};
  \draw[->] (0,0) -- (0,0.75) node[above] {};

  \foreach \lab [count=\x from 0] in {$\vphantom{\frac{1}{2}}0$,$\frac{1}{2}p$,$\vphantom{\frac{1}{2}}p$,$\frac{3}{2}p$}
    \draw ({\x},1pt) -- ({\x},-3pt)
    node[anchor=north] {\lab};
    
  \foreach \y in {0.5}
    \draw (1pt,\y) -- (-3pt,\y) 
    node[anchor=east] {$\frac{1}{4}p$};
  
  \draw[blue, line width=1pt] (-0.5,-0.5)--(0.5,0.5)--(1.5,-0.5)--(2.5,0.5)--(3.5,-0.5);
\end{tikzpicture} 
\end{subfigure}
\begin{subfigure}[b]{.45\textwidth}
  \centering
\begin{tikzpicture}[scale=0.5]
  \draw[->] (-1.5,0) -- (10.5,0) node[right] {};
  \draw[->] (0,-0.5) -- (0,1.5) node[above] {};
  \foreach \x in {1,2,...,4}
    \draw ({2*\x},1pt) -- ({2*\x},-3pt)
    node[anchor=north] {\x};
  \foreach \y in {1}
    \draw (1pt,\y) -- (-3pt,\y) 
    node[anchor=east] {\y};
  
  \draw[blue, line width=1pt] (-1.5,0) \foreach \x in {0,1,...,8} {--({\x},{mod(\x,2)})--({\x},{mod(\x+1,2)})--({\x+1},{mod(\x+1,2)})} -- (9,0) -- (10.5,0);
\end{tikzpicture} 
\end{subfigure}
\caption{The function $T_0(x,p)$ used to define $S_0(x,p)$ (left)
  and the function $S_2(x,1)$ (right) giving a square wave with 4.5 periods of length~$1$ in the interval $[0,4.5]$.}
\label{fig:square-wave}
\end{figure}
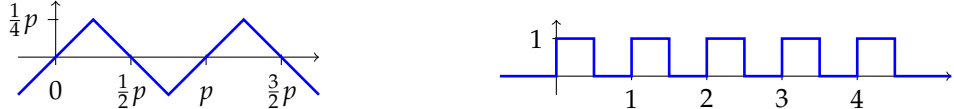

Boolean operations may be viewed as special cases of \linear functions, namely the Boolean functions $\AND(a,b)$, $\OR(a,b)$, and
$\NOT(a)$ correspond to the \linear functions $\min(a,b)$,
$\max(a,b)$, and $1-a$, for $a,b \in \zo$. The following is thus
immediate.
\begin{lemma}\label{lem:boolean-circuit-to-linear-circuit}
Let $C$ be a Boolean circuit with $n$ inputs and $m$ outputs. Then there is a \linear arithmetic circuit with $n$ inputs and $m$ outputs $C'$ with the same number of gates as $C$ such that $C(b)=C'(b)$ for all $b \in \zo^n$.
\end{lemma}
In fact, $\AND(a,b)$ and thus $\min(a,b)$ may also be viewed as the
multiplication operation for $a,b \in \zo$. Even more generally,
$\min(a,b)$ may be viewed as a multiplication when just $b \in
\zo$ and $a \in [0,1]$. As long as we ultimately only use this tool to prove PPAD-membership, we can make use of the \linoptgate to construct a stronger multiplication gadget that does not require an a priori bound on the operand~$a$.

\medskip

\begin{definition}\label{def:bit-mult}
  Define the \pseudog $\bitmult$ with two inputs
  and one output by $\bitmult(a,b) = \Heaviside(b-\frac{1}{2})a$.
\end{definition}
Note that $\bitmult$ can indeed be computed by a \pseudog by \cref{lem:mult-by-Heaviside} (which uses the \linoptgate internally). From the definition of the Heaviside function
$\Heaviside$ the following is then immediate.
\begin{lemma}\label{lem:bit-mult}
  We have $\bitmult(a,b)=ab$, for any $a \in \mathbb{R}$ and any $b \in \{0,1\}$.
\end{lemma}

We will show below that also the piecewise-linear functions and
piecewise-constant correspondences represented by Boolean circuits may
be computed by \pseudogs. In fact, for the case of
piecewise-constant correspondences we will show a considerably
stronger statement allowing the scaling of the input and output by a
variable.

An important component of both constructions is \emph{bit
  extraction}. Informally, bit extraction of a non-negative $x \geq 0$
refers to computing the binary expansion of $\floor{x}$. This is
clearly an inherently discontinuous operation, but can nevertheless be
circumvented in our applications. The bit extraction procedure we
define below will actually take 2~inputs $x$ and $y$, and the
intention is to perform bit extraction on the number $x/y$, when
$0 < y \leq 1$.

\begin{definition}\label{def:bitextract}
  For a positive integer $n$, we inductively define a \linear pseudo-circuit $\bitextract_n$ with two inputs and $n$ outputs as
  follows.  For $n=1$, we define
  $\bitextract_1(x,y) = \Heaviside(x-y)$, and for $n>1$ we define
  $\bitextract_n(x,y) = (b_{n-1},\dots,b_0)$ where
  $b_{n-1} = \Heaviside(x-2^{n-1}y)$ and
  $(b_{n-2}, \dots, b_0) = \bitextract_{n-1}(x-\bitmult(2^{n-1}y,b_{n-1}), y)$.
\end{definition}
\begin{lemma}\label{lem:bit-extraction}
  Let $n\geq 1$, assume $0<y\leq 1$, and let
  $\bitextract_n(x,y)=b=(b_{n-1},\dots,b_0)$. Then, whenever
  $x/y \in [0,2^n] \setminus \{1,2,3,\dots,2^n\}$ we have
  $b_i \in \zo$, for all~$i$ and $\bitval_n(b)=\sum_{i=0}^{n-1}b_i 2^i = \floor{x/y}$.
\end{lemma}
\begin{proof}
  Assume $x/y \in [0,2^n] \setminus \{1,2,3,\dots,2^n\}$. Since in
  particular $x/y \neq 2^{n-1}$, by definition of the Heaviside
  function $\Heaviside$ we have that $b_{n-1}=0$ when $x/y < 2^{n-1}$ and
  $b_{n-1}=1$ when $x/y>2^{n-1}$. This shows the case of $n=1$. For
  $n>1$ we have
  $x/y - b_{n-1}2^{n-1} \in [0,2^{n-1}] \setminus
  \{1,2,\dots,2^{n-1}\}$, and by induction we have
  $\floor{x/y - b_{n-1}2^{n-1}} = \sum_{i=0}^{n-2}b_i2^i$. The result
  now follows since
  $\floor{x/y - b_{n-1}2^{n-1}} = \floor{x/y} - b_{n-1}2^{n-1}$, for
  $b_{n-1} \in \zo$.
\end{proof}
\begin{remark}
  The output of $\bitextract_n$ is of course well-defined for any value
  of $x$ and $y$. In particular we can observe that
  $(\bitextract_n(x,0))_i=\Heaviside(x)$, for all $i$.
\end{remark}

As an extension of \Cref{def:bit-mult} and
\Cref{lem:bit-mult} we may in a \linear arithmetic circuit \emph{multiply}
a number given in binary encoding by variables with another
variable. This means that we may consider the output of a \linear circuit, obtained from a Boolean circuit by
\Cref{lem:boolean-circuit-to-linear-circuit}, as representing a
number given in binary encoding and then multiply that by a variable.
\begin{definition}\label{def:bit-mult2}
  For an integer~$n>0$, define the \linear pseudo-circuit
  $\bitmult_n$ having~$n+1$ inputs and one output by
  $\bitmult_n(a,b_{n-1},\dots,b_0) =
  \sum_{i=0}^{n-1}\bitmult(a,b_i)2^i$.
\end{definition}
Using \Cref{lem:bit-mult} we then have the following.
\begin{lemma}\label{lem:bit-mult2}
  For any positive integer~$n$, $b\in \zo^n$, and any $a \in \mathbb{R}$ we have
  $\bitmult_n(a,b_{n-1},\dots,b_0)=a\cdot\bitval_n(b)$.
\end{lemma}
For $b \in \zo^n$ we shall for brevity also simply write
$\bitmult(a,b) = \bitmult_n(a,b_{n-1},\dots,b_0)$.

\subsubsection{Piecewise-linear functions}
We will now develop a construction of a \linear pseudo-circuit
computing the piecewise-linear function $g$ represented by $(C,N)$, as
defined in \Cref{def:implicit-piecewise-linear}, for a
circuit $C$ with $n$ inputs and $n$ outputs and an integer $N$
satisfying $1 \leq N < 2^n$. The general idea is to perform
bit extraction, evaluate the integer-valued function $f$ represented
by $(C,N)$ on two inputs, and then perform linear interpolation. The
immediate obstacle is that bit extraction applied to $x \in [0,2^n]$
\emph{fails}, i.e., does not guarantee a Boolean output, when
$x \in \{1,\dots,2^n\}$. To overcome this obstacle we perform this
procedure on two different translations of the input, and use the
square wave function to switch between the two, thereby masking out
erroneous values. Let in the following
$G=[0,2^n]\setminus \{1,\dots,2^n\}$ be the set of points where bit
extraction succeeds.

For simplicity, we first give a description of the computation without
referring to Boolean circuits.  Suppose $x \in [0,N]$.  We consider
the following two cases, each with a special case in an endpoint of
the interval~$[0,N]$.
\begin{enumerate}
\item Assume $x \in [\floor{x},\floor{x}+\frac{1}{2}]$ and assume
  first also $x\leq N -\frac{1}{2}$. Let $w_1 = x + \frac{1}{4}$ and
  $w_2=x+\frac{5}{4}$. Note that $w_1,w_2 \in G$ and that
  $\floor{w_1}=\floor{x}$ and $\floor{w_2}=\floor{x}+1$. It thus
  follows that $g(x)=g_1(x)$, where
  \[
    g_1(x) = (\floor{w_2}-x)f(\floor{w_1})+(x-\floor{w_1})f(\floor{w_2}) \enspace .
  \]
  To also handle the special case of $x=N$, we proceed as
  follows. First, let $\delta_1 = \Heaviside(x-(N-\frac{1}{4}))$ and
  define $w'_1 = x+\frac{1}{4} - \delta_1$ and
  $w'_2 = x+\frac{5}{4} - \delta_1$. In case $x=N$ we have
  $\delta_1 = 1$ and thus $w'_1 = x - \frac{3}{4}$,
  $w'_2 = w_1 = x + \frac{1}{4}$, and $w'_1,w'_2 \in G$. We then have
  $g(x)=g'_1(x)$, where
  \[
    g_1'(x) = (\floor{w'_2}-x)f(\floor{w_1})+(x-\floor{w'_1})f(\floor{w'_2}) \enspace .
  \]
  Note that $f$ is evaluated on $\floor{w_1}$ and $\floor{w'_2}$ (and
  not $\floor{w'_1}$ and $\floor{w'_2}$), which in the case of $x=N$
  both equal $N$, thereby satisfying
  \Cref{def:implicit-piecewise-linear}.

  Finally, note that in the general case, when
  $x \leq N - \frac{1}{2}$, we have $\delta_1=0$, and thus $w'_1=w_1$
  and $w'_2=w_2$ which means that $g'_1(x)=g_1(x)$.
\item Assume $x \in [\ceil{x}-\frac{1}{2},\ceil{x}]$ and assume first also
  $x \geq \frac{1}{2}$. Let $w_3=x-\frac{1}{4}$ and
  $w_4 = x+\frac{3}{4}$. Note that $w_3,w_4 \in G$. We now have two
  sub-cases. First, assume that $x < \ceil{x}$. Then
  $\floor{w_3}=\floor{x}$ and $\floor{w_4}=\floor{x}+1$. Next assume
  that $x=\ceil{x}$. Then $\floor{w_3}=\floor{x}-1$ and $\floor{w_4}=\floor{x}$.
  In both sub-cases it thus follows that $g(x)=g_2(x)$, where
  \[
    g_2(x) = (\floor{w_4}-x)f(\floor{w_3})+(x-\floor{w_3})f(\floor{w_4}) \enspace .
  \]
  To also handle the special case of $x=0$, we proceed in a similar
  way as the case of $x=N$ above. First, let
  $\delta_2=\Heaviside(\frac{1}{4}-x)$ and define
  $w'_3=x-\frac{1}{4}+\delta_2$ and $w'_4 = x+\frac{3}{4} + \delta_2$. In case
  $x=0$ we have $\delta_2=1$ and thus $w'_3=w_4=x+\frac{3}{4}$,
  $w'_4=x+\frac{7}{4}$ and $w'_3,w'_4 \in G$. We then have $g(x)=g'_2(x)$, where
  \[
    g'_2(x) = (\floor{w'_4}-x)f(\floor{w'_3})+(x-\floor{w'_3})f(\floor{w_4}) \enspace .
  \]
  Note that $f$ is evaluated on $\floor{w'_3}$ and $\floor{w_4}$ (and
  not $\floor{w'_3}$ and $\floor{w'_4}$), which in the case of $x=0$
  both equals $0$, thereby satisfying
  \cref{def:implicit-piecewise-linear}.

  Finally, note that in the general case, when
  $x \geq \frac{1}{2}$, we have $\delta_2=0$, and thus $w'_3=w_3$
  and $w'_4=w_4$ which means that $g'_2(x)=g_2(x)$.
\end{enumerate}

We will now simply use the square wave function to select between the
two cases.  Namely, from the definition of $S_n$ we have that
$g(x) = g'_2(x) + S_n(x,1)(g'_1(x)-g'_2(x))$. Since $S_n(x,1)$ is defined
as the composition $H(T_n(x,1))$ of the Heaviside function with the
\linear function $T_n(x,1)$, we can compute the term
$S_n(x,1)(g'_1(x)-g'_2(x))$ using the \linoptgate by \cref{lem:mult-by-Heaviside}. We now present the
general construction of the \linear pseudo-circuit.
\begin{proposition}\label{prop:implicit-function}
  Let $C$ be a Boolean circuit with $n$ inputs and $n$ outputs, and
  let $N<2^n$. Then, for the purpose of proving PPAD-membership, we can construct a \pseudog computing the
  piecewise-linear function $g$ represented by $(C,N)$.
\end{proposition}
\begin{proof}
  Let $C'$ be the \linear arithmetic circuit obtained from $C$ by
  \cref{lem:boolean-circuit-to-linear-circuit}. We construct a
  \linear pseudo-circuit computing $g(x)$, operating as follows.
\begin{enumerate}
\item Compute $x' = \max(0,\min(N,x))$. Thus $x' \in [0,N]$ and $g(x)=g(x')$.
\item Compute $\delta_1=H(x'-(N-\frac{1}{4}))$ and $\delta_2=\Heaviside(\frac{1}{4}-x')$ using \cref{lem:mult-by-Heaviside}.
\item Compute $w_1=x'+\frac{1}{4}$, $w'_1=x'+\frac{1}{4}-\delta_1$, and $w'_2=x'+\frac{5}{4}-\delta_1$.
\item Compute $\mathbf{b}_1 = \bitextract_n(w_1,1)$, $\mathbf{b}'_1 = \bitextract_n(w'_1,1)$, and $\mathbf{b}'_2 = \bitextract_n(w'_2,1)$.
\item Compute $g'_1 = \bitmult(\bitval(\mathbf{b}'_2)-x',C'(\mathbf{b}_1)) + \bitmult(x'-\bitval(\mathbf{b}'_1),C'(\mathbf{b}'_2))$.
\item Compute $w_4=x'+\frac{3}{4}$, $w'_3=x'-\frac{1}{4}+\delta_2$, and $w'_4=x'+\frac{3}{4}+\delta_2$.
\item Compute $\mathbf{b}_4 = \bitextract_n(w_4,1)$, $\mathbf{b}'_3 = \bitextract_n(w'_3,1)$, and $\mathbf{b}'_4 = \bitextract_n(w'_4,1)$.
\item Compute $g'_2 = \bitmult(\bitval(\mathbf{b}'_4)-x',C'(\mathbf{b}'_3)) + \bitmult(x'-\bitval(\mathbf{b}'_3),C'(\mathbf{b}_4))$.
\item Output $g'_2 + S_n(x',1)(g'_1-g'_2)$, where the last term is computed using \cref{lem:mult-by-Heaviside}.
\end{enumerate}
The correctness of the circuit follows from the preceding case analysis.
\end{proof}

\subsubsection{Piecewise-linear correspondences}
We will now, in a similar way, develop a construction of a \linear pseudocircuit computing the piecewise-constant correspondence $g$
represented by $(C,N)$, as defined in
\cref{def:implicit-piecewise-constant}, for a circuit $C$
with $n$ inputs and $n$ outputs and an integer $N$ satisfying
$1 \leq N < 2^n$. Again the general idea is to perform bit
extraction and evaluate the integer-valued function $f$ represented by
$(C,N)$ on two inputs. Since the function we compute is piecewise-constant, 
we are now able to also scale the output by a variable. Let
again $G=[0,2^n]\setminus \{1,\dots,2^n\}$ be the set of points where
bit extraction succeeds.

We will be able to compute the correspondence $z \cdot g(x/y)$, but for
simplicity we first describe the computation of $g(x)$, also without
referring to Boolean circuits. Thus, suppose now $x \in [0,N]$.  We
consider the same cases as above.
\begin{enumerate}
\item Assume $x \in [\floor{x},\floor{x}+\frac{1}{2}]$ . Let
  $w_1 = x + \frac{1}{4}$, and note that $w_1 \in G$ and that
  $\floor{w_1}=\floor{x}$. Let $g_1(x)=f(\floor{w_1})$. It follows
  that $g_1(x)=f(\floor{x})$.
\item Assume $x \in [\ceil{x}-\frac{1}{2},\ceil{x}]$ and assume first
  also $x \geq \frac{1}{2}$. Let $w_2=x-\frac{1}{4}$, and note that
  $w_2 \in G$. Let $g_2(x)=f(\floor{w_2})$.  We now have two
  subcases. First, assume that $x < \ceil{x}$. Then
  $\floor{w_2}=\floor{x}$ and thus $g_2(x)=f(\floor{x})$. Next assume
  that $x=\ceil{x}$. Then $\floor{w_2}=\floor{x}-1$ and thus
  $g_2(x)=f(\floor{x}-1)$.  To also handle the special case of $x=0$,
  we proceed as follows. First, let $\delta=\Heaviside(\frac{1}{4}-x)$
  and define $w'_2=x-\frac{1}{4}+\delta$. Let
  $g'_2(x)=f(\floor{w'_2})$.  In case $x=0$ we have $\delta=1$ and
  thus $w'_2 = x + \frac{3}{4}$ and $w'_2 \in G$. We thus have
  $\floor{w'_2}=\floor{x}$ and thus $g'_2(x)=f(\floor{x})$.

  Finally, note that in the general case, when $x \geq \frac{1}{2}$,
  we have $\delta=0$, and thus $w'_2=w_2$ which means that $g'_2(x)=g_2(x)$.
\end{enumerate}
We will now simply use the square wave function to select between the
two cases, letting $g(x) = g'_2(x) + S_n(x,1)(g_1(x)-g'_2(x))$.

This will also be what gives the correct result in case
$x \in \{1,\dots,N\}$ where we have $g_1(x)=f(x)$ and $g'_2(x)=f(x-1)$
and selecting using the square-wave function yields gives any convex
combination of these values. For the special case of $x=0$ we also
obtain the correct result, since $g_1(0)=g'_2(0)=f(0)$.

Since $S_n(x,1)$ is defined as the composition $H(T_n(x,1))$ of the
Heaviside function with the \linear function $T_n(x,1)$, we can
compute the term $S_n(x,1)(g_1(x)-g'_2(x))$ using \cref{lem:mult-by-Heaviside}. We
now present the general construction of the \pseudog.

\begin{proposition}\label{prop:implicit-correspondence}
  Let $C$ be a Boolean circuit with $n$ inputs and $n$ outputs, and
  let $N<2^n$. Then, for the purpose of proving PPAD-membership, we can construct a \pseudog computing a function
  $h(x,y,z)$ such that when $0<y\leq 1$ we have $h(x,y,z) = z \cdot g(x/y)$, where
  $g$ is the piecewise-constant function represented by $(C,N)$.
\end{proposition}
\begin{proof}
  Let $C'$ be the \linear arithmetic circuit obtained from $C$ by
  \cref{lem:boolean-circuit-to-linear-circuit}. We construct a
  \linear pseudo-circuit computing $h(x,y,z)$, operating as follows.
  \begin{enumerate}
\item Compute $x' = \max(0,\min(yN,x))$. Thus when $y>0$ we have $\frac{x'}{y} \in [0,N]$ and $g(\frac{x}{y})=g(\frac{x'}{y})$.
\item Compute $\delta=\Heaviside(\frac{1}{4}y-x')y$ using \cref{lem:mult-by-Heaviside}.
\item Compute $w_1=x'+\frac{1}{4}y$ and $\mathbf{b}_1 = \bitextract_n(w_1,y)$.
\item Compute $g_1 = \bitmult(z,C'(\mathbf{b}_1))$.
\item Compute $w'_2=x'-\frac{1}{4}y+\delta$ and $\mathbf{b}'_2 = \bitextract_n(w'_2,y)$.
\item Compute $g'_2 = \bitmult(z,C'(\mathbf{b}'_2))$.
\item Output $g'_2 + S_n(x',y)(g_1-g'_2)$, where the last term is computed using \cref{lem:mult-by-Heaviside}.
\end{enumerate}
The correctness of the circuit follows from the preceding case
analysis, using $\frac{x}{y}$ in place of $x$ in the analysis.
\end{proof}

\section{Concave Games, Nash Equilibria and Other Equilibrium Notions} \label{sec:games1}

The first application of our \linoptgate will be to the domain of strategic games for computing Nash equilibria, as well as other related equilibrium notions. To demonstrate the use of our technique, we first ``warm up'' with two basic applications, namely (a) computing Nash equilibria in \emph{bimatrix games} (\Cref{sec:bimatrix-games}) and (b) computing fixed point solutions of \emph{digraph threshold games} (\Cref{sec:threshold-games}). These two applications will be used as examples to demonstrate the use of our \linoptgate via linear programs, similar to linear program $\mathcal{P}$ (see Program~\eqref{eq:OPT-gate-linear} in \Cref{sec:lin-opt-gate}) or via feasibility programs, similar to feasibility program $\mathcal{Q}$ (see Program~\eqref{eq:feasibility-general} in \Cref{sec:lin-opt-gate}) respectively. 

Then in \Cref{sec:LBRO-games} we define a very large class of games, which we refer to as \emph{games with \linear best response oracles (\lbro games)}. This class has a technical definition, and essentially captures every game in which the best response of the agent can be computed by a \linear arithmetic circuit (or, to be more precise, a \pseudog). The PPAD-membership of \lbro games is then immediate from our results in \Cref{sec:lin-opt-gate} via the use of our \linoptgate. This provides us with a very strong unified tool to show the PPAD-membership of a very large class of games, by merely showing that these games are \lbro games.

More precisely, in \Cref{sec:generalized-concave-games} we consider \emph{generalized equilibria} in \emph{concave games}. Our PPAD-membership proof will capture a large class of concave games that subsumes bimatrix games, polymatrix games and more generally a class of games that we call \linear succinct games (recall \Cref{rem:linear-vs-linear} for the interpretation of the term \linear), as well as bilinear games, general threshold games, and essentially all concave games for which Nash equilibria in rational numbers exist. In \Cref{sec:personalized}, we consider an alternative equilibrium notion, due to \citet{SICOMP:KintaliPRST13}, called \emph{personalized equilibrium}, and in \Cref{sec:proper} we consider $\varepsilon$-proper equilibria in \linear succinct games (including polymatrix games), a well-known equilibrium notion due to \citet{IJGT:Myerson78}. 

With the \linoptgate at hand, the proofs that we will develop in this section will be conceptually very simple, basically mimicking the simplest proofs of existence for these settings. More precisely, all of our PPAD-membership results for this section are obtained as special cases of our main theorem for \lbro games (see \Cref{thm:LBRO-equilibria}). In fact, we also employ the same theorem to obtain PPAD-membership results for congestion games in \Cref{sec:congestion-games}.

\subsection{Warm-up: Nash Equilibria in Bimatrix Games via Linear Programs}\label{sec:bimatrix-games}

We define bimatrix games below.

\begin{definition}[Bimatrix Game]\label{def:bimatrix-game}
A \emph{bimatrix game} is a strategic (normal form) game played between two players. Each player $i \in \{1,2\}$ has a finite set of pure strategies $S_i = \{1, \ldots, m_i\}$, and a payoff function $u_i = S_1 \times S_2 \rightarrow \mathbb{R}$, mapping pairs of pure strategies to payoffs. The players can also randomize over their pure strategies. A \emph{mixed strategy} $\mathbf{x}_i$ of player $i$ is a probability distribution over pure strategies in $S_i$. The domain $\Sigma_i$ of mixed strategies for player $i$ is the $(m_i-1)$-dimensional unit simplex, i.e., $\Sigma_i:=\{y \in \mathbb{R}_{\geq 0}^{m_i}\colon \sum_{j=1}^{m_i} y_j = 1\}$. Given a mixed strategy $\mathbf{x}_i$ for player $i$, the $j$th coordinate $x_{i,j}$ denotes the probability that player $i$ is playing pure strategy $j$.
The expected payoff $\tilde{u}_i(\mathbf{x}_1,\mathbf{x}_2)$ of player $i$ given a pair of mixed strategies $(\mathbf{x}_1,\mathbf{x}_2)$ is defined as:
\begin{equation}\label{eq:bimatrix-payoff}
\tilde{u}_i(\mathbf{x}_1,\mathbf{x}_2) = \sum_{j_1=1}^{m_1}\sum_{j_2=1}^{m_2}x_{1,j_1}\cdot x_{2,j_2}\cdot u_i(j_1,j_2)
\end{equation}
\end{definition}

\noindent The name ``bimatrix games'' comes from the fact that for $2$ players, the payoffs can be expressed via two $m_1 \times m_2$ matrices, one for each player. 

\begin{definition}[Nash equilibrium]\label{def:Nash-eq-bimatrix}
A pair of mixed strategies $(\mathbf{x}_1,\mathbf{x}_2) \in \Sigma_1 \times \Sigma_2$ is a (mixed) \emph{Nash equilibrium} of a bimatrix game if for any player $i \in \{1,2\}$ and any mixed strategy $\mathbf{x}'_i \in \Sigma_i$ of that player, it holds that $\tilde{u}_i(\mathbf{x}_i,\mathbf{x}_{3-i}) \geq \tilde{u}_i(\mathbf{x}'_i,\mathbf{x}_{3-i})$, i.e., no player can improve their expected payoff by unilaterally deviating to any other mixed strategy.
\end{definition}

\paragraph{Finding Nash equilibria.} We are interested in finding a Nash equilibrium of a bimatrix game, when the payoff functions $u_i$ for each player $i \in \{1,2\}$ are given explicitly as rational numbers. This is a total search problem, since a solution is always guaranteed to exist by Nash's Theorem \citep{Nash50}, the proof of which employs the Kakutani fixed point theorem \citep{kakutani1941generalization}. \citet{JCSS:Papadimitriou1994} showed that the problem lies in PPAD, which also implies the existence of rational solutions, i.e., Nash equilibria in which the mixed strategies of the players are rational numbers. Note that \citeauthor{JCSS:Papadimitriou1994}'s proof essentially appeals to an alternative proof of Nash equilibrium existence due to \citet{cottle1968complementarity} that formulates the problem as an LCP (see also \citep{lemke1964equilibrium}); we discussed the general challenges of using this approach in \Cref{sec:Other-Approaches}. Indeed, even in the simple case of bimatrix games, an argument is required against ray termination.\footnote{Such an argument can be found in \citep{cottle2009linear} and also implicitly in \citep{cottle1968complementarity} and \citep{lemke1965bimatrix}.} In contrast, our \linoptgate allows us to organically obtain PPAD-membership from the standard, textbook existence proof of \citet{Nash50}, without any further arguments. 

In a celebrated paper, \citet{chen2009settling} showed that the problem is in fact PPAD-complete (and the hardness holds even when one aims to find approximate equilibria). Of course, \citeauthor{Nash50}'s theorem applies to strategic games beyond bimatrix games, i.e., games with more than $2$ players. In that case however, it was known already from \citeauthor{Nash50}'s original work that with $3$ or more players, there are games for which all Nash equilibria are irrational. For those games, it has been shown that the problem of computing exact Nash equilibria is FIXP-complete \citep{etessami2010complexity}.

 \medskip

\noindent We will demonstrate how our \linoptgate, in particular through its capability of solving linear programs of the form $\mathcal{P}$ (see Program~\eqref{eq:OPT-gate-linear} in \Cref{sec:lin-opt-gate}), can be used to rather straightforwardly show that computing Nash equilibria in bimatrix games is in PPAD. For ease of notation, let $-i=3-i$. We start by observing that given a mixed strategy $\mathbf{x}_{-i}$ for the other player, player $i$ can find an optimal (i.e., an expected payoff-maximizing) mixed strategy $\mathbf{y}_i$ via the solution to the following linear program:

\begin{center}\underline{Linear Program $\mathcal{P}$}
\begin{equation*}
\begin{aligned}
\mbox{maximize}\quad & \tilde{u}_i(\mathbf{y}_i,\mathbf{x}_{-i})\\
\mbox{subject to}\quad & \sum_{j=1}^{m_i} y_{ij} = 1\\
& y_{ij} \geq 0 \qquad j=1,\dots,m_i
\end{aligned}
\end{equation*}
\end{center}
``Spelling out'' the objective function by substituting \Cref{eq:bimatrix-payoff} into it, it becomes:
\begin{equation}\label{eq:bimatrix-payoff-with-y}
\tilde{u}_i(\mathbf{y}_{i},\mathbf{x}_{-i}) = \sum_{j_i=1}^{m_i}\sum_{j_{-i}=1}^{m_{-i}}y_{i,j_i}\cdot x_{-i,j_{-i}}\cdot u_i(j_i,j_{-i})
\end{equation}

Notice that \Cref{eq:bimatrix-payoff-with-y} is linear in $\mathbf{y}_i$, and hence $\mathcal{P}$ is indeed a linear program. \medskip

\noindent Now, to prove membership of finding Nash equilibria in bimatrix games in PPAD, all we need to do is to construct a function whose fixed points will be the Nash equilibria of the bimatrix game, and argue that this function can be computed by a \linear arithmetic circuit containing \linoptgates. This is essentially the straightforward translation of the proof via the Kakutani fixed point theorem to a PPAD-membership proof.

\begin{theorem}\label{thm:bimatrix-nash-ppad}
Computing a Nash equilibrium of a bimatrix game is in PPAD.
\end{theorem}

\begin{proof}
We construct a \linear arithmetic circuit with \linoptgates $F: \Sigma_1 \times \Sigma_2 \rightarrow \Sigma_1 \times \Sigma_2$, in which $F_i(\mathbf{x}_1,\mathbf{x}_2)$ is an optimal solution of linear program $\mathcal{P}$ for player $i$. Since each linear program computes a best-response mixed strategy for the corresponding player, the resulting pair of mixed strategies at a fixed point of $F$ is necessarily a Nash equilibrium. It remains to show that the linear program $\mathcal{P}$ can be computed by our \linoptgate. Indeed, all the constraints are linear functions of the variables $\mathbf{y}_i$, and the \circparams $\mathbf{x}_{-i}$ do not appear in the constraints. The feasible domain is non-empty and bounded, and the gradient of the objective function (with respect to the variables $\mathbf{y}_i$) is linear in the variables $\mathbf{y}_i$ and \circparams $\mathbf{x}_{-i}$ and hence can be computed by a \linear arithmetic circuit. Thus, by \cref{thm:linoptgate} there is an equivalent \linear arithmetic circuit that does not use \linoptgates and the problem lies in PPAD by \cref{rem:linoptgate}.
\end{proof}
\noindent We remark that in the proofs of our results from now on, we will not explicitly reference \Cref{thm:linoptgate} and \Cref{rem:linoptgate}, as their application will be straightforward. \medskip

\noindent Before we conclude the section, we remark that for (general) games of $3$ or more players, the objective function would be a polynomial of degree at least $3$ in the variables and the \circparams, and its gradient could not be given by a \pseudog. This is of course not a coincidence, as, as we said earlier, games with $3$ or more players may only have irrational Nash equilibria.

\subsection{Warm-up: Equilibria of Digraph Threshold Games via Feasibility Programs}\label{sec:threshold-games}

Our next ``warm-up'' application is that of finding equilibria in \emph{digraph threshold games}. A digraph threshold game is played on a directed graph between $n$ players that choose a real number in $[0,1]$, and their payoffs depend on the chosen numbers of their incoming neighbors and a threshold. Digraph Threshold games were defined by \citet{papadimitriou2021public}\footnote{\citet{papadimitriou2021public} used the term ``threshold games'' to refer to these games. However, this term is ambiguous, as it has also been used to refer to a specific class of congestion games, see \citep{ackermann2008impact}. In fact, both variants have been used in relation to the general class of public goods games \citep{papadimitriou2021public,klimm2023complexity}, which adds to the potential for confusion. For this reason, we have added the term ``digraph'' in front of the name for distinction.} as a tool for proving PPAD-hardness of \emph{public goods games} on directed graphs, and were later on used to establish the PPAD-hardness of other problems, e.g., see \citep{chen2021throttling,chen2022computational}\footnote{To be precise, what is used for PPAD-hardness is their approximate version, which was shown to be PPAD-complete in \citep{papadimitriou2021public}. Here we show the membership in PPAD for the exact version, strengthening the membership result.}.

\begin{definition}[Digraph Threshold Game]\label{def:threshold-game} A digraph threshold game $\mathcal{G}(V,E,t)$ is defined on a directed graph $G = (V,E)$ with $|V|=n$, and $t \in (0,1)$ is a threshold. The nodes of $G$ correspond to players, and each player $i \in V$ chooses a strategy $x_i \in [0,1]$. Let $\mathbf{x} = (x_1, \ldots, x_n)$ be the vector of (pure) strategies that we will refer to as a (pure) strategy profile. Let $N(i)$ denote the (in-)\emph{neighborhood} of player $i \in V$, i.e., the set consisting of all the players that have arcs to player $i$, i.e., $N(i) = \{j: (j,i) \in E\}$. A strategy profile $\mathbf{x}$ is an equilibrium if it satisfies
\begin{equation}\label{eq:threshold-game}
x_i =  \begin{cases} 
      0 & \text{if } \sum_{j \in N(i)}x_j > t \\
      1 & \text{if } \sum_{j \in N(i)}x_j < t \\
      [0,1] & \text{if } \sum_{j \in N(i)}x_j = t
   \end{cases}
\end{equation}
where in the above expression we have abused notation, using $x_i = [0,1]$ to denote $x_i \in [0,1]$.
\end{definition}

\noindent While we did use the terminology of games, digraph threshold games can alternatively be viewed as finding fixed point solutions to a set of constraints. This is why we chose this application to demonstrate the use of our feasibility programs $\mathcal{Q}$ (see Program~\eqref{eq:feasibility-general} in \Cref{sec:lin-opt-gate}) which we can obtain from our \linoptgate. Indeed, \Cref{eq:threshold-game} can straightforwardly be written as a feasibility program $\mathcal{Q}$ as follows:\footnote{The perceptive reader might observe that digraph threshold games look a lot like the Heaviside function introduced in \cref{ex:Heaviside}. Indeed, one could obtain the PPAD-membership of digraph threshold games as fixed points of functions containing only Heaviside \pseudogs rather than \linoptgates. We elected to use the feasibility program formulation instead, to introduce the reader to its use in light of further applications to come later. Also note that the feasibility program $\mathcal{Q}$ contains equalities; these can easily be modified to inequalities by splitting the constraints into two.}

\begin{center}\underline{Feasibility Program $\mathcal{Q}$}\end{center}
\begin{equation*}
\begin{aligned}
\sum_{j \in N(i)}x_j > t \Rightarrow y_i = 0 \\
\sum_{j \in N(i)}x_j < t \Rightarrow y_i = 1 \\
0 \leq y_i \leq 1
\end{aligned}
\end{equation*}

\noindent Note that $y_i$ is the only variable of this feasibility program, while the $x_j$, $j \in N(i)$, are \circparams.

In turn, membership in PPAD follows rather easily by constructing an appropriate function whose fixed point coordinates are the outcomes of these feasibility programs (one for each player), and arguing that it can be computed by a \linear arithmetic circuit containing \linoptgates. We have the following theorem.

\begin{theorem}\label{thm:threshold-games-PPAD}
Computing an equilibrium of a digraph threshold game is in PPAD.
\end{theorem}

\begin{proof}
We construct a function $F: [0,1]^n \rightarrow [0,1]^n$ in which for a vector $\mathbf{x}=(x_1,\ldots,x_n)$, $F_i(\mathbf{x})$ is the outcome of the feasibility program $\mathcal{Q}$ for player $i \in N$. By definition, $\mathbf{x} = F(\mathbf{x})$ satisfies the constraints of \Cref{eq:threshold-game} for every player, and is thus an equilibrium of the digraph threshold game. Clearly, the \circparams only appear on the left-hand side of the conditional constraints, and $\mathcal{Q}$ is trivially always feasible. 
\end{proof}

\subsection{PPAD-membership via \Linear Best Response Oracles}\label{sec:LBRO-games}

In the applications that we presented in \cref{sec:bimatrix-games,sec:threshold-games}, what we essentially did was construct a function $F:\times_{i \in N} D_i \rightarrow \times_{i \in N} D_i$, where $D_i$ is the strategy space of player $i$, such that coordinate $F_i(\mathbf{x})$ is a best response of player $i$ to $\mathbf{x}_{-i}$. It is then clear that any fixed point of $F$ must be an equilibrium. This is in fact the application of the Kakutani fixed point theorem to show equilibrium existence. PPAD-membership then followed from the fact that the best responses could be computed via \linoptgates.

In this section we will formulate the aforementioned principle as a general theorem, which will have \cref{thm:bimatrix-nash-ppad,thm:threshold-games-PPAD}, as well as several other theorems that we will prove later in this section and in \cref{sec:congestion-games}, as corollaries. Below we define the very general notion of a game with \linear best response oracles, when the strategy spaces $D_i$ are convex polytopes.

\begin{definition}[Game with \Linear Best Response Oracles (\lbro Game)]\label{def:LBRO-game}
A \emph{game with \linear best response oracles (\lbro)} is a game in which a best response of each player $i \in N$ is outputted by an oracle $\mathcal{C}_i: D_1 \ldots \times \ldots D_{i-1} \times D_{i+1} \times \ldots \times D_n \rightarrow D_i$, which is given by a \pseudog.
\end{definition}

\noindent An equilibrium of a \lbro game is a vector of strategies $(\mathbf{x}_1,\ldots,\mathbf{x}_{n})$, one for each player, such that each player is choosing a best response. As we will see in the applications later in the section, the notion of ``best response'' as well as the notion of ``equilibrium'' can vary, depending on the application at hand. Using the \linoptgate the following theorem becomes almost trivial.

\begin{theorem}\label{thm:LBRO-equilibria}
Computing an equilibrium of a \lbro Game is in PPAD.
\end{theorem}

\begin{proof}
We construct a function $F: \times_{i \in N} D_i \rightarrow \times_{i \in N} D_i$ in which for a vector $\mathbf{x}=(x_1,\ldots,x_n)$, $F_i(\mathbf{x})$ is the outcome of the oracle $\mathcal{C}_i(\mathbf{x}_{-i})$ for player $i \in N$. By definition, when $\mathbf{x} = F(\mathbf{x})$, $\mathbf{x}$ is an equilibrium.
\end{proof}

\noindent In \Cref{thm:bimatrix-nash-ppad}, $\mathcal{C}_i$ was given by a \linoptgate (which is, by definition, a \pseudog) computing solutions of linear program $\mathcal{P}$, whereas in \Cref{thm:threshold-games-PPAD}, $\mathcal{C}_i$ was given by a \pseudog computing solutions of feasibility program $\mathcal{Q}$. In the next subsections, we present several other applications of well-known games that we prove to be \lbro games, thus establishing their PPAD-membership.

\begin{remark}\label{rem:rationality-LBRO}
We remark that even if one is not necessarily interested in PPAD-membership, our technique also provides an easy approach for establishing the rationality of equilibria, as this is also implied by \Cref{thm:LBRO-equilibria}. 
\end{remark}

\begin{remark}\label{rem:lbro-own-input}
In \Cref{def:LBRO-game}, we defined the input domain of the oracle $\mathcal{C}_i$ to not include $D_i$, i.e., the domain of strategies of player $i$, whose best-response the oracle is calculating. This is natural, since a best-response oracle can intuitively be best understood as a device which computes the best response of a player against the chosen strategies of her opponents only. Still, in our applications in \Cref{sec:congestion-games}, it will be useful to extend the definition of the oracle to be a function $\mathcal{C}_i: \times_j D_j \rightarrow D_i$, in order to be able to apply \Cref{thm:LBRO-equilibria} above, which still applies for this more general setting.
\end{remark}

\subsubsection{Concave Games and Generalized Equilibria}\label{sec:generalized-concave-games}

In this section we will generalize the PPAD-membership of Nash equilibria in the previous two sections to a class of more general games, and to a more general equilibrium notion. These will be special cases of \emph{concave games} \citep{rosen1965concave} and \emph{generalized equilibria} \citep{debreu1952social} that admit equilibria in rational numbers. 

\paragraph{Concave Games.} Concave games are generalizations of strategic games in which there is a set $N$ of $n$ players, each of which has a strategy space $\Sigma_i \in \mathbb{R}^{m_i}$ which is compact and convex. Let $\Sigma = \Sigma_1 \times \ldots \times \Sigma_n$ be the space of strategy vectors; we will refer to an element $\mathbf{x}=(\mathbf{x}_1, \ldots , \mathbf{x}_n)$ of $\Sigma$ as a \emph{strategy profile}. Each player $i$ also has a payoff function $u_i: \Sigma \rightarrow \mathbb{R}$, which is continuous in $\mathbf{x}$ and concave in $\mathbf{x}_i$, assuming that the rest of the strategy profile $\mathbf{x}_{-i}$ is fixed.

\paragraph{Generalized Equilibrium.} An \emph{equilibrium} of a concave game is a strategy profile $\mathbf{x}$ in which every player chooses a payoff-maximizing element in her strategy space, i.e., $u_i(\mathbf{x}) = \max_{\mathbf{y}_i \in \Sigma_i}u_i(\mathbf{y}_i,\mathbf{x}_{-i})$ for every player $i \in N$. \citet{debreu1952social} defined the notion of a \emph{generalized equilibrium}, in which a strategy profile $\mathbf{x}$ also restricts the choices in the strategy space of the players. In particular, letting $\Sigma_{-i} = \Sigma_1 \times \ldots \Sigma_{i-1} \times \Sigma_{i+1} \ldots \Sigma_n$, there is a correspondence $\gamma_i:\Sigma_{-i} \rightrightarrows \Sigma_i$, which specifies the set $\gamma_i(\mathbf{x}_{-i}) \subseteq \Sigma_i$ that player $i$ is ``allowed'' to use. A generalized equilibrium of the game is a strategy profile $\mathbf{x}$ in which every player chooses a payoff-maximizing element in $\gamma_i(\mathbf{x}_{-i})$, i.e., $u_i(\mathbf{x}) = \max_{\mathbf{y}_i \in \gamma_i(\mathbf{x}_{-i})}u_i(\mathbf{y}_i,\mathbf{x}_{-i})$. The existence of a generalized equilibrium was established by \citet{debreu1952social} in concave games, when $\gamma_i$ is upper and lower hemicontinuous, convex-valued and non-empty valued. Similarly to \Cref{sec:bimatrix-games}, given a preference profile $\mathbf{x}_{-i}$ of the other players, the set of payoff-maximizing (allowable) strategies for player $i$ is given by the optimal solutions to the following convex program:

\begin{center}\underline{Convex Program $\mathcal{C}_\text{gen}$}\end{center}
\begin{equation*}
\begin{aligned}
\mbox{maximize}\quad & u_i(\mathbf{y}_i,\mathbf{x}_{-i})\\
\mbox{subject to}\quad & \mathbf{y}_i \in \gamma_i(\mathbf{x}_{-i})
\end{aligned}
\end{equation*}

\begin{remark}[Concave Games and Generalized Equilibria]
Concave games were studied by \citet{debreu1952social} (in fact, ``quasi-concave games'', where the payoff functions can be quasi-concave) in the context of his generalized equilibrium result (coined a ``Social Equilibrium'' there, see also \citep{dasgupta2015debreu}). Interestingly, \citet{rosen1965concave} also studied concave games independently, without referencing \citeauthor{debreu1952social}'s work. Seemingly the only difference between \citeauthor{rosen1965concave}'s and \citeauthor{debreu1952social}'s setting is that the former does not require the strategy profile space $\Sigma$ to necessarily be the product space of the players' strategy spaces $\Sigma_i$, for $i \in N$.

For the standard notion of equilibrium (rather than generalized equilibrium) \citeauthor{debreu1952social}'s theorem was concurrently and independently proven by \citet{fan1952fixed} and \citet{glicksberg1952further}, and for this reason it is often referred to as \emph{the \citeauthor{debreu1952social}-\citeauthor{fan1952fixed}-\citeauthor{glicksberg1952further} theorem for continuous games}. In fact, we will be referring to this theorem again throughout the paper, as it has been used in previous existence proofs in the applications that we consider. 
\end{remark}

\noindent It is not hard to see that concave games generalize $n$-player strategic games: the set of mixed strategies (i.e., randomizations over the finite set of pure strategies) is compact and convex and the payoff function is continuous in $\mathbf{x}$ and linear in $\mathbf{x}_i$. Thus a mixed Nash equilibrium in a strategic game is a special case of an equilibrium (and hence of a generalized equilibrium) in a concave game.

From the above discussion, it should be obvious that (generalized) equilibria of concave games are not guaranteed to be rational (indeed, the 3-player example of irrationality of \citet{Nash50} still applies, for example). We will consider subclasses of these games (based on the structure of the strategy spaces and the form of the payoff functions) and these equilibria (based on the structure of the sets induced by the functions $\gamma_i$) for which rational equilibria always exist; in fact, our proof of PPAD membership will also provide the certificate of rationality. We remark that for general concave games (under the necessary assumptions for the OPT gate for FIXP to work), \citet{SICOMP:Filos-RatsikasH2023} proved membership in FIXP; their theorems only apply to equilibria rather than generalized equilibria, but the extension in their setting is almost immediate. 

In our setting, we need to consider special cases of convex program $\mathcal{C}_\text{gen}$ which are amenable to the use of the \linoptgate. Specifically, the objective function will be such that its supergradient\footnote{We consider the supergradient here rather than the subgradient, since $u_i$ is concave rather than convex, see also \Cref{rem:subgradient-supergradient}.} with respect to $\mathbf{y}_i$ can be given by a \pseudog, and the constraints will be linear inequalities with the \circparams appearing only on the right-hand side of the constraints. We will also impose a bound on the domain of the variables $\mathbf{y}_i$. With this we have the following convex program:

\begin{center}\underline{Convex Program $\mathcal{C}$}\end{center}
\begin{equation*}
\begin{aligned}
\mbox{maximize}\quad & u_i(\mathbf{y}_i,\mathbf{x}_{-i})\\
\mbox{subject to}\quad & A_i \cdot \mathbf{y}_i \leq b_i(\mathbf{x}_{-i})\\
& \mathbf{y}_{i} \in [-R_i,R_i]^{m_i} 
\end{aligned}
\end{equation*}
\medskip

\noindent For the payoff function $u_i$, we do not require access to the function itself but rather to its supergradient with respect to $\mathbf{y}_i$, which is provided in the input as a \pseudog. Here, the strategy spaces of the players are given by linear inequalities which can depend on the strategies of the other players, and generally lie within a bounded domain $[-R_i,R_i]^{m_i}$, for some $R_i >0$. The functions $b_i(\cdot)$ are also given as \pseudogs. The following theorem is a corollary of \Cref{thm:LBRO-equilibria}.

\begin{theorem}\label{thm:generalized-eq-concave}
Computing a generalized equilibrium of a concave game, in which the strategy spaces are given as in the constraints of convex program $\mathcal{C}$ and the supergradients of the payoff functions of the players are given by \pseudogs is in PPAD. 
\end{theorem}

\begin{proof}
By the discussion above, it follows directly that convex program $\mathcal{C}$ can be computed by a \linoptgate, which can be used as the oracle $\mathcal{C}_i$ for each player $i \in N$ in the corresponding \lbro game.
\end{proof}

\noindent The class of concave games is very general, and hence captures several games of interest. We provide some examples of the type of equilibrium results that are captured by \Cref{thm:generalized-eq-concave} below. We will state the theorems for the computation of generalized equilibria, noting that they hold under the constraints for $b_i(\mathbf{x}_{-i})$ required for the \linoptgate to work. 

\paragraph{Bimatrix and Polymatrix Games.}The theorem establishes that \emph{generalized} Nash equilibria in bimatrix games exist, are rational, and are in PPAD, thus generalizing \Cref{thm:bimatrix-nash-ppad}. It also establishes the same properties for a natural generalization of bimatrix games, called \emph{polymatrix games} \citep{janovskaja1968equilibrium,howson1972equilibria}. These are $n$-player games in which the players' payoffs are additive over several $2$-player games.

\begin{definition}[Polymatrix game]\label{def:polymatrix}
A \emph{polymatrix game} consists of a set $N$ of $n$ players, each with a set $S_i$ of pure strategies $S_i = \{1,\ldots,m_i\}$ and a domain of mixed strategies $\Sigma_i:=\{y \in \mathbb{R}_{\geq 0}^{m_i} \colon \sum_{j=1}^{m_i} y_j = 1\}$. Given a mixed strategy $\mathbf{x}_i$ for player $i$, the $j$th coordinate $x_{i,j}$ denotes the probability that player $i$ is playing the pure strategy $j$. For every pair of players $i$ and $i'$ with $i \neq i'$, there is an associated $(m_i \times m_{i'})$-dimensional \emph{payoff matrix} $\mathbf{A}_{i,i'} \in \mathbb{R}^{m_i \times m_{i'}}$, which determines the payoffs of player $i$ in the bimatrix game against player $i'$. Given a mixed strategy profile $\mathbf{x} = (\mathbf{x}_1, \ldots \mathbf{x}_n)$, the expected payoff of player $i$ is defined as $\tilde{u}_i(\mathbf{x}) = \sum_{i' \in N, i' \neq i} (\mathbf{x}_i)^\transpose\cdot\mathbf{A}_{i,i'} \cdot \mathbf{x}_{i'}$.
\end{definition}

\noindent Since the gradients of the payoff functions in polymatrix games are linear functions, \Cref{thm:generalized-eq-concave} has the following corollary.  

\begin{corollary}
Computing a generalized mixed Nash equilibrium of a polymatrix game is in PPAD.
\end{corollary}

\paragraph{\Linear Succinct Games.} We observe that the reason which allowed us to use the \linoptgate to prove PPAD-membership for bimatrix and polymatrix games above is that we were always able to construct a \pseudog computing the following quantity:
\begin{equation}\label{eq:linear-succinct-utility}
    \tilde{u}_{i}(j,\mathbf{x}_{-i}) = \mathbb{E}_{\mathbf{s}_{-i} \sim \mathbf{x}_{-i}}u_{i}(j,\mathbf{s}_{-i}),
\end{equation}
where $\mathbf{s}_{-i}$ denotes a vector of pure strategies (i.e., a pure strategy profile) of all players besides $i$. The quantity above is the expected payoff of player $i$ when using pure strategy $j \in S_i$ against the mixed strategy $\mathbf{x}_{-i}$ of the other players. As long as we have a \pseudog for computing $\tilde{u}_{i}(j,\mathbf{x}_{-i})$, \Cref{thm:generalized-eq-concave} goes through. We will use the term \linear succinct games to refer to those games.

\begin{definition}[\Linear Succinct Game]\label{def:linear-succinct-games}
A game is a \linear succinct game if for any player $i$ and for any pure strategy $j \in S_i$, there exists a \pseudog computing the expected utility $\tilde{u}_{i}(j,\mathbf{x}_{-i})$, for any profile of mixed strategies $\mathbf{x}_{-i}$ of the other players.
\end{definition}

\noindent We have the following corollary of \Cref{thm:generalized-eq-concave}.

\begin{corollary}\label{thm:linear-succinct}
Computing a generalized mixed Nash equilibrium of a \linear succinct game is in PPAD.
\end{corollary}

\noindent We draw parallels between our application and that of \citet{daskalakis2006game} and \citet{papadimitriou2008computing}. Those works define classes of succinct games for which there is an oracle for computing the expected utility of the player. In \citep{papadimitriou2008computing}, this oracle is referred to as the \emph{polynomial expectation property} and is used to show that correlated equilibria \citep{aumann1974subjectivity} of games with this property can be computed in polynomial time. In \citep{daskalakis2006game}, it is shown that if the oracle is given by a \emph{bounded division free straight-line program of polynomial length}, then these games are in PPAD. Crucially, this latter result concerns \emph{approximate equilibria}. One could view our result as a complement to those two results, one which concerns \emph{exact} equilibria in \emph{rational} numbers.

\paragraph{General Threshold Games.} \Cref{thm:generalized-eq-concave} also establishes the PPAD-membership of \emph{generalized} equilibria in digraph threshold games (see \Cref{sec:threshold-games}). To see this, we may redefine a player's payoff in a digraph threshold game (\Cref{eq:threshold-game}) as:
\begin{equation}\label{eq:threshold-games-alternative}
u_i(x_i,\mathbf{x}_{-i}) =  x_i \cdot (t - \sum_{j \in N_i}x_j),
\end{equation}
where $\mathbf{x}_{-i}=(x_1,\ldots,x_{i-1},x_{i+1},\ldots,x_n)$. Indeed, when $\sum_{j \in N_i}x_j > t$, the payoff is maximized when $x_i=1$, when $\sum_{j \in N_i}x_j < t$, the payoff is maximized when $x_i = 0$, and when $\sum_{j \in N_i}x_j = t$, any choice of $x_i \in [0,1]$ maximizes the payoff. Therefore the equilibria of the game under this utility function are the same as those under the original definition. Additionally, the derivative of the utility function in \Cref{eq:threshold-games-alternative} w.r.t.\ $x_i$ is linear and $x_i \in [0,1]$ for every player $i \in V$. The digraph threshold game is then a concave game satisfying the conditions of the statement of \Cref{thm:generalized-eq-concave}. One may also add constraints on the allowable values of $x_i$ chosen by player $i$, given the chosen strategies $\mathbf{x}_{-i}$ of the other players, as those of the constraints in convex program $\mathcal{C}$. This gives us the following corollary, which generalizes \Cref{thm:threshold-games-PPAD}.

\begin{corollary}
Computing a generalized equilibrium of a digraph threshold game is in PPAD. 
\end{corollary}

\noindent Our approach can actually capture the complexity of a larger class of threshold games, which we call \emph{general threshold games}.

\begin{definition}[General Threshold Game]A general threshold game $\mathcal{G}(g_1,\ldots,g_n)$ is given by functions $g_i: [0,1]^{n-1} \to \mathbb{R}$ for $i \in N$. A strategy profile $\mathbf{x} \in [0,1]^n$ is an equilibrium if for every agent $i$ it satisfies
\begin{equation}\label{eq:extended-threshold-game}
x_i =  \begin{cases} 
      0 & \text{if } g_i(\mathbf{x}_{-i}) > 0 \\
      1 & \text{if } g_i(\mathbf{x}_{-i}) < 0 \\
      [0,1] & \text{if } g_i(\mathbf{x}_{-i}) = 0 
   \end{cases}
\end{equation}
where in the above expression we have abused notation using $x_i = [0,1]$ to denote $x_i \in [0,1]$, and where $\mathbf{x}_{-i}=(x_1,\ldots,x_{i-1},x_{i+1},\ldots,x_n)$. 
\end{definition}

\noindent It is not hard to see that general threshold games are a generalization of digraph threshold games. General threshold games are also concave games, as, similarly to \Cref{eq:threshold-games-alternative}, we can write the payoff functions as 
\begin{equation}\label{eq:threshold-games-alternative-general}
u_i(x_i,\mathbf{x}_{-i}) =  x_i \cdot (- g_i(\mathbf{x}_{-i})).
\end{equation}
If the functions $g_i(\mathbf{x}_{-i})$ can be given by \pseudogs, then \Cref{thm:generalized-eq-concave} applies.

\begin{theorem}\label{thm:generalized-extended-threshold-games}
Computing a generalized equilibrium of a general threshold game is in PPAD, as long as the functions $g_i(\mathbf{x}_{-i})$ are given by \pseudogs.
\end{theorem}

\noindent Some natural general threshold games for which \Cref{thm:generalized-extended-threshold-games} applies can be defined as follows:
\begin{itemize}
    \item[-] digraph threshold games (\Cref{def:threshold-game}), for which $g_i(\mathbf{x}_{-i}) = (\sum_{j \in N_i} x_j) - t$.
    \item[-] \emph{max digraph threshold games}, for which $g_i(\mathbf{x}_{-i}) = (\max_{j \in N_i} x_j) - t$. In these games, player $i$ chooses $x_i=0$ if the maximum of the strategies of her incoming neighbors is above a threshold~$t$, $x_i=1$ if it is below, and any value in $[0,1]$ otherwise.
    \item[-] \emph{min digraph threshold games}, for which $g_i(\mathbf{x}_{-i}) = (\min_{j \in N_i} x_j) - t$. In these games, player $i$ chooses $x_i=0$ if the minimum of the strategies of her incoming neighbors is above a threshold~$t$, $x_i=1$ if it is below, and any value in $[0,1]$ otherwise.
\end{itemize}

\paragraph{Bilinear Games.} \Cref{thm:generalized-eq-concave} also establishes the rationality and PPAD membership of Nash equilibria in \emph{bilinear games}, introduced by \citet{garg2011bilinear}. Intuitively, these generalize bimatrix games in the sense that the strategy space does not have to be the unit simplex $\Sigma_i = \{y \in \mathbb{R}_{\geq 0}^{m_i} \colon \sum_{j=1}^{m_i}y_j =1\}$, as in \Cref{def:bimatrix-game}, but it can be any compact polytope. 

\begin{definition}[Bilinear Games \citep{garg2011bilinear}]\label{def:bilinear-games}
A \emph{bilinear game} is a two-player game represented by two $m_1 \times m_2$ payoff matrices $\mathbf{A}_1$ and $\mathbf{A}_2$, one for each player, and two compact polytopal strategy sets $X_1$ and $X_2$. Let $\mathbf{E}_1 \in \mathbb{R}^{k_1 \times m_1}$ and $\mathbf{E}_2 \in \mathbb{R}^{k_2 \times m_2}$ and let $\mathbf{e}_1 \in \mathbb{R}^{k_1}$ and $\mathbf{e}_2 \in \mathbb{R}^{k_2}$ be two vectors. The strategy spaces of players $1$ and $2$ respectively are defined as 
\[
X_1 = \{ \mathbf{x} \in \mathbb{R}^{m_1} \colon \mathbf{E}_1 \cdot \mathbf{x} = \mathbf{e}_1, \mathbf{x}\geq 0\}, \text{ and }
X_2 = \{ \mathbf{x} \in \mathbb{R}^{m_2} \colon \mathbf{E}_2 \cdot \mathbf{x} = \mathbf{e}_2, \mathbf{x}\geq 0\}.
\]
Given a pair of strategies $(\mathbf{x}_1,\mathbf{x}_2) \in X_1 \times X_2$, the payoff of player $1$ is $\mathbf{x}_1^{\transpose}\cdot \mathbf{A}_1 \cdot \mathbf{x}_2$ and the payoff of player $2$ is $\mathbf{x}_1^{\transpose}\cdot \mathbf{A}_2 \cdot \mathbf{x}_2$.
\end{definition}

\noindent A Nash equilibrium of the game is defined analogously to \Cref{def:Nash-eq-bimatrix}. From \Cref{def:bilinear-games}, it follows that bilinear games are concave games in which the gradients of the payoffs of the players are linear functions. Thus \Cref{thm:generalized-eq-concave} has the following corollary.

\begin{corollary}
Computing a generalized Nash equilibrium of a bilinear game is in PPAD.
\end{corollary}

\noindent We remark that the PPAD-membership of finding equilibria in bilinear games had not explicitly been proven before our work, but it can be recovered rather implicitly via observing that these games are very much related to the \emph{sequence form} of \citet{koller1996efficient}, for which the authors devise an LCP.

\paragraph{The \citeauthor{debreu1952social}-\citeauthor{fan1952fixed}-\citeauthor{glicksberg1952further} theorem [\citeyear{debreu1952social}] and \citeauthor{rosen1965concave}'s theorem [\citeyear{rosen1965concave}].} As we mentioned earlier, the \citeauthor{debreu1952social}-\citeauthor{fan1952fixed}-\citeauthor{glicksberg1952further} theorem for continuous games (or equivalently, \citeauthor{rosen1965concave}'s theorem) is often used in the literature to prove existence of equilibria for various games. It is in a sense stronger than Brouwer's fixed point theorem \citep{MA:Brouwer1911}, as it can be used to prove the existence of equilibria in more general games. Our PPAD-membership of this section also provides a tool for proving PPAD-membership of these other problems where the \citeauthor{debreu1952social}-\citeauthor{fan1952fixed}-\citeauthor{glicksberg1952further} theorem applies, as long as the games that they consider satisfy the conditions required for \Cref{thm:generalized-eq-concave}. Generally speaking, if the problem in question can be reduced to that of finding equilibria of a concave game in which the supergradients of the payoff functions can be computed by \linear arithmetic circuits, then \Cref{thm:generalized-eq-concave} implies its membership in PPAD. 

The existence proofs of several of our applications in this section and in \Cref{sec:congestion-games} are in fact established via the \citeauthor{debreu1952social}-\citeauthor{fan1952fixed}-\citeauthor{glicksberg1952further} theorem. One could transform those proofs into computational reductions (and then \Cref{thm:generalized-eq-concave} would apply immediately); instead, we obtain the PPAD-membership of those problems directly from the more general \Cref{thm:LBRO-equilibria}. 

\subsubsection{Personalized Equilibria}\label{sec:personalized}

Our next application of \Cref{thm:LBRO-equilibria} is to finding \emph{personalized equilibria} in graphical games, a notion introduced by \citet{SICOMP:KintaliPRST13}. Intuitively speaking, these equilibria allow players to ``match'' their strategies with those of their opponents, without obeying a product distribution. 

\begin{definition}[Hypergraph Game \citep{SICOMP:KintaliPRST13}]\label{def:kintali-hypergraph}
Consider a game $\mathcal{G}$ played on a hypergraph $G=(V,E)$ among $n$ different players of a set $N$. Each player $i \in N$ has a finite strategy set $S_i$, and $|S_i \cap S_J| = \emptyset$ for all $i,j \in N$ such that $i \neq j$. Let $V = \bigcup_{i \in N} S_i$. For each player $i \in N$, we have a set $E_i \subseteq E$ of hyperedges, which satisfy the following two conditions:
\begin{itemize}
    \item[-] if $e \in E_i$ then $|e \cap S_j| \leq 1$ for any $j \in N$,
    \item[-] if $e, e' \in E_i$ are distinct, then $e \not\subset e'$.
\end{itemize}
Lastly, each player $i \in N$ has a utility function $u_i: E_i \rightarrow \mathbb{R}$. A \emph{mixed strategy} $\mathbf{x}_i$ for player $i \in N$ is a probability distribution over $S_i$, and a \emph{weight assignment} $\mathbf{w}_i$ is a probability distribution over $E_i$.
\end{definition}

\noindent One way to interpret \Cref{def:kintali-hypergraph} above in relation to the ``standard'' games that we have seen so far, is that the hyperedges $e$ correspond to pure strategy profiles, interpreted as sets of pure strategies, at most one for each player, and each player can choose how much weight $w_i(e)$ to assign to each such profile. A personalized equilibrium of $\mathcal{G}$ is defined as follows:

\begin{definition}[Personalized Equilibrium \citep{SICOMP:KintaliPRST13}]\label{def:kintali-personalized-equilibrium}
A personalized equilibrium of $\mathcal{G}$ consists of a vector of mixed strategies $\{\mathbf{x}_1,\ldots,\mathbf{x}_n\}$ and a vector of weight assignments $\{\mathbf{w}_1,\ldots,\mathbf{w}_n\}$, such that for every player $i \in N$, $(\mathbf{x}_i,\mathbf{w}_i)$ is a solution to the following linear program (with variables $\mathbf{x}_i$ and $\mathbf{w}_i$):

\begin{equation}
\begin{aligned}
&\textrm{maximize}  &\mathclap{ \sum_{e \in E_i} w_i(e)u_i(e) } \\
&\textrm{subject to} & \sum_{e : s \in e} w_i(e) &\leq x_j(s) && \forall s \in S_j, \forall j\neq i\\
                    && \sum_{e : s \in e} w_i(e) & = x_i(s)   && \forall s \in S_i\\
                    && \sum_{e \in E_i} w_i(e) &= 1\\
                    && \sum_{s \in S_i} x_i(s) &= 1\\
                    && w_i(e) & \geq 0                        && \forall e \in E_i\\
                    && x_i(s) & \geq 0                        && \forall s \in S_i
\end{aligned}
\label{LP:PersonalizedEquilibrium}
\end{equation}
We say that $\mathcal{G}$ is \emph{well-behaved} if this LP is feasible for any player $i$ and for any choice of strategies $\mathbf{x}_j$ for the players $j \neq i$.
\end{definition}

\paragraph{Features of our proof and \citeauthor{SICOMP:KintaliPRST13}'s PPAD-membership result.}
To establish existence of a personalized equilibrium in every well-behaved hypergraph game $\mathcal{G}$ as defined above, \citet{SICOMP:KintaliPRST13} essentially reduce the game to a concave game, and then invoke the \citeauthor{debreu1952social}-\citeauthor{fan1952fixed}-\citeauthor{glicksberg1952further} theorem (\citeyear{debreu1952social}). This already hints at the fact that we could obtain PPAD-membership of the problem as a corollary of \Cref{thm:generalized-eq-concave}. We will instead obtain it as corollary of the more general \Cref{thm:LBRO-equilibria} rather straightforwardly; we do that in \Cref{thm:personalized-equilibrium} below. Additionally, to obtain PPAD-membership (and as a result, rationality of equilibria), \citet{SICOMP:KintaliPRST13} first define an approximate version of the problem (the $\varepsilon$-personalized equilibrium), and reduce that problem to \textsc{End-Of-Line} (see \Cref{def:end-of-line}), via a relatively involved construction. To obtain PPAD-membership for the exact problem (i.e., when $\varepsilon=0$) \citet{SICOMP:KintaliPRST13} construct an elaborate argument that appeals to linear programming compactness, by first showing that for sufficiently small $\varepsilon$, $\varepsilon$-personalized equilibria ``almost satisfy'' the constraints of the linear programs, and then carefully rounding the solution to obtain an exact equilibrium. Our technique allows us to reduce this whole argument to a few lines.

\begin{theorem}\label{thm:personalized-equilibrium}
Computing a personalized equilibrium of a well-behaved hypergraph game (\Cref{def:kintali-hypergraph}) is in PPAD.
\end{theorem}

\begin{proof}
We observe that the linear program~\eqref{LP:PersonalizedEquilibrium} computes a best response $\mathbf{x}_i$ for player $i$, given the other players' strategies $\mathbf{x}_j$. The \lbro game oracle $\mathcal{C}_i$ for player $i \in N$ will be given by a \linoptgate which takes as input the vector $\mathbf{x}_{-i}$ of strategies of the other players and outputs a best-response strategy $\mathbf{x}_i$ of player $i$ by solving the LP~\eqref{LP:PersonalizedEquilibrium}. The LP will always be feasible since the game is well-behaved. By \Cref{thm:LBRO-equilibria}, the theorem follows.
\end{proof}

\subsubsection{\eps-proper Equilibria in \Linear Succinct Games}\label{sec:proper}

In this section we will consider an alternative equilibrium notion, that of $\varepsilon$-proper equilibrium \citep{IJGT:Myerson78}, which refines another well-known notion, that of $\varepsilon$-perfect equilibrium \citep{selten1975reexamination}. Both of these notions allow the players to make small mistakes (``trembles'') when choosing their optimal mixed strategies, but ensure that these mistakes happen with small probability (related to the parameter $\varepsilon$). Recall the definition of \linear succinct games (\Cref{def:linear-succinct-games}); we will show that finding $\varepsilon$-proper equilibria of those games is in PPAD.

First, to demonstrate the main ideas, we will consider bimatrix games, which were studied in the past in the context of TFNP and $\varepsilon$-proper equilibria by \cite{EC:Sorensen12}. Recall that bimatrix games can be expressed by two $m_1 \times m_2$ matrices, which we will henceforth denote by $\mathbf{A}_1$ and $\mathbf{A}_2$. We provide the definition of an $\varepsilon$-proper equilibrium below.

\begin{definition}[$\varepsilon$-proper equilibrium in a bimatrix game \citep{IJGT:Myerson78}] \label{def:eps-proper-equilibrium}
Let $\varepsilon > 0$. A pair of mixed strategies $(\mathbf{x}_1,\mathbf{x}_2)$ is an $\varepsilon$-proper equilibrium of a bimatrix game if $\mathbf{x}_1$ and $\mathbf{x}_2$ are \emph{fully mixed} (i.e., $x_{i,j} >0$ for any pure strategy $j \in S_i$, for $i \in \{1,2\}$), and 
$$(\mathbf{A}_1 \cdot \mathbf{x}_2)_j < (\mathbf{A}_1 \cdot \mathbf{x}_2)_{j'} \Rightarrow x_{1,j} \leq \varepsilon \cdot x_{1,j'} \text{ for all } j, j' \in S_1,$$
$$(\mathbf{x}_1^\transpose \cdot \mathbf{A}_2)_j < (\mathbf{x}_1^\transpose \cdot \mathbf{A}_2)_{j'} \Rightarrow x_{2,j} \leq \varepsilon \cdot x_{2,j'} \text{ for all } j, j' \in S_2.$$

\noindent We remark that the \emph{$\varepsilon$-perfect equilibrium} mentioned earlier is defined analogously, with the only difference being that $\varepsilon$ is not multiplied by the mixed strategy $x_{i,j'}$ on the right-hand side of the constraints.
\end{definition}

\noindent We will obtain that computing $\varepsilon$-proper equilibria of bimatrix games is in PPAD, essentially as a corollary of \Cref{thm:LBRO-equilibria}. Before we present the theorem and the proof, we discuss the idea and its advantages over the previous PPAD-membership proof due to \citet{EC:Sorensen12}.

\paragraph{Features of our proof and the previous PPAD-membership result.} 

\citet{EC:Sorensen12} first provided a PPAD-membership result for computing $\varepsilon$-proper equilibria. His proof proceeds by showing that an $\varepsilon$-proper equilibrium can be recovered as a solution to an LCP, and thus can be found by Lemke's algorithm \citep{lemke1965bimatrix}. As we highlighted in \Cref{sec:Other-Approaches}, this approach already introduces complications, mainly arguing against ray termination, which is also explicitly done in \citep{EC:Sorensen12}. Besides that, to make sure that the constructed LCP has polynomial size, \citeauthor{EC:Sorensen12} employs an \emph{extended formulation of the generalized permutahedron} due to \citet{goemans2015smallest}. 

Our proof is conceptually much simpler: it suffices to embed the conditions of \Cref{def:eps-proper-equilibrium} in a set of feasibility programs $\mathcal{Q}$ (see Program~\eqref{eq:feasibility-general} in \Cref{sec:lin-opt-gate}), one for each player $i \in \{1,2\}$, see \Cref{fig:feasibility-eps-proper}. Each of those feasibility programs will compute the best response of the corresponding player, where a ``best response'' here is a mixed strategy $\mathbf{x}_i$ that satisfies the conditions of \Cref{def:eps-proper-equilibrium}, i.e., when a pure strategy yields smaller expected payoff, then it is played with probability which is smaller by an $\varepsilon$ multiplicative factor. Using these feasibility programs as the oracles $\mathcal{C}_1$ and $\mathcal{C}_2$ of an \lbro game, we obtain the proof of \Cref{thm:eps-proper-equilibrium} as a corollary of \Cref{thm:LBRO-equilibria}. We state the theorem next.

\begin{theorem}\label{thm:eps-proper-equilibrium}
Computing an $\varepsilon$-proper equilibrium of a bimatrix game is in PPAD.
\end{theorem}

\begin{proof}
Consider the feasibility programs $\mathcal{Q}_1$ and $\mathcal{Q}_2$ of \Cref{fig:feasibility-eps-proper}, for players $1$ and $2$ respectively. By \Cref{def:eps-proper-equilibrium}, a solution to $\mathcal{Q}_i$, for $i \in \{1,2\}$ is a best response of player $i$ in the bimatrix game. Note that the antecedents (namely, the constraints on the left-hand side of the implications) only contain \circparams. Hence, these can be computed by \linoptgates, so they can be used as oracles $\mathcal{C}_1$ and $\mathcal{C}_2$ in the corresponding \lbro game (\Cref{def:LBRO-game}) as long as they are solvable. Then, the theorem follows from \Cref{thm:LBRO-equilibria}. 
Solvability of $\mathcal{Q}_i$ for $i \in \{1,2\}$ is easy to see, by observing that it is of the form $\mathcal{Q}_\text{app}$, presented in \Cref{sec:linopt-applications}. By \Cref{lem:feasibility-of-q-graph}, it suffices to argue that the \qgraph is acyclic. The \qgraph $G_{\mathcal{Q}_i}$ consists of vertices corresponding to the different pure strategies of player $i$, and an edge $(j,j')$ is only present if the expected utility of player $i$ from $j$ is strictly lower than the expected utility for $j'$. It is straightforward to verify that $G_{\mathcal{Q}_i}$ is acyclic. 
\end{proof}

\noindent \Cref{thm:eps-proper-equilibrium} establishes the PPAD-membership of computing $\varepsilon$-proper equilibria in bimatrix games. In the following we will show that the proof can straightforwardly be extended to larger classes of games. We first offer the following remark.

\begin{remark}[Proper and Trembling hand perfect equilibria]
\citet{IJGT:Myerson78} defined a \emph{proper equilibrium} to be a limit point of $\varepsilon$-proper equilibria for $\varepsilon \rightarrow^{+} 0$. Similarly, limits points of $\varepsilon$-perfect equilibria are trembling hand perfect equilibria \citep{selten1975reexamination}. With this definition, proper equilibria are \emph{refinements} of  trembling hand perfect equilibria, which, in turn, are refinements of Nash equilibria. Intuitively, a refinement is a special class of equilibria characterized by a set of principles that make them more ``plausible''. Note that $\varepsilon$-proper equilibria are not refinements of (exact) Nash equilibria, and this is why we refer to them as ``alternative equilibrium notions''. We believe they are still natural, as they can be interpreted as a model of \emph{limited rationality}\footnote{\citet{selten1975reexamination} in fact introduces the notion of trembling hand perfect equilibria using a similar narrative.}, with the players being ``imperfectly rational'' in their decisions. 

In this context, \cite{EC:Sorensen12} showed a stronger result, namely that finding \emph{symbolic} $\varepsilon$-proper equilibria is in PPAD. Interestingly, the problem of finding a proper equilibrium is unlikely to be in TFNP \citep{hansen2018computational}, and hence \citeauthor{EC:Sorensen12}'s result does not provide a PPAD-membership result for all proper equilibria. We refer the reader to \citep{EC:Sorensen12} for the appropriate definitions and the details.
\end{remark}

\begin{figure}
\centering 
\fbox{
\centering
    \begin{minipage}{.45\textwidth}
        \begin{center}\underline{Feasibility Program $\mathcal{Q}_1$}\end{center}
\begin{equation*}
\begin{aligned}
(\mathbf{A}_1 \cdot \mathbf{x}_2)_j < (\mathbf{A}_1 \cdot \mathbf{x}_2)_{j'} \Rightarrow x_{1,j} \leq \varepsilon \cdot x_{1,j'} \\ \text{ for all }j, j' \in S_1\\
x_{1,j} \geq \frac{\varepsilon^{m_1}}{m_1}, \text{ for all }j \in S_1, \ \ \ \ \sum_{j=1}^{m_1} x_{1,j} = 1 
\end{aligned}
\end{equation*}
    \end{minipage}%
    \hfill\vline\hfill

\begin{minipage}{.45\textwidth}
   \begin{center}\underline{Feasibility Program $\mathcal{Q}_2$}\end{center}
\begin{equation*}
\begin{aligned}
(\mathbf{x}_1^\transpose \cdot \mathbf{A}_2)_j < (\mathbf{x}_1^\transpose \cdot \mathbf{A}_2)_{j'} \Rightarrow x_{2,j} \leq \varepsilon \cdot x_{2,j'} \\ \text{ for all }j, j' \in S_2\\
x_{2,j} \geq \frac{\varepsilon^{m_2}}{m_2}, \text{ for all }j \in S_2, \ \ \ \ \sum_{j=1}^{m_2} x_{2,j} = 1
\end{aligned}
\end{equation*}
\end{minipage}
    }
\caption{The feasibility programs $\mathcal{Q}_1$ and $\mathcal{Q}_2$ used in the proof of \Cref{thm:eps-proper-equilibrium}.}
\label{fig:feasibility-eps-proper}
\end{figure}

\paragraph{\Linear succinct games.}

We now explain how \Cref{thm:eps-proper-equilibrium} extends rather straightforwardly to \linear succinct games  (\Cref{def:linear-succinct-games}). The definition of $\varepsilon$-proper equilibria in such games generalizes that of \Cref{def:eps-proper-equilibrium} straightforwardly. Recall the definition of the expected utility of player $i$ when playing the pure strategy $j$ against the mixed strategy profile $\mathbf{x}_{-i}$ of the other players (\Cref{eq:linear-succinct-utility}).

\begin{definition}[$\varepsilon$-proper equilibrium (general games) \citep{IJGT:Myerson78}] \label{def:eps-proper-equilibrium-general}
Let $\varepsilon > 0$. A mixed strategy profile $\mathbf{x} = (\mathbf{x}_1,\ldots,\mathbf{x}_n)$ is an $\varepsilon$-proper equilibrium of a strategic game if for any player $i \in N$, we have that $\mathbf{x}_i$ is \emph{fully mixed} (i.e., $x_{i,j} >0$ for any pure strategy $j \in S_i$), and for any player $i \in N$
\begin{equation}\label{eq:feasibility-updated-constraints-succints-eps-proper}
 \tilde{u}_{i}(j,\mathbf{x}_{-i}) < \tilde{u}_{i}(j',\mathbf{x}_{-i}) \Rightarrow x_{i,j} \leq \varepsilon \cdot x_{i,j'} \text{ for all } j, j' \in S_i.
\end{equation}
\end{definition}

\noindent To extend the feasibility programs $\mathcal{Q}_i$, for $i \in N$, to general games, one only has to substitute the corresponding constraints for each player $i \in N$ with those in \Cref{def:eps-proper-equilibrium-general} above. In \linear succinct games, the quantities $\tilde{u}_{i}(j,\mathbf{x}_{-i})$ for any player $i$ and pure strategy $j$ can be computed by a \pseudog. The proof of the following theorem is then very similar to that of \Cref{thm:eps-proper-equilibrium}.

\begin{theorem}\label{thm:eps-proper-succinct}
Computing an $\varepsilon$-proper equilibrium of a \linear succinct game is in PPAD.
\end{theorem}

\begin{proof}
Again, feasibility program $\mathcal{Q}_i$, for $i \in N$, captures the best response of player $i \in N$. Each $\mathcal{Q}_i$ can be computed by a \linoptgate, and is solvable by the same very simple argument used in the proof of \Cref{thm:eps-proper-equilibrium}. Thus $\mathcal{Q}_i$ can be used as the oracle $\mathcal{C}_i$ in the corresponding \lbro game (\Cref{def:LBRO-game}), and the theorem follows from \Cref{thm:LBRO-equilibria}.
\end{proof}

\noindent Since polymatrix games (\Cref{def:polymatrix}) are \linear succinct games, we obtain the following corollary.

\begin{corollary}\label{thm:eps-proper-polymatrix}
Computing an $\varepsilon$-proper equilibrium of a polymatrix game is in PPAD.
\end{corollary}
\noindent A proof of \Cref{thm:eps-proper-polymatrix} was first provided by \citet{hansen2018computational}. Their proof essentially redoes all of the steps of the proof of \citet{EC:Sorensen12} for bimatrix games, extending them to the more general case, stating and proving corresponding lemmas etc. Clearly, our \linoptgate allows us to avoid all this labor and extend the PPAD-membership from bimatrix games to polymatrix games, and even beyond, rather straightforwardly. We conclude our discussion on $\varepsilon$-proper equilibria with the following remark. 

\begin{remark}[$\varepsilon$-proper equilibria can be irrational]
For general multiplayer games (not necessarily \linear succinct), \citet{SICOMP:Filos-RatsikasH2023} showed a FIXP-membership result via a system of conditional convex constraints, which generalizes the feasibility program $\mathcal{Q}_i$. FIXP is the right class for these games, because it has been shown (see Footnote 1 in \citep{etessami2014complexity}) that there are games in which all $\varepsilon$-perfect equilibria (and hence all $\varepsilon$-proper equilibria) are irrational.
\end{remark}

\section{Congestion Games}\label{sec:congestion-games}

In this section we consider models of congestion
games~\citep{Wardrop1952, Rosenthal1973}, where players compete for
resources. More precisely, we focus our study on \emph{non-atomic} and
\emph{atomic splittable} congestion games, and for both we consider in
particular the important subclasses of network congestion games. In particular, 
via the employment of our \linoptgate we will obtain PPAD-membership results 
for three domains of congestion games, namely for
\begin{itemize}
    \item[-] finding equilibria in congestion games with rational and malicious players, studied by \cite{babaioff2009congestion}, see \Cref{sec:congestion-games-with-malicious-players},
    \item[-] finding Wardrop equilibria in multi-class non-atomic network congestion games, studied by \cite{meunier2013lemke}, see \Cref{sec:nonatomic-network-congestion-ppad},
    \item[-] finding Nash equilibria multi-class atomic splittable network congestion games, studied by \citep{klimm2020complexity}, see \Cref{sec:atomic-network-congestion-ppad}.
\end{itemize}
All of our PPAD-membership results hold when the latency functions (an in some cases possibly their subgradients) are given by \pseudogs (recall \Cref{def:pseudo-circuit}) or simply, by \linear arithmetic circuits. In particular, they apply for example to all \emph{piecewise-linear} latency functions that are given explicitly as part of the input. Via the machinery that we develop in \Cref{sec:implicit}, we can in fact obtain the PPAD-membership for the more general case where the latency functions and their subgradients are provided in the input \emph{implicitly}, via \pseudogs.
As such, our results are simultaneously \emph{significiant simplifications} over the existing results in the literature, as well as \emph{generalizations}, since the only known PPAD-membership results so far were for \emph{linear} latency functions.\footnote{In the context of congestion games, the term ``affine'' is often used instead of ``linear''. \citet{Boyd2004convex} define a linear function as of the form $f(x)=c^\transpose x$, and an affine function as the sum of a linear function and a constant. We use the term ``linear'' to refer to all affine functions.} In fact, for games with malicious players, complexity results had not been proven before our work. 

We remark that all the games mentioned above fall into one of two main categories, either \emph{non-atomic congestion games} or \emph{atomic splittable congestion games}. As such, finding their equilibria does not simply fall in the class PLS \citep{johnson1988easy}, which is known to capture the complexity of atomic (non-splittable) congestion games. In fact, finding equilibria for some of the games studied in this section has been proven to be PPAD-complete \citep{klimm2020complexity}. \medskip

\noindent Before we proceed with the technical parts of the section, we present the previous works on these domains and compare the previous approaches to our proofs which use the \linoptgate. 

\paragraph{Congestion games with malicious players.} The study of congestion games with malicious players was initiated by \citet{karakostas2007equilibria}. Here, we consider the model studied by \citet{babaioff2009congestion}, in which there is a continuum of rational players, who are trying to minimize their total load, and a single malicious player, who is trying to maximize the load experienced by the rational players. \citeauthor{babaioff2009congestion} proved the existence of equilibria in such games with non-decreasing latency functions. They also prove that when the latency functions are concave, these games have pure equilibria, by appealing to the \citeauthor{debreu1952social}-\citeauthor{fan1952fixed}-\citeauthor{glicksberg1952further} theorem [\citeyear{debreu1952social}]. This already hints that PPAD-membership can be recovered as a corollary of \Cref{thm:generalized-eq-concave}; we will instead obtain it as a corollary of the more general \Cref{thm:LBRO-equilibria}. We remark that for these games, a PPAD-membership result was not previously known.\footnote{\citet{babaioff2009congestion} mention in the conclusion of their work that their existence proofs establish membership of the problem in PPAD. This claim is however not straightforward, and there is no proof to support it. In any case, such a claim would most certainly refer only to \emph{approximate} equilibria, as these games do not always have rational equilibria for any concave latency function.}

\paragraph{Multi-class non-atomic network congestion games.} In multi-class non-atomic network congestion games, there is a continuum of players divided into different classes. Existence of Wardrop equilibria in this setting was first established by \citet{schmeidler1973equilibrium}, via an application of the \citeauthor{debreu1952social}-\citeauthor{fan1952fixed}-\citeauthor{glicksberg1952further} theorem [\citeyear{debreu1952social}]. A different proof of existence was provided by \citet{milchtaich2000generic}, who also studied the equilibrium uniqueness. The first computational complexity results on the problem were provided by \citet{meunier2013lemke}, who proved that the problem of finding a Wardrop equilibrium lies in PPAD, when the latency functions on the resources of the game are linear. Their proof goes via the ``LCP approach'' (see \Cref{sec:Other-Approaches}). It turns out that their LCP formulation cannot be solved by the vanilla version of Lemke's algorithm, and so they devise a similar pivoting algorithm, tailored to their problem. As in the case of Lemke's algorithm, they argue explicitly against ray termination. 

\paragraph{Multi-class atomic splittable network congestion games.} The existence of Nash equilibria in atomic splittable congestion games follows from the \citeauthor{debreu1952social}-\citeauthor{fan1952fixed}-\citeauthor{glicksberg1952further} theorem [\citeyear{debreu1952social}], see \citep{klimm2020complexity}. The computational complexity of the problem with player-specific linear latency functions was studied by \citet{klimm2020complexity}, who proved its PPAD-completeness. What is important for us is the most challenging part of that result, which is the PPAD-membership. \citeauthor{klimm2020complexity}'s proof is rather involved, and goes via the development of a homotopy method for tracing an equilibrium given the demand rates of the players. This gives rise to a new pivoting algorithm, similar in spirit to Lemke's algorithm, or more precisely, the Lemke-Howson algorithm (see \Cref{sec:Other-Approaches}). The method solves the problem of finding a Nash equilibria as a system of linear equations involving \emph{excess flows}, \emph{vertex potentials} and \emph{block Laplacians}. At a very high level, the authors use the excess and potentials to define an undirected version of the \textsc{End-of-Line} graph (see \Cref{def:end-of-line}), and the determinant of the block Laplacians to define a unique orientiation of the edges, effectively reducing the problem to \textsc{End-of-Line}. It is interesting to mention that \citeauthor{klimm2020complexity} do mention that the Nash equilibrium problem can be formulated as an LCP, but it is unclear whether Lemke's algorithm can solve it, motivating the development of their new algorithm. 

\paragraph{Our results and proofs.} The \linoptgate allows us to avoid any of the technical complications of the proofs of \citet{meunier2013lemke} and \cite{klimm2020complexity} (which are rather involved, especially the latter), and essentially obtain the PPAD-membership for all of these problems as simple corollaries of \Cref{thm:LBRO-equilibria} (or even \Cref{thm:generalized-eq-concave}, as the games in this section are \linear concave). In fact, as we mentioned earlier, we obtain generalizations of those results, from linear latency functions to any latency function that is given (in some cases together with its subgradient) by a \linear arithmetic circuit. In particular, they capture all piecewise-linear latency functions that are given explicitly as part of the input, as well as those where the latencies and their subgradients are given implicitly by the aforementioned circuits (see \Cref{sec:implicit}). For congestion games with malicious players, our PPAD-membership result for this same class of latency functions is the first such complexity result for any version of the problem. 

\subsection{Multi-class Congestion Games}

In the models we consider, the players are divided into a finite
number $k$ of \emph{classes}, and we refer to these as \emph{multi-class
congestion games}. The main difference between the different models we consider
arise from how the classes of players are interpreted. Apart from
this, the models have many commonalities, and we start by describing
these.

\begin{remark}[Notation]
  In this section we shall for convenience make use of the (common)
  shorthand notation of writing a function argument as a
  subscript. That is, we may write $f_a$ in place of $f(a)$, for a
  given function $f \colon A \to B$.
\end{remark}

\begin{definition}[Multi-class Congestion Games]\label{def:multi-class-congestion}
A multi-class congestion game with a finite number of $k$ classes is
given by a finite set $E$ of \emph{resources} and, for each
class~$i \in [k]$, a set of \emph{pure strategies}
$\Sigma_i \subseteq 2^E$, consisting of subsets of resources.  Each
resource is equipped with class-dependent continuous \emph{latency
  functions}, denoted by $\ell^i_e \colon \RRnn \to \RRnn$ for
$e \in E$ and $i \in [k]$. Note that we do not assume that the latency
functions are non-decreasing. 
A \emph{load} on resources is a function $x \colon E \to \RRnn$. The
latency functions extend additively to latency functions for pure
strategies that are functions of loads by letting
$\ell^i_S(x) = \sum_{e \in S}\ell^i_e(x_e)$.

We associate to each class~$i \in [k]$ a positive \emph{weight demand}
$d_i$ and let $d=d_1+\dots+d_k$ be the total weight demand. A
\emph{load allocation} for class~$i\in [k]$ is a function
$f^i \colon \Sigma_i \to \RRnn$ such that
$\sum_{S \in \Sigma_i} f^i_S=d_i$. The load allocation $f^i$ induces a
\emph{load} $x^i \colon E \to \RRnn$ for class~$i \in [k]$ and a
resulting total load $\overline{x} \colon E \to \RRnn$ given as follows.
\begin{equation}
  \begin{aligned}
    x^i_e & = \sum_{\substack{S \in \Sigma_i \\ e \in S}} f^i_S \enspace,&&  e \in E, i=1,\dots,k\enspace .\\
    \overline{x}_e &= \sum_{i=1}^k x^i_e \enspace,&&  e \in E\enspace .
  \end{aligned}
  \label{EQ:load-from-load-allocation}
\end{equation}
\end{definition}

\subsubsection{Non-atomic congestion games}\label{sec:non-atomic-congestion-games}

In the \emph{non-atomic} model, each class of players represents a continuum
of players. Formally we consider the set of all players to be given by
a bounded real interval $I$ of length $d$, partitioned into $k$
measurable sets $I_1,\dots,I_k$ with respect to the Lebesgue measure
$\mu$, such that $\mu(I_i)=d_i$, for $i \in [k]$, thereby defining the $k$ classes of players. A strategy profile
is a measurable function $\rho : I \rightarrow 2^E$ such that
$\rho(I_i) \subseteq \Sigma_i$ for every $i\in [k]$. The restriction
$\rho_i = \rho_{\restriction I_i}$ of $\rho$ to the subset $I_i$ yields the
strategy profile $\rho_i \colon I_i \to \Sigma_i$ of the players of
class~$i$. The strategy profile induces a load allocation profile
$f=(f^1,\dots,f^k)$ by letting $f^i_S = \mu(\rho^{-1}(S))$ for
$S \in \Sigma_i$ and $i \in [k]$, which in turn defines a load profile
$x=(x^1,\dots,x^k)$ and the total load $x$ on resources.\medskip

\noindent
The central equilibrium notion for non-atomic congestion games is the
\emph{Wardrop equilibrium}~\citep{Wardrop1952}, here defined for multi-class
congestion games. The notion is very similar to the Nash equilibrium that we used 
in the previous sections. In fact, the Wardrop equilibrium coincides with the notion of Nash
equilibrium for this formulation in terms of a continuity of players.
We will use the term Wardrop equilibrium to refer
to equilibria in non-atomic games, and the term \emph{Nash equilibrium} or simply
\emph{equilibrium} to refer to atomic splittable games, see \Cref{sec:atomic-splittable-congestion-games}.

\begin{definition}[Wardrop equilibrium]\label{def:wardrop-equilibrium}
  A strategy profile $\rho$ is a \emph{Wardrop equilibrium} if for all
  $i \in [k]$ it holds that
  $\rho(I_i) \subseteq \argmin_{S \in \Sigma_i} \ell^i_S(\overline{x})$, where
  $\overline{x}$ is the total load induced by $\rho$.
\end{definition}

\noindent We shall also say that the load allocation profile induced by a
Wardrop equilibrium $\rho$ \emph{itself} is a Wardrop equilibrium. Note that
a different strategy profile $\rho'$ may induce the same load
allocation as profile $\rho$ without being a Wardrop
equilibrium. In that case, however, there would exist another Wardrop
equilibrium $\rho''$ such that $\rho'$ and $\rho''$ differ only on a
null set with respect to $\mu$.

\begin{definition}[Wardrop equilibrium --- load allocation formulation]
  A load allocation profile $f=(f^1,\dots,f^k)$ is a Wardrop
  equilibrium if $\ell^i_S(\overline{x})= \min_{S' \in \Sigma_i} \ell^i_{S'}(\overline{x})$,
  whenever $f^i_S>0$, for all $i$, where $\overline{x}$ is the total load defined
  by $(f^1,\dots,f^k)$.
\end{definition}

\noindent Before presenting a fixed point formulation of Wardrop equilibria, we
introduce some further notation, which will be used again in the
setting of atomic splittable congestion games in \Cref{sec:atomic-splittable-congestion-games} below.\medskip

\noindent Given a load allocation $f=(f^1,\dots,f^k)$ we let
$f^{-i}=(f^1,\dots,f^{i-1},f^{i+1},f^k)$ denote the load allocation
profile for classes different from~$i$. For $i \in [k]$ and
$S \in \Sigma_i$, we define the function $L^i_S$ as follows. For 
a load allocation profile $f$ and a load allocation $g^i$ for class~$i$, we let
\begin{equation}\label{eq:non-atomic-congestion-games-latency-LS}
  L^i_S(g^i,f^{-i}) = \sum_{e \in S} \ell^i_e\left(\sum_{\substack{T \in \Sigma_i \\ e \in T}} g^i_T + \sum_{j \neq i}  \sum_{\substack{T \in \Sigma_j \\ e \in T}} f^j_T\right)
\end{equation}

\noindent Note that if $x^{-i}=(x^1,\dots,x^{i-1},x^{i+1},x^k)$ are the induced
loads by $f^{-i}$, and $y^i$ is the induced load by $g$, then
$L^i_S(g^i,f^{-i}) = \sum_{e \in S} \ell^i_e(y^i_e+\sum_{j\neq i}
x^j_e)$. In particular, if $\overline{x}$ is the total load induced by $f$, then
$L^i_S(f) = \ell^i_S(\overline{x})$.\medskip

\noindent Let $i \in [k]$ and consider the following linear program in variables
$g^i_S$.

\begin{center}\underline{Linear Program $\mathcal{P}_\mathrm{WE}$}\end{center}
\begin{equation*}
\begin{aligned}
\mbox{minimize}\quad & \sum_{S \in \Sigma_i} L^i_S(f) g^i_S\\
\mbox{subject to}\quad 
& \sum_{S \in \Sigma_i} g^i_S = d_i\\
& g^i_S \geq 0, \ \text{ for all } S \in \Sigma_i
\end{aligned}
\end{equation*}

\noindent An optimal solution $g^i = (g^i_S)_{S \in \Sigma_i}$ must satisfy that
$L^i_S(f) = \min_{S' \in \Sigma_i} L^i_{S'}(f)$ whenever $g^i_S >
0$. This gives rise to a fixed point characterization of Wardrop
equilibrium.

\begin{proposition}\label{prop:wardrop-fixed-point-characterization}
  A load allocation profile $f=(f^1,\dots,f^k)$ is a Wardrop
  equilibrium if and only if $f^i$ is an optimal solution to the
  linear program $\mathcal{P}_\mathrm{WE}$ for all $i \in [k]$.
\end{proposition}

\noindent Looking ahead, the characterization of \Cref{prop:wardrop-fixed-point-characterization} via linear program $\mathcal{P}_\text{WE}$ will be employed to prove PPAD-membership. In particular, each $\mathcal{P}_\text{WE}$ will capture the best response of each player, and a \linoptgate will serve as the oracle for computing this best response. This effectively will render the game a \lbro game (\Cref{def:LBRO-game}), and PPAD-membership will follow from \Cref{thm:LBRO-equilibria}. We present a concrete application of this in \Cref{sec:congestion-games-with-malicious-players} below. First, we define the setting of atomic splittable congestion games. 

\subsubsection{Atomic splittable congestion games}\label{sec:atomic-splittable-congestion-games}

In the atomic splittable model each class is representing a single
player controlling the entire weight demand of the class. We shall
thus refer to class~$i$ simply as \emph{player}~$i$, and the strategy
space of player~$i$ is the set of load allocations of class~$i$. Given
a load allocation profile $f$, player~$i$ experiences latency
$\sum_{S \in \Sigma_i}L^i_S(f)f^i_S$ which becomes the cost function
$C^i$ the player wishes to minimize. We thus define
$C^i(g^i,f^{-i}) = \sum_{S \in \Sigma_i}L^i_S(g^i,f^{-i})g^i_S$.\medskip 

\noindent The
central equilibrium notion for atomic splittable congestion games we consider is
the \emph{Nash equilibrium} with respect to these cost functions.

\begin{definition}[Nash Equilibrium --- load allocation formulation]
  A load allocation profile $f=(f^1,\dots,f^k)$ is a \emph{Nash
  equilibrium} if and only if $C^i(f) \leq C^i(g^i,f^{-i})$ for all~$i \in [k]$.
\end{definition}

\noindent If one imposes an additional convexity assumption involving the latency functions, then the game essentially becomes a concave game, and thus possesses a Nash equilibrium by the \citeauthor{debreu1952social}-\citeauthor{fan1952fixed}-\citeauthor{glicksberg1952further} theorem [\citeyear{debreu1952social}]. The following
lemma follows immediately from definitions.
\begin{lemma}[Convexity assumption for latency functions]
  Suppose that for all $i \in [k]$ and $e \in E$ the latency function
  $\ell^i_e$ satisfies that the function $x \mapsto \ell^i_e(x+c)x$ is
  convex for all $c \geq 0$. Then, for all $i \in [k]$ and all load
  allocation profiles $f$ the cost function
  $g^i \mapsto C^i(g^i,f^{-i})$ is convex.
  \label{LEM:Latency-convexity-assumption}
\end{lemma}
\begin{remark}
  In the literature on congestion games, a common assumption of a
  latency function $\ell$ is that the function $x \mapsto \ell(x)x$ is
  convex. This is in general not sufficient to guarantee that the
  function $x \mapsto \ell(x+c)x$ for $c \geq 0$ is convex, as
  required in the model we consider. It is however sufficient if one
  additionally assumes that $\ell$ is non-decreasing.
\end{remark}

\noindent With the assumptions of \cref{LEM:Latency-convexity-assumption}
we consider the following convex program in variables $g^i_S$.
\begin{center}\underline{Convex Program $\mathcal{C}_\mathrm{NE}$}
\begin{equation*}
\begin{aligned}
\mbox{minimize}\quad & \sum_{S \in \Sigma_i} L^i_S(g^i,f^{-i}) g^i_S\\
\mbox{subject to}\quad 
& \sum_{S \in \Sigma_i} g^i_S = d_i\\
& g^i_S \geq 0, \ \text{ for all } S \in \Sigma_i
\end{aligned}
\end{equation*}
\end{center}
Directly from definitions we then have the following.
\begin{proposition}
  In a congestion game with latency functions satisfying the
  assumptions of \cref{LEM:Latency-convexity-assumption}, a load
  allocation profile $f=(f^1,\dots,f^k)$ is a Nash equilibrium if and
  only if $f^i$ is an optimal solution to the convex program
  $\mathcal{C}_\mathrm{NE}$ for all $i \in [k]$.
\end{proposition}

\noindent Again, convex program $\mathcal{C}_{\text{NE}}$ essentially calculates the best response for each player, and hence will be computed by an oracle given by a \linoptgate. We will apply this concretely in \Cref{sec:nonatomic-network-congestion-ppad} in the context of network congestion games, but after we will have reformulated $\mathcal{C}_{\text{NE}}$ via a more concise description of the strategies in terms of flows on resources rather than load allocations. 

\subsubsection{PPAD-membership of congestion games with malicious players}
\label{sec:congestion-games-with-malicious-players}

A model of non-atomic congestion games with a \emph{malicious} player
was introduced by \citep{babaioff2009congestion}. In this model there
are two classes of players, the \emph{rational} players and the
\emph{malicious} player. We shall use the index set $\{R,M\}$ for
these classes rather than $\{1,2\}$. The class of rational players
represents a continuum of players. As in the basic case of non-atomic
congestion games as given in
\cref{sec:non-atomic-congestion-games}, these are given by a
bounded real interval $I$ of length $d_R$ and a strategy of the
rational players is a measureable function
$\rho_R \colon I \to \Sigma_R$. The strategy $\rho$ induces a load
allocation $f^R$, and we will use that to re-express the equilibrium
condition in terms of $f^R$.

The class of the malicious player represents a single player that
controls a weight demand $d_M$, as in the case of atomic splittable
games as given in
\cref{sec:atomic-splittable-congestion-games}. However,
unlike the basic model of atomic splittable congestion games,
the malicious player can use a \emph{mixed} strategy, i.e., a
probability distribution $\rho_M$ over load allocations $f^M$.
Since the malicious player is using a mixed strategy, what the
rational players aim to minimize is the expected load of their chosen
strategy. The malicious player on the other hand aims to maximize the
expected total load experienced by the rational players. We define an equilibrium of this game explicitly below.

\begin{definition}\label{def:equilibrium-malicious}
  The strategy profile $(\rho_R,\rho_M)$ is an \emph{equilibrium} if the
  following conditions hold, where $x^R$ is the load induced by
  $\rho^R$ and $x^M$ is the load induced by $f^M$.
  \begin{enumerate}
  \item $\rho_R \subseteq \argmin_{S \in \Sigma_R} \mathbb{E}_{f^M \sim \rho_M}[\ell^R_S(x^R+x^M)]$.
  \item $\supp(\rho_M) \subseteq \argmax_{f^M} \sum_{e \in E} \ell^R_S(x^R_e+x^M_e)x^R_e$.
  \end{enumerate}
\end{definition}

\noindent We reformulate \Cref{def:equilibrium-malicious} in terms of load allocations for the rational
players, similarly to the case of Wardrop equilibria, and thus view the
pair $(f^R,\rho_M)$ as a strategy profile. For convenience we use the
function $L^i_S$ defined in
\cref{sec:non-atomic-congestion-games}.

\begin{definition}\label{def:equilibrium-malicious-load}
  The strategy profile $(f^R,\rho_M)$ is an equilibrium if 
  \begin{enumerate}
  \item $\mathbb{E}_{f^M \sim \rho_M} L_S^R(f^R,f^M) = \min_{S' \in \Sigma_R} \mathbb{E}_{f^M \sim \rho_M} L_{S'}^R(f^R,f^M)$ whenever $f^R_S>0$.
  \item $\supp(\rho_M) \subseteq \argmax_{f^M} \sum_{S \in \Sigma_R} L^R_S(f^R,f^M)f^R_S$.
  \end{enumerate}
\end{definition}

\noindent \citet{babaioff2009congestion} proved that every congestion game with a malicious player with
non-decreasing latency functions has an equilibrium. They also proved, employing the 
\citeauthor{debreu1952social}-\citeauthor{fan1952fixed}-\citeauthor{glicksberg1952further} theorem,
that when the latency functions are concave, such a game has a pure
equilibrium, i.e. an equilibrium of the form $(f^R,f^M)$. We will obtain a PPAD-membership result for this case.
The result will be based on the following fixed formulation of the problem, very similar to that of linear program
$\mathcal{P}_\text{WE}$ presented in \Cref{sec:non-atomic-congestion-games}, and convex program $\mathcal{C}_\text{NE}$
presented in \Cref{sec:atomic-splittable-congestion-games}, refined for the problem at hand. 

\begin{figure}
\centering 
\fbox{
\centering
    \begin{minipage}{.45\textwidth}
        \centering
        \begin{center}\underline{Linear Program $\mathcal{P}^R_\mathrm{WE}$}\end{center}
    \begin{equation*}
    \begin{aligned}
    \mbox{minimize}\quad & \sum_{S \in \Sigma_R} L^R_S(f^R,f^M) g^R_S\\
    \mbox{subject to}\quad 
    & \sum_{S \in \Sigma_R} g^R_S = d_R\\
    & g^R_S \geq 0, \ \text{ for all } S \in \Sigma_i
    \end{aligned}
    \end{equation*}
    \end{minipage}%
    \hfill\vline\hfill
    \begin{minipage}{0.45\textwidth}
        \centering
        \begin{center}\underline{Convex Program $\mathcal{C}^M_\mathrm{NE}$}
        \begin{equation*}
        \begin{aligned}
        \mbox{maximize}\quad & \sum_{S \in \Sigma_R} L^R_S(f^R,g^M) f^R_S\\
        \mbox{subject to}\quad 
        & \sum_{S \in \Sigma_M} g^M_S = d_M\\
        & g^M_S \geq 0, \ \text{ for all } S \in \Sigma_i
        \end{aligned}
        \end{equation*}
        \end{center}
    \end{minipage}
    }
\caption{The linear program $\mathcal{P}^R_\mathrm{WE}$ for the rational players (left) and the convex program $\mathcal{C}^M_\mathrm{NE}$ for the malicious player (right), in the fixed point formulation of the equilibrium.}
\label{fig:malicious}
\end{figure}

For the rational players, best responses are expressed by linear program $\mathcal{P}^R_\mathrm{WE}$ in variables $g^R$ of the left-hand side of \Cref{fig:malicious}.. For the malicious player, the best response is expressed by convex program $\mathcal{C}^M_\mathrm{NE}$ in variables $g^M$ of the right-hand side of \Cref{fig:malicious}.\medskip

\noindent From the fixed point formulation in terms of the programs of \Cref{fig:malicious}, PPAD-membership is almost immediate via arguing that those best response can be computed by \linoptgates, as long as the subgradients of the objective functions can be computed by a \pseudog. Hence we have the following theorem.

\begin{theorem}\label{thm:congestion-games-malicious}
Computing an equilibrium of a congestion game with a malicious player in which the latency functions are given by \pseudogs and their subgradients are given by \pseudogs computing piecewise-constant functions is in PPAD.
\end{theorem}

\begin{proof}
Linear program $\mathcal{P}^R_\mathrm{WE}$ of \Cref{fig:malicious} computes the best response of the class of rational players, and $\mathcal{C}^M_\mathrm{NE}$ computes the best response of the malicious players. As long as those optimization programs can be computed by \linoptgates, then the game is a \lbro game (\Cref{def:LBRO-game}, noting \Cref{rem:lbro-own-input} as the oracles also input a players' own strategy) and PPAD-membership follows from \Cref{thm:LBRO-equilibria}. Indeed, for both programs the feasible domain is non-empty and bounded with no \circparams appearing in the left-hand side of the constraints. By the statement of the theorem, the latency functions can be computed by \pseudogs. In the case of linear program $\mathcal{P}^R_\mathrm{WE}$, the subgradient is a sum of latency functions (one for each $S \in \Sigma_R$), which by extension can also be computed by a \pseudog. For convex program $\mathcal{C}^M_\mathrm{NE}$, by the definition of the latency function $L_S^R(f^R,g^M)$ in \Cref{eq:non-atomic-congestion-games-latency-LS}, we can use the \pseudogs that compute the latency functions, as well as those that compute their subgradients. Those subgradients are piecewise-constant functions, which we then can multiply by the variables $g^M$ by the machinery developed in \Cref{sec:implicit}. In the end, we obtain a function which can be computed by a \pseudog, and the whole linear program can be computed by a \linoptgate.
\end{proof}

\noindent A direct corollary of \Cref{thm:congestion-games-malicious} is to the case when the latency functions are piecewise-linear functions given explicitly as part of the input. Indeed, in that case the objective functions are piecewise-quadratic and their subgradients are piecewise-linear functions, hence computable by a \pseudog. \medskip

\noindent We complement \Cref{thm:congestion-games-malicious} above with an example showing the tightness of the class of latency functions considered in the statement of theorem. Indeed, if one moves to more general concave latency functions, then the only equilibria of the game may be irrational.

\begin{example}[Only irrational equilibria] \label{rem:only-irrational-malicious}
Consider a game with only rational players (the game with a malicious player is a generalization). Define the function $f \colon \RRnn \to \RRnn$ by
\[
f(x) = \begin{cases} 3x-x^2 & \text{for } x \leq 1\\ 1+x & \text{for } x > 1\end{cases} \enspace .
\]
The function $f$ passes through the origin and is continuous, differentiable, increasing, and concave. Consider the non-atomic congestion game with two resources $1$ and $2$, and latency functions $l_1(x)=f(x)$ and $l_2(x)=2f(x)$. If $(x,1-x)$ is an equilibrium, then it must hold that $3x-x^2=2(3(1-x)-(1-x)^2$ which as $x \in [0,1]$ implies that $x=(\sqrt{41}-5)/2$.
\end{example}

\noindent We remark that for games with general concave latency functions, a FIXP-membership result is implied by (although not explicitly stated in) \citep{SICOMP:Filos-RatsikasH2023}, as these games are concave games.

\subsection{Network Congestion Games}

We now define an interesting subclass of congestion games, called \emph{network congestion games}.
Network congestion games allow the players' strategies to be succinctly described by a flow network.

\paragraph{Multi-commodity flow networks}

To be able to define multi-class network congestion games we consider a
\emph{multi-commodity flow model}. More precisely, consider networks given by
a directed graph $G=(V,E)$ on which we consider routing of $k$
commodities. For each commodity $i \in [k]$ we are given a source node
$s_i \in N$, a target node $t_i \in N$, and a flow demand $d_i>0$. A
cost function $c$ is in general a function $c \colon E \to \RR$, but
we will however only consider non-negative cost functions. We assume
separate cost functions $c^i$ for each commodity~$i \in [k]$.

A flow $x$ is a function $x \colon E \to \RR$. To be feasible, the
flow $x$ is required to be non-negative. The balance $b(x)$ of $x$ is
the function given by
\[b(x)_u=\sum_{uv \in E} x_{uv}-\sum_{vu \in E} x_{vu} \text{ for } u \in
N.\]
\noindent For every $i \in [k]$, define the vector $b^i$ by
\[b^i_{s_i} = d_i, \ b^i_{t_i}=-d_i, \text{ and } b^i_u=0, \text{ for every }
u \in N \setminus \{s_i,t_i\}.\]
\noindent The cost of $x$ with respect to
commodity $i$ is $\sum_{uv \in E} c^i_{uv} x_{uv}$.\medskip

\noindent A multi-commodity flow in $G$ consists of a \emph{flow profile}
$x=(x^1,\dots,x^k)$ of flows for each commodity $i \in [k]$. The
individual flow $x^i$ is \emph{feasible} if $x^i$ is nonnegative and
satisfies the balance constraint $b(x^i)=b^i$. We will say that the flow profile $x$ is feasible if
$x^i$ is feasible for all~$i$. The total flow $\overline{x}$ induced
by such a multi-commodity flow $x=(x^1,\dots,x^k)$ is the sum
$\overline{x} = \sum_{i=1}^k x^i$.\medskip

\noindent For a directed path $P$ in $G$, we denote by $\gamma^P$ the unit path
flow given by 
\[\gamma^P_{uv}=1 \text{ for } uv \in P \text{ and } \gamma^P_{uv}=0
\text{ for } uv \notin P.\]
\noindent Similarly, for a directed cycle $C$ in $G$, we
denote by $\gamma^C$ the unit cycle flow given by 
\[\gamma^C_{uv}=1
\text{ for } uv \in C \text{ and } \gamma^C_{uv}=0 \text{ for } uv \notin C.\]  

\noindent The costs
$c^i(P)$ and $c^i(C)$ of $P$ and $C$ with respect to commodity $i$ are
given as 

\[c^i(P)=\sum_{uv \in P} c^i_{uv} \ \ \text{ and } \ \ c^i(C)=\sum_{uv \in C} c^i_{uv}.\]

 \noindent Note that $c^i(P) = c^i(\gamma^P)$ and
$c^i(C)=c^i(\gamma^C)$. Let $\calP_i$ denote the set of all directed
$(s_i,t_i)$ paths in $G$ and let $\calP = \cup_{i=1}^k \calP_i$. Let
$\calC$ denote the set of all directed cycles in $G$.

A \emph{routing} of commodity~$i$ in $G$ is a flow $x^i$ of the form
$x^i=\sum_{P \in \calP_i} \alpha_P \gamma^P$, where
$\sum_{P \in \calP_i} \alpha_P = d_i$ and $\alpha_P \geq 0$ for every
$P \in \calP_i$. Note that $x^i$ is feasible by definition. A
\emph{multi-commodity routing} is a multi-commodity flow
$(x^1,\dots,x^k)$ such that $x^i$ is a routing of commodity $i$ for
every $i \in [k]$. The network model we consider is \emph{uncapacitated}, i.e., it
does not place any upper limits on the flow on arcs. This means that a
minimum cost routing must place non-zero flow only on minimum cost paths.

\begin{lemma}
  \label{LEM:Min-Cost-Routing-Uses-Min-Weight-Paths}
  Suppose that $x^i=\sum_{P \in \calP_i} \alpha_P \gamma^P$ is a
  routing of commodity~$i$ minimizing the cost $c^i(x^i)$. Then
  $c^i(P) = \min_{P' \in \calP_i} c^i(P')$ whenever $\alpha_P > 0$.
\end{lemma}

\noindent It is well known that any flow may be decomposed path and cycle
flows. We shall make use of the following special case.

\begin{lemma}
  Let $x^i$ be a feasible flow for commodity~$i$. Then there are
  $\alpha_P \geq 0$ for $P \in \calP_i$ and $\beta_C \geq 0$ for
  $C \in \calC$ such that
  $x^i = \sum_{P \in \calP_i} \alpha_P \gamma^P + \sum_{C \in \calC}
  \beta_C \gamma^C$.
  \label{LEM:Path-Cycle-Decomposition}
\end{lemma}

\noindent If $b(x^i)=b^i$ the cost of $x^i$ with respect to commodity $i$ is
thus given as
$c(x^i) = \sum_{P \in \calP_i} \alpha_P c^i(P) + \sum_{C \in \calC}
\beta_C c^i(C)$. We have the following direct consequence.

\begin{corollary}
  \label{COR:Min-Cost-Flow-Is-Routing}
  Let $x^i$ be a feasible flow for commodity~$i$ minimizing the cost
  $c^i(x^i)$. If $c^i(C)>0$ for any cycle $C$ for which there exist
  $uv \in E$ such that $x^i_{uv}>0$, then $x^i$ is a routing of commodity~$i$.
\end{corollary}

\paragraph{Multi-class network congestion games.}
From a multi-commodity flow network $G=(V,E)$ we naturally obtain a
multi-class congestion game. The set of resources is simply the set
$E$ of arcs in the network and the set $\Sigma_i$ of pure strategies
of class~$i$ is equal to the set $\calP_i$ of directed $(s_i,t_i)$
paths in $G$. This gives a very succinct description of the strategies
that allows the relevant optimization problems to be formulated with
variables corresponding to resources rather than strategies. Note that
a load allocation of class~$i$ in the congestion game corresponds
exactly to a routing of commodity~$i$ in $G$.\medskip 

\noindent For the case of Wardrop equilibrium we consider the following linear
program in variables $y^i_{uv}$, where $y^i = (y^i_{uv})_{uv \in E}$, $x$ is a given non-negative
multi-commodity flow and $\overline{x}$ the resulting total flow.

\begin{center}\underline{Linear Program $\mathcal{P}^G_\mathrm{WE}$}\end{center}
\begin{equation*}
\begin{aligned}
\mbox{minimize}\quad & \sum_{uv \in E} \ell^i_{uv} (\overline{x}_{uv}) y^i_{uv}\\
\mbox{subject to}\quad 
& b(y^i) = b^i\\
& y^i\geq 0
\end{aligned}
\end{equation*}

\noindent This linear program is simply expressing minimum cost feasible flows
for commodity~$i$ with costs given by the latency functions and the
total flow $\overline{x}$. In order to identify these minimum cost
flows with minimum cost routing , using
\cref{COR:Min-Cost-Flow-Is-Routing}, we need to impose a very
mild restriction on the latency functions.
\begin{definition}\label{def:no-free-cycles}
  We say that the the latency functions $l^i_{uv}$ are without
  \emph{free cycles} if for every $i \in [k]$, every $C \in \calC$ and
  every feasible flow $x$, for which $x_{uv}>0$ for some $uv \in C$, we have
  $\sum_{uv \in C} \ell^i(x_{uv}) > 0$.
\end{definition}

\noindent We have the following fixed point characterization of Wardrop
equilibrium in network congestion games, using
\cref{LEM:Min-Cost-Routing-Uses-Min-Weight-Paths} and
\cref{COR:Min-Cost-Flow-Is-Routing}.
\begin{proposition}\label{prop:network-congestion-wardrop}
  In a network congestion games with latency functions without free
  cycles, a flow $x=(x^1,\dots,x^k)$ is a Wardrop equilibrium if and
  only $x^i$ is an optimal solution to the linear program
  $\mathcal{P}^G_\mathrm{WE}$ for all $i \in [k]$.
\end{proposition}

\noindent For the case of atomic splittable congestion games and with the
assumptions of \cref{LEM:Latency-convexity-assumption} we
consider the following convex program in variables in variables $y^i_{uv}$, where $y^i = (y^i_{uv})_{uv \in E}$.

\begin{center}\underline{Convex Program $\mathcal{C}^G_\mathrm{NE}$}
\begin{equation*}
\begin{aligned}
\mbox{minimize}\quad & \sum_{uv \in E} \ell^i_{uv} (y^i_{uv} + \sum_{j\neq i} x^j_{uv}) y^i_{uv}\\
\mbox{subject to}\quad 
& b(y^i) = b^i\\
& y^i \geq 0
\end{aligned}
\end{equation*}
\end{center}

\noindent Analogously to the case of the linear program
$\mathcal{P}^G_\mathrm{WE}$ we would like to ensure that optimal
solutions to $\mathcal{C}^G_\mathrm{NE}$ are routings of
commodity~$i$. Here however, the cost functions of the flow depend on
the flow variables, and we thus require a stronger assumption on the
latency functions.

\begin{proposition}\label{prop:network-congestion-nash}
  In a network congestion game with non-decreasing latency functions
  without free cycles and satisfying the assumptions of
  \cref{LEM:Latency-convexity-assumption}, a flow
  $x=(x^1,\dots,x^k)$ is a Nash equilibrium if and only if $x^i$ is an
  optimal solution to the convex program $\mathcal{C}^G_\mathrm{NE}$
  for all $i \in [k]$.
\end{proposition}
\begin{proof}
  We need to prove that each $x^i$ that is an optimal solution to
  $\mathcal{C}^G_\mathrm{NE}$ is a routing of commodity~$i$. By
  \cref{LEM:Path-Cycle-Decomposition} we have $\alpha_P \geq 0$
  for $P \in \calP_i$ and $\beta_C \geq 0$ for $C \in \calC$ such that
  $x^i = \sum_{P \in \calP_i} \alpha_P \gamma^P + \sum_{C \in \calC}
  \beta_C \gamma^C$. Suppose for contradiction that there exists
  $C \in \calC$ such that $\gamma_C>0$ and define the new flow
  $\hat{y}^i = y^i-\beta_C \gamma^C$. Since $\gamma^C$ is a cycle flow
  and $\beta_C \gamma^C \leq y^i$ it follows that $\hat{y}^i$ is
  feasible. Since the latency functions are non-decreasing it follows
  that
  $\ell^i_{uv}(\hat{y}^i_{uv}+\sum_{j\neq i}x^j_{uv}) \leq
  \ell^i_{uv}(\hat{y}^i_{uv}+\sum_{j\neq i}x^j_{uv})$ for all
  $uv \in E$, and thus also
  $\ell^i_{uv}(\hat{y}^i_{uv}+\sum_{j\neq i}x^j_{uv})\hat{y}^i_{uv}
  \leq \ell^i_{uv}(\hat{y}^i_{uv}+\sum_{j\neq i}x^j_{uv})y^i_{uv}$ for
  all $uv \in E$. Now, since the latency functions are without free
  cycles there exist $uv \in C$ such that
  $\ell^i(y^i_{uv}+\sum_{j\neq i}x^j_{uv})>0$. It follows that
  $\ell^i(\hat{y}^i_{uv}+\sum_{j\neq i}x^j_{uv})\hat{y}^i_{uv} \leq
  \ell^i(y^i_{uv}+\sum_{j\neq i}x^j_{uv})\hat{y}^i_{uv} <
  \ell^i(y^i_{uv}+\sum_{j\neq i}x^j_{uv})y^i_{uv}$. Taken together
  this contradicts the optimality of $y^i$.
\end{proof}

\subsubsection{PPAD-membership for non-atomic network congestion games}\label{sec:nonatomic-network-congestion-ppad}

The PPAD-membership result for finding Wardrop equilibria in multi-class network non-atomic congestion games follows almost immediately from the fixed point characterization of equilibria in \Cref{prop:network-congestion-wardrop}. The following result generalizes that of \citet{meunier2013lemke}.

\begin{theorem}\label{thm:PPAD-wardrop-network-congestion}
Finding a Wardrop equilibrium of a multi-class network non-atomic congestion game when the latency functions are without free cycles and they are given by \pseudogs is in PPAD.
\end{theorem}

\begin{proof}
Linear program $\mathcal{P}_\text{WE}^G$ can be interpreted as computing the best response of a player corresponding to class $i$, for the equilibrium concept of Wardrop equilibrium. $\mathcal{P}_\text{WE}^G$ obviously has a non-empty and bounded feasible domain. Since the latency functions $\ell_{uv}^i$ can be computed by \pseudogs, the subgradient of the objective function can be computed by a \pseudog, and hence the whole linear program can be computed by a \linoptgate. By using this \linoptgate as the oracle, the game becomes a \lbro game (\Cref{def:LBRO-game}), and the theorem follows as a corollary of \Cref{thm:LBRO-equilibria}.
\end{proof}

\noindent Again, a direct corollary of \Cref{thm:PPAD-wardrop-network-congestion} is to the case when the latency functions are piecewise-linear functions given explicitly as part of the input. Indeed, in that case the objective functions are piecewise-quadratic and their subgradients are piecewise-linear functions, hence computable by a \pseudog.

\begin{remark}[No free cycles]\label{rem:no-free-cycles-ppad}
As we mentioned earlier, the no free cycles condition (see \Cref{def:no-free-cycles} is rather mild one. It is in fact a milder condition that the usual condition in this literature, which requires that if $l_e(x)=0$, then $x=0$, i.e., the latencies are $0$ only on $0$ points. The works on linear latency functions implicitly make this assumption, and we obtain strict generalizations of those via considering more general latency functions. 
\end{remark}

\noindent For completeness, we also provide a simple example in which, if one goes beyond the latency functions captured by \Cref{thm:PPAD-wardrop-network-congestion}, e.g., to quadratic latency functions, then the game may possess only irrational Wardrop equilibria. 

\begin{example}[Only irrational Wardrop equilibria]\label{rem:wardrop-only-irrational}
Consider a very simple nonatomic network congestion game where there is only one class given by the interval $[0,1]$. The node set $V$ of the
network $G$ consists of two nodes $s$ and $t$. For the arc set $E$, there are two arcs $e_1$ and $e_2$ from $s$ to $t$ for the class, with latency functions $l_{e_1}(x) = x^2$ and $l_{e_2}(x) = 2x^2$. Now if $(x,1-x)$ is an equilibrium, then it must hold that $x^2 = 2\cdot (1-x)^2$ which as $x\in [0,1]$ implies that $x = 2-\sqrt{2}$.
\end{example}

\noindent \Cref{rem:wardrop-only-irrational} establishes that the class captured by our \Cref{thm:PPAD-wardrop-network-congestion} is tight, as if one moves to larger classes of latency functions, Wardrop equilibria may be irrational. We remark that for these cases, a FIXP-membership result follows from the results of \citep{SICOMP:Filos-RatsikasH2023} (althought not explicitly stated there), since multi-class non-atomic network congestion games are concave games. 

\subsubsection{PPAD-membership for atomic splittable network congestion games}\label{sec:atomic-network-congestion-ppad}

In the same fashion as in the previous section, the PPAD-membership result for finding Nash equilibria in multi-class network atomic congestion games follows almost immediately from the fixed point characterization of equilibria in \Cref{prop:network-congestion-nash}. The following result generalizes that of \citet{klimm2020complexity}. We also again note the remarkable simplification of the proof that the \linoptgate allows us to achieve, compared to that of \citet{klimm2020complexity}.

\begin{theorem}\label{thm:PPAD-Nash-network-congestion}
Finding a Nash equilibrium of a multi-class network atomic splittable congestion game when the latency functions are without free cycles, they are given by \pseudogs and their subgradients are given by \pseudogs computing piecewise-constant functions is in PPAD.
\end{theorem}

\begin{proof}
The proof is very similar to that of \Cref{thm:PPAD-wardrop-network-congestion} above. Convex program $\mathcal{C}_\text{NE}^G$ can be interpreted as computing the best response of a player corresponding to class $i$. $\mathcal{C}_\text{NE}^G$ obviously has a non-empty and bounded feasible domain. To compute the subgradient of the objective function, we can use the \pseudogs that compute the latency functions, as well as those that compute their subgradients. Those subgradients are piecewise-constant functions, which we then can multiply by the variables $y_{uv}$ by the machinery developed in \Cref{sec:implicit}. In the end, we obtain a function which can be computed by a \pseudog, and the whole linear program can be computed by a \linoptgate. By using this \linoptgate as the oracle, the game becomes a \lbro game (\Cref{def:LBRO-game} and the theorem follows as a corollary of \Cref{thm:LBRO-equilibria}.
\end{proof}

\noindent Once gain, we obtain the PPAD-membership of multi-class network atomic splittable congestion game when the latency functions are piecewise-linear functions given explicitly as part of the input as a direct corollary of \Cref{thm:PPAD-Nash-network-congestion}. The result is however more general, and also captures cases where the latency functions and their subgradients are accessed implicitly via their corresponding \pseudogs.
Additionally, similarly to the statement of \Cref{thm:PPAD-wardrop-network-congestion}, we remark that the no free cycles condition is implicitly made in the case of linear latencies, and hence we obtain a strict generalization of the previous results.  \medskip

\noindent Again, for completeness, we provide a simple example of a game that only has irrational equilibria, if the latency functions go beyond those captured by \Cref{thm:PPAD-Nash-network-congestion}.

\begin{example}[Only irrational Nash equilibria]\label{rem:nash-only-irrational}
The example is the same as the one in \Cref{rem:wardrop-only-irrational}. We consider a network with a single player and two arcs $e_1$ and $e_2$ from $s$ to $t$, with the same latency functions as before. Here, since we have an atomic splittable game, the single player has to minimize the latency $x\cdot l_1(x)+(1-x)\cdot l_2(1-x) = x^3+2\cdot (1-x)^3$ which for $x\in [0,1]$ is minimized in $x=2-\sqrt{2}$. 
\end{example}

\noindent Again,  \Cref{rem:nash-only-irrational} establishes that the class captured by our \Cref{thm:PPAD-Nash-network-congestion} is tight, as if one moves to larger classes of latency functions, Nash equilibria may be irrational. Again, for similar reasons as before, a FIXP-membership for general concave latency functions follows from the results of \citep{SICOMP:Filos-RatsikasH2023}.

\section{Competitive Equilibria in Arrow-Debreu Markets}\label{sec:markets}

In this section, we show how our \linoptgate can be used to show the PPAD-membership of finding competitive equilibria in Arrow-Debreu markets. We will consider such markets in which the utility and production functions have a certain form which allows equilibrium points to be rational numbers. For a gentle introduction, we first apply the technique to the simpler cases of exchange markets (i.e., no production) with linear utilities (\Cref{sec:exchange-markets-linear}) and of markets with linear utilities and productions (\Cref{sec:prod-markets-linear}), before we move to our most advanced applications.

The most general known class of utilities and productions for which rational competitive equilibria are possible are those of \emph{Leontief-free utilities and productions}, introduced by \citet{garg2018substitution}, who also proved the PPAD-membership of finding equilibria in those markets. We provide an alternative proof via the employment of the \linoptgate, which is \emph{conceptually} and \emph{technically} much simpler in \Cref{sec:general-markets}. 

The Leontief-free class generalizes the well-known class of \emph{separable piecewise linear concave (SPLC)} utilities, studied by several works, e.g., see \citep{vazirani2011market,garg2015complementary,SODA:GargV14}. SPLC utilities is in turn a generalization of linear utilities in which every agent has a piecewise linear concave utility for the amount of a good $j$ that she receives, and her total utility for her bundle is additive over goods.

\paragraph{SSPLC, a new class of utilities.} In \Cref{sec:SSPLC-markets}, we define a new class of utility functions, coined \emph{succinct separable piecewise linear concave (SSPLC) utilities}, which generalize the SPLC class in a different way than the Leontief-free class. Intuitively, these are functions which can have exponentially many pieces, but those are accessed \emph{implicitly}, via access to boolean circuits computing their slopes. In this sense, the class differs from SPLC utilities not in the definition of the utility function, but in the way that it is inputted in the computational problem.
As such, the existence of rational solutions for this class already implicitly follows from previous work \citep{SODA:GargV14} for SSPLC utilities. As we explain below however, the ``LCP-approach'' is inherently limited in its ability to establish computational results for this class, and in particular PPAD-membership. In contrast, by taking advantage of the machinery that we developed in \Cref{sec:implicit}, we show that our \linoptgate can be used in very much the same fashion as in all the other proofs of the section, to prove that computing competitive equilibria for SSPLC utilities is in PPAD. Our proof also incorporates (explicit) SPLC production functions; we provide a discussion on the challenges of extending our technique to the case of SSPLC production functions as well in \Cref{rem:ssplc-production-challenge}.

\subsubsection*{Features of our Proof vs Previous Approaches}

Before we proceed, we first discuss the previous results that appeared in a succession of important works in the area, and highlight the advantages of our approach over those. 

\paragraph{Previous work and previous proofs.} The origins of the literature that deals with the PPAD-membership of finding \emph{exact} and \emph{rational} competitive equilibria for markets for which this is possible can be traced back to the work of \citet{eaves1976finite}. \citeauthor{eaves1976finite} studied the exchange market model with linear utilities (the same one we present first in \Cref{sec:exchange-markets-linear}) and proved that its competitive equilibria can be found by Lemke's algorithm \citep{lemke1965bimatrix}. Referencing our discussion in \Cref{sec:Other-Approaches}, \citeauthor{eaves1976finite}' approach, similarly to most subsequent approaches, formulates the problem as an LCP. Clearly, the class PPAD had not been defined then, but as we explained in \Cref{sec:Other-Approaches}, the PPAD-membership of the problem is implied by his proof. 

\citet{vazirani2011market} considered the case of SPLC utilities. Their proof does not employ the ``LCP approach'', but rather the ``approximation and rounding approach'', see \Cref{sec:Other-Approaches}. An issue with this method is that very small changes in the prices may result in drastic changes in the optimal bundles of the consumers, and hence \citet{vazirani2011market} devise a set of technical lemmas that allow them to ``force'' certain allocations over others. Subsequently, \citet{garg2015complementary} proved the PPAD-membership of equilibria for the same class of utility functions, via instead using the ``LCP approach''. \citeauthor{garg2015complementary} point out that the LCP that they derive does not fall into any of the classes known to avoid secondary rays, and hence they necessarily devise an argument against ray termination (see \Cref{sec:Other-Approaches}). Besides that, the formulation of the LCP itself \emph{``turns out to be quite complex''} \citep{garg2015complementary}. An integral part of the proof, which is also useful in our regime, is the definition of the ``desire'' for a good $j$, to be the total amount of the good represented by its non-zero utility segments in the SPLC function.

For markets with production, the aforementioned work of \citet{vazirani2011market} also provides a PPAD-membership result for markets with SPLC utilities and SPLC productions. The same PPAD-membership result was later recovered by \citep{SODA:GargV14} via the ``LCP approach''; this is in fact the paper from which the quote of \Cref{sec:Other-Approaches} is taken. The quote highlights the increasing challenge of developing these LCPs and establishing their successful termination. Indeed, the LCP of \citep{SODA:GargV14} naturally differs from that of \citep{garg2015complementary}, in that it needs to account for production. This is done via devising a set of linear programs, and then using complementary slackness and their feasibility conditions to develop the LCP needed for production. The non-homogeneity of the resulting LCP for the equilibrium problem is dealt with in a manner which is different from previous works \citep{eaves1976finite,garg2015complementary} and, naturally, since the developed LCP is different, \citeauthor{SODA:GargV14} again need to argue against ray termination. 

The last of the works in this sequence regarding competitive markets for goods is by \citet{garg2018substitution}. \citeauthor{garg2018substitution} define a new class of utilities/productions, coined ``Leontief-free utilities/productions'', which generalize the SPLC functions mentioned earlier (see \Cref{sec:general-markets} for a precise definition). What is interesting about this class is that it contains functions which are not separable, but yet they admit rational solutions, hence dispelling a potential suspicion that non-separability is what essentially leads to irrationality. \citet{garg2018substitution} obtain the PPAD-membership of finding competitive equilibria in those markets via, again, employing the ``LCP approach''. Their LCP formulation turns out to be even more complex than those of previous works, as it has to differentiate between ``normal'' and ``abnormal'' variables, the latter preventing the employment of Lemke's algorithm. To circumvent this, they exploit some additional structure of their \emph{nonstandard} LCP, and then they also \emph{modify} Lemke's algorithm, to account for the possibility of abnormal variables becoming zero. Finally, as they devise a new LCP, they also have to argue once again against ray termination.

\paragraph{Our proofs.} We employ the \linoptgate to provide proofs which are conceptually and technically simpler than the ones discussed above. In particular, we formulate the optimal consumption (i.e., the consumers' utility maximization) and the optimal production (i.e, the firms' profit maximization) as linear programs (e.g., see \Cref{fig:production-markets-P1-and-P2}), which can be effectively substituted by \linoptgates in a \linear arithmetic circuit. Intuitively, one could view these as the ``parallel'' of best responses that we used in previous sections, in the domain of markets. What differentiates the proofs of this section is that these consumptions/productions need to be supported at a set of market-clearing prices. 

To make sure that we can work with linear constraints (in the inputs to \linear arithmetic circuit), we apply a standard variable change which was first used by \citet{eaves1976finite}, and then subsequently on all of the works that apply the ``LCP approach'': instead of sets of goods for consumption and production, we use the \emph{expenditures} of the consumers to buy goods in their bundles, and of the firms to produce goods using other goods as raw material. Interestingly, \citet{eaves1976finite} attributes this idea to Gale, and hence we refer to it as \emph{Gale's substitution}, see \Cref{rem:gale}. For the prices, we develop a feasibility program (e.g., see \Cref{fig:production-markets-Q}), which at a high level establishes that if certain goods (or segments) have smaller expenditures, then they are sufficiently cheaper. This is the main property that we use to argue that the prices computed by the feasibility program at a fixed point are market-clearing. 

For our results in \Cref{sec:SSPLC-markets}, we take advantage of the capability of the \linoptgate to compute implicit functions and correspondences, as detailed in \Cref{sec:implicit}. In contrast, the ``LCP approach'' would need to have explicit variables for each segment, resulting in an LCP of exponential size in the size of the input (which is the size of the given input circuits). This is sufficient to prove rationality of solutions, but not to construct the computational reduction needed for a PPAD-membership. 

Compared to the applications in other sections of our paper, the proofs that we provide in this section are probably the most technically involved. The main complication lies in arguing the market clearing of the outputted prices, which still requires a relatively short proof. We remark that most of the space in the following sections is used for introducing the corresponding settings and putting our technique in context, rather than the proofs themselves.

Finally, we emphasize that, as we mentioned also in \Cref{sec:Other-Approaches}, the ``LCP approach'' satisfies some other desirable properties, e.g., see \citep{garg2018substitution} for a discussion. As our focus here is the PPAD-membership of the problems, we do not discuss them further.

\paragraph{Roadmap.} We provide a brief roadmap of the section. In \Cref{sec:exchange-markets-linear} we introduce our technique in the simple setting of exchange markets with linear utilities. To gently incorporate production, in \Cref{sec:prod-markets-linear} we consider markets with linear utilities and linear productions. In \Cref{sec:general-markets} we consider the case of Leontief-free utilities and productions. Finally in \Cref{sec:SSPLC-markets} we prove the PPAD-membership of finding competitive equilibria in markets with SSPLC utilities and SPLC productions.

\subsection{Exchange Markets with Linear Utilities}\label{sec:exchange-markets-linear}

First, we explain how our \linoptgate can be used to show the PPAD-membership of finding competitive equilibria in a simple but fundamental variant of the Arrow-Debreu market model, that of \emph{exchange markets} with \emph{linear utilities}. An exchange market is one in which there is no production, and each consumer brings an endowment to the market. The PPAD-membership for these markets demonstrates the strength of the technique for a simple model, before moving on to apply it to more general market settings. \medskip

\noindent Our main theorem in this section is the following: 

\begin{theorem}\label{thm:exchange-market-PPAD}
    Computing a competitive equilibrium in an exchange market with linear utilities is in PPAD.
\end{theorem}

\paragraph{Exchange Markets.} In an exchange market $\mathcal{M}$, we have a set $N$ of consumers and a set $G$ of infinitely divisible goods. Let $n=|N|$ and $m=|G|$. We will typically use index $i$ to refer to consumers and $j$ or $g$ to refer to goods. Each consumer brings an endowment $w_i = (w_{i1}, \ldots, w_{im})$ to the market, with $w_{ij} \geq 0$ for all $i \in N$ and $j \in G$. We may assume without loss of generality that for every good $j$, we have $\sum_{i \in N} w_{ij} =1$, i.e., that the total endowment of each good is $1$. We will use $\mathbf{x}_i = (x_{i1}, \ldots, x_{im})$ to denote the vector of quantities of goods allocated to consumer $i \in N$ in $\mathcal{M}$, and we will call it the \emph{bundle} of consumer $i$. Let $\mathbf{x} = (\mathbf{x}_1, \ldots, \mathbf{x}_n)$ be the vector of such bundles. We will use $\mathbf{p} = (p_1, \ldots, p_m)$ to denote the vector of \emph{prices} in $\mathcal{M}$, one for each good $j \in G$. Prices are non-negative, so $p_j \geq 0$ for all $j \in G$. Given a vector of prices $\mathbf{p}$ the \emph{budget} of consumer $i \in N$ is defined as $\sum_{j \in G}w_{ij}p_j$; intuitively, this is the amount of money that the consumer acquires by selling her endowment at prices $\mathbf{p}$.  

\paragraph{Utility Functions.} Every consumer has a utility function $u_i : \mathbb{R}_{\geq 0}^m \rightarrow \mathbb{R}_{\geq 0}$ mapping a bundle $\mathbf{x}_i$ to a non-negative real number. In this section, these utilities are \emph{linear}, i.e., every consumer $i \in N$ has a utility $u_{ij}$ for every good $j \in G$, and her utility for the bundle $\mathbf{x}_i$ is $u_i(\mathbf{x}_i) = \sum_{j \in G} u_{ij}x_{ij}$, where $u_{ij}$ is the value of consumer $i$ for good $j$.  

\paragraph{Competitive Equilibrium.} Next, we provide the definition of a competitive equilibrium (often referred to as a \emph{market} or a \emph{Walrasian} equilibrium).

\begin{definition}[Competitive equilibrium]\label{def:exchange-markets-competitive-equilibrium}
A competitive equilibrium of an exchange market $\mathcal{M}$ is a pair $(\mathbf{p^*}, \mathbf{x^*})$ consisting of a prices $\mathbf{p^*}$ and and bundles $\mathbf{x^*} = (\mathbf{x}_1^*, \ldots, \mathbf{x}_n^*)$ such that 
\begin{enumerate}
    \item $u_i(x_i^*) = \sum_{j \in G}x_{ij}^* u_{ij}$ is maximized, subject to \\$\sum_{j \in G} x_{ij}^{*} p_j^{*} \leq \sum_{j \in G} w_{ij}p_j^{*}$ and $x_{ij}^{*} \geq 0$, for any consumer $i \in N$. \hfill \textbf{\emph{(bundle optimality)}} \label{con1:optimality}
    \item $\sum_{i \in N}x_{ij}^*=1$, for any good $j \in G$. \hfill \textbf{\emph{(market clearing)}} \label{con2:clearing}
\end{enumerate}
\end{definition}

\noindent Condition~\ref{con1:optimality} above guarantees that at the equilibrium prices $\mathbf{p}^*$, every consumer receives the best possible bundle that they can buy given their budget and that each good is allocated in a non-negative amount. Condition~\ref{con2:clearing} guarantees that the market clears, i.e., that the total quantity of goods sold is equal to the total endowment $\sum_{i \in N} w_{ij}$, which, recall, we have assumed without loss of generality that is equal to $1$. 

\paragraph{Sufficiency conditions.} A competitive equilibrium as defined above exists for every exchange market $\mathcal{M}$, under some \emph{sufficiency conditions}. Following \citet{vazirani2011market}, we will use the sufficiency conditions used by \citet{eaves1976finite}, namely that:
\begin{enumerate}
    \item For every good $j \in G$, there exists some consumer $i \in N$ that values the good positively, i.e., $u_{ij} >0$. \label{suff:1}
    \item For every consumer $i \in N$, there exists some good $j \in G$ that the consumer endows in a positive amount, i.e., $w_{ij} >0$.\label{suff:2}
    \item The market is \emph{not reducible}. A market is reducible where there exists a proper subset $N' \subset N$ of the consumers and a subset $G' \subseteq G$ of the goods such that
    \begin{enumerate}
        \item all of the goods in $G'$ are entirely endowed by consumers in $N'$, i.e., $\sum_{i \in N'}w_{ij}=1$ for all $j \in G$,
        \item all of the consumers in $N'$ only value positively goods in $G'$, i.e., for all $i \in N'$ and $g \in G\setminus G'$, it holds that $u_{ig}=0$. 
    \end{enumerate}
    Reducibility intuitively captures the fact that if a strict subset of consumers completely endow a subset of the goods and only value those goods positively, then they could form their own market $\mathcal{M'}$ and exchange between them. Looking ahead, the non-reducibility of the market is a special case of a property called \emph{strong connectivity of the economy graph} \citep{maxfield1997general}, which we will make use of for the more general market settings that we consider later on. \label{suff:3}
\end{enumerate}

\paragraph{Optimality and bang-per-buck (BPB).} The bundle optimality condition of the competitive equilibrium stipulates that each consumer buys the best possible bundle at the given prices. For linear markets, these optimal bundles have a crisp characterization, in terms of the \emph{bang-per-buck (BPB)}. Given consumer $i \in N$ and prices $\mathbf{p}$, the BPB of a good $j \in G$ is defined as $\text{BPB}_i(j)=\frac{u_{ij}}{p_j}$. An optimal bundle only contains non-zero quantities of goods for which the BPB is maximum, i.e., $j \in \arg\max_{j \in G} \text{BPB}_i(j)$. \medskip 

\noindent \emph{Bounds on the prices.} We may without loss of generality assume that all the prices are strictly positive, i.e., $p_j>0$ for all $j\in G$. Indeed, if $p_j=0$, then there would be infinite demand for good $j$, contradicting market clearing. As such, in any equilibrium, the prices are strictly positive. Note that by this assumption, the quantity $\text{BPB}_i(j)$ is well-defined for every $j \in G$.

\begin{remark}[Normalized Prices]\label{rem:exchange-markets-normalized-prices}
Given that $p_j > 0$ for all $j \in G$, we can normalize the prices to sum to $1$ without loss of generality, i.e., we may assume that $\sum_{j \in G} p_j = 1$. 
\end{remark}

We will use a parameter $\varepsilon$ to capture the fact that if the price $p_j$ for a good $j \in G$ is sufficiently smaller than the price $p_{j'}$ for a good $j' \in G$, then $\text{BPB}(j) > \text{BPB}(j')$. This will allow us to ``control'' which goods certain consumers buy in their optimal bundles. Specifically, we can compute $\varepsilon > 0$ such that
\[
\text{If } p_j \leq \varepsilon \cdot p_{j'} \text{ and } u_{ij} > 0 \text{ then } \text{BPB}_i(j) > \text{BPB}_i(j').
\]
Additionally, we can pick $\varepsilon$ to be sufficiently small such that $\varepsilon < \frac{w_{ij}}{m}$ for all $i \in N$ and all $j \in G$ with $w_{ij}>0$.

Given $\varepsilon$, we will impose a stricter lower bound on the prices, which will be useful later on: in particular, we will assume that for all $j \in G$, $p_j \geq \frac{\varepsilon^m}{m}$. Our PPAD-membership proof will also establish that a competitive equilibrium still always exists even under this additional restriction.

\subsubsection{Preprocessing}

The high-level overview of the proof will be the following. We will define a function $F$ from a convex compact domain $D$ to itself, in a way that ensures that a competitive equilibrium of $\mathcal{M}$ can be recovered from a fixed point of $F$. The input and output variables of the function will be outputs of a linear program $\mathcal{P}$ (see Program~\eqref{eq:OPT-gate-linear} in \Cref{sec:lin-opt-gate}) and a feasibility program $\mathcal{Q}$ (see Program~\eqref{eq:feasibility-general} in \Cref{sec:lin-opt-gate}). Intuitively, $\mathcal{P}$ will be used to compute optimal bundles $\mathbf{x}_i$ for the consumers, and $\mathcal{Q}$ will be used to compute equilibrium prices $\mathbf{p}$. These programs will be of a particular form amenable to the use of our \linoptgate. As such, we will be be able to encode $F$ via a \linear arithmetic circuit containing \linoptgate.

\paragraph{Straightforward choice: $\mathbf{x}$ and $\mathbf{p}$.} A straightforward choice would be to have the pair $(\mathbf{x}, \mathbf{p})$ as the input/output to the function $F$. In that case, the linear program $P$ would be the following:
\begin{align}\label{LP-markets-straightforward}
	\mbox{maximize } & \sum_{i \in N}\sum_{j \in M} u_{ij}x_{ij}  \nonumber\\
	\mbox{subject to }& \sum_{j \in G} x_{ij}p_j \leq \sum_{j \in G}w_{ij}p_j,\;\forall i \in N \\
        & x_{ij} \geq 0, \;\forall i \in N, j \in G \nonumber
\end{align}

\noindent While the variables of this linear program are the variables $x_{ij}$, the prices $p_j$ are \circparams (see \Cref{def:circuit-parameters}), and the budget constraints of the form $\sum_{j \in G} x_{ij}p_j \leq \sum_{j \in G}w_{ij}p_j$ cannot be present in a linear program that we would like to compute by the \linoptgate.

\paragraph{Change of variables: Expenditure.} To circumvent this obstacle, we follow an idea used by \citet{eaves1976finite} that ``linearizes'' the constraints via an appropriate change of variables. \citet{eaves1976finite} attributes this idea to Gale, via private communication. Thus, we refer to this change of variables as \emph{Gale's substitution}.

\begin{remark}[Gale's substitution]\label{rem:gale}
Let $q_{ij} = x_{ij}p_j$ be the \emph{expenditure} of consumer $i$ on good $j$, i.e., how much money the consumer spends on the good given an allocation $x_{ij}$ and a price $p_j$. This change of variables is by now very much standard in the literature, used among others by \citet{eaves1976finite,SODA:GargV14,garg2018substitution,chaudhury2022complementary}.
\end{remark}

\noindent With this at hand, the budget constraints now become 
\[\sum_{j \in G} q_{ij} \leq \sum_{j \in G} w_{ij} p_j.\] 
These are clearly linear in the variables $q_{ij}$ with no \circparams appearing in the expression. In turn, the objective function of $\mathcal{P}$ becomes
\[
\sum_{i \in N} \sum_{j \in M} \frac{u_{ij}}{p_j}q_{ij}.
\]
Interestingly, the coefficient of each variable $q_{ij}$ is the bang-per-buck $\text{BPB}_i(j)$ of good $j$ for consumer $i$. In an optimal solution to $\mathcal{P}$, the consumer will only spend money on goods with maximum BPB. That being said, the format of the objective function (with $p_j$ appearing in the denominator) is not one that can be handled by the \linoptgate, as its subgradient cannot be computed by a \pseudog. For this reason, we transform the objective function to 
\[
\sum_{i \in N} \sum_{j \in M} \frac{p_j}{u_{ij}}q_{ij}.
\]
and the maximization in linear program~(\ref{LP-markets-straightforward}) to minimization. Clearly, in an optimal solution, the consumer is again only spending on goods with minimum ratio $\frac{p_j}{u_{ij}}$, i.e., with maximum BPB, therefore purchasing an optimal bundle. This objective function form is of a form amenable to the use of the \linoptgate, as its subgradient is a linear function.

However, there is a potential issue that arises when converting the maximization objective to minimization. If we keep the budget constraint $\sum_{j \in G} q_{ij} \leq \sum_{j \in G} w_{ij} p_j$ intact, then setting $q_{ij} = 0$ for all $i \in N$ and $j \in G$ is an optimal solution to the linear program, while it is clearly not an expenditure consistent with a competitive equilibrium. To handle this, we observe that in an optimal solution to linear program~(\ref{LP-markets-straightforward}), where the objective is maximization, we may assume that all the budget constraints have zero slack, i.e., that $\sum_{j \in G} x_{ij}^*p_j^* = \sum_{j \in G}w_{ij}p_j^*$. Indeed, if one could increase some variable $x_{ij}$ without violating the budget constraint for some consumer $i \in N$, then the consumer would certainly not receive a lower utility in the objective function. Correspondingly, the constraints $\sum_{j \in G} q_{ij} \leq \sum_{j \in G} w_{ij} p_j$ can be converted to equality constraints, i.e., $\sum_{j \in G} q_{ij} = \sum_{j \in G} w_{ij} p_j$. In turn, these constrain the expenditure variables to receive positive values in an optimal solution, ensuring that each consumer spends their entire budget. \medskip

\noindent Summarizing, the linear program $\mathcal{P}$ can be seen on the left-hand side of \Cref{fig:exchange-markets-P-and-Q}.

\begin{figure}
\centering 
\fbox{
\centering
    \begin{minipage}{.45\textwidth}
        \centering
        \underline{Linear Program $\mathcal{P}$}
        \begin{align*}%
	\mbox{minimize } & \sum_{i \in N}\sum_{j \in G} \frac{p_j}{u_{ij}}q_{ij}  \\
	\mbox{subject to }& \sum_{j \in G} q_{ij} = \sum_{j \in G}w_{ij}p_j \\
        & q_{ij} \geq 0, \;\forall j \in G 
    \end{align*}
    \end{minipage}%
    \hfill\vline\hfill
    \begin{minipage}{0.45\textwidth}
        \centering
        \underline{Feasibility Program $\mathcal{Q}$}
        \begin{align*}
	& e_j<e_{j'}\Rightarrow p_j\leq \varepsilon \cdot p_{j'},\;\forall j,j' \in G \nonumber \\
	& p_j\geq \frac{\varepsilon^m}{m},\;\forall j \in G\\
	& \sum_{j \in G}p_j = 1\nonumber
        \end{align*}
    \end{minipage}
    }
\caption{The linear program $\mathcal{P}$ used to recover optimal bundles (left) and the feasibility program $\mathcal{Q}$ used to find market-clearing prices (right).}
\label{fig:exchange-markets-P-and-Q}
\end{figure}

\paragraph{Excess expenditure and the program $\mathcal{Q}$.} While linear program $\mathcal{P}$ will be used to obtain the optimal expenditures $q_{ij}$ (and as consequence the optimal bundles $x_{ij}$), as we mentioned earlier, the equilibrium prices will be obtained via a feasibility program $\mathcal{Q}$. Before we define the program, we introduce the notion of excess expenditure, which will be useful later on. The \emph{excess expenditure} $e_j$ of a good $j \in G$ is defined as the difference between the total expenditure of all consumers $i \in N$ for that good and the price of the good, i.e., 
\[
e_j = \sum_{i \in N}q_{ij} - p_j.
\]
Then, at equilibrium prices $\mathbf{p}^*$ we will have market clearing, i.e., $e_j^* = 0$ for all $j\in G$. \medskip

\noindent We are now ready to define our feasibility program $\mathcal{Q}$; see the right-hand side of \Cref{fig:exchange-markets-P-and-Q}. The equilibrium prices will be obtained as the output $\mathbf{p}^{*}$ of $\mathcal{Q}$. 

\subsubsection{Membership in PPAD: The proof of \crtCref{thm:exchange-market-PPAD}.}

\noindent We will develop the proof in three steps, namely (a) construction of the function $F$ and arguing that it can be represented by a \linear arithmetic circuit containing \linoptgates, (b) showing that the \linoptgate can compute all the necessary components, and (c) arguing that a competitive equilibrium can be recovered from a fixed point of $F$. 

\paragraph{The function $F$.} Given the above, we will define $F: D \rightarrow D$ with domain 
\[D = \left\{p \in \Delta^{m-1}: \forall j \in G, p_j \geq \frac{\varepsilon^m}{m}\right\} \times [0,1]^{nm}.\]
An input to $F$ is a pair $(\bar{\mathbf{p}}, \bar{\mathbf{q}})$ of prices and expenditures, where $\bar{\mathbf{q}} = (\bar{q}_{ij})_{i \in N, j \in G}$ and the output is another such pair $(\mathbf{p}^*, \mathbf{q}^*)$. The domain $D$ is the one above since $\sum_{j \in G} \bar{p}_j = 1$ (see \Cref{rem:exchange-markets-normalized-prices}), and by the fact that in any feasible set of expenditures, $\sum_{j \in G} \bar{q}_{ij} \leq \sum_{j \in G} w_{ij} \bar{p}_j \leq 1$.

\begin{remark}\label{rem:input-variables-lhs-feasibility}
Note that while the linear program $\mathcal{P}$ and the feasibility program $\mathcal{Q}$ were written in a general form in \Cref{fig:exchange-markets-P-and-Q}, when using $(\bar{\mathbf{p}}, \bar{\mathbf{q}})$ as the input, we would have to substitute $p_{j}$ with $\bar{p}_{j}$ in $\mathcal{P}$ (with the variables being $q_{ij}$) and $q_{ij}$ and $p_{ij}$ with $\bar{q}_{ij}$ and $\bar{p}_{j}$ respectively in the expression for $e_j$ in $\mathcal{Q}$ (with the variables being $p_j$). While we do not write this out explicitly in all of the applications in our paper, here we thought it would be important to mention, in order to avoid any confusion with $p_j$ potentially being on the left-hand side of the conditional constraints in $\mathcal{Q}$. It is $\bar{p}_j$ that appears on the left-hand side, which is a \circparam rather than a variable, and thus is in accordance with the definition of feasibility programs in \eqref{eq:feasibility-general} of \Cref{sec:lin-opt-gate}.
\end{remark}

\paragraph{Computation by the \linoptgate.} In the next lemma, we argue formally what we have alluded to earlier, namely that the solutions to the linear program $\mathcal{P}$ and the feasibility program $\mathcal{Q}$ of \Cref{fig:exchange-markets-P-and-Q} can be computed by our \linoptgate.

\begin{lemma}\label{lem:ExchangeMarkets-OPT-gate-can-compute}
Consider the linear program $\mathcal{P}$ and the feasibility program $\mathcal{Q}$ of \Cref{fig:exchange-markets-P-and-Q}. An optimal solution to $\mathcal{P}$ can be computed by the \linoptgate. $\mathcal{Q}$ is solvable, and a solution can be computed by the \linoptgate. 
\end{lemma}

\begin{proof}
    For the linear program $\mathcal{P}$, we observe that the feasible domain $[0,1]^n$ is non-empty and bounded. The \circparams $p_j$ appear only on the right-hand side of the constraints, and the subgradient of the objective function is linear, and hence can be given by a \pseudog. For the feasibility program $\mathcal{Q}$, there are no variables appearing in the left-hand side of the conditional constraints (see also \Cref{rem:input-variables-lhs-feasibility}), and no \circparams appearing on the right-hand side. This means that the \linoptgate correctly computes the outcome of the program, assuming that the program is solvable. To argue solvability, notice that $\mathcal{Q}$ is of the form $\mathcal{Q}_{\text{app}}$ as defined in \Cref{sec:linopt-applications}, and hence by \Cref{lem:feasibility-of-q-graph} it suffices to argue that its \qgraph $G_\mathcal{Q}$ is acyclic. In our case here, the vertices $V_{\mathcal{Q}}$ of $G_{\mathcal{Q}}$ are the goods, and an edge $(j, j')$ exists if and only if $e_j < e_{j'}$. Clearly, $G_{\mathcal{Q}}$ is acyclic. 
\end{proof}

\paragraph{Arguing optimality and market clearing.} To conclude the proof, what is left to show is that a fixed point $(\mathbf{p}^*, \mathbf{q}^*)$ of $F$ indeed corresponds to a competitive equilibrium. We argue that in the following lemma.

\begin{lemma}\label{lem:exchange-markets-correctness-lemma}
Let $(\mathbf{p}^*, \mathbf{q}^*)$ be a fixed point of $F$. Then the pair $(\mathbf{p}^*, \mathbf{x}^*)$, where $x_{ij}^* = q_{ij}^*/p_j^*$ is a competitive equilibrium of $\mathcal{M}$.
\end{lemma}

\begin{proof}
By the form of the objective function of the linear program $\mathcal{P}$ of \Cref{fig:exchange-markets-P-and-Q}, each consumer spends money only on goods with maximum BPB and therefore the allocation that she receives is only one of goods for which she has maximum BPB. Therefore, the consumer receives an optimal bundle which satisfies Condition~\ref{con1:optimality} of the competitive equilibrium definition in \Cref{def:exchange-markets-competitive-equilibrium}. The allocation quantities $x_{ij}^*$ can straightforwardly be recovered from the values of $q_{ij}^*$. 

What remains to show is that $\mathbf{p^*}$ is a vector of market-clearing prices. By the definition of $e_j$, this is equivalent to arguing that for all $j \in G$, we have that $e_j^*=0$. We first argue that $\sum_{j \in G}e_j^*=0$. Indeed:
\begin{align*}
\sum_{j \in G} e_j^* &= \sum_{j \in G}\left(\sum_{i \in N}q_{ij}^* - p_j^*\right) = \sum_{j\in G}\sum_{i \in N}q_{ij}^* - \sum_{j \in G} p_j^* = \sum_{j \in G}\sum_{i \in N} w_{ij}p_j^* - 1 \\
&= \sum_{j \in G} p_j^* \cdot \sum_{i \in N} w_{ij} - 1 = \sum_{j \in G}p_j^* - 1 = 0,
\end{align*}
where in the calculation above we used:
\begin{itemize}
    \item[-] in Equations 3 and 6, that $\sum_{j \in G}p_j^* = 1$, which is without loss of generality (see \Cref{rem:exchange-markets-normalized-prices}),
    \item[-] in Equation 5, that $\sum_{i \in N}w_{ij}=1$, which is without loss of generality, and
    \item[-] in Equation 3, that $\sum_{j \in G} \sum_{i \in N}q_{ij}^* = \sum_{j \in G} \sum_{i \in N}w_{ij}p_j^{*}$, which follows from the constraints of linear program $\mathcal{P}$ of \Cref{fig:exchange-markets-P-and-Q}. 
\end{itemize}

\noindent Define the set $J$ to be the set of goods with minimum excess expenditure, i.e., 
\[
J = \{j \in G: e_j^* \leq e_{j'}^*, \text{ for all } j' \in G\}.
\]
It now suffices to show that $e_j^*\geq 0$ for all $j\in J$. Assume by contradiction that there exists some $j_1 \in J$ such that $e_{j_1}^* <0$. By definition of the set $J$, we have that
\begin{equation}\label{eq:min-expenditure}
\text{for all } j \in J, \text{ we have } e_j^* < 0.
\end{equation}

\noindent Also, by the fact that $\sum_{j \in G}e_j^*=0$ which we established above, it must be the case that there also exists $j_2 \in G\setminus J$ such that $e_{j_2}^* > 0$. In particular $G\setminus J\neq\emptyset$. Define
\[
N' = \{i \in N: \text { there exists } j\in G\setminus J \text{ such that } w_{ij} > 0\}.
\]

\noindent In words, $N'$ is the set of consumers with strictly positive endowment for some good which is not in the set $J$. We will argue that
\[
\text{there exists } i_0 \in N' \text{ and } j \in J \text{ such that } u_{i_0 j} > 0.
\]
To see this, assume by contradiction that for all $i \in N'$ and $j \in J$, it was the case that $u_{ij}=0$. Note also that for any good $j' \in G\setminus J$, it holds that $w_{ij'} >0$ only for consumers $i \in N'$. That would imply that $N'$ and $G\setminus J$ would constitute a reduced market, contradicting Feasibility Condition~\ref{suff:3}. 

Consider any good $j \in J$ and any good $j' \in G \setminus J$ such that $w_{i_0j'>0}$. By the definition of $J$, we have that $e_j^* < e_{j'}^*$. By the fact that $\mathbf{p}^*$ is a solution to the feasibility program $\mathcal{Q}$ of \Cref{fig:exchange-markets-P-and-Q}, we have that $p_{j}^* \leq \varepsilon \cdot p_{j'}^*$.  We have the following inequalities:
\begin{equation}\label{eq:two-inequalities}
\sum_{g \in G}w_{i_0 g}\cdot p_{g}^* \geq w_{i_0 j'}\cdot p_{j'}^* \ \ \  \text{ and } \ \ \ p_{j'}^* \geq \frac{m \cdot p_j^*}{w_{i_0 j'}}.
\end{equation}
The first inequality above is trivial, since $j' \in G$. The second inequality follows since by the choice of $\varepsilon$,
\[
p_{j'}^* \geq \frac{p_j^*}{\varepsilon} \geq \frac{p_j^*}{w_{i_0j'}/m} = \frac{m \cdot p_j^*}{w_{i_0j'}}.
\]
Since $j$ was chosen to be any good in $J$, the inequalities in (\ref{eq:two-inequalities}) hold for $j$ such that $p_{j}^*$ is the maximum price amongst all goods in $J$. Therefore in (\ref{eq:two-inequalities}) we can substitute $p_j^*$ by $\max_{g \in J}p_g^*$ and we have
\begin{equation}\label{eq:enough-budget-to-buy-J}
\sum_{g \in G}w_{i_0 g}\cdot p_{g}^* \geq w_{i_0 j'}\cdot p_{j'}^* \geq m \cdot \max_{g \in J} p_{g}^* \geq \sum_{g \in J} p_g^*.
\end{equation}
Inequality~\ref{eq:enough-budget-to-buy-J} implies that consumer $i_0$ has enough budget to buy all the goods in $J$. At the same time, by the choice of $\varepsilon$, consumer $i_0$ prefers to buy quantities of good $j$ rather than $j'$, i.e., $\text{BPB}_{i_0}(j) > \text{BPB}_{i_0}(j')$. Since consumer $i_0$ only spends money on goods with maximum BPB, we have that 
\begin{equation}\label{eq:only-expenditure-on-J}
\sum_{g \in G}q_{i_0 g}^* = \sum_{g \in J}q_{i_0 g}^*. 
\end{equation}
From Inequality~\ref{eq:enough-budget-to-buy-J}, Equation~\ref{eq:only-expenditure-on-J}, and since $\mathbf{q^*}$ is a feasible solution to linear program $\mathcal{P}$ of \Cref{fig:exchange-markets-P-and-Q} (and hence satisfies the constraints $\sum_{g \in G} q_{i_0g}^* = \sum_{g \in G}w_{i_0 g}\cdot p_g^*$), we have that 
\[
\sum_{g \in J}q_{i_0 g}^* \geq \sum_{g \in J}p_g^*.
\]
This last inequality however implies that there must exist $j_0 \in J$ for which $q_{i_0 j_0}^* \geq p_{j_0}^*$ and therefore $e_{j_0}^* = \sum_{i \in N} q_{ij_0}^* - p_{j_0}^* \geq 0$, contradicting (\ref{eq:min-expenditure}). This completes the proof.
\end{proof}

\subsection{Arrow-Debreu Markets with Linear Utilities and Productions}\label{sec:prod-markets-linear}

We now move on to the next fundamental variant of the Arrow-Debreu market model, that of markets with \emph{linear utilities} as well as \emph{linear productions}. Our results about exchange markets in \Cref{sec:exchange-markets-linear} were used to demonstrate the application of our technique to one of the simplest variants of the main market model. The results of this section can be seen as extending this exposition to the case where the markets also have linear production functions. This is rather informative for the reader, as it illustrates the approach that will be used for the most general market setting that we prove PPAD-membership for, which we present in \Cref{sec:general-markets}. \medskip

\noindent Our main theorem in this section is the following: 

\begin{theorem}\label{thm:production-market-PPAD}
    Computing a competitive equilibrium in an Arrow-Debreu market with linear utilities and linear productions is in PPAD.
\end{theorem}

\noindent We provide the main definitions for Arrow-Debreu markets with linear utilities and linear productions below. To ensure that the section is self-contained, we have elected to fully define the setting, rather than to only highlight the changes to the definitions of exchange markets of \Cref{sec:exchange-markets-linear}. Still, to avoid unnecessary repetition, in the preprocessing and setup steps for the proof, we do refer to the appropriate arguments presented in \Cref{sec:exchange-markets-linear}.

\paragraph{Markets with Production.} In an Arrow-Debreu market with production $\mathcal{M}$, we have a set $N$ of consumers, a set $G$ of infinitely divisible goods, and a set $F$ of firms. Let $n=|N|$, $m=|G|$, and $\ell = |F|$. We will typically use index $i$ to refer to consumers, $j$ or $g$ to refer to goods and $f$ to refer to firms. Each consumer brings an endowment $w_i = (w_{i1}, \ldots, w_{im})$ to the market, with $w_{ij} \geq 0$ for all $i \in N$ and $j \in G$. We may assume without loss of generality that for every good $j$, we have $\sum_{i \in N} w_{ij} =1$, i.e., that the total endowment of each good is $1$. We will use $\mathbf{x}_i = (x_{i1}, \ldots, x_{im})$ to denote the vector of quantities of goods allocated to consumer $i \in N$ in $\mathcal{M}$, and we will call it the \emph{bundle} of consumer $i$. Let $\mathbf{x} = (\mathbf{x}_1, \ldots, \mathbf{x}_n)$ be the vector of such bundles. We will use $\mathbf{p} = (p_1, \ldots, p_m)$ to denote the vector of \emph{prices} in $\mathcal{M}$, one for each good $j \in G$. Prices are non-negative, so $p_j \geq 0$ for all $j \in G$. Given a vector of prices $\mathbf{p}$, the \emph{budget} of consumer $i \in N$ is defined as $\sum_{j \in G}w_{ij}p_j$; intuitively, this is the amount of money that the consumer acquires by selling her endowment at prices $\mathbf{p}$.  

\paragraph{Utility Functions.} Every consumer has a utility function $u_i : \mathbb{R}_{\geq 0}^m \rightarrow \mathbb{R}_{\geq 0}$ mapping a bundle $\mathbf{x}_i$ to a non-negative real number. In this section, these utilities are \emph{linear}, i.e., every consumer $i \in N$ has a utility $u_{ij}$ for every good $j \in G$, and her utility for the bundle $\mathbf{x}_i$ is $u_i(\mathbf{x}_i) = \sum_{j \in G} u_{ij}x_{ij}$, where $u_{ij}$ is the value of consumer $i$ for good $j$.  

\paragraph{Firm Shares.} Each consumer $i \in N$ has a share $\theta_{if} \in [0,1]$ of the profit of each firm $f \in F$. We assume that the profits are entirely shared among the consumers, i.e., for every firm $f \in F$, we have that $\sum_{i \in N}\theta_{if}=1$. 

\paragraph{Production Functions and Conversion Rates.} Each firm $f \in F$ produces a set of goods using a set of goods as raw material. Following \citet{SODA:GargV14}, for simplicity we will assume without loss of generality that each firm produces a single good. We may also assume without loss of generality that each firm uses a single good as raw material for the production. This because the production functions that we will consider are separable, and hence a firm using/producing multiple goods can be split into multiple firms using/producing single goods, with the shares of the agents' being duplicated (see also \citep{SODA:GargV14}). Given the above we will let:
\begin{itemize}
    \item[-] $\fout$ be the good produced by firm $f \in F$; we will refer to this as the \textbf{\emph{output good}},
    \item[-] $\fin$ be the good used as raw material by firm $f \in F$; we will refer to this as the \textbf{\emph{input good}}.
\end{itemize}

\noindent For every firm $f \in F$ there is a function $P_{f}$ which determines the firm's ability to produce units of the output good $\fout$ as a function of quantities of the input good $\fin$. In this section, we will assume that these functions are linear, of the form $P_{f}(y) = c_f \cdot y$, where $c_f \geq 0$ is a fixed \emph{conversion rate} for firm $f$, specifying that the firm can use $y$ units of good $\fin$ to produce $c_f \cdot y$ units of good $\fout$. Given a conversion rate $c_f$, and prices $p_{\fout}$ and $p_{\fin}$ for the goods $\fout$ and $\fin$ respectively, the \emph{profit} of $f$ from using $y$ units of $\fin$ to produce $c_f \cdot y$ units of $\fout$ is defined as $p_{\fout} \cdot c_f \cdot y - p_{\fin}\cdot y$.

\paragraph{Competitive Equilibrium.} We are now ready to define the notion of a competitive equilibrium in markets with production.

\begin{definition}[Competitive Equilibrium - Markets with Production]\label{def:prod-markets-competitive-equilibrium}
A competitive equilibrium of an Arrow-Debreu market with linear utilities and linear production functions is a triple $(\mathbf{p}^*,\mathbf{x^*}, \mathbf{y^*})$ consisting of non-negative prices $\mathbf{p}^*$, non-negative bundles $\mathbf{x}^* = (\mathbf{x}_1^*,\ldots, \mathbf{x}_n^*)$ and non-negative amounts of input goods $\mathbf{y}^* = (y_1^*, \ldots, y_\ell^*)$, such that
\begin{enumerate}
    \item $p^*_{\fout} \cdot c_f \cdot y_f^* - p^*_{\fin}\cdot y_f^*$ is maximized, for any firm $f \in F$. \hfill \textbf{\emph{(firm profit maximization)}} \label{con1-prod:firm-profit}
   \item $u_i(\mathbf{x}_i)$ is maximized for every consumer $i \in N$, subject to \\
   $\sum_{j \in G}p_j^*\cdot x_{ij}\leq \sum_{j \in G}p_j^*\cdot w_{ij}+\sum_{f \in F} \theta_{if}\cdot (p^*_{\fout} \cdot c_f \cdot y_f^* - p^*_{\fin}\cdot y_f^*)$.\hfill	\textbf{\emph{(bundle optimality)}} \label{con2-prod:optimality}
   \item $z_j^*\leq 0$, and $z_j^* p_j^* =0$, where $z_j^* = \sum_{i \in N} x_{ij}^*+\sum_{f \in F \colon \fin = j}y_{f}^*-\sum_{f \in F \colon \fout = j} c_f\cdot y_{f}^*-1$, \\for every good $j \in G$. \hfill	\textbf{\emph{(market clearing)}} \label{con4-prod:market-clearing}
\end{enumerate}
\end{definition}

\noindent Condition~\ref{con1-prod:firm-profit} requires that at the chosen set of prices $\mathbf{p}^*$, each firm maximizes its profit, given its production functions. Condition~\ref{con2-prod:optimality} requires that at the chosen set of prices $\mathbf{p}^*$, each consumer maximizes her utility subject to their budget constraints, where the budget consists of the amount earned from selling all the consumer's endowments $\sum_{j \in G}p_j^*w_{ij}$ and the profit share $\sum_{f \in F} \theta_{if}\cdot (p^*_{\fout} \cdot c_f \cdot y_f^* - p^*_{\fin}\cdot y_f^*)$ of the consumer from the production of the firms. Finally, Condition~\ref{con4-prod:market-clearing} is the market clearing condition, which requires that the total consumption of each good is at most the total production plus the total endowment of the consumers, and supply equals demand for all goods which are not priced at $0$. As we detail later on in the section, we may in fact assume without loss of generality that in any competitive equilibrium \emph{all the prices are positive}, and hence Condition~\ref{con4-prod:market-clearing} reduces to $z_j^* =0$ for all $j \in G$. Note that in Condition~\ref{con4-prod:market-clearing} we have used that $\sum_{\in N}w_{ij}=1$ for each good $j \in G$.

\paragraph{Sufficiency Conditions.} A competitive equilibrium as defined above exists for every market $\mathcal{M}$, under some sufficiency conditions. The weakest known such conditions are the ones provided by \citet{maxfield1997general} (see also \citep{SODA:GargV14}), which generalize the sufficiency conditions of \citet{eaves1976finite} that we used in \Cref{sec:exchange-markets-linear}. We use the following conditions, which are very close to that of \citet{maxfield1997general}.
\begin{enumerate}
    \item For every good $j \in G$, there exists some consumer $i \in N$ that values the good positively, i.e., $u_{ij} > 0$.
    \item For every consumer $i \in N$, there exists some good $j \in G$ that the consumer endows in a positive amount, i.e., $w_{ij} > 0$.
    \item Consider a graph $\mathcal{G}_F(\mathcal{M})$ in which the nodes are the goods, and an edge $(j,j')$ has weight 
    \[
    \alpha_{jj'}=\max_{f \in F\colon \fout=j', \fin = j}c_f,
    \]
    i.e., $j'$ can be produced from $j$ at conversion rate $\alpha_{jj'}$ by some firm $f \in F$. Then, for any cycle $C=(g_0, g_1), (g_1,g_2)\ldots, (g_{k-2},g_{k-1})$ of $G_F(\mathcal{M})$ the product of the weights of the edges is less than $1$, i.e., $\prod_{e \in C} \alpha_e < 1$.

    This condition is known as the \textbf{\emph{no production out of nothing and no vacuous production}} condition. Indeed, if $\prod_{e \in C} \alpha_e > 1$, then it would be possible to increase the quantity of some good, without decreasing the quantity of any other good. The case of $\prod_{e \in C} \alpha_e = 1$ refers to the case of vacuous production, which is also disallowed in our model.\footnote{\citet{maxfield1997general} technically allows for vacuous production. Vacuous production is disallowed in the conditions imposed in the original setting of \citet{arrow1954existence} and is also disallowed by \citet{SODA:GargV14}.} \label{cond3-prod-markets-no-prod-out-of-nothing}

    \item Consider the \emph{economy graph} $\mathcal{G}_E(\mathcal{M})$ of the market $\mathcal{M}$ in which the nodes are the consumers and the firms, and
    \begin{itemize}
        \item[-] there is an edge $(i,i')$ between consumer/firm-node $i$ and consumer-node $i'$ if $i$ endows/produces a good $j$ for which $u_{i'j}>0$,
        \item[-] there is an edge $(i,f)$ between consumer/firm-node $i$ and firm-node $f$, if $i$ endows/produces the raw good $\fin$ that $f$ uses for production. 
    \end{itemize}
    Then, $\mathcal{G}_E(\mathcal{M})$ contains a strongly connected component containing all the consumer-nodes. This condition generalizes that of the market non-reducibility condition that we used in \Cref{sec:exchange-markets-linear}.
\end{enumerate}

\begin{remark}[Bounds on Production]\label{rem:bounds-on-production}
Condition~\ref{cond3-prod-markets-no-prod-out-of-nothing} above imposes a bound on the total amount of a production of any firm $f \in F$ in a competitive equilibrium. Note that the production starts from finite endowments $\sum_{i \in N}w_{ij}=1$ for all $j \in G$. Since no firm operates at a loss at an equilibrium, any cycle of production would violate Condition~\Cref{cond3-prod-markets-no-prod-out-of-nothing}. Since there are no such cycles, production can take place along chains. The longest such chain is obviously bounded by $m$, and the maximum production of any firm can be bounded by some sufficienctly large global constant $L^c=$ (e.g., some constant such that $L^c \geq m^m (\max_{f}c_f+1)^n$. See also \citep{SODA:GargV14} for a very similar argument. Looking ahead, this will allow us to impose ``loose'' upper bounds on the production and consumption in the linear programs that we will devise, without compromising the existence of a competitive equilibrium.
\end{remark}

\paragraph{Optimality and bang-per-buck (BPB).}

Similarly to \Cref{sec:exchange-markets-linear}, the optimal bundles of Condition~\ref{con2-prod:optimality} in \Cref{def:prod-markets-competitive-equilibrium} are characterized by their \emph{bang-per-buck (BPB)}. Given consumer $i \in N$ and prices $\mathbf{p}$, the BPB of a good $j \in G$ is defined as $\text{BPB}_i(j)=\frac{u_{ij}}{p_j}$. An optimal bundle only contains non-zero quantities of goods for which the BPB is maximum, i.e., $j \in \arg\max_{j \in G} \text{BPB}_i(j)$. \medskip 

\noindent \emph{Bounds on the prices.} Again, similarly to \Cref{sec:exchange-markets-linear}, we may assume without loss of generality that all the prices are strictly positive, i.e., $p_j > 0$ for all goods $j \in G$. Indeed, not that for any good $j$, there is some consumer $i$ with $u_{ij} >0$. Then its price $p_j$ cannot be $0$ in any equilibrium. If $j$ is not being produced, then the demand of consumer $i$ cannot be satisfied. If its being produced, it has to be produced using an input good of $0$ price as well, for the production to be profitable. This in fact implies that it has to be produced along a chain of goods with $0$ price, contradicting the market clearing condition. Very similar arguments have been made in the related literature, e.g., see \citep{SODA:GargV14,garg2018substitution}. Note that by these assumptions, the quantity $\text{BPB}_i(j)$ is well defined for every $j \in G$. 

\begin{remark}[Normalized Prices]\label{rem:prod-markets-normalized-prices}
Given that $p_j > 0$ for all $j \in G$, we can normalize the prices to sum to $1$ without loss of generality, i.e., we may assume that for every good $j \in G$, we have that $\sum_{j \in G} p_j = 1$. 
\end{remark}

\noindent Again, in a similar manner to \Cref{sec:exchange-markets-linear}, we will use a parameter $\varepsilon$ to capture the fact that if the price $p_j$ for a good $j \in G$ is sufficiently smaller than the price $p_{j'}$ for a good $j' \in G$, then $\text{BPB}(j) > \text{BPB}(j')$. Specifically, we can compute $\varepsilon > 0$ such that
\[
\text{If } p_j \leq \varepsilon \cdot p_{j'} \text{ and } u_{ij} > 0 \text{ then } \text{BPB}_i(j) > \text{BPB}_i(j').
\]
Additionally, we can pick $\varepsilon$ to be sufficiently small such that $\varepsilon < \frac{w_{ij}}{m}$ for all $i \in N$ and all $j \in G$. Given $\varepsilon$, we will impose a stricter lower bound on the prices, which will be useful later on: in particular, we will assume that for all $j \in G$, $p_j \geq \frac{\varepsilon^m}{m}$.

\subsubsection{Preprocessing}

The high-level approach will again be to define a function $F$ from a convex compact domain $D$ to itself, and show that a competitive equilibrium of $\mathcal{M}$ can be recovered from a fixed point of $F$. We will need to show that $F$ can be computed by a \linear arithmetic circuit and hence, similarly to \Cref{sec:exchange-markets-linear}, we cannot work directly with the allocations $\mathbf{x}$ and the quantities of input goods $\mathbf{y}$ as in \Cref{def:prod-markets-competitive-equilibrium}, as those would need to be multiplied with the the prices $\mathbf{p}$, which will also be inputs to the circuit. To circumvent this, we will again apply a change of variables very similar to \emph{Gale's substitution} (see \Cref{rem:gale}) and work with the expenditure variables as we did in \Cref{sec:exchange-markets-linear}. In more detail, we will need such variables for the expenditure of the consumers as before, but we will also need variables for the expenditure and the revenue of the firms in production. The same variable change was used by \citet{SODA:GargV14}. 

We will then make use of our \linoptgate to obtain the equilibrium quantities as outcomes of linear programs and a feasibility program. In particular, linear programs $\mathcal{P}_1$ and $\mathcal{P}_1$ will be used to compute the optimal consumer and firm expenditures respectively, and from those the optimal bundles $\mathbf{x}$ and optimal quantities of input goods $\mathbf{y}$ will be recovered. 
The prices will be the outcome of a feasibility program $\mathcal{Q}$. 

\begin{figure}
\centering 
\fbox{
\centering
    \centering
    \begin{minipage}{.5\textwidth}
        \centering
        \underline{Linear Program $\mathcal{P}_1$}
        \begin{align*}%
	\mbox{minimize} & \sum_{i \in N}\sum_{j \in G} \frac{p_j}{u_{ij}}q_{ij}  \\
	\mbox{subject to}& \sum_{j \in G} q_{ij} = \sum_{j \in G}w_{ij}p_j + \sum_{f \in F} \theta_{if}\cdot (r_f-s_f) \\
        & 0 \leq q_{ij} \leq (C \cdot L + 1), \;\forall j \in G 
    \end{align*}
    \end{minipage}%
    \hfill\vline\hfill
    \begin{minipage}{.4\textwidth}
        \centering
        \underline{Linear Program $\mathcal{P}_2$}
        \begin{align*}%
	\mbox{maximize} \ \ \ \ & (c_f \cdot p_{\fout} - p_{\fin}) \cdot s_f  \\
	\mbox{subject to} \ \ \ \ & 0 \leq s_f \leq L\cdot p_{\fin} \\
    \hfill \\
    \hfill 
    \end{align*}
    \end{minipage}%
    }
\caption{The linear programs $\mathcal{P}_1$ and $\mathcal{P}_2$ used to recover optimal bundles (left), and optimal production (right) respectively.}
\label{fig:production-markets-P1-and-P2}
\end{figure}

\paragraph{Consumer expenditure.}
For the consumers, we will use the standard change of variables that we used in \Cref{sec:exchange-markets-linear} for exchange markets. Namely, we will let $q_{ij} = x_{ij} p_j$ be the \emph{expenditure} of consumer $i$ on good $j$, i.e., how much money the consumer spends on the good given an allocation $x_{ij}$ and a price $p_j$.

\paragraph{Firm expenditure.} For the firms, we will use a similar substitution for the case of production. Namely, we will let 
\begin{itemize}
    \item[-] $s_f = p_{\fin} \cdot y_f$ be the \emph{expenditure} of firm $f$ on its input good $\fin$, i.e., the amount of money spend on $y_f$ units of the input good. 
    \item[-] $r_f = p_{\fout}\cdot c_f \cdot y_f$ be the \emph{revenue} of firm $f$ from its output good $\fout$, i.e., the amount of money earned from $c_f\cdot y_f$ units of the output good.\footnote{We remark that \citet{SODA:GargV14} used $r_f$ to refer to the expenditure (spending) and $s_f$ to refer to the revenue. We have switched those so that the variables names are more indicative of the quantities that they represent.} 
\end{itemize}
\noindent By definition, the profit of firm $f \in F$ is then exactly $r_f - s_f$.

\paragraph{The program $\mathcal{P}_1$ for optimal consumer expenditure.} 

We can now write the linear program $\mathcal{P}_1$, the solution of which will give us the optimal expenditures $q_{ij}^*$ for the consumers; see the left-hand side of \Cref{fig:production-markets-P1-and-P2}. It now becomes evident why we need to apply the aforementioned substitution: without it the quantity $\sum_{f \in F} \theta_{if}\cdot (p_{\fout} \cdot c_f \cdot y_f - p_{\fin}\cdot y_f)$ would appear in the constraints of the linear program $\mathcal{P}_1$. The constants $p_\fin$, $p_\fout$ and $y_f$ are \circparams for the linear program, and such a constraint cannot be handled by our \linoptgate. With the \circparams $r_f$ and $s_f$ instead, the constraint becomes linear and thus respects the conditions of the \linoptgate. Note that the linear program $\mathcal{P}_1$ is parameterized by two constants $L$ and $C$. Intuitively, $C\cdot L + 1$ is an upper bound on the amount of money that consumer $i$ can spend on good $j$. Recall that there is a global upper bound $L^c$ on the total possible production, and as a result there are global bounds on the total possible expenditure for both production and consumption. We will set $L$ and $C$ to be sufficiently large to not constrain these global upper bound; the precise bounds are established in the proofs of \Cref{claim:prod-claim1} and \Cref{claim:prod-claim2}. 

Similarly to \Cref{sec:exchange-markets-linear}, in an optimal solution, the consumer is again only spending on goods with minimum ratio $\frac{p_j}{u_{ij}}$, i.e., with maximum BPB, therefore purchasing an optimal bundle. At the same time, the subgradient of this objective function can be computed by a \pseudog. 

\paragraph{The program $\mathcal{P}_2$ for optimal firm expenditure.} We are now ready to define the linear program for each firm $f \in F$, the optimal solution of which will be the optimal expenditure $s_f^*$ of firm $f$ on the amount $y_f^*$ of the input good that the firm uses for production; see the right-hand side of \Cref{fig:production-markets-P1-and-P2}. The linear program $\mathcal{P}_2$ is also parameterized by the constant $L$, which we mentioned above; intuitively, $L$ imposes an upper bound on the quantity of the input good that firm $f$ can use, and will be large enough to not constrain the global upper bound $L^c$ on production. 

\begin{remark}[$\mathcal{P}_2$ ensures optimal production]\label{rem:optimal-production}
At first glance, the objective function of linear program $\mathcal{P}_2$ might seem a bit unintuitive: why are we multiplying the expenditure with the profit? How do we guarantee that in the competitive equilibrium, each firm maximizes its profit? To provide some intuition, we remark that linear program $\mathcal{P}_2$ could equivalently be substituted by the following feasibility program $\mathcal{Q}(\mathcal{P}_2)$:

\begin{center}\underline{Feasibility Program $\mathcal{Q}(\mathcal{P}_2)$}\end{center}
\vspace{-5mm}
\begin{align*}
        & c_f \cdot p_{\fout} - p_{\fin} > 0 \Rightarrow s_f = L \cdot p_{\fin} \nonumber \\
	& c_f \cdot p_{\fout} - p_{\fin} < 0 \Rightarrow s_f = 0 \nonumber \\
	& 0 \leq s_f \leq L \cdot p_{\fin} \nonumber
\end{align*}
Indeed, linear program $\mathcal{P}_2$ has a very simple form, and hence in an optimal solution
\begin{itemize}
    \item[-] the firm spends as much as possible on the input good (i.e., $s_f$ is maximum), when it is profitable to produce, (i.e., when $c_f \cdot p_{\fout} - p_{\fin} > 0$),
    \item[-] the firm spends as little as possible on the input good (i.e., $s_f$ is $0$), when it is not profitable to produce, (i.e., when $c_f \cdot p_{\fout} - p_{\fin} < 0$),
    \item[-] the firm spends any feasible amount when producing or not producing yield the same profit (i.e., when $c_f \cdot p_{\fout} - p_{\fin} = 0$).
\end{itemize}
We could have in fact used $\mathcal{Q}(\mathcal{P}_2)$ rather than $\mathcal{P}_2$, as it is a valid feasibility program that can be solved using the \linoptgate. However, we elected to go the the linear program $\mathcal{P}_2$ instead, as it sets up the machinery that we will use in \Cref{sec:general-markets}, where the corresponding linear program will be more involved and cannot simply be substituted by a feasibility program as above.  

We remark that in a competitive equilibrium, we will establish that $c_f \cdot p_{\fout} - p_{\fin} \leq 0$, and hence the firm will never be required to use exactly $L$ units of $\fin$ for production. 
\end{remark}

\begin{remark}[Calculating $r_f$ from $s_f$]\label{rem:r_f-from-s_f}
Looking ahead, the terms $r_f$ which appear in the constraints of linear program $\mathcal{P}_1$ (and also in the feasibility program $\mathcal{Q}$, see below), will not be inputs to the function $F$ that we will construct to compute a competitive equilibrium. These will need to be recovered from the values of $s_f$, or more precisely, the terms themselves will need to be substituted with expressions of the term $s_f$. This can be done as follows:
\begin{equation}\label{eq:prod:rf-sf}
    r_f = s_f + \max\{0, c_f \cdot p_{\fout} - p_{\fin}\}\cdot L
\end{equation}
As we will only be concerned with the fixed point behavior of our function $F$, this equation needs to correctly capture the revenue of firm $f$ at a competitive equilibrium. Indeed, given our discussion above, we have that:
\begin{itemize}
    \item[-] When $c_f \cdot p_{\fout} - p_{\fin} < 0$, the firm does not spend any amount on producing and as a result it does not produce anything. In that case $r_f^* = s_f^* = 0$.
    \item[-] When $c_f \cdot p_{\fout} - p_{\fin} > 0$, the firm spends an amount of $s_f^* = L \cdot p_{\fin}$ for $L$ units of the input good and obtains a revenue of $r_f^* = L \cdot c_f \cdot p_{\fout}$. 
    \item[-] When $c_f \cdot p_{\fout} - p_{\fin} = 0$, the firm can spend any amount $s_f^*$ for production, and we have $r_f^* = s_f^*$ by definition. 
\end{itemize}
In each case, \Cref{eq:prod:rf-sf} correctly computes the value of $r_f^*$ from $s_f^*$ at a competitive equilibrium.
\end{remark}

\paragraph{Excess expenditure and the program $\mathcal{Q}$.}

Similarly to \Cref{sec:exchange-markets-linear}, the equilibrium prices $\mathbf{p}^*$ will be obtained via a feasibility program $\mathcal{Q}$. Before we define the program, we introduce the notion of excess expenditure; the definition is very similar to the one we used in \Cref{sec:exchange-markets-linear}, except that it now takes into account the expenditure/revenue due to production. The \emph{excess expenditure} $e_j$ of a good $j \in G$ is defined as the difference between the total expenditure of all consumers $i \in N$ and firms $f \in F$ for that good and the price of the good plus the total revenue of firms $f \in F$ from that good, i.e., 
\[
e_j = \sum_{i \in N}q_{ij} +\sum_{f \in F \colon \fin=j}s_f - p_j - \sum_{f \in F \colon \fout=j}r_f.
\]
Then, at equilibrium prices $p_j^*$ we will have market clearing, i.e., $e_j^* = 0$. \medskip

\noindent We are now ready to define our feasibility program $\mathcal{Q}$; see \Cref{fig:production-markets-Q}. The equilibrium prices will be obtained as the output $\mathbf{p}^{*}$ of $\mathcal{Q}$. 

\begin{figure}
\centering 
\fbox{
\begin{minipage}{0.5\textwidth}
\centering
        \underline{Feasibility Program $\mathcal{Q}$}
        \begin{align*}
	& e_j<e_{j'}\Rightarrow p_j\leq \varepsilon \cdot p_{j'},\;\forall j,j' \in G \nonumber \\
	& p_j\geq \frac{\varepsilon^m}{m},\;\forall j \in G\\
	& \sum_{j \in G}p_j = 1\nonumber
        \end{align*}
\end{minipage}
}
\caption{The feasibility program $\mathcal{Q}$ used to find market-clearing prices.}
\label{fig:production-markets-Q}
\end{figure}

\subsubsection{Membership in PPAD: The proof of \crtCref{thm:production-market-PPAD}.}

\noindent We will again develop the proof in three steps, namely (a) construction of the function $F$ and arguing that it can be represented by a \linear arithmetic circuit containing \linoptgates, (b) showing that the \linoptgate can compute all the necessary components, and (c) arguing that a competitive equilibrium can be recovered from a fixed point of $F$. 

\paragraph{The function $F$.} Given the above, we will define $F: D \rightarrow D$ with domain 
\[D = \left\{p \in \Delta^{m-1}: \forall j \in G, p_j \geq \frac{\varepsilon^m}{m}\right\} \times [0,C\cdot L + 1]^{nm} \times [0,L]^{\ell}.\]
An input to $F$ is a triple $(\bar{\mathbf{p}}, \bar{\mathbf{q}}, \bar{\mathbf{s}})$ of prices, consumer expenditures and firm expenditures, where $\bar{\mathbf{q}} = (\bar{q}_{ij})_{i \in N, j \in G}$ and $\bar{\mathbf{s}}=(\bar{s}_1,\ldots \bar{s}_\ell)$, and the output is another such triple $(\mathbf{p}^*, \mathbf{q}^*, \mathbf{s}^*)$. The domain $D$ is the one above since $\sum_{j \in G} \bar{p}_j = 1$ (see \Cref{rem:prod-markets-normalized-prices}), and by the upper bounds of $C\cdot L +1$ and $L$ set on the consumer expenditures in linear program $\mathcal{P}_1$ and the firm expenditures in linear program $\mathcal{P}_2$ respectively. 

\paragraph{Computation by the \linoptgate.} Next, we argue that the solutions to the linear programs $\mathcal{P}_1$ and $\mathcal{P}_2$ and the feasibility program $\mathcal{Q}$ of \Cref{fig:exchange-markets-P-and-Q} can be computed by our \linoptgate.

\begin{lemma}\label{lem:productionMarkets-OPT-gate-can-compute}
Consider the linear programs $\mathcal{P}_1$, $\mathcal{P}_2$ of \Cref{fig:production-markets-P1-and-P2}, and the feasibility program $\mathcal{Q}$ of \Cref{fig:production-markets-Q}. An optimal solution to $\mathcal{P}_1$ and $\mathcal{P}_2$ can be computed by the \linoptgate. $\mathcal{Q}$ is solvable, and a solution can be computed by the \linoptgate. 
\end{lemma}

\begin{proof}
For the linear program $\mathcal{P}_1$, the feasible domain $[0,1]^n$ is non-empty and bounded. The \circparams $p_j$ and $s_f$ appear only on the right-hand side of the constraints, and the subgradient of the objective function is linear, and hence can be given by a \pseudog. For the linear program $\mathcal{P}_2$, the feasible domain $[0,L]$ is non-empty and bounded, and the \circparams $p_{\fin}$ only appear on the right-hand side of the constraints. The subgradient of the objective function is linear, and hence can be given by a \pseudog. For the feasibility program $\mathcal{Q}$, the arguments that it is of the correct form and that it is solvable are identical to those of \Cref{lem:ExchangeMarkets-OPT-gate-can-compute}, noting the updated definition of the expenditure $e_j$. 
\end{proof}

\paragraph{Arguing optimality and market clearing.} To conclude the proof, what is left to show is that a fixed point $(\mathbf{p}^*, \mathbf{q}^*, \mathbf{s}^*)$ of $F$ indeed corresponds to a competitive equilibrium. We argue that in the following lemma.

\begin{lemma}\label{lem:production-markets-correctness-lemma}
Let $(\mathbf{p}^*, \mathbf{q}^*, \mathbf{s^*})$ be a fixed point of $F$. Then the triple $(\mathbf{p}^*, \mathbf{x}^*,\mathbf{y}^*)$, where $x_{ij}^* = q_{ij}^*/p_j^*$ and $y_f^* = s_f^*/p_{\fin}^*$ is a competitive equilibrium of $\mathcal{M}$.
\end{lemma}

\begin{proof}
By the form of the objective function of the linear program $\mathcal{P}_1$ of \Cref{fig:production-markets-P1-and-P2}, each consumer spends money only on goods with maximum BPB and therefore the allocation that she receives is only one of goods for which she has maximum BPB. Therefore, the consumer receives an optimal bundle which satisfies Condition~\ref{con2-prod:optimality} of definition in \Cref{def:prod-markets-competitive-equilibrium}. The allocation quantities $x_{ij}^*$ can straightforwardly be recovered from the values of $q_{ij}^*$. As we discussed earlier in \Cref{rem:optimal-production}, an optimal solution to linear program $\mathcal{P}_2$ also results in the optimal expenditure for the firms.

What remains to show is that $\mathbf{p^*}$ is a vector of market-clearing prices. By the definition of $e_j$, this is equivalent to arguing that for all $j \in G$, we have that $e_j^* = 0$. We first argue that $\sum_{j \in G}e_j^*=0$. Indeed:
\begin{align*}
\sum_{j \in G} e_j^* &= \sum_{j \in G} \left(\sum_{i \in N} q_{ij}^* - p_j^* +\sum_{f \in F \colon \fin=j}s_f^* - \sum_{f \in F \colon \fout=j}r_f^* \right) \\ 
&= \sum_{i \in N} \sum_{j \in G} q_{ij}^* - \sum_{j \in G}p_j^* + \sum_{f \in F}(s_f^*-r_f^*)  \\ 
&= \sum_{j \in G} \sum_{i \in N} w_{ij}p_j^* - 1 + \sum_{i \in N}\sum_{f \in F}\theta_{if}\cdot (r_f^* - s_f^*) + \sum_{f \in F}(s_f^*-r_f^*) \\ 
& = \sum_{j \in G}p_j^*\cdot \sum_{i \in N}w_{ij} - 1 + \sum_{f \in F}(r_f^* - s_f^*)\cdot \sum_{i \in N}\theta_{if}\cdot  \sum_{f \in F}(s_f^*-r_f^*) \\
& = \sum_{i\in N}w_{ij}-1 + \sum_{f \in F}(r_f^* - s_f^*) + \sum_{f \in F}(s_f^* - r_f^*) \\
& = 0,
\end{align*}
where in the calculations above we used:
\begin{itemize}
    \item[-] in Equation 2, that $\sum_{j \in G}\sum_{f \in F \colon \fin = j}s_f^* = \sum_{f \in F}s_j^*$ and that $\sum_{j \in G}\sum_{f \in F \colon \fout = j}r_f^* = \sum_{f \in F}r_f^*$,
    \item[-] in Equation 3, that $\sum_{i \in N} \sum_{j \in G}q_{ij}^* = \sum_{i\in N}\sum_{j \in G}w_{ij}p_j^* + \sum_{i \in N} \sum_{f \in F}\theta_{if}(r_f^*-s_f^*)$, which follows from the constraints of linear program $\mathcal{P}_2$ of \Cref{fig:production-markets-P1-and-P2},
    \item[-] in Equations 3 and 5, that $\sum_{j \in G}p_j^* =1$, which is without loss of generality, see \Cref{rem:prod-markets-normalized-prices},
    \item[-] in Equation 5, that $\sum_{i \in N}\theta_{if} = 1$ for every $f \in F$, by definition, and
    \item[-] in Equation 6, that $\sum_{i \in N}w_{ij} = 1$, which is without loss of generality.
\end{itemize}
\noindent From the above, it suffices to prove that $e_j^* \geq 0$, for all $j \in G$. Assume by contradiction that there exists some $j_1 \in G$ such that $e_{j_1}^* < 0$; we will obtain a contradiction to the strong connectivity of the economy graph of the market. We define the following three sets:
\begin{itemize}
    \item[-] $J = \{j \in G \colon e_j^* \leq e_{j'}^*, \text{ for all } j' \in G\}$. In other words, $J$ is the set of goods with minimum excess expenditure.
    \item[-] $N_J = \{i \in N \colon \text{ there exists } j \in J \text{ such that } u_{ij} > 0\}$. In other words, $N_J$ contains the set of consumers that value at least one item in $J$ positively.
    \item[-] $F_J=\{f \in F \colon \fin \in J\}$. In other words, $F_J$ contains the set of firms for which the input good is in $J$. 
\end{itemize}
Since $e_{j_1}^* < 0$, by the definition of $J$ we know that
\begin{equation}\label{eq:prod-min-expenditure}
\text{for all } j \in J, \text{ we have } e_j^* < 0.
\end{equation}
Also, since $e_{j_1}^* < 0$, and $\sum_{j \in G}e_j^*=0$, there must exist some other good $j_2$ such that $e_{j_2}^* > 0$. In particular, that implies that $J$ is a strict subset of $G$, i.e., $J \subsetneq G$.
We state and prove the following claim.

\begin{claim}\label{claim:prod-claim1}
For all $i \in N_J$ and all $j' \in G\setminus J$, it holds that $w_{ij'}=0$.
\end{claim}

\begin{proof}
Assume by contradiction that there exists some consumer $i_0 \in N_J$ and some good $j' \in G\setminus J$ such that $w_{i_0j'}>0$. We will show that this implies that there exists some $j_\ell \in J$ with $e_{j_\ell}^* \geq 0$, contradicting Statement~\eqref{eq:prod-min-expenditure}. Since $\mathbf{p}^*$ is a solution to feasibility program $\mathcal{Q}$ of \Cref{fig:production-markets-Q}, we have that $p_j^* \leq \varepsilon \cdot p_{j'}^*$, for every $j \in J$. Hence, by choosing $\varepsilon$ to be sufficiently small, we may lower-bound the budget of consumer $i$ by
\begin{equation*}
w_{i_0j'} \cdot p_{j'}^* \geq \frac{w_{i_0j'}}{\varepsilon \cdot m} \cdot \sum_{j \in J}p_j^* \geq (C\cdot L+1)\sum_{j \in J} p_j^*
\end{equation*}
Since $i_0 \in N_J$, by the choice of $\varepsilon$, there is at least one good $g \in J$ such that $\text{BPB}_{i_0}(g) > \text{BPB}_{i_0}(g')$, for all $g' \in G\setminus J$, i.e., the consumer prefers to buy quantities of good $g$ rather than any good which is not in the set $J$. Let $j \in J$ be a good with maximum BPB for consumer $i_0$. Since $q_{i_0j}^*$ is an optimal solution to linear program $\mathcal{P}_1$ of \Cref{fig:production-markets-P1-and-P2}, it should be that $q_{i_0j}^* = C \cdot L +1$, which implies that $q_{i_0j}^* \geq (C \cdot L+1)p_j^*$. By substituting that into the definition of the excess expenditure $e_j^*$, we obtain:
\begin{align*}
e_j^* &= \sum_{i \in N}q_{ij}^* +\sum_{f \in F \colon \fin=j}s_f^* - p_j^* - \sum_{f \in F \colon \fout=j}r_f^* \\
&\geq q_{i_0j}^* +\sum_{f \in F \colon \fin=j}s_f^* - p_j^* - \sum_{f \in F \colon \fout=j}r_f^* \\
&\geq (C\cdot L)p_j^* - \sum_{f \in F \colon \fout=j}r_f^*
\end{align*}
Therefore, it suffices to show that $\sum_{f \in F \colon \fout=j}r_f^* \leq (C\cdot L)p_j^*$; in that case we will have $e_j^* \geq 0$ and will obtain a contradiction, by setting $j = j_\ell$. Now consider any firm $f \in F$. Referencing also \Cref{rem:optimal-production}, note that \begin{itemize}
    \item[-] if $c_f \cdot p_{\fout}^* < p_{\fin}^*$, then the firm does not produce anything, and we have $r_f^*=0$,
    \item[-] if $c_f \cdot p_{\fout}^* > p_{\fin}^*$, the firm produces $L$ units of $\fout$ and we have $r_f^* = L \cdot c_f \cdot p_{\fout}$,
    \item[-] if $c_f \cdot p_{\fout}^* > p_{\fin}^*$, the firm may produce any amount and we have $r_f^* = s_f^* \leq L \cdot c_f \cdot p_{\fout}$, since $s_f^*$ is a feasible solution to linear program $\mathcal{P}_2$ of \Cref{fig:production-markets-P1-and-P2}.
\end{itemize} 
\noindent By letting $C \geq |F| \cdot \max_{f \in F} c_f$, it follows that $\sum_{f \in F \colon \fout=j}r_f^* \leq (C\cdot L)p_j^*$ and we obtain a contradiction.
\end{proof}

\noindent \Cref{claim:prod-claim1} establishes that a consumer in the set $N_J$ cannot endow any good that is not in the set $J$ in a positive quantity. Since $J \subsetneq G$, this implies that $N_J \subsetneq N$, since every good is positively endowed by some consumer $i \in N$. Now consider any consumer $i \in N_J$ and any consumer $i' \in N\setminus N_J$ or any firm $f \in F\setminus F_J$. \Cref{claim:prod-claim1} implies that there cannot be an edge $(i,i')$ or an edge $(i,f)$ in the economy graph of the market. Indeed, for $(i,i')$, notice that by the claim consumer $i$ only endows positively goods that are in $J$, but consumer $i'$ has value $0$ for them by virtue of being in $N\setminus N_J$. Similarly, for $(i,f)$, $f$ uses good $\fin \in G\setminus J$ by definition, but consumer $i$ only positively endows goods in $J$. \medskip

\noindent Next, we state and prove the following claim.

\begin{claim}\label{claim:prod-claim2}
For any $f \in F_J$, it holds that $\fout \in J$.
\end{claim}

\begin{proof}
Assume by contradiction that there exists some firm $f_0 \in F_J$ with $g_{f_0}^{\text{out}} \notin J$. For convenience, let $j = g_{f_0}^{\text{out}}$ be that good, and let $j_0 = g_{f_0}^{\text{in}}$ be the firm's input good. Note that by definition of $F_J$, we have that $j_0 \in J$. By Statement~\eqref{eq:prod-min-expenditure}, we have that $e_{j_0}^* < 0$. Since $\mathbf{p}^*$ is a solution to feasibility program $\mathcal{Q}$ of \Cref{fig:production-markets-Q}, we have that $p_{j_0}^* \leq \varepsilon \cdot p_{j}^*$. The value of $\varepsilon$ can be chosen sufficiently small for this to imply that $c_{f_0} \cdot p_{j}^* - p_{j_0}^* >0$. Since this is the multiplier of the variable $s_{f_0}$ in the objective function of linear program $\mathcal{P}_2$, the optimality of $s_{f_0}^*$ as a solution to the linear program implies that $s_{f_0}^* = L \cdot p_{j_0}^*$. Using this we may bound the expenditure $e_{j_0}^*$ as follows:
\begin{align*}
    e_{j_0}^* &\geq L\cdot p_{j_0}^* - \sum_{f \in F \colon \fout=j_0}r_f^* - p_{j_0}^* \\
    &= (L-1)p_{j_0}^* - \sum_{f \in F \colon \fout=j_0}r_f^*
\end{align*}
where in the calculations above we used the obvious fact that $\sum_{f \in F \colon \fin=j}s_f^* \geq s_{f_0}$. 

Obviously, if good $j_0$ is not being produced by any firm $f \in F$, we can simply set $L \geq 1$ and obtain that $e_{j_0}^* \geq 0$, a contradiction, since $r_f^* = 0$ for all firms that have $j_0$ as their output good. Now consider the sequence of goods $j_0,j_1,\ldots,j_k$, where good $j_{\ell-1}$ is being produced (in positive quantity) by some firm $f_{\ell}$ using $j_{\ell}$ as the input good. Note that by the discussion above, we may assume that $k \geq 1$. Additionally, by the ``no production out of nothing and no vacuous production'' feasibility condition of the market $\mathcal{M}$, it also holds that $k \leq m$, as otherwise the graph $G_F(\mathcal{M})$ would have a cycle in which the product of the weight of the edges would be at least $1$. 

Finally, we remark that for each $\ell \in \{0,1,\ldots,k\}$, we have that $j_{\ell} \in J$. To see this, observe that from the conditional constraints of feasibility program $\mathcal{Q}$ of \Cref{fig:production-markets-Q}, the price for any good in $J$ is at least $\varepsilon$ times smaller than the price of any good in $G\setminus J$. By taking $\varepsilon$ to be small enough, we can guarantee that if a firm $f$ is producing a good $\fout \in J$, then it must use a good $\fin \in J$ as input, otherwise $s_f^*$ would not be an optimal expenditure. Since $j_0 \in J$ by the definition of $F_J$, it must hold that $j_\ell \in J$ for all $\ell \in \{0,1,\ldots,k\}$.  \medskip

\noindent We will argue by induction that for $1 \leq \ell' \leq k$, we have:
\begin{equation}\label{eq:prod-markets-induction}
r_{f_\ell}^* \geq \left(\frac{L-1}{|F|^{\ell} \cdot \prod_{i=1}^{\ell-1}c_{f_i}} - \sum_{i=1}^{\ell-1} \frac{1}{|F|^i \cdot \prod_{j = \ell - i +1}^{\ell-1}c_{f_{j}}}\right) \cdot p_{j_{\ell-1}}^*
\end{equation}

\noindent With this at hand, at step $k$ we will have
\[
e_{j_k}^* \geq \left(\frac{L-1}{c_{f_k}\cdot|F|^{k} \cdot \prod_{i=1}^{k-1}c_{f_i}} - \sum_{i=1}^{k-1} \frac{1}{c_{f_k}\cdot |F|^i \cdot \prod_{j = k - i +1}^{k-1}c_{f_{j}}} -1\right)p_{j_k}^* - \sum_{f \in F \colon \fout=j_k}r_f^*
\]
Since $k$ is the last index in the sequence, i.e, good $j_{k-1}$ is the last to be produced, we know that $r_f^* =0$ for any firm $f$ that has $j_k$ as its output good. From the above, we have that  
\[
e_{j_k}^* \geq \left(\frac{L-1}{c_{f_k}\cdot|F|^{k} \cdot \prod_{i=1}^{k-1}c_{f_i}} - \sum_{i=1}^{k-1} \frac{1}{c_{f_k}\cdot |F|^i \cdot \prod_{j = k - i +1}^{k-1}c_{f_{j}}} -1\right)p_{j_k}^*
\]
By picking $L$ to be sufficiently large, e.g., large enough for the following to hold
\[
\frac{L}{(|F|\cdot \max_{f}c_f)^m} - \frac{m}{(\min_{f}c_f)^m} \geq 0,
\]
we obtain a contradiction. We compute the proof below by proving Inequality~\eqref{eq:prod-markets-induction}.\medskip 

\noindent \textbf{Base Case:} Let $\ell =1$, and consider firm $f_1$ that produces good $j_0$ from good $j_1$. It follows that in the objective function of linear program $\mathcal{P}_2$ of \Cref{fig:production-markets-P1-and-P2} corresponding to firm $f_1$, the multiplier of the expenditure $s_{f_1}$ is non-negative, i.e., $c_{f_1}\cdot p_{j_0}^* - p_{j_{1}^*}  \geq 0$. There are two cases:
\begin{itemize}
    \item[-] If $c_{f_1}\cdot p_{j_0}^* - p_{j_{1}^*}  > 0$, then $s_{f_1}^* = L \cdot p_{j_1}^*$.
    \item[-] If $c_{f_1}\cdot p_{j_0}^* - p_{j_{1}^*}  =0$, then $s_{f_1}^* = r_{f_1}^* \geq \frac{L-1}{|F|}\cdot p_{j_0}^*$
\end{itemize}
where the last inequality follows from the fact that $e_{j_0} < 0$, by Statement~\eqref{eq:prod-min-expenditure}. In either case, in the second case using the fact that $p_{j_0}^* \geq p_{j_1}^*/c_{f_1}$, we can bound the expenditure $e_{j_1}^*$ as follows:
\begin{equation*}
    0 > e_{j_1}^* \geq \left(\frac{L-1}{|F|\cdot c_{f_1}}-1\right)\cdot p_{j_1}^* - \sum_{f \in F \colon \fout=j_1}r_f^*
\end{equation*}
From this, we get that
\[
r_{f_2}^* \geq \left(\frac{L-1}{|F|^2\cdot c_{f_1}}-\frac{1}{|F|}\right)\cdot p_{j_1}^*
\]

\noindent \textbf{Induction Step:} Consider some step $\ell$ and consider firm $f_{\ell}$ that produces good $j_{\ell-1}$ from good $j_{\ell}$. It follows that in the objective function of linear program $\mathcal{P}_2$ of \Cref{fig:production-markets-P1-and-P2} corresponding to firm $f_{\ell}$, the multiplier of the expenditure $s_{f_\ell}$ is non-negative, i.e., $c_{f_\ell}\cdot p_{j_{\ell-1}}^* - p_{j_{\ell}^*}  \geq 0$. There are two cases:
\begin{itemize}
    \item[-] If $c_{f_\ell}\cdot p_{j_{\ell-1}}^* - p_{j_{\ell}^*}  > 0$, then $s_{f_\ell}^* = L \cdot p_{j_\ell}^*$.
    \item[-] If $c_{f_\ell}\cdot p_{j_{\ell-1}}^* - p_{j_{\ell}^*}  =0$, then \[s_{f_\ell}^* = r_{f_\ell}^* \geq \left(\frac{L-1}{|F|^{\ell} \cdot \prod_{i=1}^{\ell-1}c_{f_i}} - \sum_{i=1}^{\ell-1} \frac{1}{|F|^i \cdot \prod_{j = \ell - i +1}^{\ell-1}c_{f_{j}}}\right) \cdot p_{j_{\ell-1}}^*\]
\end{itemize}
where the inequality in the second case follows from the induction hypothesis. In either case, in the second case using the fact that $p_{j_{\ell-1}}^* \geq p_{j_\ell}^*/c_{f_\ell}$, we can bound the expenditure $e_{j_\ell}^*$ as follows:
\begin{equation*}
    0 > e_{j_\ell}^* \geq \left(\frac{L-1}{|F|^{\ell} \cdot \prod_{i=1}^{\ell}c_{f_i}} - \sum_{i=1}^{\ell-1} \frac{1}{|F|^i\cdot \prod_{j = \ell - i +1}^{\ell}c_{f_{j}}}-1\right)\cdot p_{j_\ell}^* - \sum_{f \in F \colon \fout=j_1}r_f^*
\end{equation*}
From this, we obtain that 
\begin{equation*}
r_{f_{\ell+1}}^* \geq \left(\frac{L-1}{|F|^{\ell+1} \cdot \prod_{i=1}^{\ell}c_{f_i}} - \sum_{i=1}^{\ell} \frac{1}{|F|^i \cdot \prod_{j = \ell - i +1}^{\ell}c_{f_{j}}}\right) \cdot p_{j_{\ell}}^*
\end{equation*}
\end{proof}

\noindent \Cref{claim:prod-claim2} implies that for any firm $f \in F_J$, there cannot be an edge to any consumer $i \in N \setminus N_J$ or any firm $f' \in F \setminus F_J$ in the economy graph. Indeed, for an edge $(f,i)$, the claim asserts that the production good $\fout$ of $f$ will be in the set $J$, whereas consumer $i$ only values goods in the set $G \setminus J$ positively, by virtue of being in $N \setminus N_J$. Similarly, for an edge $(f,f')$, firm $f$ produces the good $\fout \in J$, whereas the input good $g_{f'}^{\text{in}}$ of firm $f'$ is in $G \setminus J$, by virtue of $f'$ being in $F\setminus F_J$. \medskip

\noindent From the two paragraphs succeeding \Cref{claim:prod-claim1} and \Cref{claim:prod-claim2}, $N_J$ and $N\setminus N_J$ are two strongly connected components in the economy graph of the market, contradicting the sufficiency condition requiring that in the economy graph there is a strongly connected component containing all the consumer-nodes. This completes the proof.
\end{proof}

\subsection{Markets with Leontief-free Utilities and Productions}\label{sec:general-markets}

In this section we show how our \linoptgate can be used to obtain the PPAD-membership of finding competitive equilibria in Arrow-Debreu markets with Leontief-free utilities and productions. Recall that, as we mentioned in the beginning of the section, this is the largest class of functions for which membership in PPAD (and hence rationality of solutions) has been proven \citep{garg2018substitution}. In the following section we define another class of functions which generalizes all previous ones, but its generality is incomparable to the Leontief-free case. Our main theorem of the section is the following.

\begin{theorem}\label{thm:Leontief-free}
Computing a competitive equilibrium of an Arrow-Debreu market with Leontief-free utilities and productions is in PPAD.
\end{theorem}

\noindent Our proof in this section provides a significant simplification over that of \citet{garg2018substitution}. To keep the section as self-contained as possible, we define these markets in detail here, rather than explain how they generalize the markets of \Cref{sec:prod-markets-linear}, but we make appropriate references to that section.  

\paragraph{Markets with Production.} In an Arrow-Debreu market with production $\mathcal{M}$, we have a set $N$ of consumers, a set $G$ of infinitely divisible goods, and a set $F$ of firms. Let $n=|N|$, $m=|G|$, and $\ell = |F|$. We will typically use index $i$ to refer to consumers, $j$ or $g$ to refer to goods and $f$ to refer to firms. Each consumer brings an endowment $w_i = (w_{i1}, \ldots, w_{im})$ to the market, with $w_{ij} \geq 0$ for all $i \in N$ and $j \in G$. We may assume without loss of generality that for every good $j$, we have $\sum_{i \in N} w_{ij} =1$, i.e., that the total endowment of each good is $1$. We will use $\mathbf{x}^i = (x_{i1}, \ldots, x_{im})$ to denote the vector of quantities of goods allocated to consumer $i \in N$ in $\mathcal{M}$, and we will call it the \emph{bundle} of consumer $i$. Let $\mathbf{x} = (\mathbf{x}^1, \ldots, \mathbf{x}^n)$ be the vector of such bundles. We will use $\mathbf{p} = (p_1, \ldots, p_m)$ to denote the vector of \emph{prices} in $\mathcal{M}$, one for each good $j \in G$. Prices are non-negative, so $p_j \geq 0$ for all $j \in G$. Given a vector of prices $\mathbf{p}$, the \emph{budget} of consumer $i \in N$ is defined as $\sum_{j \in G}w_{ij}p_j$; intuitively, this is the amount of money that the consumer acquires by selling her endowment at prices $\mathbf{p}$. 

\paragraph{Utility Functions.} Every consumer has a utility function $u_i\colon\RRnn^m\rightarrow\RRnn$ mapping a bundle 
$\mathbf{x}^i$ to a non-negative real number. In this section, we consider \emph{Leontief-free utilities}.
Such utility functions are specified by a finite list $K_i$ of \emph{\segs} $\{s^i_k\}_{k \in K_i}$.\footnote{\citet{garg2018substitution} us the term ``segments'' to refer to these ``portions'' of the utility function. We use the term ``segments'' for the parts of the SSPLC functions in \Cref{sec:SSPLC-markets} which has a different meaning, so we adopt the term ``\seg'' here instead.}
Associated 
with every \seg $s^i_k$ is a number $L^i_k\in\RRnn\cup\{\infty\}$, which is an upper bound for the total utility 
that can be accrued from this \seg, and for every $j \in G$ there is a number
$u^i_{jk}\geq 0$, which is the rate at which good $j$ provides utility for consumer $i$ on \seg $k$. That is, the
utility consumer $i$ receives from the bundle $\mathbf{x}^i$ is calculated by solving the following linear
program.
\begin{align*}
	\mbox{maximize } & \sum_{k \in K_i}\sum_{j \in G}u^i_{jk}x^i_{jk}  \nonumber\\
	\mbox{subject to }& \sum_{j \in G} u^{i}_{jk}x^i_{jk}\leq L^i_k,\;\ \ \text{ for all } k \in K_i \\
        & \sum_{k \in K_i}x^{i}_{jk}\leq x^i_{j},\;\ \ \ \text{ for all } j\in\mathcal{G} \nonumber
\end{align*}
\noindent Note that the class of linear utilities that we considered in \Cref{sec:exchange-markets-linear,sec:prod-markets-linear} is a special case of this construction where every \seg is 
unbounded and can accrue utility from only a single good. In the case of SPLC utilities which we mentioned in the beginning of \Cref{sec:markets}, the restriction
that the \segs must be unbounded is dropped.

\paragraph{Firm Shares.} Each consumer $i \in N$ has a share $\theta_{if} \in [0,1]$ of the profit of each firm $f \in F$. We assume that the profits are entirely shared among the consumers, i.e., for every firm $f \in F$, we have that $\sum_{i \in N}\theta_{if}=1$. 

\paragraph{Production functions.}

Every firm $f$ has a \emph{Leontief-free production function}. Here we will use slightly different terminology from
\Cref{sec:prod-markets-linear}, following the one used by \citet{garg2018substitution}.
We will say that each firm $f$ has a \emph{partitioning} of the set of goods $G$ into \emph{raw goods} and \emph{produced goods}, i.e., $\mathcal{G}= \mathcal{R}_f\sqcup\mathcal{P}_f$. The firm then first converts the raw goods into
\emph{raw units}, and these in turn are converted into produced goods. \medskip

\noindent More precisely,
the production function of firm $f$ is specified by a list of \segs $\{s^f_k\}_{k \in K_f}$. To each of these \segs we have the following associated quantities:
\begin{itemize}
    \item[-] a number $L^f_k\in\RRnn\cup\{\infty\}$, which is a bound on how many raw units can be produced on this \seg, for every raw good $j\in\mathcal{R}_{f,k}\subseteq\mathcal{R}_f$, 
    \item[-] a number $\alpha^f_{jk}> 0$, which denotes the rate at which good $j$ can be converted into raw units on this \seg (i.e., the raw-good-to-raw-unit conversion rate), and, 
    \item[-] for every produced good $j'\in\mathcal{P}_{f,k}\subseteq\mathcal{P}_f$, a number $\beta^f_{j'k}> 0$, which denotes the rate at which raw units can be converted into good $j'$ on this \seg (i.e,, the raw-unit-to-produced-good conversion rate).  
\end{itemize}

\noindent Let $\mathbf{y}^f = (y_j)_{j \in G}$ be the \emph{production vector} of firm $f$, i.e., the vector of amounts of raw goods and produced goods involved in its production function. 

\begin{remark}[One \seg for each firm]\label{rem:one-seg-per-firm}
Note that whether a firm produces on a \seg $k$ is independent of whether it produces on another \seg 
$k'$. Therefore, to simplify the exposition, we may assume that the production function of any firm has only one \seg. If a firm's production function has multiple \seg, then we
may simply ``split'' the firm into one firm for every \seg of the function, and allocate to the consumers shares in the new firms that are identical to the shares they had in the original firm. This change will not impact the production of the firms or the budget constraints of the consumers (see \Cref{def:leontief-free-market-equilibrium} below). Given this, we can write $\alpha_{j}^f$ and $\beta_{j'}^f$ for the conversion rate of firm $f$ from quantities of the raw good $j$ to raw units, and for the conversion rate of the firm from raw units to quantities of the produced good $j'$, respectively. Similarly, we can also write $L^f$ for the bound on the number of raw units that the firm can produce on its \seg.  
\end{remark}

\noindent Next, we define the optimal production for firms and the optimal consumption for consumers. We start from the former.

\paragraph{Optimal Production.}
Consider firm $f$, and let $p_j$ and $p_j'$ be the prices of goods $j$ and $j'$ in the market respectively. 
Using some raw good $j\in\mathcal{R}_{f}$, the firm can produce
1 raw unit on its \seg at a \emph{cost} of $p_j/\alpha^f_j$. For any produced good $j'\in\mathcal{P}_{f}$,
the firm can use 1 raw unit to produce $\beta^f_{j'}$ units of good $j'$, resulting in a \emph{revenue} of 
$p_{j'}\beta^{f}_{j'k}$. With this in mind, we define the \emph{cost-per-unit (cpu)} and \emph{revenue-per-unit (rpu)} of firm $f$ as follows:
\begin{align*}\label{eq:cpu-and-rpu}
\mbox{cpu}^f(\mathbf{p}) = \min_{j\in\mathcal{R}_{f}}\frac{p_j}{\alpha^f_{j}}\quad\mbox{ and }\quad \mbox{rpu}^f(\mathbf{p}) = \max_{j'\in\mathcal{P}_{f}}p_{j'}\beta^f_{j'}
\end{align*}
Intuitively, the cpu of a firm is the minimum cost that it incurs from converting any of its raw goods into raw units, and the rpu is the maximum revenue that it obtains by converting raw units into any of its produced goods.\medskip

\noindent The \emph{profit-per-unit (ppu)} of firm $f$ is then defined as 

$$\mbox{ppu}^f(\mathbf{p})=\mbox{rpu}^f(\mathbf{p})-\mbox{cpu}^f(\mathbf{p})$$ 

\noindent We can now define a firm's optimal production.

\begin{definition}[Optimal Production] \label{def:optimal-prod-leontief-free}
Given the production vector $\mathbf{y}^f$ of firm $f$, a production for the firm is \emph{optimal} if the following conditions are satisfied:
\begin{enumerate}\label{LF:optimal-production}
\item $\sum_{j'\in\mathcal{P}_f}y^f_{j'}/\beta^f_{j'}\leq \sum_{j\in\mathcal{R}_f}\alpha^f_j y^f_j\leq L^f$. \hfill (\emph{\textbf{feasibility}}) \label{cond:LF-1}
\item If $\mbox{ppu}^f(\mathbf{p})<0$, then $\mathbf{y}^f=\allzeros$. \hfill (\emph{\textbf{zero unprofitable production}}) \label{cond:LF-2}
\item If $\mbox{ppu}^f(\mathbf{p})>0$, then $\sum_{j\in\mathcal{R}_f}\alpha^f_j y^f_j=L^f$. \hfill (\emph{\textbf{maximum profitable production of raw units}}) \label{cond:LF-3}
\item If $\mbox{rpu}^f(\mathbf{p})>0$, then $\sum_{j'\in\mathcal{P}_f}y^f_{j'}/\beta^f_{j'} = \sum_{j\in\mathcal{R}_f}\alpha^f_j y^f_j$. \label{cond:LF-4}
\hfill (\emph{\textbf{maximum usage of raw units}})
\item For every $j\in\mathcal{R}_f$, if $y^f_j>0$, then $p_j/\alpha^f_j = \mbox{cpu}^f(\mathbf{p})$. \hfill (\emph{\textbf{only use cost-optimal raw goods}}) \label{cond:LF-5}
\item For every $j'\in\mathcal{P}_f$, if $y^f_{j'}>0$, then $p_{j'}\beta^f_{j'} = \mbox{rpu}^f(\mathbf{p})$. \hfill (\emph{\textbf{only produce revenue-optimal produced goods}}) \label{cond:LF-6}
\end{enumerate}
\end{definition}

\noindent Before we proceed, with offer the following remark with respect to Condition~\cref{cond:LF-3} above.

\begin{remark}[Maximum Profitable Production]\label{rem:L-cant-be-infinity}
In \Cref{def:optimal-prod-leontief-free}, when $L^f = \infty$, Condition~\ref{cond:LF-3} stipulates that the firm should use an infite amount of raw goods for production. As we explained in the previous section (see \Cref{rem:bounds-on-production}) and as we we will reiterate in the context of the markets of this section in \Cref{rem:leontief-bounds-on-production} below, in a competitive equilibrium (see \Cref{def:leontief-free-market-equilibrium}) there is a finite, global upper bound to how much a firm can produce. In turn, this also imposes a global upper bound on how much a consumer can consume. Looking ahead, in the linear programs that we will construct, we will use some upper bounds on production and consumption, which will be sufficiently large to not constrain the aforementioned global upper bounds. For the case of production, that bound will be represented by $L$, and then, \Cref{cond:LF-3} should be interpreted as 
\[
\text{If } \mbox{ppu}^f(\mathbf{p})>0, \text{ then } \sum_{j\in\mathcal{R}_f}\alpha^f_j y^f_j=\min\{L^f,L\},
\]
where $L$ will be guaranteed to be large enough such that $\min\{L^f,L\} = L$ only when $L^f = \infty$. 
\end{remark}

\paragraph{Optimal Consumption and bang-per-buck (BPB).}

For the consumers, the optimal consumption amounts to them utilizing \segs in order of decreasing \emph{bang-per-buck}. Given consumer $i \in N$ and prices $\mathbf{p}$, the BPB of a pair $(j,k)$ consisting of a good $j \in G$ and a \seg $k \in K_i$ is defined as $\text{BPB}_{i}(j,k)=u^i_{jk}/p_j$. In particular, $\text{BPB}_{i}(j,k)$ has a maximum BPB over pairs $(j,k) \in G \times K_i$, then consumer $i$ will buy quantities of good $j$, and ``allocate'' the good to \seg $k$ until the upper bound $L_{k}^i$ on the utility that the consumer gain from this \seg has been met. This process is repeated until the budget of the consumer is exhausted.

\paragraph{Competitive Equilibrium.} We are now ready to define the notion of a competitive equilibrium in markets with Leontief-free utilities and productions.

\begin{definition}[Competitive Equilibrium - Markets with Leontief-free Utilites and Productions]\label{def:leontief-free-market-equilibrium}
A competitive equilibrium of an Arrow-Debreu market with Leontief-free utilities and Leontief-free productions is a tuple $(\mathbf{\bar{p}},\mathbf{\bar{x}},\mathbf{\bar{y}})$ consisting of non-negative prices $\mathbf{\bar{p}}$, non-negative bundles $\mathbf{\bar{x}}^{i}$ for $i\in\mathcal{N}$, and non-negative amounts of goods $\bar{y}^{f}\in\RRnn^{|\mathcal{R}_f\cup\mathcal{P}_f|}$ for $f\in F$ satisfying the following conditions:
\begin{enumerate}
	\item For every firm $f\in F$, $\bar{y}^f$ is an optimal production vector for $f$. \hfill \textbf{\emph{(firm profit maximization)}} \label{con1-prod:firm-profit-lf}
	\item For every consumer $i\in\mathcal{N}$, $\mathbf{\bar{x}}^i$ is an optimal consumption vector for $i$
    under the budget constraint $\sum_{j \in G}\sum_{k \in K_i} \bar{p}_j \cdot \bar{x}^i_{jk}\leq\sum_{j \in G} w^i_j \bar{p}_j+\sum_{f \in F}\theta^i_f \left(\sum_{j\in\mathcal{P}_f}\bar{p}_j\cdot \bar{y}^f_j-\sum_{j\in\mathcal{R}_f}\bar{p}_j\cdot \bar{y}^f_j\right)$,  
    \hfill \textbf{\emph{(bundle optimality)}} \label{con2-prod:optimality-lf}
    
	\item $\bar{z}_j \leq 0$, and $\bar{z}_j \cdot \bar{p}_j$, where for every good $j\in\mathcal{G}$, 
 
    $\bar{z}_j = \sum_{i \in N} \sum_{k \in K_i}\bar{x}^i_{jk}+\sum_{f \in F \colon j\in\mathcal{R}_f}\bar{y}^f_j-\sum_{f \in F\colon j\in\mathcal{P}_f}\bar{y}^f_j-1$ \label{con3-prod:market-clearing-lf} \hfill \textbf{\emph{(market clearing)}} 
\end{enumerate}
\end{definition}

\noindent Condition~\ref{con1-prod:firm-profit-lf} requires that at the chosen set of prices $\mathbf{\bar{p}}$, each firm maximizes its profit, given its production functions. Condition~\ref{con2-prod:optimality-lf} requires that at the chosen set of prices $\mathbf{\bar{p}}$, each consumer maximizes her utility subject to their budget constraints, where the budget consists of the amount earned from selling all the consumer's endowments and the profit share of the consumer from the production of the firms on the different \segs. Finally, Condition~\ref{con3-prod:market-clearing-lf} is the market clearing condition, which requires that the total consumption of each good is at most the total production plus the total endowment of the consumers, and supply equals demand for all goods which are not priced at $0$. Similarly to \Cref{sec:prod-markets-linear}, and as we explain later, we may in fact assume without loss of generality that in any competitive equilibrium \emph{all the prices are positive}, and hence Condition~\ref{con3-prod:market-clearing-lf} reduces to $\bar{z}_j =0$ for all $j \in G$. Note that in Condition~\ref{con3-prod:market-clearing-lf} we have used that $\sum_{\in N}w_{ij}=1$ for each good $j \in G$.

\paragraph{Sufficiency Conditions.} A competitive equilibrium as defined above exists for every market $\mathcal{M}$, under some sufficiency conditions. Similarly to \Cref{sec:prod-markets-linear}, we will use the set of conditions used by \citet{maxfield1997general}, also used in the series of papers on market equilibria that we mentioned in the beginning of \Cref{sec:markets}. 

\begin{enumerate}
    \item For every consumer $i \in N$, there exists some good $j \in G$ that the consumer endows in a positive amount, i.e., $w_{ij} > 0$. 
    \item We will say that a consumer $i \in N$ is \emph{\textbf{nonsatiated} for good $j \in G$} if there exists some \seg $k$ such that $u_{jk}^i >0$ and $L_k^i = \infty$. Following \citet{garg2018substitution}, we will assume that for any good $j \in G$, there exists some consumer $i \in N$ that is nonsatiated with respect to $j$. Similarly, a firm is \emph{\textbf{nonsatiated} for good $j \in G$} if $j \in \mathcal{R}_f$ and $L^f = \infty$. 

    We remark that this condition naturally generalizes the condition that ``for any good there exists some consumer with positive utility for that good'', which we used in \Cref{sec:exchange-markets-linear,sec:prod-markets-linear}.

    \item Consider a graph $\mathcal{G}_F(\mathcal{M})$ in which the nodes are the goods, and an edge $(j,j')$ has weight 
    \[
    w_{jj'}=\max_{f \in F \colon j\in\mathcal{R}_f\wedge j'\in\mathcal{P}_f}\alpha_j^f \cdot b_j'^f,
    \]
    If the set is empty, the weight is defined to be $0$.

    The above weight captures the fact that $j'$ can be produced from $j$, via the intermediate raw units, at combined conversion rate $w_{jj'}$ by some firm $f \in F$. Then, for any cycle $C=(g_0, g_1), (g_1,g_2)\ldots, (g_{k-2},g_{k-1})$ of $G_F(\mathcal{M})$ the product of the weights of the edges is less than $1$, i.e., $\prod_{e \in C} \alpha_e < 1$.\label{cond3-leontief-free-no-production-out-of-nothing}

    This condition is known as the \textbf{\emph{no production out of nothing and no vacuous production}} condition. Indeed, if $\prod_{e \in C} \alpha_e > 1$, then it would be possible to increase the quantity of some good, without decreasing the quantity of any other good. The case of $\prod_{e \in C} \alpha_e = 1$ refers to the case of vacuous production, which is also disallowed in our model, similarly to the related works. e.g., see \citet{SODA:GargV14}.

    \item Consider the \emph{economy graph} $\mathcal{G}_E(\mathcal{M})$ of the market $\mathcal{M}$ in which the nodes are the consumers and the firms, and 
    \begin{itemize}
        \item[-] there is an edge $(i,i')$ between consumer-node $i$ and consumer/firm-node $i'$ if $i$ endows a good $j$ for which $i'$ is nonsatiated. 
        \item[-] there is an edge $(f,i)$ between firm-node $f$ and consumer/firm-node $i$, if $L^f = \infty$ and there exists some good $j \in \mathcal{P}_f$ for which $i$ is nonsatiated. 
    \end{itemize}
    Then, $\mathcal{G}_E(\mathcal{M})$ contains a strongly connected component containing all the consumer-nodes. This condition generalizes the corresponding strong connectivity condition that we used in \Cref{sec:prod-markets-linear}.
\end{enumerate}

\begin{remark}[Bounds on Production]\label{rem:leontief-bounds-on-production}
Very similarly to \Cref{sec:prod-markets-linear}, we remark that there are some inherent bounds on how much a firm can produce on a \seg, even if for that segment we have $L^f = \infty$. The idea here is very similar to that of \Cref{rem:bounds-on-production}; given that we start from a finite set of endowed goods, the absence of cycles of profitable production in an equilibrium restricts the production to take place in chains of length bounded by $m$. This implies a global upper bound $L^c$ on how many raw goods can be used or produced. Again, looking ahead, this will allow us to impose ``loose'' upper bounds on the production and consumption in the linear programs that we will devise, without compromising the existence of a competitive equilibrium.
\end{remark}

\noindent \emph{Bounds on the prices.} Again, similarly to \Cref{sec:prod-markets-linear}, we may assume without loss of generality that all the prices are strictly positive, i.e., $p_j > 0$ for all goods $j \in G$. The argument to achieve this is very similar to the one used in \Cref{sec:prod-markets-linear}, and has also been established in \citep{garg2018substitution}. Note that by this assumption, the quantity $\text{BPB}_i(j,k)$ is well-defined for every $j \in G$ and $k \in K_i$.  

\begin{remark}[Normalized Prices]\label{rem:leontief-markets-normalized-prices}
Given that $p_j > 0$ for all $j \in G$, we can normalize the prices to sum to $1$ without loss of generality, i.e., we may assume that for every good $j \in G$, we have that $\sum_{j \in G} p_j = 1$. 
\end{remark}

\noindent Again, in a similar manner to \Cref{sec:exchange-markets-linear,sec:prod-markets-linear}, we will use a parameter $\varepsilon$ to capture the fact that if the price $p_j$ for a good $j \in G$ is sufficiently smaller than the price $p_{j'}$ for a good $j' \in G$, then $\text{BPB}(j,k) > \text{BPB}(j',k')$, for any $k,k' \in K_i$. Specifically, we can compute $\varepsilon > 0$ such that
\[
\text{If } p_j \leq \varepsilon \cdot p_{j'} \text{ and } u^i_{jk} > 0 \text{ then } \text{BPB}_i(j,k) > \text{BPB}_i(j',k').
\]
Additionally, we can pick $\varepsilon$ to be sufficiently small such that $\varepsilon < \frac{w_{ij}}{m}$ for all $i \in N$ and all $j \in G$. Given $\varepsilon$, we will impose a stricter lower bound on the prices, which will be useful later on: in particular, we will assume that for all $j \in G$, $p_j \geq \frac{\varepsilon^m}{m}$.

\subsubsection{Preprocessing}

The approach that we will take for the proof is very similar to the one that we used in \Cref{sec:prod-markets-linear}. The main difference comes from the fact that we now need to consider pairs of (goods, \segs) rather than just goods. This mainly complicates notation, but the type of arguments that we make are very much along the same lines as in the previous section. Again, we employ a standard variable change (see \emph{Gale's Substitution} in \Cref{rem:gale}) to work with expenditures for the consumers and the firms, rather than with quantities of goods. This is to avoid multiplications of parameters which will be input to the circuit that we will construct, in particular the prices $\mathbf{p}$ and the quantities $x_{jk}^i$, and $y_{jk}^i$ in the bundle optimality condition of the competitive equilibrium in \Cref{def:leontief-free-market-equilibrium}.

\paragraph{Consumer expenditure.}
For the consumers, we will use the following standard change of variables. We will let $q^i_{jk} = x^i_{jk} \cdot p_j$ be the \emph{expenditure} of consumer $i$ on good $j$ and \seg $k$, i.e., how much money the consumer spends on the good at a given \seg given an allocation $x^i_{jk}$ and a price $p_j$.

\paragraph{Firm expenditure.} For the firms, we will use a similar substitution for the case of production. Namely, we will let 
\begin{itemize}
    \item[-] For $j \in \mathcal{R}_f$, let $s_j^f = p_j\cdot y^f_j$ be the \emph{expenditure} of firm $f$ on its raw good $j$ on all \segs , i.e., the amount of money spend on $y_f$ units of the raw good $j$, at given set of prices $\mathbf{p}$. 
    
    Let $\mathbf{s}^f = \left(s^f_j\right)_{j \in \mathcal{R}_f}$ and let $\mathbf{s} = \left(\mathbf{s}^f\right)_{f \in F}$.
    
    \item[-] For $j \in \mathcal{P}_f$, let $r_j^f = p_j\cdot y^f_j$ be the \emph{revenue} of firm $f$ from its produced good $j$ on all \segs, i.e., the amount of money earned from $y_f$ units of the produced good $j$, at a given set of prices $\mathbf{p}$.

    Let $\mathbf{r}^f = \left(r^f_j\right)_{j \in \mathcal{P}_f}$ and let $\mathbf{r} = \left(\mathbf{r}^f\right)_{f \in F}$.
\end{itemize}

\paragraph{The program $\mathcal{P}_\text{con}$ for optimal consumer expenditure.} 
\begin{figure}
\centering
    \fbox{
\begin{minipage}{0.5\textwidth}
\centering
\underline{Linear Program $\mathcal{P_\text{con}}$}
\begin{align*}
\mbox{minimize}\quad & \sum_{j\in\mathcal{G}}\sum_{k \in K_i}\frac{p_j}{u^i_{jk}}\cdot q^i_{jk}\\
\mbox{subject to}\quad & \sum_{j\in\mathcal{G}}q^i_{jk}\leq \min\{L^i_k, CL+1\}\cdot \min_{j}\frac{p_j}{u^i_{jk}},\;\ \forall k \in K_i\\
& \sum_{j\in\mathcal{G}}\sum_{k \in K_i}q^i_{jk} =\sum_{j\in\mathcal{G}}w^i_jp_j+\sum_{f\in\mathcal{F}}\theta^i_f\cdot \left(\sum_{j \in \mathcal{P}_f} r_j^f-\sum_{j \in \mathcal{R}_f} s_j^f\right)\\
& q^i_{jk}\geq 0,\;\ \forall  j\in\mathcal{G}, \ \   \forall k \in K_i
\end{align*}
\end{minipage}}
\caption{The linear program $\mathcal{P_\text{con}}$ for the optimal expenditures in consumption. Note that $CL+1$ will be chosen to be large enough such that $\min\{L_k^i, CL+1\} = CL+1$ only when $L_k^i = \infty$.}  
\label{fig:leontief-free-consumption-lp}
\end{figure}

We are now ready to write the linear program $\mathcal{P}_1$, the solution of which will give us the optimal expenditures for the consumers, see \Cref{fig:leontief-free-consumption-lp}. In the linear program, notice the quantity $\min\{L_k^i, CL+1\}$. $CL+1$ will be chosen to be large enough, so that $\min\{L_k^i, CL+1\} = CL+1$ only when $L_k^i = \infty$. This is to capture the cases in which the consumer should be allowed to spend an infinite amount of money on goods on a \seg, since there is upper bound on the utility the she can accrue from this \seg. Still, we remark again (see \Cref{rem:leontief-bounds-on-production}) that there are inherent bounds on how much a consumer can spend, which do not come from $L_k^i$ but rather from the inherent global upper bounds on production. In particular, it holds that $\sum_{j \in G} u_{jk}^ix_{jk}^i\leq L^p$, where $L^p$ is a global upper bound on the consumption resulting from the global upper bound $L^c$ on the production. 

In previous sections, it was straightforward to see that the from the optimal solution $\mathbf{q}^i$ of the corresponding linear program, one could recover the optimal bundle $\mathbf{x}^i$ of consumer $i$. Here, we prove that in a simple lemma below. Note that we only need to guarantee that the linear program computes the optimal expenditures correctly at a competitive equilibrium, and so we can use the fact that $\sum_{j \in G} u_{jk}^ix_{jk}^i\leq L^p$.

\begin{lemma}\label{lem:leontief-free-optimal-consumption-captured-by-lp}
Suppose that $\mathbf{q}^i$ is an optimal solution to linear program $\mathcal{P}_\text{con}$ of \Cref{fig:leontief-free-consumption-lp}, and let Let $x^i_{jk} = q^i_{jk}/p_j$ for all $j \in G$ and $k \in K_i$. Then, the resulting bundle $\mathbf{x}^i$ is an optimal consumption vector for consumer $i$ at a competitive equilibrium of the market.
\end{lemma}

\begin{proof}
    First note that the linear program is feasible, because by the strong connectivity of the economy graph $G_E(\mathcal{M})$, $i$ is non-satiated with respect to some good, which means that she can spend its entire budget on this good. By construction, consumer $i$ will spend money in order of increasing $\text{BPB}_i(j,k)$, establishing the optimality of the resulting bundle $\mathbf{x}^i$. It remains to show that the feasibility constraints are satisfied. Note that if $x^i_{jk}>0$, then $q^i_{jk}>0$, establishing that $p_j/u^i_{jk}=\min_{j'}p_{j'}/u^i_{j'k}$. Hence, for any $i \in G$ an $k \in K_i$, we have that
    \begin{align*}
        \sum_{j \in G}u^i_{jk}x^i_{jk} & = \sum_{j \in G} \left(\frac{u^i_{jk}}{p_j} \right)q^i_{jk} 
         = \frac{1}{\underset{{j' \in G}}{\min} \ \  p_{j'}/{u^i_{j'k}}}\sum_{j \in G} q^i_{jk}  \\
        &\leq \frac{1}{\underset{{j' \in G}}{\min} \ \  p_{j'}/u^i_{j'k}}\cdot \min\{L^i_k,CL+1\} \cdot \min_{j' \in G}\frac{p_{j'}}{u^i_{j'k}}
         = \min\{L^i_k,CL+1\}
    \end{align*}
This completes the proof of the lemma, since $CL+1$ will be chosen to be larger than any finite $L_k^i$ and larger than $L^p$.
\end{proof}

\paragraph{The programs $\mathcal{P}_\text{prod1}$ and $\mathcal{P}_\text{prod2}$ for optimal firm expenditure.}

In \Cref{sec:prod-markets-linear} we devised a single linear program for the optimal production of the firms. Here, since we have variables for both the expenditures on raw goods and the expenditures on produced goods, we will devise two linear programs, one for each case. Those can be seen in \Cref{fig:leontief-free-productions-lp}. Again, in \Cref{sec:prod-markets-linear} the corresponding linear program almost straightforwardly captured the optimal production of the firm which it represented; here we argue that this is the case with another simple lemma. 

\begin{figure}
\centering
    \fbox{
\begin{minipage}{0.4\textwidth}
\centering
\underline{Linear Program $\mathcal{P_\text{prod1}}$}
\begin{align*}
\mbox{maximize}\quad & \sum_{j\in\mathcal{R}_f}\left(\mbox{rpu}^f(\mathbf{p})-\frac{p_j}{\alpha^f_j}\right) s^f_j\\
\mbox{subject to}\quad & \sum_{j\in\mathcal{R}_f}s^f_j\leq \min\{L^f,L\}\cdot\mbox{cpu}^f(\mathbf{p})\\
& s^f_j \geq 0,\;\forall j\in\mathcal{R}_f
\end{align*}
\end{minipage}
\hfill\vline\hfill
\begin{minipage}{0.5\textwidth}
\centering
\underline{Linear Program $\mathcal{P_\text{prod2}}$}
\begin{align*}
\ \ \mbox{maximize}\quad & \sum_{j'\in\mathcal{P}_f}\left(p_{j'}\beta^f_{j'}\right)\cdot r^f_{j'}\\
\ \ \mbox{subject to}\quad & \sum_{j'\in\mathcal{P}_f}r^f_{j'} = \sum_{j\in\mathcal{R}_f}s^f_j+ B\\
& B = \max\{0,\mbox{rpu}^f(\mathbf{p})-\mbox{cpu}^f(\mathbf{p})\}\cdot \min\{L^f,L\}\\
& r^f_{j'}\geq 0,\;\forall j'\in\mathcal{P}_f
\end{align*}
\end{minipage}}
\caption{The linear programs $\mathcal{P_\text{prod1}}$ and $\mathcal{P_\text{prod2}}$ for the optimal expenditures in production, for the raw goods (left) and for the produced goods (right). Note that $L$ will be sufficiently large such that $\min\{L^f,L\}=L$ will only hold when $L^f =\infty$, in both programs. Also note that the variable $B$ is only used for notational convenience.}  
\label{fig:leontief-free-productions-lp}
\end{figure}

\begin{lemma}  \label{lem:leontief-free-optimal-production-captured-by-lp}
Suppose that $\mathbf{s}^f$ and $\mathbf{r}^f$ are solutions to the linear programs $\mathcal{P_\text{prod1}}$ and $\mathcal{P_\text{prod2}}$ of \Cref{fig:leontief-free-productions-lp} respectively. Let $y^f_j = s^f_j/p_j$ for $j\in\mathcal{R}_f$ and $y^f_{j'}=r^f_{j'}/p_{j'}$ for $j'\in\mathcal{P}_f$. Then $\mathbf{y}^f$ is an optimal production vector for firm $f$ at a competitive equilibrium of the market.
\end{lemma}

\begin{proof}
We will argue that $\mathbf{y}^f$ satisfies all the conditions of \Cref{def:optimal-prod-leontief-free}. 
Since $\mathbf{y}^f$ consists of optimal solutions to the two linear programs, Condition~\ref{cond:LF-5} and Condition~\ref{cond:LF-6} are clearly satisfied. Now, suppose that $\mbox{ppu}^f(p)<0$. This implies that all the coefficients of linear program $\mathcal{P}_\text{prod1}$ are negative, meaning that $s^f_j = 0$ for all $j\in\mathcal{R}_f$. By the constraints of linear program $\mathcal{P}_\text{prod1}$, it follows that also $r^f_j=0$ for all $j\in\mathcal{P}_f$. Hence $y^f=\allzeros$, and Condition~\ref{cond:LF-2} also holds. If $\mbox{ppu}^f(p)>0$, then at least one coefficient of linear program $\mathcal{P}_\text{prod1}$ is positive, implying that $\sum_{j\in\mathcal{R}_f}s^f_j = \min\{L^f,L\}\cdot\mbox{cpu}^f(\mathbf{p})$. Also, if $s^f_j > 0$, then $p_j/\alpha^f_j = \mbox{cpu}^f(\mathbf{p}).$ Hence, we have that
\begin{align*}
    \sum_{j\in\mathcal{R}_f}\alpha^f_j y^f_j  = \sum_{j\in\mathcal{R}_f}\frac{\alpha^f_j}{p_j}\cdot s^f_j
     = \frac{1}{\mbox{cpu}^f(\mathbf{p})}\sum_{j\in\mathcal{R}_f} s^f_j
     = \frac{1}{\mbox{cpu}^f(\mathbf{p})}\cdot \min\{L^f,L\}\cdot \mbox{cpu}^f(\mathbf{p})
     = \min\{L^f,L\}
\end{align*}
and hence Condition~\ref{cond:LF-3} is satisfied (see also \Cref{rem:L-cant-be-infinity}). 
What remains is to argue is that $\mathbf{y}^f$ satisfied Condition~\Cref{cond:LF-1} of \Cref{def:optimal-prod-leontief-free}, i.e. that it is a feasible production vector. If $\mbox{ppu}^f(\mathbf{p})<0,$ then $y^f=\allzeros$ is feasible. If $\mbox{ppu}^f(\mathbf{p})>0,$ then it can be verified that $\sum_{j\in\mathcal{R}_f}\alpha^f_j y^f_j = \sum_{j\in\mathcal{P}_f}y^f_j/\beta^f_j$. Finally, if $\mbox{ppu}^f(\mathbf{p})=0$, then
\begin{align*}
    \sum_{j'\in\mathcal{P}_f}\frac{y^f_{j'}}{\beta^f_{j'}} 
     = \sum_{j'\in\mathcal{P}_f}\frac{r^f_{j'}}{p_{j'}\cdot \beta^f_{j'}}  
     = \sum_{j'\in\mathcal{P}_f}\frac{r^f_{j'}}{\mbox{rpu}^f(\mathbf{p})} 
     = \sum_{j\in\mathcal{R}_f}\frac{s^f_{j}}{\mbox{cpu}^f(\mathbf{p})} 
     = \sum_{j\in\mathcal{R}_f}\frac{\alpha^f_j s^f_{j}}{p_{j}} \
     = \sum_{j\in\mathcal{R}_f}\alpha^f_j y^f_j.
\end{align*}
This establishes all of the properties of \Cref{def:optimal-prod-leontief-free}. Finally, let $\hat{G}=\{j \in R_f \colon \mbox{rpu}(\mathbf{p}) - p_j/\alpha_j^f > 0\}$. The corresponding constraint of linear program $\mathcal{P}_\text{prod1}$ stipulates that $\sum_{j \in \hat{G}}s_j^f = \min\{L^f,L\}\cdot \mbox{cpu}^f(\mathbf{p})$. Our proof will establish that in a competitive equilibrium, if $L^f > L^c$ (the global upper bound on production), then $\hat{G}$ does not contain any goods. This is to establish that the optimality of $s_j^f$ does not impose any artificial constraints on the expenditure on raw goods due to the constraints of the linear program.
\end{proof}

\paragraph{Excess expenditure and the program $\mathcal{Q}$.} Similarly to \Cref{sec:exchange-markets-linear,sec:prod-markets-linear}, the equilibrium prices will be obtained via a feasibility program $\mathcal{Q}$. Before we define the program, we introduce the notion of excess expenditure; the definition is very similar to the one we used in previous sections. The \emph{excess expenditure} $e_j$ of a good $j \in G$ is defined as the difference between the total expenditure of all consumers $i \in N$ and firms $f \in F$ for that good and the price of the good plus the total revenue of firms $f \in F$ from that good, i.e., 
\[
e_j = \sum_{i \in N}\sum_{k \in K_i} q^i_{jk} +\sum_{f \in F \colon j \in \mathcal{R}_f}s^f_j - p_j - \sum_{f \in F \colon j \in \mathcal{P}_f}r_j^f.
\]
Then, at equilibrium prices $\mathbf{\bar{p}}$ we will have market clearing, i.e., $\bar{e}_j = 0$, for all $j \in G$. \medskip

\noindent We are now ready to define our feasibility program $\mathcal{Q}$; see \Cref{fig:leontief-free-markets-Q}. The program looks in fact identical with that which we used in the previous sections; the only difference being the updated definition of the expenditures $e_j$. The equilibrium prices will be obtained as the output $\mathbf{\bar{p}}$ of $\mathcal{Q}$. 

\begin{figure}
\centering 
\fbox{
\begin{minipage}{0.5\textwidth}
\centering
        \underline{Feasibility Program $\mathcal{Q}$}
        \begin{align*}
	& e_j<e_{j'}\Rightarrow p_j\leq \varepsilon \cdot p_{j'},\;\forall j,j' \in G \nonumber \\
	& p_j\geq \frac{\varepsilon^m}{m},\;\forall j \in G\\
	& \sum_{j \in G}p_j = 1\nonumber
        \end{align*}
\end{minipage}
}
\caption{The feasibility program $\mathcal{Q}$ used to find market-clearing prices.}
\label{fig:leontief-free-markets-Q}
\end{figure}

\subsubsection{Membership in PPAD: The proof of \crtCref{thm:Leontief-free}}

\noindent We will again develop the proof in three steps, namely (a) construction of the function $F$ and arguing that it can be represented by a \linear arithmetic circuit containing \linoptgates, (b) showing that the \linoptgate can compute all the necessary components, and (c) arguing that a competitive equilibrium can be recovered from a fixed point of $F$. 

\paragraph{The function $F$.} Given the above, we will define $F: D \rightarrow D$ with domain 
\[D = \left\{p \in \Delta^{m-1}: \forall j \in G, p_j \geq \frac{\varepsilon^m}{m}\right\}  \times_{i \in N} \left([0,C\cdot L + 1]^{m|K_i|}\right) \times [0,L]^{\ell m}.\]
An input to $F$ is a tuple $(\hat{\mathbf{p}}, \hat{\mathbf{q}}, \hat{\mathbf{s}},\hat{\mathbf{r}})$ of prices, consumer expenditures and firm expenditures for raw goods and produced goods, and the output is another such tuple $(\mathbf{\bar{p}}, \mathbf{\bar{q}}, \bar{\mathbf{s}},\bar{\mathbf{r}})$. The domain $D$ is the one above since $\sum_{j \in G} \hat{p}_j = 1$ (see \Cref{rem:leontief-markets-normalized-prices}), and by the upper bounds of $C\cdot L +1$ and $L$ imposed on the maximum expenditures for consumption and production respectively. 

\paragraph{Computation by the \linoptgate.} Next, we argue that the solutions to the linear programs $\mathcal{P}_\text{con}$ of \Cref{fig:leontief-free-consumption-lp}, $\mathcal{P}_\text{prod1}$ and $\mathcal{P}_\text{prod2}$ of \Cref{fig:leontief-free-productions-lp}, and the feasibility program $\mathcal{Q}$ of \Cref{fig:leontief-free-markets-Q} can be computed by our \linoptgate.

\begin{lemma}\label{lem:leontief-free-Markets-OPT-gate-can-compute}
Consider the linear program $\mathcal{P}_\text{con}$ of \Cref{fig:leontief-free-consumption-lp}, the linear programs $\mathcal{P}_\text{prod1}$ and $\mathcal{P}_\text{prod2}$ of \Cref{fig:leontief-free-productions-lp}, and the feasibility program $\mathcal{Q}$ of \Cref{fig:leontief-free-markets-Q}. An optimal solution to $\mathcal{P}_\text{con}$, to $\mathcal{P}_\text{prod1}$ and to $\mathcal{P}_\text{prod2}$ can be computed by the \linoptgate. $\mathcal{Q}$ is solvable, and a solution can be computed by the \linoptgate. 
\end{lemma}

\begin{proof}
The proof is very similar to that of \Cref{lem:productionMarkets-OPT-gate-can-compute}. For all of the linear programs, the feasibly domains non-empty and bounded. The \circparams appear only on the right-hand side of the constraints: these are the prices $\hat{p}_j$, and the production expenditures $\hat{s}_j^f$ and $\hat{r}_j^f$ for linear program $\mathcal{P}_\text{con}$, the prices $\hat{p}_j$ (as part of $\mbox{cpu}^f(\mathbf{\hat{p}})$) for linear program $\mathcal{P}_\text{prod1}$ , the prices $\hat{p}_j$ (as part of $\mbox{cpu}^f(\mathbf{\hat{p}})$ and $\mbox{cpu}^f(\mathbf{\hat{p}})$) and the expenditures $s_j^f$ for linear program $\mathcal{P}_\text{prod2}$. The subgradients of the objective functions for all three linear programs are linear, and hence can be computed by a \pseudog. For the feasibility program $\mathcal{Q}$, the arguments that it is of the correct form and that it is solvable are identical to those of \Cref{lem:productionMarkets-OPT-gate-can-compute} (and hence also of \Cref{lem:ExchangeMarkets-OPT-gate-can-compute}), noting the updated definition of the expenditure $e_j$. 
\end{proof}

\paragraph{Arguing optimality and market clearing.} To conclude the proof, what is left to show is that a fixed point $(\mathbf{\bar{p}}, \mathbf{\bar{q}}, \bar{\mathbf{s}},\bar{\mathbf{r}})$ of $F$ indeed corresponds to a competitive equilibrium. We argue that in the following lemma.

\begin{lemma}\label{lem:leontief-free-markets-correctness-lemma}
Let $(\mathbf{\bar{p}}, \mathbf{\bar{q}}, \bar{\mathbf{s}},\bar{\mathbf{r}})$ be a fixed point of $F$. Then the triple $(\mathbf{\bar{p}}, \mathbf{\bar{x}}, \bar{\mathbf{y}},\bar{\mathbf{r}})$, where $\bar{x}^i_{jk} = \bar{q}^i_{jk}/\bar{p}_j$, $\bar{y}^f_j = \bar{s}^f_j/\bar{p}_j$ for $j\in\mathcal{R}_f$ and $\bar{y}^f_{j'}=\bar{r}^f_{j'}/\bar{p}_{j'}$ for $j'\in\mathcal{P}_f$ is a competitive equilibrium of $\mathcal{M}$.
\end{lemma}

\begin{proof}
In \Cref{lem:leontief-free-optimal-consumption-captured-by-lp} and \Cref{lem:leontief-free-optimal-production-captured-by-lp} we argued that the linear programs $\mathcal{P}_\text{con}$, $\mathcal{P}_\text{prod1}$ and $\mathcal{P}_\text{prod2}$ correctly calculate the optimal consumption vectors of the consumers and the optimal production vectors of the firms. What is left to argue is that $\mathbf{\bar{p}}$ is a vector of market-clearing prices. The proof will follow along very much the same lines as the one of \Cref{lem:production-markets-correctness-lemma}, using updated the updated terminology and notation of this section.

By the definition of $e_j$, arguing about market-clearing is equivalent to arguing that for all $j \in G$, we have that $\bar{e}_j = 0$. Using a very similar calculation to the one that we used in \Cref{lem:production-markets-correctness-lemma}, this is equivalent to arguing that $\bar{e}_j \geq 0$ for all $j \in G$. We provide the calculation below.

\begin{align*}
\sum_{j \in G} \bar{e}_j &= \sum_{j \in G} \left(\sum_{i \in N}\sum_{k \in K_i} \bar{q}^i_{jk} - \bar{p}_j +\sum_{f \in F \colon j \in \mathcal{R}_f}\bar{s}_f - \sum_{f \in F \colon j \in \mathcal{P}_f}\bar{r}_f \right) \\ 
&= \sum_{i \in N} \sum_{j \in G}\sum_{k \in K_i} \bar{q}^i_{jk} - \sum_{j \in G}\bar{p}_j + \sum_{f \in F}(\bar{s}_f-\bar{r}_f)  \\ 
&= \sum_{j \in G} \sum_{i \in N} w_{ij}\bar{p}_j - 1 + \sum_{i \in N}\sum_{f \in F}\theta_{if}\cdot (\bar{r}_f - \bar{s}_f) + \sum_{f \in F}(\bar{s}_f-\bar{r}_f) \\ 
& = \sum_{j \in G}\bar{p}_j\cdot \sum_{i \in N}w_{ij} - 1 + \sum_{f \in F}(\bar{r}_f - \bar{s}_f)\cdot \sum_{i \in N}\theta_{if}\cdot  \sum_{f \in F}(\bar{s}_f-\bar{r}_f) \\
& = \sum_{i\in N}w_{ij}-1 + \sum_{f \in F}(\bar{r}_f - \bar{s}_f) + \sum_{f \in F}(\bar{s}_f - \bar{r}_f) \\
& = 0,
\end{align*}
where in the calculations above we used:
\begin{itemize}
    \item[-] in Equation 2, that $\sum_{j \in G}\sum_{f \in F \colon j \in \mathcal{R}_f}\bar{s}_f = \sum_{f \in F}\bar{s}_j$ and that $\sum_{j \in G}\sum_{f \in F \colon j \in \mathcal{P}_f}\bar{r}_f = \sum_{f \in F}\bar{r}_f$,
    \item[-] in Equation 3, that $\sum_{i \in N} \sum_{j \in G}\sum_{k \in K_i}\bar{q}_{ij} = \sum_{i\in N}\sum_{j \in G}w_{ij}\bar{p}_j + \sum_{i \in N} \sum_{f \in F}\theta_{if}(\bar{r}_f-\bar{s}_f)$, which follows from the bundle optimality of the market equilibrium in \Cref{def:leontief-free-market-equilibrium}, which at an optimal solution is sastisfied with equality,
    \item[-] in Equations 3 and 5, that $\sum_{j \in G}\bar{p}_j =1$, which is without loss of generality, see \Cref{rem:leontief-markets-normalized-prices},
    \item[-] in Equation 5, that $\sum_{i \in N}\theta_{if} = 1$ for every $f \in F$, by definition, and
    \item[-] in Equation 6, that $\sum_{i \in N}w_{ij} = 1$, which is without loss of generality.
\end{itemize}

\noindent Assume by contradiction that there exists some $j_1 \in G$ such that $\bar{e}_{j_1} < 0$; we will obtain a contradiction to the strong connectivity of the economy graph of the market. We define the following three sets:
\begin{itemize}
    \item[-] $J = \{j \in G \colon \bar{e}_j \leq \bar{e}_{j'}, \text{ for all } j' \in G\}$. In other words, $J$ is the set of goods with minimum excess expenditure.
    \item[-] $N_J = \{i \in N \colon \text{ there exists } j \in J \text{ such that consumer } i \text{ is nonsatiated for good } j\}$. 
    \item[-] $F_J=\{f \in F \colon \text{ there exists } j \in J \text{ such that firm } f \text{ is nonsatiated for good } j\}$. 
\end{itemize}
Since $e_{j_1}^* < 0$, by the definition of $J$ we know that
\begin{equation}\label{eq:leontief-prod-min-expenditure}
\text{for all } j \in J, \text{ we have } e_j^* < 0.
\end{equation}
Also, since $e_{j_1}^* < 0$, and $\sum_{j \in G}e_j^*=0$, there must exist some other good $j_2$ such that $e_{j_2}^* > 0$. In particular, that implies that $J$ is a strict subset of $G$, i.e., $J \subsetneq G$.
We state and prove the following claim.

\begin{claim}\label{claim:leontief-prod-claim1}
For all $i \in N_J$ and all $j' \in G\setminus J$, it holds that $w_{ij'}=0$.
\end{claim}

\begin{proof}
The proof of the claim is very similar to that of \Cref{claim:prod-claim1} which we presented in \Cref{sec:prod-markets-linear}. Assume by contradiction that there exists some consumer $i_0 \in N_J$ and some good $j' \in G\setminus J$ such that $w_{i_0j'}>0$. We will show that this implies that there exists some $j_\ell \in J$ with $\bar{e}_{j_\ell} \geq 0$, contradicting Statement~\eqref{eq:leontief-prod-min-expenditure}. Since $\mathbf{\bar{p}}$ is a solution to feasibility program $\mathcal{Q}$ of \Cref{fig:production-markets-Q}, we have that $\bar{p}_j \leq \varepsilon \cdot \bar{p}_{j'}$, for every $j \in J$. Hence, by choosing $\varepsilon$ to be sufficiently small, we may lower-bound the budget of consumer $i$ by
\begin{equation*}
w_{i_0j'} \cdot p_{j'}^* \geq \frac{w_{i_0j'}}{\varepsilon \cdot m} \cdot \sum_{j \in J}\bar{p}_j \geq (C\cdot L+1)\sum_{j \in J} p_j^*
\end{equation*}
Since $i_0 \in N_J$, by the choice of $\varepsilon$, there is at least one good $g \in J$ such that $\text{BPB}_{i_0}(g,k) > \text{BPB}_{i_0}(g',k)$, for all $g' \in G\setminus J$, and for any two \segs $k,k \in K_i$, i.e., the consumer prefers to buy quantities of good $g$ rather than any good which is not in the set $J$, on any \seg. Since consumer $i_0$ is nonsatiated for some good $j \in J$ (by virtue of being in $N_J$), by the fact that $\bar{q}_{i_0j}$ is an optimal solution to linear program $\mathcal{P}_\text{con}$ of \Cref{fig:leontief-free-consumption-lp}, it should be that $\bar{q}_{i_0j} = C \cdot L +1$, which implies that $\bar{q}_{i_0j} \geq (C \cdot L+1)p_j^*$.  By substituting that into the definition of the excess expenditure $\bar{e}_j$, we obtain the following inequality:
\begin{align*}
0 > \bar{e}_j^* \geq (CL+1)\bar{p}_j - \sum_{f \in F \colon j \in \mathcal{P}_f} \bar{r}_j^f - \bar{p}_j = CL\cdot \bar{p}_j - \sum_{f \in F \colon j \in \mathcal{P}_f} \bar{r}_j^f.
\end{align*}
Therefore, it suffices to show that 
\[\sum_{f \in F \colon j \in \mathcal{P}_f} \bar{r}_j^f \leq (C\cdot L)\cdot \bar{p}_j;\] in that case we will have $\bar{e}_j \geq 0$ and we will obtain a contradiction. Now consider any firm $f \in F$. Note that if $\bar{r}^f \leq 0$, then the inequality above clearly follows. Otherwise, if $\bar{r}_f^j >0$, then $\mbox{rpu}(\mathbf{\bar{p}}) = \bar{p}_j \cdot \beta_j^f$. In this case, using the constraints of linear program $\mathcal{P}_\text{prod2}$ of \Cref{fig:leontief-free-productions-lp}, we may bound the production of good $j$ by firm $f$ as follows: 
\begin{align*}
   \bar{r}^f_j&\leq  \sum_{j\in\mathcal{R}_f}\bar{s}^f_j+\max\{0,\mbox{rpu}^f(\mathbf{\bar{p}})-\mbox{cpu}^f(\mathbf{\bar{p}})\}\cdot \min\{L^f,L\}\\
   &\leq\mbox{rpu}^f(\mathbf{\bar{p}})\cdot \min\{L^f,L\}\\
   &=(\beta^f_j \min\{L^f,L\})\cdot \bar{p}_j\\
   &\leq (\beta^f_j L)\bar{p}_j.
\end{align*}
This clearly bounds $\bar{r}_j^f$ as long as the $L^f \neq \infty$, as in that case $L^f \leq L$. $L^f = \infty$, then $\bar{r}j^f$ is upper bounded by the global upper bound $L^c$ on the production. Again, $L^c \leq L$, and the above inequality follows. 
\noindent By letting $C\geq |\mathcal{F}|\cdot\max\{\beta^f_j\colon f\in\mathcal{F}\wedge j\in\mathcal{P}_f\}$, we obtain the contradiction $0>\bar{e}_j\geq 0$. This concludes the proof of the claim. 
\end{proof}

\noindent \Cref{claim:leontief-prod-claim1} establishes that a consumer in the set $N_J$ cannot endow any good that is not in the set $J$ in a positive quantity. Since $J \subsetneq G$, this implies that $N_J \subsetneq N$, since every good is positively endowed by some consumer $i \in N$. Now consider any consumer $i \in N_J$ and any consumer $i' \in N\setminus N_J$ or any firm $f \in F\setminus F_J$. \Cref{claim:leontief-prod-claim1} implies that there cannot be an edge $(i,i')$ or an edge $(i,f)$ in the economy graph of the market. Indeed, for $(i,i')$, notice that by the claim consumer $i$ only endows positively goods that are in $J$, but consumer $i'$ has value $0$ for them by virtue of being in $N\setminus N_J$, and hence she is not nonsatiated for any good endowed positively by consumer $i$. Similarly, for $(i,f)$, we have that $\mathcal{R}_f \subseteq G\setminus J$ by definition, but consumer $i$ only positively endows goods in $J$, so clearly $f$ cannot be nonsatiated for any of the goods endowed by consumer $i$. \medskip

\noindent Next, we state and prove the following claim.

\begin{claim}\label{claim:leontief-prod-claim2}
Consider and $f \in F_J$, with $L^f = \infty$. For any good $j' \in \mathcal{P}_f$, it holds that $j' \in J$.
\end{claim}

\begin{proof}
The claim also follows a very similar idea to that of \Cref{claim:leontief-prod-claim2} which we presented in \Cref{sec:prod-markets-linear}. Let $\alpha = \max_{f',j}\alpha^{f'}_j$ and $\beta = \max_{f',j}\beta^{f'}_j$. Assume by contradiction that for firm $f$, there is some good $j'\in\mathcal{P}_f$ such that $j' \in G\setminus J$. Note that $f$ is nonsatiated with respect to some good in $J$ (by virtue of being in $F_J$ and $L^f = \infty$). Additionally, since $\mathbf{\bar{p}}$ is an optimal solution to feasibility program $\mathcal{Q}$ of \Cref{fig:leontief-free-markets-Q}, we have that $\bar{p}_j \cdot \varepsilon \bar{p}_{j'}$ for all goods $j \in J$. By the fact that $\mathbf{\bar{s}^f}$ is an optimal solution to linear program $\mathcal{P}_\text{prod1}$, it follows that firm $f$ will utilize only goods from $J$ as raw goods for its production. Hence, there exists some good $j_1\in J$ such that
    \begin{align*}
        \bar{s}^{f}_{j_1}\geq \frac{L\cdot\mbox{cpu}^{f_1}(\mathbf{p})}{m}\geq\frac{L}{m\alpha }\cdot \bar{p}_{j_1}
    \end{align*}
    Using this, we conclude that 
 \begin{align}\label{eq:leontief-free-chain-first}
        0>\bar{e}_{j_1}\geq \Big(\frac{L}{m \alpha}-1\Big)\bar{p}_{j_1}-\sum_{f'\colon j_1\in\mathcal{P}_{f_1}}\bar{r}^{f'}_{j_1}
    \end{align}
If $j_1$ were not being produced, we would reach the contradiction that $0>\bar{e}_{j_1}\geq 0$ by picking $L$ sufficiently large. If $j_1$ \emph{is} being produced, then it must be produced using some good $j_2$ from $J$. This is because by the choice of $\varepsilon$, it is only profitable for any firm to produce units of goods $j \in J$ (that have very low prices) only by using units of goods $j \in J$ as raw goods. Again, if $j_2$ is not being produced, we reach a contradiction by choosing $L$ sufficiently large, and otherwise we may repeat the argument. Since there cannot be a cycle of profitable production in the graph $G_\mathcal{F}(\mathcal{M})$, by the no production out of nothing and no vacuous production sufficiency condition, the argument will need to be repeated at most $m$ times, one for each good in the possibly longest production chain of the graph. In the end, we can pick $L$ to be large enough to ensure that we get a contradiction. 

We remark that we made the very same argument in the proof of \Cref{claim:prod-claim2} via establishing appropriate bounds inductively for every step of the chain. In the case of markets with Leontief-free productions however, calculating the resulting quantities becomes very tedious, so we elected to present the proof a bit differently.\medskip

\noindent As it stands, Inequality~\ref{eq:leontief-free-chain-first} implies that there exists a firm $f_1$ such that $r^{f_1}_{j_1}\geq \Big(\frac{L}{\ell m\alpha}-\frac{1}{\ell} \Big)\bar{p}_{j_1}$. As $f_1$ produces a good in $J$, it holds that 
\begin{enumerate}
    \item $\mbox{rpu}^{f_1}(\mathbf{p})\geq \mbox{cpu}^{f_1}(\mathbf{p})$ and \label{enum:claim2-leontief}
    \item $f_1$ produces only using goods from $J$.
\end{enumerate}
We consider two cases, depending on whether the inequality in Inequality~\ref{enum:claim2-leontief} is strict or not.

\paragraph{Case 1: $\mbox{rpu}^{f_1}(\mathbf{\bar{p}}) = \mbox{cpu}^{f_1}(\mathbf{\bar{p}})$.}
        In this case, we find that
        \begin{align*}
            \Big(\frac{L}{\ell m\alpha}-\frac{1}{\ell}\Big)\bar{p}_{j_1}\leq \bar{r}^{f_1}_{j_1}\leq\sum_{j\in\mathcal{P}_{f_1}}\bar{r}^{f_1}_j = \sum_{j\in\mathcal{R}_{f_1}} \bar{s}^{f_1}_j
        \end{align*}
        From this it follows that there exist a good $j_2$ such that $\bar{s}^{f_1}_{j_2}\geq \big(\frac{L}{\ell m^2\alpha}-\frac{1}{\ell m}\big)\bar{p}_{j_1}$. Using that $\bar{p}_{j_2}/\alpha^{f_1}_{j_2}=\mbox{cpu}^{f_1}(\mathbf{\bar{p}})=\mbox{rpu}^{f_1}(\mathbf{\bar{p}})=\beta^{f_1}_{j_1}\bar{p}_{j_1}$, this implies that $\bar{s}^{f_1}_{j_2}\geq \big(\frac{L}{\ell m^2\alpha^2\beta}-\frac{1}{\ell m\alpha\beta}\big)\bar{p}_{j_2}$. Now we may bound $\bar{e}_{j_2}$ as before and repeat the argument.

\paragraph{Case 2: $\mbox{rpu}^{f_1}(\mathbf{\bar{p}}) > \mbox{cpu}^{f_1}(\mathbf{\bar{p}})$.}
        First, we show that $L^{f_1}=\infty$. Assume that it is not. Then we have:
        \begin{align}
            \Big(\frac{L}{\ell m\alpha}-\frac{1}{\ell}\Big)\bar{p}_{j_1}\leq \bar{r}^{f_1}_{j_1}\leq\sum_{j\in\mathcal{P}_{f_1}}\bar{r}^{f_1}_j\leq \min\{L^{f_1},L\}\cdot\mbox{rpu}^{f_1}(\mathbf{\bar{p}})\leq (\min\{L^{f_1},L\}\beta)\bar{p}_{j_1} \leq (L^{f_1}\beta)\bar{p}_{j_1},
        \end{align}
        where the last inequality follows because $L$ is chosen sufficiently large for $L^f_1 \leq L$ to hold, since $L^f_1 \neq \infty$ by assumption. For $L$ sufficiently large though, this leads to a contradiction. Hence we may assume henceforth that $L^{f_1}=\infty$. \medskip
        
        \noindent Given the above, as $\mbox{rpu}^{f_1}(\bar{\mathbf{p}}) > \mbox{cpu}^{f_1}(\mathbf{\bar{p}})$, we have that $\sum_{j\in\mathcal{R}_{f_1}}\bar{s}^{f_1}_j = L\cdot\mbox{cpu}^{f_1}(\bar{\mathbf{p}})$. As $f_1$ produces a good from $J$, $f_1$ can only produce using goods from $J$ as raw goods. This is because, as we discussed earlier, the prices of goods in $J$ are sufficiently smaller that those for goods in $G\setminus J$, and hence any profitable production that uses goods in $J$ as raw goods should also involve only goods in $J$ as produced goods. Therefore, we establish that there exists a $j_2\in J$ such that $\bar{s}^{f_1}_{j_2}\geq\frac{L}{m}\mbox{cpu}^{f_1}(\mathbf{\bar{p}})\geq \frac{L}{m\alpha}\bar{p}_{j_2}$. Using this, we may again bound $\bar{e}_{j_2}$ and repeat the arguments above.
        
    \paragraph{Finding $j_k \in J$ with $0>\bar{e}_{j_k} \geq 0$.}
        As we mentioned earlier, by the \emph{no production out of nothing and no vacuous production} sufficiency condition, the above argument has to be repeated at most $m$ times, as that is the maximum possible length of any profitable production chain. In both of the two cases consider above, $L$ becomes smaller by a certain factor as a function of the input parameters. This means that $L$ can be picked large enough to guarantee that at the end of the chain, when the last good is not being produced, we obtain that $0 > \bar{e}_{j_k} \geq 0$, a contradiction.\medskip

        \noindent This concludes the proof of the claim. 
\end{proof}

\noindent \Cref{claim:leontief-prod-claim2} implies that for any firm $f \in F_J$, with $L^f = \infty$, there cannot be an edge to any consumer $i \in N \setminus N_J$ or any firm $f' \in F \setminus F_J$ in the economy graph. 
Indeed, for an edge $(f,i)$, the claim asserts that the any produced good $j'$ of $f$ will be in the set $J$, whereas consumer $i$ only values goods in the set $G \setminus J$ positively, by virtue of being in $N \setminus N_J$. Similarly, for an edge $(f,f')$, firm $f$ produces only goods in $J$, whereas the for firm $f'$ we have $\mathcal{R}_{f'} \subseteq G \setminus J$, by virtue of $f'$ being in $F\setminus F_J$. \medskip

\noindent From the two paragraphs succeeding \Cref{claim:leontief-prod-claim1} and \Cref{claim:leontief-prod-claim2}, $N_J$ and $N\setminus N_J$ are two strongly connected components in the economy graph of the market, contradicting the sufficiency condition requiring that in the economy graph there is a strongly connected component containing all the consumer-nodes. This completes the proof.
\end{proof}

\subsection{Arrow-Debreu Markets with Succinct SPLC Utilities and Production}\label{sec:SSPLC-markets}

In this section we provide our PPAD-membership result for \emph{succinct} separable piecewise-linear utilities (SSPLC). These are SPLC utilities, i.e., each consumer has a piecewise-lineaer utility function over different pieces (or segments) of each good, and her total utility for her bundle is additive over those utility functions. The ``succinct'' part comes from the fact that, in contrast with the SPLC utilities that have been studied before in the literature (e.g., see \citep{garg2015complementary,SODA:GargV14,vazirani2011market}), these functions can be accessed implicitly, given access to a boolean circuit that evaluates the function at different points. Effectively, these allows us to model functions with exponentially many pieces, while still developing a reduction which is polynomial in the size of the input, i.e., the size of those circuits. Our main theorem of the section is the following.

\begin{theorem}\label{thm:SSPLC-SPLC-PPAD}
    Computing a competitive equilibrium in an Arrow-Debreu market with SSPLC utilities and SPLC production functions is in PPAD.
\end{theorem}

\noindent Note that we are also incorporating (explicit) SPLC productions in our result. Whether we can generalize those to SSPLC production functions is an interesting technical open problem; we discuss the challenges and why our technique falls short in this regard in \Cref{rem:ssplc-production-challenge}.

\paragraph{Markets with Production.} In an Arrow-Debreu market with production $\mathcal{M}$, we have a set $N$ of consumers, a set $G$ of infinitely divisible goods, and a set $F$ of firms. Let $n=|N|$, $m=|G|$, and $\ell = |F|$. We will typically use index $i$ to refer to consumers, $j$ or $g$ to refer to goods and $f$ to refer to firms. Each consumer brings an endowment $w_i = (w_{i1}, \ldots, w_{im})$ to the market, with $w_{ij} \geq 0$ for all $i \in N$ and $j \in G$.
We may assume without loss of generality that for every good $j$, we have $\sum_{i \in N} w_{ij} =1$, i.e., that the total endowment of each good is $1$. We will use $\mathbf{x}_i = (x_{i1}, \ldots, x_{im})$ to denote the vector of quantities of goods allocated to consumer $i \in N$ in $\mathcal{M}$, and we will call it the \emph{bundle} of consumer $i$. Let $\mathbf{x} = (x_1, \ldots, x_n)$ be the vector of such bundles. We will use $\mathbf{p} = (p_1, \ldots, p_m)$ to denote the vector of \emph{prices} in $\mathcal{M}$, one for each good $j \in G$. Prices are non-negative, so $p_j \geq 0$ for all $j \in G$. Given a vector of prices $\mathbf{p}$, the \emph{budget} of consumer $i \in N$ is defined as $\sum_{j \in G}w_{ij}p_j$; intuitively, this is the amount of money that the consumer acquires by selling her endowment at prices $\mathbf{p}$. 

\paragraph{Utility functions.}

Every consumer has a utility function $u_i\colon \RRnn^m\rightarrow\RRnn$ mapping a bundle $\mathbf{x}_i$ to a non-negative real number. In this section, these utilities are \emph{succinct separable piecewise-linear concave (SSLPC)}.  A utility $u_i$ being SPLC means that it is separable over the goods and that the utility accrued from each good is given by a piecewise-linear concave function $u_{ij}$, that is, $u_i(\mathbf{x}_i)=\sum_{j=1}^m u_{ij}(x_{ij})$. The word \emph{succint} refers to the piecewise-linear utilities concave $u_{ij}$ being given by boolean circuits computing their slopes, thus allowing for an exponential number of changes to the value of the slope. Concretely, we assume that the utility functions are given by the following data:

\begin{enumerate}
    \item A natural number $A\in\NN$ such that the piecewise linear functions $u_{ij}$ have constant slope in the interval $(A,\infty)$.
    \item A natural number $B\in\NN$ such that the slopes of the utility functions can change only in the points $\frac{k}{B}\cdot A$ for $k\in \{0,1,\dots,B\}$. 
    \item For every $i\in N$ and $j\in G$, a boolean circuit $C_{ij}$ with $\lfloor\log_2(B)\rfloor+1$ input bits and output bits which on input $k\in \{0,1,\dots,B\}$ computes $1/u_{ijk}$, where $u_{ijk}$ is the slope of $u_{ij}$ in the interval $((k/B)A,((k+1)/B)A)$. Note that this makes sense only if the slope is strictly positive. Also, the utility functions are concave, meaning that $u_{ij0}\geq u_{ij1}\geq\cdots\geq u_{ijB}$. 
    \item For every $i\in N$ and $j\in J$, if the slope of $u_{ij}$ eventually becomes zero, then we are given a natural number $D_{ij}$ such that the slope becomes zero in the point $(D_{ij}/B)A$. 
\end{enumerate}

\begin{remark}
    The reason that we assume that the boolean circuits $C_{ij}$ compute the inverse of the slopes is that this allows us to compute bang-per-buck ratios which is required in our construction.
\end{remark}

\paragraph{Firm Shares.} Each consumer $i \in N$ has a share $\theta_{if} \in [0,1]$ of the profit of each firm $f \in F$. We assume that the profits are entirely shared among the consumers, i.e., for every firm $f \in F$, we have that $\sum_{i \in N}\theta_{if}=1$. 

\paragraph{Production functions.}

In this section we assume that the firms have \emph{separable piecewise-linear concave} production functions given explicitly. Every firm $f$ has a single \emph{output good} $\fout$ which it can produce, and for every $j\in G,$ it has a piecewise-linear concave function $P^f_j\colon\RRnn\rightarrow\RRnn$ defining the ability of $f$ to produce $\fout$ as a function of the good $j$. The overall production of firm $f$ from bundle $\mathbf{y}=(y_1,\dots,y_m)$ is then given by $P^f(\mathbf{y})=\sum_{j\in G}P^f_j(y_j)$. Because of separability, we may assume that every firm has an \emph{input good} $\fin$, which is the only good with which it can produce a strictly positive amount of $\fout$. As such the production function $P^f$ is simply a piecewise-linear concave function. We assume that this function is given to us explicitly as follows:
\begin{enumerate}
    \item The nonnegative real line $\RRnn$ is divided into a finite number $K_f$ of subsequent segments of length $L_{fk}$ for $1\leq k\leq K_f.$ The length of the final segment is infinite.
    \item To every segment is associated a \emph{conversion rate} $c_{fk}\geq 0$ which is the slope of $P_f$ in this segment, that is, $c_{fk}$ is the rate at which $f$ can convert good $\fin$ to good $\fout$ on this segment. The production function $P^f$ being concave means that $c_{f1}\geq c_{f2}\geq\cdots\geq c_{fK_f}$. 
\end{enumerate}
One may wonder why we do not assume that the slopes of the production functions are given succinctly by a boolean circuit as we did for the utility functions. We discuss the problems arising from this representation later in the section.

\paragraph{Competitive Equilibrium.} We are now ready to define the notion of a competitive equilibrium in markets with production.

\begin{definition}[Competitive Equilibrium - Markets with SSPLC Utilites and SPLC Production]\label{def:ssplc-markets-competitive-equilibrium}
A competitive equilibrium of an Arrow-Debreu market with SSPLC utilities and PLC production functions is a triple $(\mathbf{p}^*,\mathbf{x^*}, \mathbf{y^*})$ consisting of non-negative prices $\mathbf{p}^*$, non-negative bundles $\mathbf{x}^* = (\mathbf{x}_1^*,\ldots, \mathbf{x}_n^*)$ and non-negative amounts of input goods $\mathbf{y}^* = (y_1^*, \ldots, y_\ell^*)$, such that
\begin{enumerate}
    \item $p^*_{\fout} \cdot P_f(y_f^*) - p^*_{\fin}\cdot y_f^*$ is maximized, for any firm $f \in F$. \hfill \textbf{\emph{(firm profit maximization)}} \label{con1-ssplc:firm-profit}
   \item $u_i(\mathbf{x}_i)$ is maximized for every consumer $i \in N$, subject to \\
   $\sum_{j \in G}p_j^*\cdot x_{ij}\leq \sum_{j \in G}p_j^*\cdot w_{ij}+\sum_{f \in F} \theta_{if}\cdot (p^*_{\fout} \cdot P_f(y_f^*) - p^*_{\fin}\cdot y_f^*)$.\hfill	\textbf{\emph{(bundle optimality)}} \label{con2-ssplc:optimality}
   \item $z_j^*\leq 0$, and $z_j^* p_j^* =0$, where $z_j^* = \sum_{i \in N} x_{ij}^*+\sum_{f \in F \colon \fin = j}y_{f}^*-\sum_{f \in F \colon \fout = j} P_f(y_{f}^*)-1$, \\for every good $j \in G$. \hfill	\textbf{\emph{(market clearing)}} \label{con4-ssplc:market-clearing}
\end{enumerate}
\end{definition}

\noindent Condition~\ref{con1-ssplc:firm-profit} requires that at the chosen set of prices $\mathbf{p}^*$, each firm maximizes its profit, given its production functions. Condition~\ref{con2-ssplc:optimality} requires that at the chosen set of prices $\mathbf{p}^*$, each consumer maximizes her utility subject to their budget constraints, where the budget consists of the amount earned from selling all the consumer's endowments $\sum_{j \in G}p_j^*w_{ij}$ and the profit share $\sum_{f \in F} \theta_{if}\cdot (p^*_{\fout} \cdot P_f(y_f^*) - p^*_{\fin}\cdot y_f^*)$ of the consumer from the production of the firms. Finally, Condition~\ref{con4-ssplc:market-clearing} is the market clearing condition, which requires that the total consumption of each good is at most the total production plus the total endowment of the consumers, and supply equals demand for all goods which are not priced at $0$. As we detail later on in the section, we may in fact assume without loss of generality that in any competitive equilibrium \emph{all the prices are positive}, and hence Condition~\ref{con4-ssplc:market-clearing} reduces to $z_j^* =0$ for all $j \in G$. Note that in Condition~\ref{con4-ssplc:market-clearing} we have used that $\sum_{\in N}w_{ij}=1$ for each good $j \in G$.

\paragraph{Sufficiency Conditions.} A competitive equilibrium as defined above exists for every market $\mathcal{M}$, under some sufficiency conditions. Similarly to \Cref{sec:prod-markets-linear}, we will use the set of conditions used by \citet{maxfield1997general}, also used in the series of papers on market equilibria that we mentioned in the beginning of \Cref{sec:markets}. 

\begin{enumerate}
    \item Consider a graph $\mathcal{G}_F(\mathcal{M})$ in which the nodes are the goods, and an edge $(j,j')$ has weight 
    \[
    w_{jj'}=\max_{f \in F \colon j = \fin \wedge j'=\fout}c_{f1}
    \]
    If the set is empty, the weight is defined to be $0$.

    The above weight captures the fact that $j'$ can be produced from $j$ at conversion rate $w_{jj'}$ by some firm $f \in F$. Then, for any cycle $C=(g_0, g_1), (g_1,g_2)\ldots, (g_{k-2},g_{k-1})$ of $G_F(\mathcal{M})$ the product of the weights of the edges is less than $1$, i.e., $\prod_{e \in C} \alpha_e < 1$.\label{cond3-ssplc-markets-no-prod-out-of-nothing}

    This condition is known as the \textbf{\emph{no production out of nothing and no vacuous production}} condition. Indeed, if $\prod_{e \in C} \alpha_e > 1$, then it would be possible to increase the quantity of some good, without decreasing the quantity of any other good. The case of $\prod_{e \in C} \alpha_e = 1$ refers to the case of vacuous production, which is also disallowed in our model, similarly to the related works, e.g., see \citep{SODA:GargV14}.

    \item We will say that a consumer $i \in N$ is \emph{\textbf{nonsatiated} for good $j \in G$} if $C_{ij}(B)>0$. We will assume that for any good $j\in G$, there exists some consumer $i\in N$ that is nonsatiated with respect to $j$. Similarly, we say that a firm $f$ \emph{\textbf{nonsatiated}} for good $j \in G$ if $j = \fin$ and $c_{fK_f}>0$.

    \item Consider the \emph{economy graph} $\mathcal{G}_E(\mathcal{M})$ of the market $\mathcal{M}$ in which the nodes are the consumers and the firms, and there is an edge from node $a$ to node $b$ in node $a$ owns/produces a good for which node $b$ is nonsatiated. Then, $\mathcal{G}_E(\mathcal{M})$ contains a strongly connected component containing all the consumer-nodes. This condition generalizes the corresponding strong connectivity condition that we used in \Cref{sec:prod-markets-linear}.
\end{enumerate}

\begin{remark}[Bounds on Production]\label{rem:ssplc-bounds-on-production}
Condition~\ref{cond3-ssplc-markets-no-prod-out-of-nothing} above imposes a bound on the total amount of a production of any firm $f \in F$ in a competitive equilibrium. Note that the production starts from finite endowments $\sum_{i \in N}w_{ij}=1$ for all $j \in G$. Since no firm operates at a loss at an equilibrium, any cycle of production would violate Condition~\Cref{cond3-ssplc-markets-no-prod-out-of-nothing}. Since there are no such cycles, production can take place along chains. The longest such chain is obviously bounded by $m$, and the maximum production of any firm can be bounded by some sufficiently large global constant $L^c$ (e.g., some constant such that $L^c \geq m^m (\max_{f}C_f(0)+1)^n$). See also \citep{SODA:GargV14} for a very similar argument. Looking ahead, this will allow us to impose ``loose'' upper bounds on the production and consumption in the linear programs that we will devise, without compromising the existence of a competitive equilibrium.
\end{remark}

\paragraph{Optimality and bang-per-buck (BPB).}

Similarly to \Cref{sec:exchange-markets-linear}, the optimal bundles of Condition~\ref{con2-prod:optimality} in \Cref{def:prod-markets-competitive-equilibrium} are characterized by their \emph{bang-per-buck (BPB)}. Given consumer $i \in N$ and prices $\mathbf{p}$, the bang-per-buck of good good $j$ on segment $k$ is defined as $\text{BPB}_i(j,k)=\frac{u_{ijk}}{p_j}$. \medskip

\noindent \emph{Bounds on the prices.} Recall that we have assumed that for every good $j$, there is some agent $i$ that is nonsatiated with respect to $j$. Hence, using very similar arguments to \Cref{sec:prod-markets-linear} and \Cref{sec:general-markets}, we may assume that all the prices are strictly positive.

\begin{remark}[Normalized Prices]\label{rem:ssplc-markets-normalized-prices}
Given that $p_j > 0$ for all $j \in G$, we can normalize the prices to sum to $1$ without loss of generality, i.e., we may assume that for every good $j \in G$, we have that $\sum_{j \in G} p_j = 1$. 
\end{remark}

\noindent Again, in a similar manner to \Cref{sec:exchange-markets-linear}, we will use a parameter $\varepsilon$ to capture the fact that if the price $p_j$ for a good $j \in G$ is sufficiently smaller than the price $p_{j'}$ for a good $j' \in G$, then $\text{BPB}(j) > \text{BPB}(j')$. Specifically, we can compute $\varepsilon > 0$ such that
\[
\text{If } p_j \leq \varepsilon \cdot p_{j'} \text{ and } u_{ijk} > 0 \text{ then } \text{BPB}_i(j,k) > \text{BPB}_i(j',k')\text{ for all }k'.
\]
Because the inverses of the utilities are in the set $\{1,2,\dots,B\}$ by assumption, the actual utilities are in the set $\{1/B,1/(B-1),\dots,1/2,1\}$. For the above to hold, it is thus sufficient that if $p_j\leq\eps p_{j'}$, then $(1/B)/p_j\geq 1/p_{j'}$. For this implication to hold, it suffices to choose $\eps\leq 1/B.$
Additionally, we can pick $\varepsilon$ to be sufficiently small such that $\varepsilon < \frac{w_{ij}}{m}$ for all $i \in N$ and all $j \in G$ such that $w_{ij}>0$. Given $\varepsilon$, we will impose a stricter lower bound on the prices, which will be useful later on: in particular, we will assume that for all $j \in G$, $p_j \geq \frac{\varepsilon^m}{m}$.

\paragraph{Expenditures.} As in all of the previous sections, we do not work directly with the quantities $\mathbf{x}$, $\mathbf{y}$. Instead we work with the consumer and firm \emph{expenditures} respectively, which are defined similarly to those of \Cref{sec:prod-markets-linear}. Specifically, for every for agent $i$ we will have a variable $q_{ij}$ for every $j\in G$ which denotes the spending of agent $i$ on good $j$, and for every $f$ and every $1\leq k\leq K_f$, we will have a variable $s_{fk}$ which denotes the spending of firm $f$ on its input good $\fin$ on the $k$th segment.

\subsubsection{Membership in PPAD: The proof of \crtCref{thm:SSPLC-SPLC-PPAD}}

\noindent We will again develop the proof in three steps, namely (a) construction of the function $F$ and arguing that it can be represented by a \linear arithmetic circuit containing \linoptgates, (b) showing that the \linoptgate can compute all the necessary components, and (c) arguing that a competitive equilibrium can be recovered from a fixed point of $F$. 

\paragraph{The function $F$.}

We define a function $F\colon D\rightarrow D$ where
\begin{align}
    D = \{p\in\Delta^{m-1}\colon \forall j\in G,p_j\geq \frac{\eps^m}{m}\}\times [0,CL+1]^{nm}\times_{f} [0,L]^{K_f}
\end{align}
An input to $F$ is a tuple $(\mathbf{p},\mathbf{q},\mathbf{s})$ of prices, consumer expenditures and firm expenditures for raw goods, and the output is another such tuple $(\mathbf{\bar{p}}, \mathbf{\bar{q}}, \bar{\mathbf{s}})$.  \medskip

\noindent For any firm $f$, the output spending variable $\bar{\mathbf{s}}_f=(s_{fk})_{1\leq k\leq K_f}$ is set to be the output of the \linoptgate for the following linear program.

\begin{center}
\underline{\text{Linear Program } $\mathcal{P}_\text{prod}$}
\end{center}
\begin{align}\label{LP:production-SSPLC}
	\mbox{maximize} \ \ \ \ & \sum_{k=1}^{K_f}(c_{fk}p_{\fin}-p_{\fout})\cdot v_{fk}  \\
	\mbox{subject to} \ \ \ \ & 0 \leq v_{fk} \leq \min\{L_{fk},L\}\cdot p_{\fin} \nonumber
\end{align}
As in the case considered in \Cref{sec:prod-markets-linear}, we define 
\[r_{fk} = s_{fk}+\max\{0,c_{fk}\cdot p_{\fout}-p_{\fin}\}\cdot\min\{L_{fk},L\}.\]
Furthermore, we let $s_{f} = \sum_{k=1}^{K_f}s_{fk}$ and $r_f = \sum_{k=1}^{K_f}r_{fk}$. Then, for every agent $i$, we let the output expenditure $\mathbf{q}_i$ be the output of the \linoptgate of the following linear program. 
\begin{center}
\underline{\text{Linear Program } $\mathcal{P}_\text{con}$}
\end{center}
\begin{align}\label{LP:consumption-SSPLC}
	\mbox{minimize} & \sum_{j \in G} \left(p_j\cdot C_{ij}\left(\frac{q_{ij}}{p_j}\right)+(B+1)\cdot H\left(q_{ij}-\frac{D_{ij}}{B}Ap_j\right)\right)\cdot v_{ij} \nonumber \\
	\mbox{subject to}& \sum_{j \in G} v_{ij} = \sum_{j \in G}w_{ij}p_j + \sum_{f \in F} \theta_{if}\cdot (r_f-s_f) \\
        & 0 \leq v_{ij} \leq (C \cdot L + 1), \;\forall j \in G \nonumber
\end{align}
The first term $p_j\cdot C_{ij}(q_{ij}/{p_j})$ of the coefficient in front of $v_{ij}$ is the \pseudogs from \Cref{prop:implicit-correspondence} computing the product of $p_j$ and the piecewise constant correspondence represented by $(C_{ij},B)$. This term should represent the inverse bang-per-buck ratio of agent $i$ for good $j$ when her spending on the good is $q_{ij}$. The second term in the coefficient in front of $v_{ij}$ in the objective function only appears if the slope of the PLC function that agent $i$ has for good $j$ eventually becomes zero. This term is meant to ensure that in an optimal solution, an agent does not spend money on a segment where the slope is zero.\medskip

\paragraph{Excess Expenditure and the Feasibility Program for the prices.} Finally, for every $j\in G$, similarly to \Cref{sec:prod-markets-linear} and \Cref{sec:general-markets}, we define the \emph{excess expenditure} of good $j$ as 

\[e_j = \sum_{i}q_{ij}+\sum_{f\colon \fin = j}s_f-\sum_{f\colon\fout = j}r_f-p_j.\] As in previous sections, we let the output price vector $\bar{p}$ be the output of the following feasibility program.

\begin{center}
\underline{\text{Feasibility Program } $\mathcal{Q}$}
\end{center}
\begin{align}\label{LP:prices-SSPLC}
	& e_j<e_{j'}\Rightarrow p_j\leq \varepsilon \cdot p_{j'},\;\forall j,j' \in G \nonumber \\
	& p_j\geq \frac{\varepsilon^m}{m},\;\forall j \in G\\
	& \sum_{j \in G}p_j = 1\nonumber
\end{align}

\paragraph{Computation by the \linoptgate.} Next, we argue that the solutions to the linear programs and the feasibility program can be computed by our \linoptgate. The subgradients of the LPs for production and consumption can be computed by \linear circuits. The main part of the construction is the \pseudog from \Cref{prop:implicit-correspondence}, computing the correspondence $p_j\cdot C_{ij}(q_{ij}/{p_j})$. For the feasibility program, the arguments that it is of the correct form and that it is solvable are identical to those of \Cref{lem:productionMarkets-OPT-gate-can-compute} (and hence also of \Cref{lem:ExchangeMarkets-OPT-gate-can-compute}), noting the updated definition of the expenditure $e_j$.

\paragraph{Fixed points.}

Suppose that $(\mathbf{p},\mathbf{q},\mathbf{s})$ is a fixed point of $F$. Let $f\in F$ be some firm. As $\mathbf{s}_f$ is a solution to LP (\ref{LP:production-SSPLC}), we have that $s_{fk}>0$ only if $c_{fk} \cdot p_{\fout}\geq p_{\fin}$. Furthermore, if the inequality is strict, then the firm produces using all of that segment. We conclude that $\mathbf{s}$ is an optimal production vector at prices $\mathbf{p}$. Next, we prove the following:

\begin{claim}\label{claim:SSPLC-consumption-claim}
Let $i\in N$ and let $x_{ij}=q_{ij}/p_j$ for every $j\in G$. Then $\mathbf{x}_i$ is an optimal bundle for agent $i$ at prices $\mathbf{p}$. 
\end{claim}
\begin{proof}
    First we show that if the slope of $u_{ij}$ becomes zero at some point, then $q_{ij}\leq\frac{D_{ij}}{B}\cdot A \cdot p_j$. If this were not the case, then the coefficient in front of $v_{ij}$ would be at least  $B+1$. However, by assumption there is a good $j'$ for which agent $i$ is nonsatiated, and the coefficient in front of $v_{ij'}$ is at most 
    \[p_{j'}\cdot C_{ij'}\left(\frac{q_{ij'}}{p_{j'}}\right)\leq p_{j'}\cdot B\leq B<B+1.\] 
    Because $\mathbf{q}_i$ is an optimal solution to LP (\ref{LP:production-SSPLC}), this would then imply that $q_{ij}=0\leq\frac{D_{ij}}{B}\cdot A \cdot p_j$. This means that any agent $i$ can only spend money on a good $j$ up until the point where the inverse of the bang-per-buck ratio is no longer well-defined.
    Now the result follows, because all the goods $j\in G$ such that $q_{ij}>0$ must have the minimal coefficient. From this it follows that agent $i$ has bought the goods in a greedy manner according to the bang-per-buck ratio, meaning that $\mathbf{q}_i$ represents an optimal spending for agent $i$. 
\end{proof}

\noindent What remains to show is that the market clears, that is, $e_j=0$ for all $j\in G$. As in previous sections, we can use a standard calculation to show that $\sum_{j\in G}e_j=0$. Hence, it suffices to argue that $e_j\geq 0$ for all $j\in G$. We define:
\begin{enumerate}
    \item[-] $J = \{j\in G\colon e_j\leq e_{j'}\text{ for all }j'\in G\}$
    \item[-] $N_J = \{i\in N\colon \text{there exists } j\in J\text{ such that } i\text{ is nonsatiated for good }j\}$
    \item[-] $F_J = \{f\in F\colon \text{there exists } j\in J\text{ such that } f\text{ is nonsatiated for good }j\}$
\end{enumerate}
Assume by contradiction that $e_j<0$ for $j\in J$. As $\sum_{j\in G}e_j=0$, this implies that there must exists a good $j'$ such that $e_{j'}>0$, meaning that $J\subsetneq G$. The contradiction will be obtain via the following two claims, very similar to \Cref{claim:prod-claim1} and \Cref{claim:prod-claim2}. Using an identical reasoning to that of \Cref{lem:production-markets-correctness-lemma}, the two claims together will establish that the economy graph $\mathcal{G}_E(\mathcal{M})$ is not strongly connected, violating the corresponding sufficiency condition.

\begin{claim}\label{claim:SSPLC-prod-claim2}
If $i\in N_J$ and $j'\in G\setminus J,$ it holds that $w_{ij'}=0.$
\end{claim}
\begin{proof}
    Assume for contradiction that $i\in N_J$ and $j' \in G\setminus J$ and that $w_{ij'}>0$. As in previous sections, this implies that we can lower bound the budget of agent $i$ by $(CL+1)\sum_{j\in J}p_j$ by picking $\eps\leq \frac{w_{min}}{m\cdot (CL+1)}$. By assumption, $i$ is nonsatiated with respect to some $j\in G$. As $p_j\leq\eps p_{j''}$ for all $j''\in G\setminus J$, we get by choice of $\eps$ that $i$ will spend her entire budget on goods in $J.$ Hence, there is some $j\in J$ such that $q_{ij}\geq (CL+1)p_j.$ This implies that
    \begin{align}
        e_j\geq (CL+1)p_j-\sum_{f\colon\fout = j}r_f-p_j
    \end{align}
    Now for any $f$ with $\fin  = j$, we have that $r_f =\sum_{k=1}^{K_f}r_{fk}\leq\sum_{k=1}^{K_f}Lc_{fk}p_j\leq (K_f\cdot\max_{f\in F}c_{f1})Lp_j$. By picking $C\geq |F|\cdot \max_f K_f\cdot\max_f c_{f1}$, we obtain the contradiction that $e_j\geq 0.$ We conclude that if $i\in N_J$ and $j'\in G\setminus J$, then $w_{ij'}=0.$
\end{proof}

\begin{claim}\label{claim:SSPLC-prod-claim1}
If $f \in F_J$, then $\fout\in J$.
\end{claim}
\begin{proof}
    Let $f\in F_J$ and let $j=\fin$. Assume for contradiction that $\fout\notin J.$ Then $p_j\leq\eps p_{\fout}$, meaning that $f$ will produce fully if $\eps>0$ is chosen sufficiently small. Hence $s_f\geq s_{fK_f}=Lp_j$, and so
    \begin{align}
        0>e_j\geq (L-1)\cdot p_j-\sum_{f'\colon g^{\text{out}}_{f'}=j}r_{f'}
    \end{align}
    If no firm produces $j$, then we reach a contradiction by choosing $L\geq 1.$ 
    The above inequality implies that there must exist a firm $f_1$ such that $g^{\text{out}}_{f_1}=j$ and $r_{f_1}\geq \frac{L-1}{\ell}p_j$. Let $j_1 = g^{\text{in}}_{f_1}$ and note that $j_1\in J$. As in previous sections, we may thus bound $\frac{L-1}{\ell}p_j\leq r_{f_1}\leq c_{f_1 1}s_{f_1}$. If any production occurs, then $c_{f_1 1}p_j\geq p_{j_1}.$ These two inequalities combine to show that $s_{f_1}\geq \frac{L-1}{\ell\cdot c_{f_1 1}^2}p_{j_1}$. Using this, we may again bound
    \begin{align}
        0>e_{j_1}\geq \Big(\frac{L-1}{\ell\cdot c_{f_1 1}^2}-1\Big)\cdot p_{j_1}-\sum_{f'\colon g^{\text{out}}_{f'}=j_1}r_{f'}
    \end{align}
    Again, if $j_1$ is not produced, then we may pick $L$ large enough to reach a contradiction. Otherwise we repeat the previous argument. This can happen at most $m$ times, because of the \emph{no assumption out of nothing} assumption. 
\end{proof}

\subsubsection*{Beyond SSPLC Utilities?}

We conclude the section with two remarks. The first one briefly discusses the challenges of using the machinery of \Cref{sec:implicit} for SSPLC productions also.

\begin{remark}\label{rem:ssplc-production-challenge}
    The reason for not allowing a succinct representation of the production functions is that this seemingly necessitates multiplying a piecewise-linear function by an input variable of the circuit. A natural idea would be to have the cumulative sum of conversion rates given succinctly by a boolean circuit $C'$. However, in order to compute the revenue of firm $f$ given spending $s_f$, one would then have to compute the piecewise-linear function $p_{\fout}\cdot C'(s_f/p_{\fin})$ which is something that is not captured by the results in \Cref{sec:implicit}.
\end{remark}

\noindent The second remark regards the possibility of defining a class of ``succinct Leontief-free functions'', similarly to the case of SSPLC. 

\begin{remark}
    One might hope to define a class which captures the case of Leontief-free utility and production functions and the case of SSPLC utility functions and SPLC production functions simultaneously. For the utility functions, one approach would be to consider utility functions $u^i$ given by a finite list $K_i$ of tranches $\{s^i_k\}_{k\in K_i}$. Associated with a tranche $s^i_k$ would again be a number $L^i_k\in\RRnn\cup\{\infty\}$ which is an upper bound for the total amount of utility that can be accrued from this tranche. Also, for every good $j$, instead of having a single number $u^i_{jk}$ which is the rate at which good $j$ provides utility for agent $i$ on tranche $k$, we would have that $j$ would provide utility for agent $i$ on segment $k$ according to a piecewise-linear concave function $U^i_{jk}$ given to us succinctly in the form of a boolean circuit. Such a class of utility functions would generalize both Leontief-free utility functions and SSPLC utility functions. However, because the rates at which goods provide utility is non-constant, it is unclear how to obtain constraints that mimic the constraint $\sum_{j\in G}U^i_{jk}(x^i_{jk})\leq L^i_k$ after transforming the variables via Gale's substitution (\Cref{rem:gale}).
\end{remark}

\section{Pacing Equilibria in Auto-Bidding Auctions}\label{sec:pacing}

In this section, we consider environments which we will broadly refer to as ``auto-bidding auctions''. In these, a set of items are sold separately to a set of buyers via parallel single-item auctions (e.g., first-price auctions or second-price auctions). To participate in these auctions, each buyer $i$ \emph{scales} her valuations for the items by the same multiplier $\alpha_i$, which is known as a \emph{pacing multiplier}, under some constraints on her expenditure, typically provided by budgets, or return-on-investment (ROI) thresholds. The objective is to find a \emph{pacing equilibrium} (see \Cref{def:PE,def:RPE}), i.e., a vector of pacing multipliers and allocations which ensure that the auction is run correctly and that it satisfies the spending constraints of all the buyers simultaneously. 

In \Cref{sec:pacing-sp-budgets} we consider the former case, i.e., that of single-item auctions with budget constraints. For the case of first-price auctions, \cite{borgs2007dynamics} showed that a pacing equilibrium always exists and can be computed in polynomial time via a convex program. Later on \citet{conitzer2022pacing} showed that in this case the pacing multipliers are unique in every pacing equilibrium, and presented an interesting connection between pacing equilibria and solutions of the well-known Eisenberg-Gale convex program for quasi-linear utilities \citep{cole2017convex}.
For second-price auctions, the existence of a pacing equilibrium was established in \citep{conitzer2022multiplicative}. \citet{chen2021complexity} later on proved that the problem of computing a pacing equilibrium in these auctions is in PPAD, which also implied that a solution described by rational numbers also exists. We provide an alternative PPAD-membership result for the same problem. Compared to the proof of \citet{chen2021complexity}, our proof makes uses of our \linoptgate and is \emph{significantly} simpler. We highlight the main ideas of our proof and how they differ from those of \citet{conitzer2022multiplicative} and \citet{chen2021complexity} in \Cref{sec:pacing-sp-original-proof} below. We remark that \citet{chen2021complexity} also showed the problem to be PPAD-hard, thus establishing its PPAD-completeness.

In \Cref{sec:pacing-sp-roi}, we also consider the case of (on average) ROI-constrained buyers, in the model introduced by \citet{li2022auto}. The existence of pacing equilibria for this case was proven by \citet{li2022auto}, following a very similar approach to that of \citet{conitzer2022multiplicative}. Our proof follows along very much the same lines as the one we present for the case of budget-constrained buyers in \Cref{sec:pacing-sp-budgets}, and is again significantly simpler. Interestingly, it turns out that our \linoptgate cannot be used to obtain PPAD-membership of the problem in this case. Still, using the \fixpoptgate developed by \cite{SICOMP:Filos-RatsikasH2023}, we prove membership of the problem in the class FIXP. This is the first membership result for this variant of the problem; we provide more details in \Cref{sec:RPE-FIXP-membership}. In \Cref{sec:RPE-irrational} we show that membership in FIXP (rather than PPAD) is necessary, because pacing equilibria in this case can be irrational. Whether the problem of computing pacing equilibria in second-price auctions for ROI-constrained buyers is actually FIXP-complete is an intriguing open problem. We note that for a notion of approximate pacing equilibria, a PPAD-hardness result is known from \citep{li2022auto}.

\subsection{Pacing Equilibria in Auctions with Budgets}\label{sec:pacing-sp-budgets}

Following \citet{chen2021complexity}, we will refer to the auction market environment as a \emph{second-price pacing game (SPPG)} $\mathcal{G}$, consisting of
\begin{itemize}
    \item[-] a set $N$ of buyers, with $n=|N|$,
    \item[-] a set $G$ of items, with $m = |G|$,
    \item[-] a set of vectors of \emph{valuations} $\mathbf{v}_{i} = (v_{i1}, \ldots v_{im})$ for every buyer $i \in N$, where $v_{ij}$ denotes the value of buyer $i \in N$ for item $j\in G$; without loss of generality we assume that for all $j \in G$, there exists some $i \in N$ such that $v_{ij}>0$,
    \item[-] a vector of \emph{budgets} $\mathbf{B}=(B_1, \ldots, B_n)$, such that $B_i >0$ for every buyer $i \in N$.
\end{itemize}
For each buyer $i \in N$, there is a \emph{pacing multiplier} $\alpha_i \in [0,1]$, and let $\alpha=(\alpha_1,\ldots,\alpha_n)$ be the vector of those multipliers. Given $\alpha$, the \emph{bid} of buyer $i$ on item $j \in G$ is given by $\alpha_i \cdot v_{ij}$, i.e., every buyer \emph{scales} their value $v_{ij}$ by the multiplier $\alpha_i$. The items are sold in parallel second-price auctions, which implies that
\begin{enumerate}
    \item the item is allocated to the buyer with the highest bid, breaking ties using some tie-breaking rule in case there are multiple such buyers. Let $h_j(\alpha) = \max_{i \in N} \alpha_i v_{ij}$ be the highest bid on item $j \in G$.\label{item:SP-1}
    \item the price of item $j \in G$ is the second-highest bid. Let $p_j(\alpha)=\min_{i \in N}\max_{i' \in N, i\neq i'}\alpha_{i'}v_{i'j}$ be that price. \label{item:SP-2}
\end{enumerate}
Another way to interpret Condition~\ref{item:SP-1} above is that each one of the (potential) winners achieves a fraction of the item, which can be interpreted as the probability that that buyer is the final winner of the auction for that item. Formally, let $W_j(\alpha)=\{i \in N: \alpha_i v_{ij} = h_j(\alpha)\}$ be the set of possible winners (those with the highest bid) for item $j \in G$. Each buyer $i \in N$ receives an allocation $x_{ij}$ of item $j \in N$, with the restriction that $x_{ij}=0$ for all $i \in N\setminus W_j(\alpha)$. 

\paragraph{Pacing Equilibria.} We are now ready to define the notion of a pacing equilibrium. Informally, a pacing equilibrium is a set of pacing multipliers and allocations that satisfy some constraints that guarantee the validity of the auction and that the buyers are in a sense meeting their maximum buying capabilities. Following \citet{conitzer2022pacing,conitzer2022multiplicative}, we formally refer to those equilibria as \emph{second-price pacing equilibria.} 

\begin{definition}[Second-price Pacing Equilibrium (SPPE)]\label{def:PE}
A pair $(\alpha, \mathbf{x})$, with $\alpha \in [0,1]^n$ and $\mathbf{x} \in [0,1]^{nm}$, is a \emph{second-price pacing equilibrium (SPPE)} of a SPPG $\mathcal{G}$, if the following conditions hold:
\begin{conditions}[label={I.\alph*.}]
    \item $\sum_{i \in N}x_{ij} \leq 1$, for all $j \in G$, \hfill \textbf{\emph{(feasible allocation)}} \label{cond-PE1}
    \item $x_{ij} > 0 \Rightarrow \alpha_i v_{ij} = h_j(\alpha)$, for all $i \in N$ and $j \in G$, \hfill \textbf{\emph{(only highest bids win)}} \label{cond-PE2}
    \item $h_j(\alpha) > 0 \Rightarrow \sum_{i \in N}x_{ij}=1$, for all $j \in G$, \hfill \textbf{\emph{(full allocation of items with positive bids)}} \label{cond-PE3}
    \item $\sum_{j \in G}x_{ij}p_j(\alpha) \leq B_i$, for all $i \in N$, \hfill \textbf{\emph{(budget constraints)}} \label{cond-PE4} 
    \item $\sum_{j \in G} x_{ij}p_j(\alpha) < B_i \Rightarrow \alpha_i = 1$, for all $i \in N$, \hfill \textbf{\emph{(maximum pacing)}} \label{cond-PE5} 
\end{conditions}
\end{definition}

\noindent We remark that \Cref{def:PE} requires the flexibility to choose the way the items are allocated to the possible winners; this is in fact necessary to guarantee the existence of a SPPE for every SPPG. As we mentioned earlier, a SPPE is always guaranteed to exist \citep{conitzer2022multiplicative,chen2021complexity}. Our PPAD-membership proof also provides an alternative, easier proof of existence. We provide a high-level overview of the proof, as well as a comparison to the existing proofs in the next section below.

\subsubsection{Features of our proof}\label{sec:pacing-sp-original-proof}

\paragraph{Previous proofs of existence and PPAD-membership.} The existence of SPPE in SPPGs was first established by \citet{conitzer2022multiplicative}. Their proof proceeds by defining an \emph{$(\varepsilon,H)$-smoothed pacing game}, for $\varepsilon > 0$ and $H >0$, which is derived from the SPPG $\mathcal{G}$ as follows:
\begin{itemize}
    \item[-] There is an artificial \emph{reserve bid} of $2\varepsilon$ on all items.
    \item[-] The set of possible winners $W_j(\alpha)$ (i.e., those buyers $i \in N$ with $x_{ij} >0$ for item $j \in G$) contains all buyers with \emph{almost} the highest bid (i.e., a bid of at least $h_j(\alpha)-\varepsilon$). Each buyer in $W_j$ receives a specific fraction of item $j$ (i.e., $x_{ij}$ is fixed for all $i \in N$ and $j \in N$, as a function of the $\alpha_i$'s.) and pays a specific price $s_{ij}$ (see \citep{conitzer2022multiplicative}, Definition 3 for the exact allocations and prices).
    \item[-] There is an artificial item $j_a$ of unlimited supply. Buyer $i \in N$ receives a quantity $\alpha_i$ of that good and values that good at $v_{ij_{a}}=2\varepsilon$.  
    \item[-] Every buyer $i \in N$ has a utility which is a step function depending on the budget $B_i$, the price $s_{ij}$ she pays for quantities of item $j$, $\alpha_i$, $\varepsilon$ and the parameter $H$ in the definition of the $(\varepsilon,H)$-smoothed pacing game (see \citep{conitzer2022multiplicative}, Definition 3 for the exact utility function).
\end{itemize}
The reason for introducing all these seemingly convoluted conditions is that they ensure that the $(\varepsilon,H)$-smoothed pacing game satisfies a set of desired properties, namely (a) compactness and convexity of the strategy space, (b) continuity of the utility function in all strategies and (c) quasi-concavity of the utility function in a buyer's own strategy. With these at hand, one can then invoke the \citeauthor{debreu1952social}-\citeauthor{fan1952fixed}-\citeauthor{glicksberg1952further} theorem for continuous games [\citeyear{debreu1952social}], which we presented in \Cref{sec:generalized-concave-games}, to establish the existence of a Nash equilibrium. 

The proof of \citet{conitzer2022multiplicative} proceeds by proving Properties (a-c) above, and then obtains the existence of a Nash equilibrium for the $(\varepsilon,H)$-smoothed pacing game. To obtain the existence of a SPPE, they consider a sequence of Nash equilibria for $(\varepsilon,H)$-smoothed pacing games, and show that as $\varepsilon$ converges to $0$ and $H$ converge to infinity, a SPPE is obtained as a limit point. 

Besides the overhead of defining $(\varepsilon,H)$-smoothed pacing games and establishing their properties, the proof uses off-the-shelf existence theorems and limit arguments. Therefore, it is not surprising that it cannot be used to obtain PPAD-membership, or even to establish that rational solutions are always possible. Both of these properties were established via an alternative proof provided by \citet{chen2021complexity}, which we describe briefly next. \medskip

\noindent Similarly in spirit to \cite{conitzer2022multiplicative}, \citet{chen2021complexity} define an \emph{approximate} version of a SPPE, called a $(\delta, \gamma)$-SPPE, with $\delta, \gamma >0$, in which 
\begin{itemize}
    \item[-] \Cref{cond-PE2} in \Cref{def:PE} has been substituted by $x_{ij} > 0 \Rightarrow a_{i}v_{ij} = (1-\delta)h_j(\alpha)$, for all $i \in N$ and all $j \in G$.
    \item[-] \Cref{cond-PE5} in \Cref{def:PE} has been substituted by $\sum_{j \in G} x_{ij}p_j(\alpha) < (1-\gamma)B_i \Rightarrow \alpha_i \geq (1-\gamma)$, for all $i \in N$.
\end{itemize}
In other words, they consider an approximate SPPE in which the set of winners $W_j(\alpha)$ contains buyers that have approximately the highest bid, and in which a buyer that does not spend almost all of her budget uses almost a maximum pacing multiplier. On top of that, they define a variant of $(\delta, \gamma)$-SPPE, which they refer to as \emph{smooth} $(\delta, \gamma)$-SPPE in which the allocation rule is specified, and is in fact similar to that used by \citet{conitzer2022multiplicative} in their definition of a $(\varepsilon,H)$-smoothed pacing game defined above. By definition, a  smooth-$(\delta, \gamma)$-SPPE is $(\delta, \gamma)$-SPPE. Their PPAD-membership proof then proceeds in three steps:
\begin{itemize}
    \item[-] They show the PPAD-membership of smooth $(\delta, \gamma)$-SPPE via a reduction to the computational version of Sperner's lemma \citep{sperner1928neuer}. This implies the existence of a $(\delta, \gamma)$-SPPE.
    \item[-] They devise an intricate rounding procedure, which converts $(\delta, \gamma)$-SPPE into $\gamma$-SPPE, i.e., $(\delta, \gamma)$-SPPE for which $\delta = 0$, and establish the correctness of this procedure via a series of arguments.
    \item[-] They employ a general technique (introduced by \cite{etessami2010complexity}) of converting approximate solutions of inversely exponential precision to exact solutions via a carefully constructed linear program. They use this to convert a $\gamma$-SPPE into a SPPE.
\end{itemize}

\paragraph{Our proof.} Our proof is markedly different from the aforementioned ones, and conceptually much simpler, while at the same time obtaining both existence and PPAD-membership, and as a result establishing the existence of SPPE described by rational numbers. Contrary to \cite{conitzer2022multiplicative} and \cite{chen2021complexity}, we do not need to define relaxations of the SPPE. Instead, we apply Gale's substitution (see \Cref{rem:gale}), i.e., the standard change of variables from allocations $x_{ij}$ to expenditures $q_{ij} = x_{ij}p_j(\alpha)$, see \Cref{sec:PE-preprocess} below.\footnote{Interestingly, \citet{conitzer2022pacing} also apply this change of variables in their Mixed-Integer Linear Programming (MILP) formulation (see Section 6.3 in their paper), but not in the existence proof.} This is in fact the very same change of variables that we used in \Cref{sec:markets} for competitive markets, and serves as a means to escape having to deal with quantities of the form $x_{ij}p_j(\alpha)$, as such multiplications cannot be handled by \linear arithmetic circuits. 

Membership in PPAD then follows from constructing a simple function $F$ whose fixed points will be the SPPE. For each buyer $i \in N$, the expenditures $q_{ij}$ will be obtained as the solutions to a linear program which can be handled by our \linoptgate, and the pacing multipliers will be the solutions to a single equation involving the inputs and the outputs of the function. Arguing that a fixed point of $F$ indeed corresponds to a SPPE reduces to establishing that \Cref{cond-PE1,cond-PE2,cond-PE3,cond-PE4,cond-PE5} of \Cref{def:PE} are satisfied, which turns out to be rather simple. 

\subsubsection{Preprocessing}\label{sec:PE-preprocess}

Before we present our proof, we will first modify and simplify the constraints of \Cref{def:PE}, to make them amenable to the use of our \linoptgate. En route to that, we will prove that without loss of generality, we may assume that in any SPPE $(\alpha,\mathbf{x})$, both the pacing multipliers $\alpha_i$ and the prices $p_j(\alpha)$ can be assumed to be bounded away from $0$, which we will use later on.

\paragraph{Lower bounds on the pacing multipliers and prices.} Below we prove simple lower bounds on the value of the pacing multipliers $\alpha_i^*$ and the prices $p_j(\alpha^*)$ in any SPPE $(\alpha^*, \mathbf{x}^*)$. 

\begin{lemma} \label{lem:alphas-lower-bound-PE}
Let $(\alpha^*, \mathbf{x}^*)$ be any SPPE. Then, for any $i \in N$, we have that $\alpha_i^* \in [\ell_i,1]$, where $\ell_i= \min\{1,B_i/\sum_{j \in G}v_{ij}\}$. 
\end{lemma}

\begin{proof}
If $\alpha^* = 1$ then the lemma holds trivially. Otherwise, by \Cref{cond-PE5}, we have that
\[
B_i = \sum_{j \in G}x_{ij}^*p_j(\alpha^*) \leq \sum_{j \in G: x_{ij}^* >0}h_j(\alpha^*) = \sum_{j \in G: x_{ij}^* >0}\alpha_i^*v_{ij} \leq \alpha_i^* \sum_{j \in G}v_{ij}, 
\]
where the first inequality above holds by the fact that $x_{ij}^*p_j(\alpha^*) \leq x_{ij}^*h_j(\alpha^*) \leq h_j(\alpha^*)$, since $h_j(\alpha^*) \geq p_j(\alpha^*)$ by the definition of the auction, and the second equation holds by \Cref{cond-PE2}.
\end{proof}

\begin{lemma} \label{lem:prices-lower-bound-PE}
Let $(\alpha^*, \mathbf{x}^*)$ be any SPPE. Then, for any $j \in G$, we may assume without loss of generality that $p_j(\alpha^*) > 0$.
\end{lemma}

\begin{proof}
Assume that there is some item $j'$ for which $p_j'(\alpha^*)=0$. For any choice of pacing multipliers $\alpha^*$, i.e., in any PE, we may choose an arbitrary allocation of the good $j'$ for which $x_{ij'} > 0$ only for buyers $i \in W_j(\alpha)$. This way, \Cref{cond-PE2,cond-PE3} of \Cref{def:PE} are satisfied. Since $p_j'(\alpha^*) = 0$, the term for $j'$ contributes $0$ to the sum $\sum_{j \in G} x_{ij}^* p_j(\alpha^*)$ for every buyer $i \in N$, and hence \Cref{cond-PE4,cond-PE5} are unaffected. \Cref{cond-PE1} is trivially satisfied as well. Therefore, given any SPPE $(\alpha^*,\mathbf{x}^*_{-j})$ for items $j \in G\setminus\{j'\}$, we get the same SPPE, together with item $j'$ allocated as described above. 
\end{proof}

\paragraph{Change of variables.} Looking at the conditions of \Cref{def:PE}, we observe that in \Cref{cond-PE4}, the allocation $x_{ij}$ is multiplied by the price $p_j(\alpha)$. Looking ahead, this condition will appear as a constraint in a linear program $\mathcal{P}$ that we will aim to solve via our \linoptgate, as part of a \linear arithmetic circuit. This linear program would have $x_{ij}$ as the variables, but $p_j(\alpha)$ would be \circparams in the context of the \linoptgate. Hence, such a constraint could not be handled by the \linoptgate. To remedy this, we work in a very similar fashion as we did in \Cref{sec:markets}, via applying Gale's substitution (\Cref{rem:gale}). Namely, we will let $q_{ij} = x_{ij} p_j(\alpha)$ be the \emph{expenditure} of buyer $i \in N$ on item $j \in G$. Additionally, \Cref{lem:alphas-lower-bound-PE} implies that $h_j(\alpha^*) > 0$ for all items $j \in G$. In turn, this implies that \Cref{cond-PE1,cond-PE3} can be combined into a single condition, namely
\[
\sum_{i \in N} x_{ij} = 1, \text{ for all } j \in G.
\]
Putting everything together, we have the following equivalent definition of a \emph{second-price pacing expenditure equilibrium (SPPEE)}, from which a SPPE can be straightforwardly recovered. 
\begin{definition}[Second-price Pacing Expenditure Equilibrium (SPPEE)]\label{def:PEE}
A pair $(\alpha, \mathbf{q})$, with $\alpha \in [0,1]^n$ and $\mathbf{q} \in [0,1]^{nm}$, is a \emph{second-price pacing expenditure equilibrium (SPPEE)} of an SPPG $\mathcal{G}$, if the following conditions hold:
\begin{conditions}[label={II.\alph*.}]
    \item $\sum_{i \in N}q_{ij} = p_j(\alpha)$, for all $j \in G$, \hfill \textbf{\emph{(feasible expenditure)}} \label{cond-PEE1}
    \item $q_{ij} > 0 \Rightarrow \alpha_i v_{ij} = h_j(\alpha)$, for all $i \in N$ and $j \in G$, \hfill \textbf{\emph{(only highest bids spend)}} \label{cond-PEE2}
    \item $\sum_{j \in G}q_{ij}(\alpha) \leq B_i$, for all $i \in N$, \hfill \textbf{\emph{(budget constraints)}} \label{cond-PEE3} 
    \item $\sum_{j \in G} q_{ij} < B_i \Rightarrow \alpha_i = 1$, for all $i \in N$. \hfill \textbf{\emph{(maximum pacing)}} \label{cond-PEE4} 
\end{conditions}
\end{definition}

\subsubsection{Membership in PPAD}

We are now ready to prove the main result of the section.

\begin{theorem}
Computing a SPPE of any SPPG $\mathcal{G}$ is in PPAD.
\end{theorem}

\noindent Following our standard methodology, we will develop the proof in three steps, namely (a) construction of the function $F$ and arguing that it can be computed by a \linear arithmetic circuit containing \linoptgates, (b) showing that the \linoptgate can compute all the necessary components, and (c) arguing that a SPPE can be recovered from a fixed point of $F$. 

\paragraph{The function $F$.} We define the function $F: D \rightarrow D$ with domain $D = \prod_{i \in N}[l_i,1] \times \prod_{i \in N}\prod_{j \in G}[0,v_{ij}]$. Let $(\alpha,\mathbf{q})$ and $(\alpha^*, \mathbf{q}^*)$ denote inputs and outputs to $F$ respectively. For any item $j \in G$, the vector of expenditures $\mathbf{q}^*_j = (q_{1j}^*, \ldots, q_{nj}^*)$ is obtained as the output of the following linear program.

\begin{center}\underline{Linear Program $\mathcal{P}$}\end{center}
\begin{align*}
\mbox{maximize } & \sum_{i \in N}(\alpha_i v_{ij})q_{ij}  \\
\mbox{subject to }& \sum_{i \in N} q_{ij} \leq p_j(\alpha) \\
& q_{ij} \geq 0, \;\forall i \in N
\end{align*}
The pacing multipliers will be obtained via the following equation
\begin{equation}\label{eq:alpha-fixed-point-PPE}
\alpha_i^* = \max \left\{\frac{B_i}{\sum_{j \in G}v_{ij}}, \min \left\{1,\alpha_i+B_i-\sum_{j\in G}q_{ij}\right\}\right\}
\end{equation}
\Cref{eq:alpha-fixed-point-PPE} uses $\max, \min$ and addition operations on the variables, and hence can be computed by a \linear arithmetic circuit. Similarly, recall that $p_j(\alpha)=\min_{i \in N}\max_{i' \in N, i\neq i'}\alpha_{i'}v_{i'j}$, and hence the linear program $\mathcal{P}$ can also be computed by a \linear arithmetic circuit.

\paragraph{Computation by the \linoptgate.} In the next lemma, we argue that the output $\mathbf{q}_j^*$ of the linear program $\mathcal{P}$ can be computed by the \linoptgate.

\begin{lemma}
Consider the linear program $\mathcal{P}$. An optimal solution to $\mathcal{P}$ can be computed by a \linoptgate.
\end{lemma}

\begin{proof}
We will argue that $\mathcal{P}$ satisfies all the properties required for the \linoptgate, as highlighted in \Cref{sec:linopt-applications}. The feasible domain $[0,1]^n$ of $\mathcal{P}$ is non-empty and bounded, and the \circparams $p_j(\alpha)$ appear only on the right-hand side of the constraints. The subgradient of the objective function is a linear function, and hence can be computed by a \pseudog.  
\end{proof}

\paragraph{Fixed points lead to SPPE.} We will show that a fixed point $(\alpha^*, \mathbf{q}^*)$ of $F$ is a SPPEE (\Cref{def:PEE}); a SPPE can then be straightforwardly recovered by setting $x_{ij} = q_{ij}/p_j(\alpha^*)$ for all $i \in N$ and $j \in G$, since all the prices can be assumed to be non-zero by \Cref{lem:prices-lower-bound-PE}. In particular, we show that $(\alpha,\hat{\mathbf{q}})=(\alpha^*, \mathbf{q}^*)$ satisfies \Cref{cond-PEE1,cond-PEE2,cond-PEE3,cond-PEE4} of \Cref{def:PEE}.
\begin{enumerate}[leftmargin=*]
    \item Let $j \in G$ be any item. Since $\mathbf{q}_j$ is an optimal solution to linear program $\mathcal{P}$, and since $\alpha_i v_{ij}$ is positive for all $i \in N$, the ``payment'' constraint of $\mathcal{P}$ must be tight, i.e., $\sum_{i \in N} q_{ij}^* = p_j(\alpha) = p_j(\alpha^*)$. Hence \Cref{cond-PEE1} is satisfied. \label{enum-peeproof-1}

    \item Consider any buyer $i \in N$ and item $j \in N$ such that $q_{ij}^*>0$. By the optimality of $\mathbf{q}^*_j$, $\alpha_iv_{ij}$ must be maximum, i.e., $\alpha_iv_{ij} = \alpha_i^* \cdot v_{ij} = h_j(\alpha^*)$. Hence \Cref{cond-PEE2} is satisfied.\label{enum:peeproof-2}

    \item Let $i \in N$ be any buyer. Assume by contradiction that \Cref{cond-PEE3} is violated, i.e., that $B_i < \sum_{i \in N}q_{ij}^*$. Since $\alpha_i = \alpha_i^*$, \Cref{eq:alpha-fixed-point-PPE} implies that $\alpha_i^* = B_i / \sum_{j \in G}v_{ij}$. However, since $\mathbf{q}_j^*$ is a feasible solution to the linear program $\mathcal{P}$, we have
    \begin{align*}
    \sum_{j \in G} q_{ij}^* &= \sum_{j \in G: q_{ij}^* > 0} q_{ij}^* = \sum_{j \in G: \alpha_{i}^{*} v_{ij} = h_j(\alpha^*)} q_{ij}^*
    \leq \sum_{j \in G: \alpha_{i}^{*} v_{ij} = h_j(\alpha^*)} p_j(\alpha^*) \\
    & \leq \sum_{j \in G: \alpha_{i}^{*} v_{ij} = h_j(\alpha^*)} h_j(\alpha^*) = \sum_{j \in G: \alpha_{i}^{*} v_{ij} = h_j(\alpha^*)} \alpha_i^*v_{ij} \leq \sum_{j \in G} \alpha_i^* v_{ij} \\
    & = \frac{B_i}{\sum_{j \in G}v_{ij}}\sum_{j \in G}v_{ij} = B_i,
    \end{align*}
    where in the calculations above:
    \begin{itemize}
        \item[-] the second equation follows from \Cref{cond-PEE2}, which we we established in \Cref{enum:peeproof-2} above,
        \item[-] the first inequality follows from the ``payment'' constraints of the linear program $\mathcal{P}$, and
        \item[-] the third inequality follows from the fact that $h_j(\alpha^*) \geq p_j(\alpha^*)$ by definition.
    \end{itemize}
    We have reached a contradiction. \label{enum-peeproof-3}
    \item Consider some buyer $i \in N$ such that $\sum_{j \in G}q_{ij}^* < B_i$. Since $\alpha_i = \alpha_i^*$, \Cref{eq:alpha-fixed-point-PPE} in that case implies that $\alpha_i^* = 1$, and hence \Cref{cond-PEE4} is satisfied. \label{enum-peeproof-4}
\end{enumerate}
This completes the proof.

\subsection{Pacing Equilibria in ROI-constrained Auctions}\label{sec:pacing-sp-roi}
We now turn our attention to market environments where the total expenditure of each bidder is not constrained by a fixed budget, but by a constraint on (average) return-on-investement (ROI), in the setting proposed by \citet{li2022auto}. For consistency with \Cref{sec:pacing-sp-budgets}, we will refer to these environments as \emph{second-price pacing ROI games (SPPRG)}. A SPPRG $\mathcal{G}$ consists of 
\begin{itemize}
    \item[-] a set $N$ of buyers, with $n=|N|$,
    \item[-] a set $G$ of items, with $m = |G|$, 
    \item[-] a set of vectors of \emph{valuations} $\mathbf{v}_{i} = (v_{i1}, \ldots v_{im})$ for every buyer $i \in N$, where $v_{ij}$ denotes the value of buyer $i \in N$ for item $j\in G$; without loss of generality we assume that for all $j \in G$, there exists some $i \in N$ such that $v_{ij}>0$.
\end{itemize}
Similarly to \Cref{sec:pacing-sp-budgets}, for each buyer $i \in N$ there is a pacing multiplier $\alpha_i$ and the buyer's bid on item $j$ is $\alpha_i \cdot v_{ij}$. For this setting to be meaningful, following \citep{li2022auto}, we will not constrain the pacing multipliers to lie in $[0,1]$, but in $[1,A]$ for some constant $A \in \mathbb{R}, A\geq 1$. Again, the items are sold in parallel second-price auctions, which implies that
\begin{enumerate}
    \item the item is allocated to the buyer with the highest bid, breaking ties using some tie-breaking rule in case there are multiple such buyers. Let $h_j(\alpha) = \max_{i \in N} \alpha_i v_{ij}$ be the highest bid on item $j \in G$.\label{item:SP-ROI-1}
    \item the price of item $j \in G$ is the second-highest bid. Let $p_j(\alpha)=\min_{i \in N}\max_{i' \in N, i\neq i'}\alpha_{i'}v_{i'j}$ be that price. \label{item:SP-ROI-2}
\end{enumerate}
Again, let $W_j(\alpha)=\{i \in N: \alpha_i v_{ij} = h_j(\alpha)\}$ be the set of possible winners (those with the highest bid) for item $j \in G$. Each buyer $i \in N$ receives an allocation $x_{ij}$ of item $j \in N$, with the restriction that $x_{ij}=0$ for all $i \in N\setminus W_j(\alpha)$.

\paragraph{ROI Pacing Equilibria.} Next, we define the notion of a second-price ROI pacing equilibrium.\footnote{\citet{li2022auto} refer to these equilibria as ``auto-bidding equilibria''; we choose the term ``second-price ROI pacing equilibria'' for consistency with the results of \Cref{sec:pacing-sp-budgets}.} Note that the only significant difference between \Cref{def:RPE} and \Cref{def:PE} of \Cref{sec:pacing-sp-budgets} is that the budget constraints have now been replaced by ROI constraints.

\begin{definition}[Second-price ROI Pacing Equilibrium (SPRPE)]\label{def:RPE}
 A pair $(\alpha, \mathbf{x})$, with $\alpha \in [1,A]^n$ and $\mathbf{x} \in [0,1]^{nm}$, is a \emph{second-price ROI pacing equilibrium (SPRPE)} of an SPPRG $\mathcal{G}$, if the following conditions hold:
\begin{conditions}[label={III.\alph*.}]
    \item $\sum_{i \in N}x_{ij} \leq 1$, for all $j \in G$, \hfill \textbf{\emph{(feasible allocation)}} \label{cond-RPE1}
    \item $x_{ij} > 0 \Rightarrow \alpha_i v_{ij} = h_j(\alpha)$, for all $i \in N$ and $j \in G$, \hfill \textbf{\emph{(only highest bids win)}} \label{cond-RPE2}
    \item $h_j(\alpha) > 0 \Rightarrow \sum_{i \in N}x_{ij}=1$, for all $j \in G$, \hfill \textbf{\emph{(full allocation of items with positive bids)}} \label{cond-RPE3}
    \item $\sum_{j \in G}x_{ij}p_j(\alpha) \leq \sum_{j \in G}x_{ij}v_{ij}$, for all $i \in N$, \hfill \textbf{\emph{(ROI constraints)}} \label{cond-RPE4} 
    \item $\sum_{j \in G} x_{ij}p_j(\alpha) < \sum_{j \in G}x_{ij}v_{ij} \Rightarrow \alpha_i = A$, for all $i \in N$. \hfill \textbf{\emph{(maximum pacing)}} \label{cond-RPE5} 
\end{conditions}
\end{definition}
\noindent \citet{li2022auto} established that a SPRPE always exists. Their proof follows closely that of \citet{conitzer2022multiplicative} for second-price pacing equilibria in auctions with budgets, which we referred to in \Cref{sec:pacing-sp-original-proof}. In particular, they also define a very similar $(\varepsilon,H)$-smoothed game that satisfies the properties required for the \citeauthor{debreu1952social}-\citeauthor{fan1952fixed}-\citeauthor{glicksberg1952further} theorem [\citeyear{debreu1952social}], and recover a SPRPE as a limit point of the sequence of Nash equilibria of this game. The proof that we present in this section follows very much along the same lines of the proof that we presented in \Cref{sec:pacing-sp-budgets} and is thus conceptually much simpler than that of \citet{li2022auto}. In addition, since it does not exhibit the discontinuous arguments of the proof of \citet{li2022auto}, it can also be used to show FIXP-membership of the problem, as we highlight in \Cref{sec:RPE-FIXP-membership}.  \medskip

\noindent Since $\alpha_i^* \geq 1$ for any $i \in N$, and since for every $j \in G$, there exists some buyer $i \in N$ with $v_{ij}>0$, we can again merge \Cref{cond-RPE1,cond-RPE3} above and obtain the following equivalent definition for a SPRPE.
\begin{definition}[Second-price ROI Pacing Equilibrium (SPRPE)-simplified]\label{def:RPE-s}
A pair $(\alpha, \mathbf{x})$, with $\alpha \in [1,A]^n$ and $\mathbf{x} \in [0,1]^{nm}$, is a \emph{second-price ROI pacing equilibrium (SPRPE)} of an SPPRG $\mathcal{G}$, if the following conditions hold:
\begin{conditions}[label={IV.\alph*.}]
    \item $\sum_{i \in N}x_{ij} = 1$, for all $j \in G$, \hfill \textbf{\emph{(feasible allocation)}} \label{cond-RPE1-s}
    \item $x_{ij} > 0 \Rightarrow \alpha_i v_{ij} = h_j(\alpha)$, for all $i \in N$ and $j \in G$, \hfill \textbf{\emph{(only highest bids win)}} \label{cond-RPE2-s}
    \item $\sum_{j \in G}x_{ij}p_j(\alpha) \leq \sum_{j \in G}x_{ij}v_{ij}$, for all $i \in N$, \hfill \textbf{\emph{(ROI constraints)}} \label{cond-RPE3-s} 
    \item $\sum_{j \in G} x_{ij}p_j(\alpha) < \sum_{j \in G}x_{ij}v_{ij} \Rightarrow \alpha_i = A$, for all $i \in N$. \hfill \textbf{\emph{(maximum pacing)}} \label{cond-RPE4-s} 
\end{conditions}
\end{definition}

\subsubsection{A new proof of existence}\label{sec:RPE-proof-existence}

We state the main theorem of this section.

\begin{theorem}\label{thm:ROI-equilibrium-exists}
A SPRPE exists for any SPPRG $\mathcal{G}$.
\end{theorem}

\begin{proof} We define the function $F: D \rightarrow D$ with domain $D = \prod_{i \in N}[1,A] \times \prod_{i \in N}\prod_{j \in G}[0,1]$. Let $(\alpha,\mathbf{x})$ and $(\alpha^*, \mathbf{x}^*)$ denote inputs and outputs to $F$ respectively. For any item $j \in G$, the vector of allocations $\mathbf{x}^*_j = (x_{1j}^*, \ldots, x_{nj}^*)$ is obtained as the output of the following linear program.\bigskip \bigskip \bigskip \bigskip

\begin{center}\underline{Linear Program $\mathcal{P}$}\end{center}
\begin{align*}
\mbox{maximize } & \sum_{i \in N}(\alpha_i v_{ij})x_{ij}  \\
\mbox{subject to }& \sum_{i \in N} x_{ij} = 1,\;\forall i \in N \\
& x_{ij} \geq 0, \;\forall i \in N
\end{align*}
The pacing multipliers will be obtained via the following equation
\begin{equation}\label{eq:alpha-fixed-point-RPE}
\alpha_i^* = \max \left\{1, \min \left\{A,\alpha_i+\sum_{j \in G}x_{ij}v_{ij}-\sum_{j\in G}x_{ij}p_{j}(\alpha_i)\right\}\right\}
\end{equation}

\noindent We will show that a fixed point $(\alpha^*, \mathbf{x}^*)$ of $F$ is a RPE. In particular, we show that $(\alpha,\mathbf{x})=(\alpha^*, \mathbf{x}^*)$ satisfies \Cref{cond-RPE1-s,cond-RPE2-s,cond-RPE3-s,cond-RPE4-s} of \Cref{def:RPE-s}.
\begin{enumerate}[leftmargin=*]
\item Let $j \in G$ be any item. Since $\mathbf{x}_j$ is an optimal solution to linear program $\mathcal{P}$, and since $\alpha_i v_{ij}$ is positive for all $i \in N$, the allocation constraint of $\mathcal{P}$ must be tight, i.e., $\sum_{i \in N} x_{ij}^* = 1$. Hence \Cref{cond-RPE1-s} is satisfied. \label{enum-rpeproof-1}

\item Consider any buyer $i \in N$ and item $j \in N$ such that $x_{ij}^*>0$. By the optimality of $\mathbf{x}^*_j$, $\alpha_iv_{ij}$ must be maximum, i.e., $\alpha_iv_{ij} = \alpha_i^* \cdot v_{ij} = h_j(\alpha^*)$. Hence \Cref{cond-RPE2-s} is satisfied.\label{enum:rpeproof-2}

\item Let $i \in N$ be any buyer. Assume by contradiction that \Cref{cond-RPE3-s} is violated, i.e., that $\sum_{j \in G}x_{ij}^*v_{ij} < \sum_{i \in N}x_{ij}^*p_j(\alpha^*)$. Since $\alpha_i = \alpha_i^*$, \Cref{eq:alpha-fixed-point-RPE} implies that $\alpha_i^* = 1$. In turn, this implies that for any item $j \in G$, we have $v_{ij} = \alpha_i^* v_{ij}$. Consider any such an item $j' \in G$ for which $x_{ij'}^* > 0$. By \Cref{cond-RPE2-s} which was established in \Cref{enum:rpeproof-2} above, we have that $\alpha_{ij'}^* v_{ij'} = h_j'(\alpha^*)$. Combining the above two expressions, and the fact that $h_j'(\alpha^*) \geq p_j'(\alpha^*)$, we obtain that $v_{ij'} \geq p_j'(\alpha^*)$. By summing over all items $j \in G$ with $x_{ij}>0$, we obtain a contradiction. \label{enum-rpeproof-3}

\item Consider some buyer $i \in N$ such that $\sum_{j \in G}x_{ij}^*p_j(\alpha^*) < \sum_{j \in G}x_{ij}^* v_{ij}$. Since $\alpha_i = \alpha_i^*$, \Cref{eq:alpha-fixed-point-RPE} in that case implies that $\alpha_i^* = A$, and hence \Cref{cond-RPE4} is satisfied. \label{enum-rpeproof-4}
\end{enumerate}
This completes the proof.
\end{proof}

\subsubsection{Membership in the class FIXP}\label{sec:RPE-FIXP-membership}
The main technical difference between our proof in \Cref{sec:pacing-sp-budgets} and that in \Cref{sec:pacing-sp-roi} is that in the latter case, we did not perform a variable change. Indeed, the nature of the ROI constraints in \Cref{cond-RPE3-s} prevents us from doing that, as the variables $x_{ij}$ appear in both sides of the constraints, once multiplied by the constants $v_{ij}$ and once with the \circparams $p_j(\alpha)$. The problem with having expressions of the form $\sum_{j \in G}x_{ij}p_j(\alpha)$ is that we are not allowed to multiply two parameters which are both inputs to a \linear arithmetic circuit. We can do that however with (general) arithmetic circuits, see \Cref{def:arithmetic circuit}.

As we discussed in \Cref{sec:preliminaries}, computation of fixed points via (general) arithmetic circuits corresponds to the class FIXP \citep{etessami2010complexity}, which allows for computation of solutions that are described by irrational numbers. Here we can make use of the \fixpoptgate that was designed by \citet{SICOMP:Filos-RatsikasH2023}, in order to show FIXP-membership of the problem. We have the following theorem.

\begin{theorem}\label{thm:ROI-equilibrium-FIXP}
Computing a SPRPE of any SPPRG $\mathcal{G}$ is in FIXP.
\end{theorem}

\begin{proof}
\noindent To turn the existence proof of \Cref{thm:ROI-equilibrium-exists} to a FIXP-membership proof, we have to
\begin{itemize}
    \item[-] argue that the function $F$ that we defined in \Cref{sec:RPE-proof-existence} can be computed by an arithmetic circuit,
    \item[-] argue that the \fixpoptgate of \citep{SICOMP:Filos-RatsikasH2023} can solve the linear program $\mathcal{P}$ of \Cref{sec:RPE-proof-existence}.
\end{itemize}
The first part follows by observing that all the parameters of the linear program $\mathcal{P}$ and those of \Cref{eq:alpha-fixed-point-PPE} can be computed using the standard operations of the circuit. Note that compared to the corresponding argument in \Cref{sec:pacing-sp-budgets}, here we indeed need the capability to multiply input variables, in order to compute $\alpha_i^*$ in \Cref{eq:alpha-fixed-point-PPE}. Linear program $\mathcal{P}$ is very simple and can be easily seen to satisfy all the properties required for the \fixpoptgate to work, stated in \citep{SICOMP:Filos-RatsikasH2023}.
\end{proof}

\subsubsection{Limits of rationality: second-price ROI pacing equilibria can be irrational}\label{sec:RPE-irrational}

To round off the section, it remains to determine whether our FIXP-membership result for SPRPE, rather than a PPAD-membership result, is due to our proof technique falling short, or whether there is some inherent reason for this. As we show below, there are simple examples of SPRPGs $\mathcal{G}$ for which all SPRPE are irrational and hence a PPAD-membership result is not possible. 

\begin{example}[Example of a SPRPG $\mathcal{G}$ that only has irrational SPRPE]\label{ex:irrational-roi} Consider a SPRPG $\mathcal{G}$ with 2 buyers and 3 items, with $\mathbf{v}_1 = (1,0,1)$ and $\mathbf{v}_2 = (0,1,2)$. Let $A$ be sufficiently large, e.g., $A > 4$; the proof can easily be adapted for other values of $A$ by appropriate rescaling of the parameters.
Via a case analysis, we show that in any SPRPE, the pacing multipliers $\alpha_i$ for $i \in \{1,2\}$ are irrational. \medskip

\noindent Let $(\alpha,\mathbf{x})$ be a SPRPE. Clearly, 
\begin{itemize}
    \item[-] buyer 1 will be allocated the entire quantity of item 1, i.e., $x_{11}=1$, at price $p_1(\alpha)=0$,
    \item[-] buyer 2 will be allocated the entire quantity of item 2, i.e., $x_{22}=1$, at price $p_2(\alpha)=0$.
\end{itemize} 

\paragraph{Case 1: $\alpha_1 < 2\alpha_2$.}

In this case, buyer 2 wins the entirety of item $3$ at price $\alpha_1$, thus we have $p_3(\alpha) = \alpha_1$, $x_{13}=0$ and $x_{23} = 1$. By the ROI constraint (\Cref{cond-RPE4} of \Cref{def:RPE}) for buyer 2 it holds that 
\begin{align*}
\alpha_1 = p_3(\alpha)\cdot x_{23} \leq 1+2 \cdot x_{23} = 3
\end{align*}
However, we have that
\begin{align}
x_{11}\cdot p_1(\alpha) = 0 < 1 = v_{11}\cdot x_{11}
\end{align}
By the maximum pacing constraint (\Cref{cond-RPE5} of \Cref{def:RPE}) for buyer 1, it follows that $\alpha_1 = A$. Hence $A = \alpha_1\leq 3$, which is not possible by the choice of $A$. This implies that $(\alpha,\mathbf{x})$ is not an SPRPE, a contradiction.

\paragraph{Case 2: $\alpha_1 > 2\alpha_2$.}

In this case, buyer 1 wins the entirety of item $3$ at price $2\alpha_2$, thus we have $p_3(\alpha) = 2\alpha_2$, $x_{13} = 1$ and $x_{23}=0$. By the ROI constraint (\Cref{cond-RPE4} of \Cref{def:RPE}) for buyer 1 it holds that 
\begin{align}
2\alpha_2\cdot x_{13} = p_3(\alpha)\cdot x_{13}\leq 1+1\cdot x_{13} = 2
\end{align}
and hence $\alpha_2\leq 1$, which implies that $\alpha_2=1$, since $\alpha_2 \in [1,A]$. Therefore, we have that $\alpha_2\cdot v_{21}=1$, whereas $\sum_{j \in \{1,2,3\}} x_{2j}\cdot p_j(\alpha) = 0$. By the maximum pacing constraint (\Cref{cond-RPE5} of \Cref{def:RPE}) for buyer 2, it follows that $\alpha_2 = A \leq 1$, which is not possible since $A > 1$. This implies that $(\alpha,\mathbf{x})$ is not an SPRPE, a contradiction.

\paragraph{Case 3: $\alpha_1 = 2\alpha_2$.}

In this case, both buyers $1$ and $2$ submit the same bid, and hence they are both eligible to receive positive quantities of item 3. Since the first- and second-highest bids coincide, the price of the item is $p_3(\alpha)=\alpha_1=2\alpha_2\geq 2$, since $\alpha_i \in [1,A]$ for $i \in \{1,2\}$. By the ROI constraint (\Cref{cond-RPE4} of \Cref{def:RPE}) for both buyers, we have
\begin{align}
& p_3(\alpha)\cdot x_{13}\leq 1+x_{13}\\
& p_3(\alpha)\cdot x_{23}\leq 1+2x_{23}
\end{align}
Next, we consider whether these inequalities can be strict or not. 
First, if the second inequality is strict, then by the maximum pacing constraint (\Cref{cond-RPE5} of \Cref{def:RPE}) for buyer 2, we have that $\alpha_2 = A$. However, then $\alpha_1 = 2\alpha_2 = 2A$,
which contradicts the fact that $\alpha_1\in [1,A]$. This implies that the second inequality is in fact an equality, as otherwise $(\alpha,\mathbf{x})$ is not an SPRPE.\medskip

\noindent Now suppose that the first inequality is strict. By the maximum pacing constraint (\Cref{cond-RPE5} of \Cref{def:RPE}) for buyer 1, we have that $\alpha_1 = A$ and hence $p_3(\alpha)=A$. Now adding the two inequalities together we have
\begin{align}
A = \alpha_1 = p_3(\alpha)\cdot (x_{13}+x_{23}) < 2+x_{13}+2x_{23} = 3+x_{23}\leq 4
\end{align}
which is a contradiction since $A>4$. This implies that the first inequality is also in fact an equality, as otherwise $(\alpha,\mathbf{x})$ is not an SPRPE.\medskip

\noindent We are now left with two equations, namely that 
\[x_{12} = \frac{1}{p_3(\alpha)-1} \text{ and } x_{23} = \frac{1}{p_3(\alpha)-2}\] 

\noindent Adding these together, we have that
\begin{align*}
1 = x_{13}+x_{23} = \frac{1}{p_3(\alpha)-1}+\frac{1}{p_3(\alpha)-2}
\end{align*}
Multiplying through by the denominators, we obtain the quadratic equation $p_3(\alpha)^2-5p_3(\alpha)+5=0$ which has
only irrational solutions $(5\pm\sqrt{5})/2$. Since $p_3(\alpha)\geq 2$, this means that $p_3(\alpha)=(5+\sqrt{5})/2$ is the unique solution corresponding to a SPRPE.
We conclude that $\alpha_1 = p_3(\alpha)$ and $\alpha_2 = p_3(\alpha)/2$
are irrational.
\end{example}

\section{Fair Division}\label{sec:fair-division}

In this section we consider applications of our \linoptgate to fair division. In particular, we will mainly consider two rather fundamental settings, namely \emph{envy-free cake cutting} and \emph{rental harmony}. The former is one of the prototypical problems in fair division \citep{gamow1958puzzle}, the origins of which date back to the late 1940s and the work of \citet{steinhaus1949division}, and concerns the fair division of a single infinitely divisible resource among a set of agents with heterogeneous preferences over parts of the resource. The latter was studied famously by \citet{AMM:Su1999}, who credits its origins to the chore division problem of \citet{gardner1978aha}, and is concerned with the fair allocation or rooms to tenants, taking into account their preferences over rooms and the rent for each room. 

The PPAD-membership for \emph{approximate} envy-free cake cutting has been known since the work of \citet{OR:DengQS2012}, and has been hinted at by the existence proof of Simmons \citep{AMM:Su1999} that in fact goes via an approximate version and applies Sperner's Lemma \citep{sperner1928neuer}. \citet{SICOMP:Filos-RatsikasH2023} recently showed that finding an exact envy-free division is in FIXP in general. Still, there are cases for which exact solutions are rational (e.g., when the preferences are captured by piecewise-constant densities), and for those FIXP is not the appropriate class. Our main theorem in \Cref{sec:ef-cake-cutting} below extends the ideas of \citeauthor{SICOMP:Filos-RatsikasH2023} to show the desired PPAD-membership result for these cases. We remark that a FIXP-hardness result is also known from \citep{SICOMP:Filos-RatsikasH2023}, as well as a PPAD-hardness result from \citep{OR:DengQS2012} for approximate divisions, but those only concern very general versions of the problem and leave much room for improvement.

For rental harmony, there were no known complexity results before our work, to the best of our knowledge. It turns out that our \linoptgate allows us to obtain the PPAD-membership of finding fair partitions using very much the same approach as for the case of envy-free cake cutting, again for these cases for which partitions in rational numbers exist. For more general settings (for which all partitions might be irrational), we complement our results with a FIXP-membership proof, based on the OPT-gate for FIXP of \citep{SICOMP:Filos-RatsikasH2023}.

We remark that very recently \citet{caragiannis2023complexity} already used the \linoptgate (which we made them aware of via personal communications) to establish that computing envy-free and Pareto-optimal allocations of multiple divisible goods is in PPAD.

\subsection{Envy-free cake cutting}\label{sec:ef-cake-cutting}

We start from arguably the most fundamental fair division problem, that of \emph{envy-free cake cutting}.

\begin{definition}[(Contiguous) envy-free cake cutting]\label{def:ef-cake-cutting}
Let the interval $[0,1]$ be called a \emph{cake}, which is to be divided into $n$ subintervals (pieces) using $n-1$ cuts. Let $\mathbf{x} \in \Delta^{n-1}$ denote a division of the cake, i.e., each division is a point in the simplex, and let $x_j$ be the $j$'th coordinate (or ``the $j$'th piece''). There is a set $N$ of $n$ agents, and for each piece $j$, each agent $i \in N$ has a \emph{valuation function} $u_{ij}: \Delta^{n-1} \rightarrow \mathbb{R}_{\geq 0}$, assigning a real number to a division of the cake. Given a division $\mathbf{x}$, we will say that agent $i \in N$ \emph{prefers} the $j$-th piece if $u_{ij}(\mathbf{x}) \geq u_{ij'}(\mathbf{x})$ for any piece $j'$. A division $\mathbf{x}$ is envy-free if there exists a permutation $\pi$ of $\{1,\ldots,n\}$ such that for every $i \in N$, agent $i$ prefers piece $\pi(i)$. 
\end{definition}

\noindent We offer the following remarks regarding \Cref{def:ef-cake-cutting}:
\begin{itemize}
    \item[-] The definition considers the \emph{contiguous} version of the problem, where the pieces are single intervals. In the more general version of the problem, each piece can be a collection of possibly disconnected intervals. Since a contiguous division is clearly a division, the PPAD-membership of computing contiguous divisions is a stronger result.
    \item[-] The definition above is very general, in the sense that the valuations of the agents for the pieces do not only depend on the pieces themselves, but they could depend on the whole division $\mathbf{x}$, e.g., on how the remaining pieces have been allocated. This captures for example valuations that exhibit externalities. In contrast, several textbooks on the problem (e.g., see \citep{robertson1998cake,brams1996fair,procaccia2013cake}) often consider the problem where the valuation of an agent only depends on her own piece, and in fact these valuations are captured by additive measures over the cake (i.e., the value of an agent for the union of two intervals is the sum of her value for each interval). Again, with regard to PPAD-membership, the more general setting that we consider here makes the result stronger.
\end{itemize}

\paragraph{Sufficiency condition: Hungriness.} We will say that an agent $i \in N$ is \emph{hungry}, if she prefers a non-empty piece of cake to an empty piece. An instance of the envy-free cake cutting problem satisfies the hungriness condition if all of the agents are hungry.

\paragraph{Known results for existence and complexity.} The existence of envy-free cake cutting divisions (in the sense of \Cref{def:ef-cake-cutting}) was proven independently by \citet{woodall1980dividing} and Simmons (credited in \citep{AMM:Su1999}), under the hungriness condition. Both of these proofs go via proving existence for an approximate version of the problem (via Sperner's Lemma \citep{sperner1928neuer} or the related K-K-M Lemma \citep{FM:KnasterKM1929}) and then take limits to obtain the existence of exact divisions. An alternative proof by \citet{woodall1980dividing} uses Brouwer's fixed point theorem \citep{MA:Brouwer1911}. The first continuous proof of existence was given by \citet{SICOMP:Filos-RatsikasH2023} in the context of proving that computing an envy-free division is in the class FIXP, via the employment of their OPT-gate for FIXP:

\begin{theorem}[\citep{SICOMP:Filos-RatsikasH2023}]\label{thm:FIXPpaper-cake-cutting}
Computing an envy-free division of the cake is in FIXP.
\end{theorem}

\noindent In general, FIXP is indeed the right class for the computation of exact envy-free divisions, as there are simple examples showing that even when the valuations are given by linear density functions, all envy-free divisions might be irrational (e.g., see \citep{bei2012optimal}). We provide a simple one below for completeness:

\begin{example}[Example where all envy-free cake cutting divisions are irrational]\label{ex:cake-irrational}
Consider an instance with two identical agents and let $u$ be their common valuation function. For $u$, we will specify the densities over the cake $[0,1]$, which are both given by $f(t)=2t$. Let $z$ be the point where the cake is cut at an envy-free division. It must hold that $u([0,z]) = u([z,1])$, i.e., 
\begin{align*}
    \int_0^z 2t\; dt =\int_z^1 2t\; dt 
\end{align*}
The first integral evaluates to $z^2$ and the second integral evaluates to $1-z^2$. It follows that the only envy-free division is obtained by cutting the cake at $z=1/\sqrt{2}$, an irrational number.
\end{example}

\noindent Still, there are interesting valuation functions for which rational envy-free divisions exist, and for those cases a FIXP-membership result is unsatisfactory. For example, it is known that when agents have piecewise-constant density functions (i.e., step functions), an envy-free division in rational numbers exists (e.g., see \citep{goldberg2020contiguous}). For these cases for which rationality is possible, we would like to obtain a PPAD-membership result instead. This is achievable via the ``approximation and rounding'' approach which was discussed in \Cref{sec:Other-Approaches}, i.e., to start from an $\varepsilon$-approximately envy-free division of the cake, known to be computable in PPAD from \citep{OR:DengQS2012}, and then round it to an exact solution, as long as $\varepsilon$ is small enough \citep{goldberg2020contiguous}. We provide an alternative, conceptually much simpler proof, via the employment of our \linoptgate, which does not require approximations or rounding. We will obtain the PPAD-membership of the problem for any valuation function that is given by a \pseudog. We state the main theorem of the section.

\begin{theorem}\label{thm:cake-cutting}
Computing an envy-free division of the cake when the agents' valuation functions are given by \pseudogs is in PPAD.
\end{theorem}

\noindent The valuation functions are the integrals of the density functions, so \Cref{thm:cake-cutting} immediately implies a PPAD-membership result for envy-free divisions with piecewise-constant densities, which we mentioned above. Using the machinery developed in \Cref{sec:implicit}, we can also capture simple interesting cases where the inputs are given as the integrals of the density functions rather than the utilities themselves. We develop the proof in \Cref{sec:cake-cutting-proof} below. 

\subsubsection{The proof of \crtCref{thm:cake-cutting} \label{sec:cake-cutting-proof}}

The proof follows very much along the same lines of that of \Cref{thm:FIXPpaper-cake-cutting} presented in \citep{SICOMP:Filos-RatsikasH2023}, with a crucial modification to make it amenable to the use of \linear arithmetic circuits. For that reason, it is instructive to first present the proof of \citep{SICOMP:Filos-RatsikasH2023}, and then explain how to obtain the proof of \Cref{thm:cake-cutting} from there.

\paragraph{Envy-free divisions as matchings in bipartite graphs.} Let $\mathbf{x} \in \Delta^{n-1}$ be a division of the cake. Let $G=(N,P,E)$ be the bipartite graph in which $N$ is the set of agents, $P$ is the set of pieces and an edge $(i,j)$, $i \in N, j \in P$ is in $E$ if and only if agent $i$ prefers piece $j$. Given this interpretation, $\mathbf{x}$ is envy-free if and only if $G$ admits a perfect matching. We will find that perfect matching via a maximum flow argument, therefore we define capacities $c_e$ on the edges of the corresponding flow network. 

\paragraph{Envy-free divisions as fixed points.} \citet{SICOMP:Filos-RatsikasH2023} construct a function $F: D \rightarrow D$ and recover envy-free divisions from its fixed points. In particular the domain of $F$ will be $D=\Delta^{n-1} \times \left(\Delta^{n-1}\right)^n \times ([0,1]^n)^n$. In $D$:
\begin{itemize}
    \item[-] a point $\mathbf{x} \in \Delta^{n-1}$ will be the division of the cake,
    \item[-] a point $\mathbf{c}_{i} \in \Delta^{n-1}$ represents the capacities of the edges from agent $i \in N$,
    \item[-] a point $\mathbf{y}_i$ in $[0,1]^n$ represents the flow along the edges from agent $i \in N$.
\end{itemize}
In other words, the input of $F$ will be a vector $(\mathbf{x}, \mathbf{c}_1,\ldots,\mathbf{c}_n,\mathbf{y}_1,\ldots,\mathbf{y}_n)$ and let $(\mathbf{x}^{*}, \mathbf{c}_1^{*},\ldots,\mathbf{c}_n^{*},\mathbf{y}_1^{*},\ldots,\mathbf{y}_n^{*})$ denote the output. Let $r_k = \max\{0,1-\sum_{i=1}^{n}y_{ik}\}$ denote the \emph{flow excess} incoming into piece $k$. The division $\mathbf{x^*}$ will be computed via the following equations:
\begin{equation}\label{eq:fixp-cake-cutting}
x_j^* = \frac{x_j + r_j}{1+\sum_{k=1}^n r_k}, \text{ for all pieces } j
\end{equation}
The capacities $\mathbf{c}_i^*$ and the flows $\mathbf{y}_i^*$ will be obtained as the outcomes of the linear programs of \Cref{fig:cake-cutting-P1-and-P2}. $\mathcal{P}_1$ is a set of linear programs (one for each agent $i \in N$), the optimal solutions of which will give the capacities $c_{ij}^*$, and $\mathcal{P}_2$ is a linear program that gives the values $y_{ij}^*$ of the flow. 

\begin{figure}
\centering 
\fbox{
\centering
    \begin{minipage}{.4\textwidth}
        \centering
        \underline{Linear Program $\mathcal{P}_1$}
        \begin{align*}%
	\mbox{maximize } & \quad \sum_{j=1}^n c_{ij} \cdot u_{ij}(\mathbf{x})  \\
	\mbox{subject to }& \quad \sum_{j=1}^n c_{ij} = 1 \\
        & c_{ij} \geq 0, \; \text{ for any piece } j 
    \end{align*}
    \end{minipage}%
    \hfill\vline\hfill
    \begin{minipage}{0.6\textwidth}
       \centering
        \underline{Linear Program $\mathcal{P}_2$}
        \begin{align*}%
	\mbox{maximize } & \quad \sum_{1 \leq i,j \leq n} y_{ij}  \\
	\mbox{subject to }& \quad 0 \leq y_{ij} \leq c_{ij} + \frac{1}{n^3}, \text{ for any agent }i \text{ and piece } j \\ 
        & \quad \sum_{i=1}^n y_{ij} \leq  1 \text{ for any piece } j\\
        & \quad \sum_{i=1}^n y_{ij} \geq 1 \; \text{ for all agents } i \in N 
        \end{align*}
    \end{minipage}
    }
\caption{The linear program $\mathcal{P}_1$ used for capacities $\mathbf{c}_i$ (left) and the linear program $\mathcal{P}_2$ used for the flow $\mathbf{y}$.}
\label{fig:cake-cutting-P1-and-P2}
\end{figure}

\begin{remark}
The $1/n^3$ term in the constraints of $\mathcal{P}_2$ ensures that there exists a \emph{strictly feasible} point (i.e., a point $\hat{\mathbf{z}}$ such that the general constraint $\mathbf{C} \cdot \mathbf{z} \leq \mathbf{d}$ is satisfied with strict inequality, i.e., $\mathbf{C} \cdot \hat{\mathbf{z}} < \mathbf{d}$, and hence satisfies the Slater condition \citep{Slater50}). This is required for the OPT-gate for FIXP of \citet{SICOMP:Filos-RatsikasH2023} but is not required for our \linoptgate. In particular, we could have the constraints be $0 \leq y_{ij} \leq c_{ij}$ instead, and in fact that would even slightly simplify the proofs. However, we elected to keep linear program $\mathcal{P}_2$ with the term $1/n^3$ in the constraints, for two reasons. First, it allows us to directly compare with the proof of \citep{SICOMP:Filos-RatsikasH2023}. Secondly, it allows us to refer to linear program $\mathcal{P}_2$ when we prove the FIXP-membership for the rental harmony problem in \Cref{sec:rental-harmony}.
\end{remark}

\noindent \citet{SICOMP:Filos-RatsikasH2023} state and prove the following simple lemma. 

\begin{lemma}[\citep{SICOMP:Filos-RatsikasH2023}]\label{lem:fixp-cake-cutting-lemma-1}
If the total flow $\sum_{i \in N} \sum_j y_{ij}^* = n$, then $\mathbf{x}^*$ is an envy-free division. 
\end{lemma}

\begin{proof}[Proof Sketch]
We sketch the proof here, and refer the reader to \citep{SICOMP:Filos-RatsikasH2023} for the (only slightly longer and more detailed) complete proof. By Hall's Theorem \citep{hall1935representatives}, it follows that unless $\mathbf{x}^*$ is an envy-free division indeed, there will be some set of agents $A \subseteq N$ such that $|N(A)| < |A|$, where $N(A)$ is the set of pieces preferred by agents in $A$. Linear program $\mathcal{P}_1$ dictates that in an optimal solution, $c_{ij}^* > 0$ only for pieces that agent $i$ prefers, i.e., $c_{ij}^* = 0$ for any $i \in A$ and $j \neq N(A)$. Using this in linear program $\mathcal{P}_2$, we get that $y_{ij}^* \in [0, 1/n^3]$, which allows the total flow out of $A$ to be bounded by a quantity strictly smaller than $|A|$ and consequently, the total flow in the network by a quantity strictly smaller than $n$, obtaining a contradiction.
\end{proof}

\noindent \Cref{lem:fixp-cake-cutting-lemma-1} reduces the problem of finding an envy-free division to that of finding a flow of total value $n$. In turn, by the definition of $r_j$, this is equivalent to stating that $r_j =0$ for all pieces $j$. \citet{SICOMP:Filos-RatsikasH2023} use the definition of $x_j^*$ to establish that in a fixed point of $F$, this reduces to proving the statement of \Cref{lem:fixp-cake-cutting-lemma-2} below. Indeed, the definition of $x_j^*$ establishes that in a fixed point (where $x_j^* = x_j)$, it must be the case that
\begin{equation*}\label{eq:fixp-cake-cutting-2}
x_j \cdot \sum_{k=1}^n r_k = r_j, \text{ for all pieces }j,
\end{equation*}
which implies that
\[x_j >0 \text{ and } r_j = 0 \Rightarrow \sum_{k=1}^n r_k = 0, \]
as desired.

\begin{lemma}[\citep{SICOMP:Filos-RatsikasH2023}]\label{lem:fixp-cake-cutting-lemma-2}
There exists some piece $j$ with $r_j^*=0$ and $x_j^* > 0$.
\end{lemma}

\begin{proof}
Since $\mathbf{x}$ is a valid division, there exists some $\ell$ such that $x_\ell^* > 0$. If $r_\ell^* =0$ we are done, so assume that $r_\ell^* >0$. By the definition of $r_\ell^*$, that means that piece $\ell$ has positive residual capacity, which implies that the total flow in the network is less than $n$. Hence, there exists some agent $i \in N$, such that the total flow out of $i$ is less than $1$, i.e., $\sum_{k=1}^n y_{ik}^* < 1$. The next step is to establish that there exists some piece $\ell' \in N(i)$, where $N(i)$ is the set of preferred pieces for agent $i$, with $r_{\ell'}^* = 0$, which will conclude the proof, since $x_{\ell'}>0$ follows from the hungriness condition. 
Assume by contradiction that for all $k \in N(i)$, it holds that $r_k^* >0$. Note that $\sum_{k=1}^N y_{ik}^* < 0$ as established before, and that $\sum_{k \in N(i)}c_{ik}^*=1$, which means that it is possible to send more flow from $i$ to $N(i)$. Since each piece in $N(i)$ is preferred by agent $i$, this contradicts the optimality of $\mathbf{y}^*$ as a solution to linear program $\mathcal{P}_2$. 
\end{proof}

\noindent We now move on to the proof of \Cref{thm:cake-cutting}. The proof mimics that presented above, except for one crucial part: we can no longer use \Cref{eq:fixp-cake-cutting} to retrieve the envy-free division as a fixed point of the function $F$. This is because \citet{SICOMP:Filos-RatsikasH2023} only need to argue that $F$ can be represented by an arithmetic circuit that uses OPT-gates for FIXP, whereas we need to argue that it is possible to represent $F$ by a \emph{\linear} arithmetic circuit using \linoptgates. In particular, we cannot divide by $1+\sum_{k=1}^n r_k$ in a \linear arithmetic circuit. \medskip

\noindent It turns out that there is an easy ``fix'' for this. Consider the following linear program $\mathcal{P}_0$.

\begin{center}\underline{Linear Program $\mathcal{P}_0$}\end{center}
\begin{equation*}
\begin{aligned}
\mbox{maximize}\quad & \sum_{k=1}^n r_k\cdot x_k\\
\mbox{subject to}\quad & \sum_{k=1}^n x_k = 1\\
& x_k \geq 0 \text{ for any piece } k
\end{aligned}
\end{equation*}

\noindent The division $\mathbf{x}^*$ will be obtained as the optimal solution to a set of linear programs of the form $\mathcal{P}_0$, one for each piece. We argue that a fixed point of the function $F$ is an envy-free division below. The PPAD-membership of the problem then follows by restricting the functions $u_{ij}(\mathbf{x})$ to assume an appropriate form that allows $F$ to be computable by a \linear arithmetic circuit.

\begin{proof}[Proof of \Cref{thm:cake-cutting}]
Consider a fixed point of the function $F$ and let $\mathbf{x}^*$ be the corresponding division, which is an optimal solution of linear program $\mathcal{P}_0$ by design. Note that in such an optimal solution, positive weight is only placed on variables $x_\ell$ for which $r_\ell$ is maximum, i.e., $x_\ell^* > 0 \Rightarrow r_\ell^* = \max_k r_k^*$. Hence, if there exists some piece $j$ such that $x_j^* > 0$ and $r_j^* = 0$, that means that $\max_k r_k^* = 0$. Since $r_k^* \geq 0$ for all $k$, this is equivalent to $r_k^*= 0$ for all $k$. Thus, we have established exactly what \Cref{eq:fixp-cake-cutting} does without dividing in the circuit. 

$r_k^*$ can obviously be computed by a \linear arithmetic circuit. Linear programs $\mathcal{P}_0$, $\mathcal{P}_1$ and $\mathcal{P}_2$ use only linear constraints with the \circparams appearing only on the right-hand side of the constraints. The \circparams do not appear in the objective function of $\mathcal{P}_2$, whereas the objective function of $\mathcal{P}_0$ has a linear subgradient. For $\mathcal{P}_1$, the subgradient of the objective function is a linear function in $u_{ij}(\mathbf{x})$, and since $u_{ij}(\mathbf{x})$ is given by a \pseudog, so is the subgradient of the objective function. This establishes that $F$ can be represented by a \linear arithmetic circuit, which concludes the proof.  
\end{proof}

\subsection{Rental Harmony} \label{sec:rental-harmony}

We define the \emph{rental harmony} problem, notably studied by \citet{AMM:Su1999}. 

\begin{definition}[Rental harmony]\label{def:rental-harmony}
In the rental harmony problem, the total rent of a house has to be divided among the $n$ rooms of the house, in a way that makes it possible to assign the rooms to $n$ tenants in an envy-free manner, i.e., no tenant would prefer to have someone else's room (with the respective rent) to their own. Formally, we denote the rent by the unit interval $[0,1]$. Each tenant has a valuation function $u_{ij}:\Delta^{n-1}\rightarrow \mathbb{R}_{\geq 0}$, for each room $j$, assigning a real number to each division of the rent. Given a division $\mathbf{x}$, we will say that tenant $i$ \emph{prefers} the $j$-th room if $u_{ij}(\mathbf{x}) \geq u_{ij'}(\mathbf{x})$ for any room $j'$. A division $\mathbf{x}$ is envy-free if there exists a permutation $\pi$ of $\{1,\ldots, n\}$ such that for every $i \in \{1,\ldots,n\}$, tenant $i$ prefers room $\pi(i)$.
\end{definition}

\paragraph{Known results for existence.} The existence of a solution to the rental harmony problem was proven by \citet{AMM:Su1999}, via an interesting adaptation of the idea of Simmons for proving the existence of envy-free cake cutting solutions, one which employs a ``dual Sperner labelling''. Similarly to the proofs for cake-cutting, the proof also appeals to limits of approximate solutions. To the best of our knowledge, computational complexity results about the general version of the problem did not exist before our work. \medskip

\noindent One can see the similarities between \Cref{def:rental-harmony} above and \Cref{def:ef-cake-cutting} of \Cref{sec:ef-cake-cutting}. In fact, those are even more clear if we consider the following problem, which we refer to as \emph{envy-free chore division}.

\begin{definition}[(Contiguous) envy-free chore division]\label{def:ef-chore-division}
Let the interval $[0,1]$ be called a \emph{chore}, which is to be divided into $n$ subintervals (pieces) using $n-1$ cuts. Let $\mathbf{x} \in \Delta^{n-1}$ denote a division of the chore, i.e., each division is a point in the simplex, and let $x_j$ be the $j$'th coordinate (or ``the $j$'th piece''). There is a set $N$ of $n$ agents, and for each piece $j$, each agent $i \in N$ has a \emph{valuation function} $u_{ij}: \Delta^{n-1} \rightarrow \mathbb{R}_{\geq 0}$, assigning a real number to a division of the cake. Given a division $\mathbf{x}$, we will say that agent $i \in N$ \emph{prefers} the $j$-th piece if $u_{ij}(\mathbf{x}) \geq u_{ij'}(\mathbf{x})$ for any piece $j'$. A division $\mathbf{x}$ is envy-free if there exists a permutation $\pi$ of $\{1,\ldots,n\}$ such that for every $i \in N$, agent $i$ prefers piece $\pi(i)$. 
\end{definition}

\noindent Looking at \Cref{def:ef-chore-division} above, one can observe two things:
\begin{enumerate}[label=(\alph*)]
    \item \label{enum:ef-1} It defines the same problem as the one in \Cref{def:rental-harmony}, by simply substituting the term ``room'' with ``piece'' (and the term ``tenant'' with ``agent'').
    \item \label{enum:ef-2} It defines the same problem as the one in \Cref{def:ef-cake-cutting} in \Cref{sec:ef-cake-cutting}, by simply substituting the term ``chore'' with ``cake''.   
\end{enumerate}

\noindent Observation~\ref{enum:ef-1} is because we are using a very general definition of envy-free chore-division (as we used a very general definition of envy-free cake-cutting in \Cref{sec:ef-cake-cutting}), one in which there is a different valuation function for each piece. In that generality, there is no difference between a piece and a room. For restricted valuation functions, the two problems could be different. We remark again that for proving membership results, considering these very general versions makes the results stronger.

Now looking at Observation~\ref{enum:ef-2}, it seems as if the envy-free chore division problem and the envy-free cake cutting problem are the same. Is that really the case? The difference lies in the \emph{sufficiency condition}. Recall that for cake-cutting, the sufficiency condition was hungriness, i.e., that agents always prefer a non-empty piece to any empty piece. For chore division/rental harmony, we will have the exact opposite.

\paragraph{Sufficiency condition: Miserly agents.} We will say that an agent $i$ is a \emph{miser}, if she prefers an empty piece of the chore (respectively a room with zero rent) to a non-empty piece (respectively a room with non-zero rent). An instance of the chore division/rental harmony problem satisfies the miserly agents condition if all of the agents are misers.\medskip

\noindent Given the above, in our setting envy-free chore division (\Cref{def:ef-chore-division}) and rental harmony (\Cref{def:rental-harmony}) are equivalent, so we will use the terminology of envy-free chore division in the proof that we will develop, mainly for consistency with \Cref{sec:ef-cake-cutting}. Our main theorem of the section is the following, which is reminiscent of \Cref{thm:cake-cutting} for envy-free cake cutting.

\begin{theorem}\label{thm:chore-division}
Computing an envy-free division of the chore (i.e., a rental harmony partition) when the agents' valuation functions are given by \linear arithmetic circuits in PPAD.
\end{theorem}

\noindent As in the case of \Cref{sec:ef-cake-cutting}, using the machinery developed in \Cref{sec:implicit}, we can also capture simple interesting cases where the inputs are given as the integrals of the density functions rather than the utilities themselves. \medskip

\noindent Before we proceed with the proof, we make the following remark.
\begin{remark}
One may be inclined to believe that there would be an easy reduction from envy-free chore division to envy-free cake cutting, e.g., by setting $u^\text{cake}_{ij}(\mathbf{x}) = - u^\text{chore}_{ij}(\mathbf{x})$ for all agents $i$ and pieces $j$. Indeed, this would turn a hungry agent into a miser, thus establishing the correct sufficiency condition. However, letting $\mathbf{x}^{\text{cake}} = \mathbf{x}^{\text{chore}}$ would not result in an envy-free chore division, as shown by the simple following example with $2$ agents. Let $x^\text{cake}_1=[0,1/2]$ and $x^{\text{cake}}_2=(1/2,1]$ and let $u^{\text{cake}}_{11}(x^{\text{cake}}_1,x^{\text{cake}}_2) = 1-\epsilon$ and $u^{\text{cake}}_{12}(x^{\text{cake}}_1,x^{\text{cake}}_2) = \epsilon$ for some $\epsilon < 1$. Let agent 2 be indifferent between any allocation of the cake. We have that $u^{\text{chore}}_{11}(x^{\text{chore}}_1,x^{\text{chore}}_2 )= -u^{\text{cake}}_{11}(x^{\text{cake}}_1,x^{\text{cake}}_2) = -1+\epsilon$ and $u^{\text{chore}}_{12}(x^{\text{chore}}_1,x^{\text{chore}}_2) = -u^{\text{cake}}_{12}(x^{\text{cake}}_1,x^{\text{cake}}_2) =-\epsilon$. Agent 1 obviously prefers piece $2$ rather than piece $1$, and this is not an envy-free chore division. 

In this example it is easy to remedy that, by labelling the pieces oppositely, i.e., by having $x_1^\text{cake} = x_2^\text{chore}$ and vice-versa. What happens however in more complex scenarios, with more agents and different valuations? Finding the appropriate division would essentially end up constructing the dual Sperner labelling used by \citet{AMM:Su1999}. Our results in this section provide an alternative way of finding an envy-free chore division, which is very much in line with the approach that we used in \Cref{sec:ef-cake-cutting}. In other words, our technique based on the \linoptgate provides a \emph{unified proof of existence} of solutions for both envy-free cake cutting and rental harmony.
\end{remark}

\subsection{The proof of \crtCref{thm:chore-division}}\label{sec:chore-division-proof}

Exactly in the same manner as in \Cref{sec:ef-cake-cutting}, we will consider a bipartite graph with agents $i$ on one side and pieces $j$ on the other side, with edges $(i,j)$ indicating that agent $i$ prefers pieces $j$. An envy-free chore division will correspond to a perfect matching in the graph, which we will again compute via a maximum flow argument in a fixed point.

Again, we construct a function $F: D \rightarrow D$ and recover envy-free chore divisions from its fixed points. In particular the domain of $F$ will be $D=\Delta^{n-1} \times \left(\Delta^{n-1}\right)^n \times ([0,1]^n)^n$. In $D$:
\begin{itemize}
    \item[-] a point $\mathbf{x} \in \Delta^{n-1}$ will be the division of the chore,
    \item[-] a point $\mathbf{c}_{i} \in \Delta^{n-1}$ represents the capacities of the edges from agent $i \in N$,
    \item[-] a point $\mathbf{y}_i$ in $[0,1]^n$ represents the flow along the edges from agent $i \in N$.
\end{itemize}
In other words, the input of $F$ will be a vector $(\mathbf{x}, \mathbf{c}_1,\ldots,\mathbf{c}_n,\mathbf{y}_1,\ldots,\mathbf{y}_n)$ and let $(\mathbf{x}^{*}, \mathbf{c}_1^{*},\ldots,\mathbf{c}_n^{*},\mathbf{y}_1^{*},\ldots,\mathbf{y}_n^{*})$ denote the output. As in \Cref{sec:ef-cake-cutting}, the capacities and the flow values will be obtained as optimal solutions to the linear programs $\mathcal{P}_1$ and $\mathcal{P}_2$ of \Cref{fig:cake-cutting-P1-and-P2} respectively. 

Again, we will use $r_k = \max\{0,1-\sum_{i=1}^n y_{ik}\}$ to denote the flow excess incoming into piece $k$. \Cref{lem:fixp-cake-cutting-lemma-1} holds verbatim for our setting here as well, which establishes that it suffices to argue that in our fixed point, $r_j^*  = 0$ for all pieces $j$. Now, the chore division $\mathbf{x}^*$ will be an optimal solution to the following linear program, which is a straightforward adaptation of linear program $\mathcal{P}_0$ presented in \Cref{sec:ef-cake-cutting}.

\begin{center}\underline{Linear Program $\mathcal{P}^c_0$}\end{center}
\begin{equation*}
\begin{aligned}
\mbox{minimize}\quad & \sum_{k=1}^n r_k\cdot x_k\\
\mbox{subject to}\quad & \sum_{k=1}^n x_k = 1\\
& x_k \geq 0 \text{ for any piece } k
\end{aligned}
\end{equation*}

\noindent The following lemma is analogous to \Cref{lem:fixp-cake-cutting-lemma-2}, although the proof is somewhat different.

\begin{lemma}\label{lem:chore-division-lemma-2}
In a fixed point of $F$, $r_j^{*}=0$ for all pieces $j$.
\end{lemma}

\begin{proof}
Consider a fixed point $(\mathbf{x}^{*}, \mathbf{c}_1^{*},\ldots,\mathbf{c}_n^{*},\mathbf{y}_1^{*},\ldots,\mathbf{y}_n^{*})$ of $F$. For any piece $j$, $x_j^* > 0 \Rightarrow r_j^* = \min_{k}r_k^*$, which follows from the fact that $x_j^*$ is an optimal solution to linear program $\mathcal{P}^c_0$, and hence only places positive weight to pieces with minimum $r_k^*$. 

We will first argue that $x_{j}^* > 0$ for all pieces $j$. Assume first that $x_{j_0}^* = 0$ for some piece $j_0$ and consider any piece $j$ such that $x_j^* > 0$; note that at least one such $j$ must exist by the fact that $\mathbf{x}^*$ is a valid division of the chore. By the miserly agents sufficiency condition, we have that any agent $i \in N$ prefers $j_0$ to $j$, i.e., $u_{ij_{0}}(\mathbf{x}^*) > u_{ij}(\mathbf{x}^*)$. Since $c_{ij}^*$ is an optimal solution to linear program $\mathcal{P}_1$, it follows that $c_{ij}^*=0$. In turn, by the corresponding constraint of linear program $\mathcal{P}_2$, it follows that $y_{ij}^* \leq 1/n^3$, and that $\sum_{i=1}^n y_{ij}^* \leq 1/n^2$. From this, we conclude that $r_j^* \geq 1-1/n^2$. 

Since $x_j^* >0$, by the discussion above this implies that $\min_k r_k^* \geq 1-1/n^2$, which of course implies that $r_\ell^* \geq 1-1/n^2$ for all pieces $\ell$. It follows that $\sum_{i=1}^n y_{i\ell}^* \leq 1/n^2$ for all pieces $\ell$, and by summing over $\ell$ we obtain that $\sum_{\ell=1}^n \sum_{i=1}^n y_{ij}^* < 1/n$. Since $\mathbf{c}_i$ is a feasible solution to linear program $\mathcal{P}_1$, there exists a set of pieces $A_\ell$ with $\sum_{i \in A_\ell}c_{ij}^*=1$ for at least one agent $i \in N$ (and those are precisely the pieces with $x_\ell^* = 0$). That implies that there is a solution to linear program $\mathcal{P}_2$ for which $\sum_{\ell=1}^n\sum_{i=1}^n y_{ij}\geq 1$. This contradicts the fact that $\mathbf{y}^*$ is an optimal solution to linear program $\mathcal{P}_2$.

Now since $x_j^*> 0$ for all pieces $j$, by the discussion in the first paragraph of the proof it follows that $r_j^* = r_{j'}^*$ for all pieces $j$ and $j'$. Let $r^*$ denote the value of $r_j$ for any $j$. We will argue that $r^*=0$, which will conclude the proof. Assume by contradiction that $r^* > 0$, or, equivalently, that $\sum_{i=1}^n y_{ij}^* < 1$ for piece $j$. This implies that in linear program $\mathcal{P}_2$ the corresponding constraint for piece $j$ is not be tight. In turn, that means that there is a feasible flow $\mathbf{y}$ of value larger than that of $\mathbf{y}^*$, contradicting the fact that $\mathbf{y}^*$ is an optimal solution to linear program $\mathcal{P}_2$.
\end{proof}

\noindent We can now conclude the proof of \Cref{thm:chore-division}.

\begin{proof}[Proof of \Cref{thm:chore-division}]
\Cref{lem:chore-division-lemma-2} establishes that a fixed point of $F$ corresponds to an envy-free chore division. It remains to show that $F$ can be computed by a \linear arithmetic circuit containing \linoptgates. This is virtually identical to the corresponding argument in the proof of \Cref{thm:cake-cutting}. 
\end{proof}

\subsubsection{FIXP-membership for (unrestricted) rental harmony}

In the envy-free cake cutting application, \Cref{thm:cake-cutting} provides a crisper complexity result than that of \Cref{thm:FIXPpaper-cake-cutting}, for the case when the valuation functions are given by \pseudogs, which admits envy-free divisions in rational numbers. For rental harmony/chore division, the situation is similar, except, as we mentioned earlier, a FIXP-membership result (or any other kind of complexity result for that matter) was not known before our work. It is almost straightforward to generalize our proof developed in \Cref{sec:chore-division-proof} to obtain a theorem analogous to \Cref{thm:FIXPpaper-cake-cutting}, for the case of fairly dividing a chore as well, when the valuation functions are represented by general arithmetic circuits (not necessarily \linear).

\begin{theorem}\label{thm:rental-harmony-fixp}
Computing an envy-free division of the chore is in FIXP.
\end{theorem}

\begin{proof}
Construct a function $F$ exactly as in \Cref{sec:chore-division-proof} and consider a fixed point $(\mathbf{x}^{*}, \mathbf{c}_1^{*},\ldots,\mathbf{c}_n^{*},\mathbf{y}_1^{*},\ldots,\mathbf{y}_n^{*})$ of $F$. The arguments on why $\mathbf{x^*}$ constitutes an envy-free division are identical to those developed in the previous section. What we need to establish is that $F$ can be computed by an arithmetic circuit that contains instances of the OPT-gate for FIXP developed by \citet{SICOMP:Filos-RatsikasH2023}. In reality, this has practically already been established in \citep{SICOMP:Filos-RatsikasH2023}: linear programs $\mathcal{P}_1$ and $\mathcal{P}_2$ are identical to those used in \citep{SICOMP:Filos-RatsikasH2023}, and linear program $\mathcal{P}_0$ is also amenable to the use of their OPT-gate for FIXP. We note that it is possible to multiply two inputs in arithmetic circuits, and hence the valuations $u_{ij}(\mathbf{x})$ in the objective function of linear program $\mathcal{P}_1$ are not restricted as in \Cref{thm:chore-division}. 
\end{proof}

\noindent We conclude the section with an example showing that if one goes beyond linear valuations, there are cases where all envy-free divisions are irrational, and hence the FIXP-membership result of \Cref{thm:rental-harmony-fixp} (rather than a PPAD-membership result) is justified.

\begin{example}[Example where all envy-free chore divisions are irrational]\label{ex:irrational-rental-harmony}
Consider an instance of the rental harmony/envy-free chore division problem with two agents. A division of the chore is represented by a point $(x,1-x)\in\Delta^1$. The common utility functions of the two agents are given by $u_1(x,1-x)=2(1-x)^2$ and $u_2(x,1-x)=x^2$. Note that if $x=0$ the agents will prefer piece 1, and if $x=1$ the agents will prefer piece 2. Therefore, the agents are indeed misers. For $(x,1-x)$ to be an envy-free chore division, then it must hold that $2(1-x)^2=x^2$. As $x\in [0,1],$ this implies that $x=2-\sqrt{2}$, so the only solution is irrational.
\end{example}

\section{Conclusion and Future Work}\label{sec:conclusion}

In this work, we developed the \linoptgate, a powerful general-purpose tool for showing the PPAD-membership of problems that have \emph{exact rational solutions}. We demonstrated its strength by applying it to a plethora of domains related to game theory, competitive markets, auto-bidding auctions and fair division. For those applications, we obtain new results and generalizations of the state-of-the-art complexity results, as well as \emph{significant} simplifications in terms of the proof techniques. 

There are some interesting open directions related to our work, mainly in the domain of competitive markets. First, it will be very interesting to see whether one could extend our machinery in \Cref{sec:implicit} to also capture markets with SSPLC production sets; we discussed the challenges of this task in \Cref{rem:ssplc-production-challenge}. Similarly, it would be interesting to try to design a class of succinct utility and production functions that subsumes all the known classes for which rational solutions are known to exist, i.e., one that would generalize the Leontief-free class of functions. Finally, one application that we did not study in our work is that of \emph{competitive markets for mixed manna}, where there are goods but also bads to be allocated to the consumers. \citet{chaudhury2022complementary} studied these markets for SPLC utility functions and provided a PPAD-membership result. There does not seem to be any technical obstacle to applying our technique on those markets as well (and also possibly incorporating production functions as well); the details are still to be worked out.

Looking at the big picture, our \linoptgate complements and refines the \fixpoptgate of \citet{SICOMP:Filos-RatsikasH2023} as a tool to proving computational membership of \emph{exact} problems in the appropriate complexity classes. One interesting question is whether one could hope to develop a similar gate for \emph{approximate} problems, i.e., an optimization gate that could be used in a very similar manner to those other two gates to establish PPAD-membership of more general problems (with irrational solutions), for their approximate versions. This certainly introduces new challenges and intriguing questions. One would have to work with approximate rather than exact fixed points. How should the gate be constructed to be useful in this regard? Should the gate work approximately as well? Applications domains like competitive markets, where the approximation in the competitive equilibrium notion comes from relaxing the clearing condition rather than the bundle optimality of the consumers, seem to suggest otherwise.

\subsubsection*{Acknowledgments}
Kristoffer Arnsfelt Hansen and Kasper Høgh were supported by the Independent Research Fund Denmark under grant no.~9040-00433B. Alexandros Hollender was supported by the Swiss State Secretariat for Education, Research and Innovation (SERI) under contract number MB22.00026.

\bibliographystyle{plainnat}
\bibliography{PPAD-OPT}

\end{document}